\documentclass[authorcolumns,numberwithinsect,cleveref,thm-restate]{no-lipics-v2022}
\usepackage{amssymb}
\usepackage{mathtools}
\usepackage{bm}
\usepackage{todonotes}
\usepackage{stmaryrd}
\usepackage[draft]{fixme}
\usepackage{booktabs}
\usepackage{comment}
\usepackage{ltablex}

\usepackage[T1]{fontenc}
\usepackage[utf8]{inputenc}

\usepackage[numbers,square,comma,longnamesfirst]{natbib}

\usetikzlibrary{calc}

\SetKwComment{Comment}{$\triangleright$\ \rm}{}
\fxsetup{mode=multiuser,theme=color, layout=inline}
\FXRegisterAuthor{pw}{apw}{\color{purple} {\bf Philip}}
\FXRegisterAuthor{pc}{apc}{\color{red} {\bf  Panos}}
\FXRegisterAuthor{tk}{atk}{\color{blue} {\bf Tomasz}}

\usetikzlibrary{arrows,arrows.meta,shapes.misc, decorations.pathreplacing,decorations.pathmorphing,calligraphy, patterns.meta}
\tikzset{dotmark/.style={circle,fill,inner sep=1.5pt}}
\tikzset{emptymark/.style={circle,draw,fill=white,inner sep=1.5pt}}
\tikzset{crossmark/.style={thick,inner sep=1.5pt}}

\newcommand{\A}{\mathcal{A}}
\newcommand{\B}{\mathcal{B}}

\newcommand{\Oh}{\mathcal{O}}
\newcommand{\tOh}{{\widetilde{\Oh}}}
\newcommand{\cOtilde}{\tOh}

\newcommand{\eps}{\varepsilon}
\newcommand{\lcp}{\mathsf{lcp}}

\newcommand{\sm}{\setminus}

\newcommand{\cO}{\mathcal{O}}

\newcommand{\ceil}[1]{\left\lceil #1 \right\rceil}
\newcommand{\floor}[1]{\left\lfloor #1 \right\rfloor}
\newcommand{\per}{\operatorname{per}}
\newcommand{\rot}{\operatorname{rot}}

\newcommand{\OccEx}{\mathrm{Occ}}

\newcommand{\OccE}{\mathrm{Occ}^E}
\newcommand{\OccD}{\mathrm{Occ}^D}

\newcommand{\ed}{\delta_E}
\newcommand{\DD}{\delta_D}

\newcommand{\G}{\mathcal{G}}

\newcommand{\X}{\mathcal{X}}
\newcommand{\Y}{\mathcal{Y}}
\renewcommand{\S}{\mathcal{S}}

\newcommand{\ktot}{\kappa}
\newcommand{\Z}{\mathbb{Z}}
\newcommand{\Zz}{\mathbb{Z}_{\ge 0}}
\newcommand{\Zp}{\mathbb{Z}_{> 0}}

\def\sp#1{\textsf{Special}(#1)}
\def\spb#1{\textsf{Special}_{\beta\varphi}(#1)}

\newcommand{\Comp}{\textsf{Trim}}
\newcommand{\Comps}{\textsf{Tm}}
\newcommand{\comp}{\alpha}
\newcommand{\rred}{\textsf{Red}}

\newcommand{\update}{\textsf{update}}
\newcommand{\rpl}{\textsf{Plain}}

\newcommand\markf{\ensuremath \mathrm{mk}}
\newcommand\tradeoff{\ensuremath \eta}
\newcommand\ktotm{\ensuremath \hat{\ktot}}
\newcommand\lpref{\ensuremath \partial_P}
\newcommand{\Hv}{\mathbb{H}}
\newcommand{\light}{\mathbb{L}}

\newcommand{\LP}{\mathcal{L}^P}
\newcommand{\LT}{\mathcal{L}^T}
\newcommand{\res}[2]{#1_{#2}}

\newcommand{\gen}{\textsf{gen}}
\newcommand{\val}{\textsf{val}}
\newcommand{\brkp}{\mathsf{B}}

\newcommand{\sub}{\subseteq}

\newcommand{\edl}[2]{{\delta_E}(#1,{}^*\!#2^*)}
\newcommand{\eds}[2]{{\delta_E}(#1,{}^*\!#2)}
\newcommand{\edp}[2]{{\delta_E}(#1,#2^*)}

\def\twoheadleadsto{\tikz[baseline=(a.base)]{\draw[%
    decorate,decoration={zigzag,segment length=4, amplitude=.9},%
    ] (0,0) -- (.25, 0);%
    \draw[%
    -{Classical TikZ Rightarrow}.{Classical TikZ Rightarrow},%
    ] (.25, 0) -- (.4, 0);%
    \node (a) at (.4/2,-.03) {\phantom{\(\leadsto\)}};%
}}

\newcommand{\onto}{\twoheadleadsto}

\def\modelname{{\tt PILLAR}\xspace}
\def\lceOp#1#2{{\tt LCP}(#1, #2)}
\def\lcbOp#1#2{{\tt LCP}^R(#1, #2)}
\def\ipmOp#1#2{{\tt IPM}(#1, #2)}
\def\perOp#1{{\tt Period}(#1)}

\def\accOp#1#2{#1\position{#2}}
\def\ipmOpName{{\tt IPM}\xspace}

\def\accOpName{{\tt Access}\xspace}
\def\extractOpName{{\tt Extract}\xspace}
\def\lenOpName{{\tt Length}\xspace}

\def\lceOpName{{\tt LCP}\xspace}
\def\lcbOpName{{\tt LCP$^R$}\xspace}

\def\cycEqOp#1#2{{\tt Rotations}(#1,#2)}

\newcommand{\concat}{\ensuremath{\mathtt{concat}}}
\newcommand{\makestring}{\ensuremath{\mathtt{makestring}}}
\newcommand{\splitOp}{\ensuremath{\mathtt{split}}}

\DeclareMathOperator*{\downshift}{\searrow\,}

\def\substr{\ensuremath \preccurlyeq}
\def\fragmentco#1#2{\bm{[}\,#1\,\bm{.\,.}\,#2\,\bm{)}}
\def\fragmentoc#1#2{\bm{(}\,#1\,\bm{.\,.}\,#2\,\bm{]}}
\def\fragmentoo#1#2{\bm{(}\,#1\,\bm{.\,.}\,#2\,\bm{)}}
\def\fragment#1#2{\bm{[}\,#1\,\bm{.\,.}\,#2\,\bm{]}}
\def\position#1{\bm{[}\,#1\,\bm{]}}

\SetKwFunction{verify}{Verify}

\def\problembox#1{%
    \vspace{2mm}%
    \noindent\fbox{%
    \begin{minipage}{.985\textwidth}%
        #1
    \end{minipage}%
    }%
    \vspace{2mm}%
}

\makeatother

\usepackage{pgf}

\def\alphav{128}
\def\betav{8}
\def\deltavN{3}
\def\deltavD{8}
\pgfmathsetmacro{\betavh}{int(\betav/2)}

\def\threehalfs{{}^3{\mskip -4mu/\mskip -3.5mu}_2\,}
\def\onehalf{{}^1{\mskip -4mu/\mskip -3.5mu}_2\,}

\def\Q{\ensuremath Q^{\infty}}

\def\posa{\ensuremath v}
\def\posb{\ensuremath w}

\makeatletter
\renewenvironment{cases}{%
  \matrix@check\cases\env@cases
}{%
  \endarray\right.%
}
\def\env@cases{%
  \let\@ifnextchar\new@ifnextchar
  \left\lbrace
  \def\arraystretch{1.1}%
  \array{@{\;}c@{\quad}l@{}}%
}
\makeatother

\def\mid{\ensuremath :}

\newcommand{\SM}{\textsc{New\-Periodic\-Matches}\xspace}
\newcommand{\PMWE}{\textsc{PM\-with\-Edits}\xspace}

\def\emptyset{\varnothing}

\def\Sb{\S_{\beta}}
\def\Sm{\S_{\mu}}
\def\Sf{\S_{\varphi}}

\def\torn{torsion\xspace}
\DeclareMathOperator{\tor}{tor}
\newcommand{\DPM}{\textsc{Dynamic\-Puzzle\-Matching}\xspace}
\newcommand{\I}{\mathcal{I}}

\SetKwFunction{swaps}{TrimmedPM}

\title{Faster Pattern Matching under Edit Distance}
\subtitle{A Reduction to Dynamic Puzzle Matching and
\vspace{-.5ex}\linebreak The Seaweed Monoid of Permutation Matrices}

\author{Panagiotis Charalampopoulos}{Reichman University, Herzliya, Israel}{pcharalampo@gmail.com}{https://orcid.org/0000-0002-6024-1557}{Partly supported by Israel Science Foundation grant 810/21.}
\author{Tomasz Kociumaka}{University of California, Berkeley, U.S.}{kociumaka@berkeley.edu}{https://orcid.org/0000-0002-2477-1702}{Partly supported by NSF 1652303, 1909046, and HDR TRIPODS 1934846 grants, and an Alfred P. Sloan Fellowship.}
\author{Philip Wellnitz}{Max Planck Institute for Informatics, SIC,
Saarbrücken, Germany}{wellnitz@mpi-inf.mpg.de}{https://orcid.org/0000-0002-6482-8478}{}

\authorrunning{P. Charalampopoulos, T. Kociumaka, and P. Wellnitz}

\Copyright{Panagiotis Charalampopoulos, Tomasz Kociumaka, and Philip Wellnitz}

\nolinenumbers

\begin{document}
\maketitle
\begin{abstract}
We consider the approximate pattern matching problem under the \emph{edit distance}. Given a text $T$ of length $n$, a pattern $P$ of length $m$, and a threshold~$k$, the task is to find the starting positions of all substrings of~$T$ that can be transformed to $P$ with at most $k$ \emph{edits}.
More than 20 years ago, Cole and Hariharan [SODA'98, J.~Comput.'02] gave an $\mathcal{O}(n+k^4 \cdot n/m)$-time algorithm for this classic problem, and this runtime has not been improved since.

Here, we present an algorithm that runs in time $\mathcal{O}(n+k^{3.5} \sqrt{\log m \log k} \cdot n/m)$, thus breaking through this long-standing barrier.
In the case where $n^{1/4+\varepsilon} \leq k \leq n^{2/5-\varepsilon}$ for some arbitrarily small positive constant $\varepsilon$, our algorithm improves over the state-of-the-art by polynomial factors: it is polynomially faster than both the algorithm of Cole and Hariharan and the classic $\mathcal{O}(kn)$-time algorithm of Landau and Vishkin [STOC'86, J.~Algorithms'89].

We observe that the bottleneck case of the alternative $\mathcal{O}(n+k^4 \cdot n/m)$-time algorithm of Charalampopoulos, Kociumaka, and Wellnitz [FOCS'20] is when the text and the pattern are (almost) periodic.
Our new algorithm reduces this case to a new dynamic problem (Dynamic Puzzle Matching), which we solve by building on tools developed by Tiskin [SODA'10, Algorithmica'15] for the so-called seaweed monoid of permutation matrices.
Our algorithm relies only on a small set of primitive operations on strings and thus also applies to the fully-compressed setting (where text and pattern are given as straight-line programs) and to the dynamic setting (where we maintain a collection of strings under creation, splitting, and concatenation), improving over the state of the art.
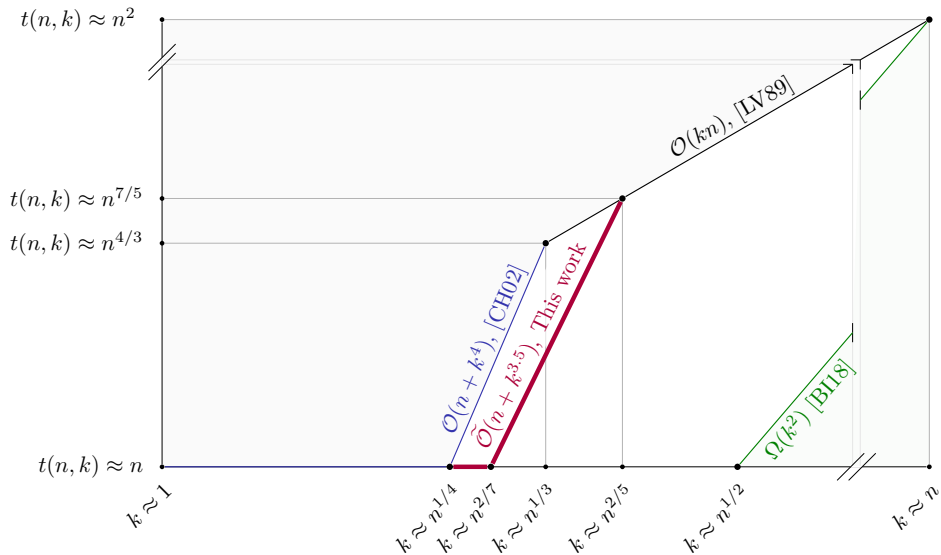
\begin{figure}[htpb!]
\centering
\scalebox{.85}{

\begin{tikzpicture}

    \pgfmathsetmacro{\sx}{6}
    \pgfmathsetmacro{\sy}{3.5}

    \path[fill=green!40!black!2] (\sx*1.5,0) -- (\sx*1.8,\sy*.6) -- (\sx*1.82,\sy*1.64)
    -- (\sx*2,\sy*2) -- (\sx*2,0) -- cycle;
    \path[fill=black!2!white] (0,0) -- (\sx*0.75,0) -- (\sx* 1,\sy*1) -- (\sx*2,\sy*2) --
(0,\sy*2) -- cycle;

    \node[fill,circle, inner sep=.75pt] (o) at (0,0) {};
    \node[fill,circle, inner sep=1pt] (1) at (\sx*0.75,0) {};
    \node[fill,circle, inner sep=1pt] (2) at  (\sx*0.85714285714,0) {};
    \node[fill,circle, inner sep=1pt] (3) at  (\sx*1.5,0) {};
    \node[fill,circle, inner sep=1pt] (4) at  (\sx*1,\sy*1) {};
    \node[fill,circle, inner sep=.75pt] (4a) at  (\sx*1,0) {};
    \node[fill,circle, inner sep=.75pt] (4b) at  (0,\sy*1) {};
    \node[fill,circle, inner sep=1pt] (5) at  (\sx*1.2,\sy*1.2) {};
    \node[fill,circle, inner sep=.75pt] (5a) at  (\sx*1.2,0) {};
    \node[fill,circle, inner sep=.75pt] (5b) at  (0,\sy*1.2) {};
    \node[fill,circle, inner sep=1pt] (6) at  (\sx*2,\sy*2) {};
    \coordinate (6p) at (\sx*3,\sy*3) {};
    \node[fill,circle, inner sep=.75pt] (6a) at  (\sx*2,0) {};
    \coordinate (6ap) at  (\sx*3,0) {};
    \node[fill,circle, inner sep=.75pt] (6b) at  (0,\sy*2) {};
    \coordinate (6bp) at  (0,\sy*3) {};

    \coordinate (62) at (\sx*1.8,\sy*1.8) {};
    \coordinate (622) at (\sx*1.82,\sy*1.82) {};
    \coordinate (92) at (\sx*1.8,\sy*.6) {};
    \coordinate (922) at (\sx*1.82,\sy*1.64) {};

    \coordinate (7a) at  (\sx*1.795,0) {};
    \coordinate (7a2) at  (\sx*1.825,0) {};
    \coordinate (7a') at  (\sx*1.8,0) {};
    \coordinate (7a2') at  (\sx*1.82,0) {};
    \coordinate (7b) at  (0,\sy*1.79) {};
    \coordinate (7b2) at  (0,\sy*1.83) {};
    \coordinate (7b') at  (0,\sy*1.8) {};
    \coordinate (7b2') at  (0,\sy*1.82) {};

    \draw[very thin,black!10] (7a') -- (62) -- (7b');
    \draw[very thin,black!10] (7a2') -- (622) -- (7b2');

    \draw[thin] (62) -- ++(180:.15);
    \draw[thin] (62) -- ++(90:-.15);
    \draw[thin] (622) -- ++(180:.15);
    \draw[thin] (622) -- ++(90:-.15);

    \draw[thin] (92) -- ++(90:.15);
    \draw[thin] (92) -- ++(90:-.15);
    \draw[thin] (922) -- ++(90:.15);
    \draw[thin] (922) -- ++(90:-.15);

    \draw[very thin,black!30] (6a) -- (6) -- (6b);
    \draw[very thin,black!30] (5a) -- (5) -- (5b);
    \draw[very thin,black!30] (4a) -- (4) -- (4b);

    \draw (o) -- (6a);
    \draw (o) -- (6b);
    \draw[white,thick] (7b) -- (7b2);
    \draw[thin] (7b) -- ++(45:.3);
    \draw[thin] (7b) -- ++(45:-.3);
    \draw[thin] (7b2) -- ++(45:.3);
    \draw[thin] (7b2) -- ++(45:-.3);

    \draw[white,thick] (7a) -- (7a2);
    \draw[thin] (7a) -- ++(45:.3);
    \draw[thin] (7a) -- ++(45:-.3);
    \draw[thin] (7a2) -- ++(45:.3);
    \draw[thin] (7a2) -- ++(45:-.3);

    \draw[color=purple!90!black, line width=2pt] (1) -- (2); 
    \draw[color=purple!90!black, line width=2pt] (2) -- (5) node[midway, above, sloped,fill=white,fill
    opacity=.65,scale=.8,outer sep=.5ex] {\phantom{\(\tOh(n + k^{3.5})\), This work}} node[midway, above,
    sloped,inner sep=0.3ex]{
    \(\tOh(n + k^{3.5})\), This work\qquad\mbox{}}; 
    \draw[color=black, line width=.5pt] (622) -- (6) (4) -- (5) -- (62) node[midway, above, sloped,
    inner sep=0.3ex]
    {\(\Oh(kn)\), \cite{LandauV89}}; 
    \draw[color=blue!60!black!80, line width=.5pt] (o) -- (1);
    \draw[color=blue!60!black!80, line width=.5pt] (1) -- (4) node[midway, above,
    sloped,inner sep=0.3ex] {\(\Oh(n +
        k^4)\), \cite{ColeH98}}; 
        \draw[color=green!50!black, line width=.5pt] (922) -- (6) (3) -- (92) node[midway, below,
    sloped,align=center,inner sep=0.5ex,text width=10em]
    {\(\Omega(k^2)\) \cite{bi18}}; 

    \def\slfrac#1#2{\ensuremath{}^{#1}\!/\!{}_{#2}}

    \node[anchor=north east,outer sep=.1em,rotate=45] at (o)  { $k \approx 1$};
    \node[anchor=north east,outer sep=.1em,rotate=45] at (1)  { $k \approx n^{1/4}$};
    \node[anchor=north east,outer sep=.1em,rotate=45] at (2)  { $k \approx n^{2/7}$};
    \node[anchor=north east,outer sep=.1em,rotate=45] at (4a)  { $k \approx n^{1/3}$};
    \node[anchor=north east,outer sep=.1em,rotate=45] at (5a)  { $k \approx n^{2/5}$};
    \node[anchor=north east,outer sep=.1em,rotate=45] at (3)  { $k \approx n^{1/2}$};
    \node[anchor=north east,outer sep=.1em,rotate=45] at (6a)  { $k \approx n$};

    \node[anchor=east,outer sep=.5em] at (o) { $t(n,k) \approx n$};
    \node[anchor=east,outer sep=.5em] at (4b) { $t(n,k) \approx n^{4/3}$};
    \node[anchor=east,outer sep=.5em] at (5b) { $t(n,k) \approx n^{7/5}$};
    \node[anchor=east,outer sep=.5em] at (6b) { $t(n,k) \approx n^2\,\,$};
\end{tikzpicture}

}
\caption{\footnotesize
    The running time $t(n,k)$ of algorithms for the approximate pattern matching
    problem under the edit distance as a function of $k$ for the important special
    case where $m=\Theta(n)$.
    The scale is doubly logarithmic and sub-polynomial factors are hidden; running
    times below \(n\) and above \(n^2\) are not relevant, neither are values of \(k\)
    that lie above \(n\).
    Any point that lies strictly to the bottom-right of the green line segment is unattainable unless SETH fails.
}\label{fig:tradeoffs}
\end{figure}
\vspace*{-6ex}
\end{abstract}
\clearpage
\setcounter{page}{0}%
\tableofcontents
\clearpage
\section{Introduction}

Almost every introductory algorithms textbook covers the \emph{pattern
matching} problem: in a given text~\(T\) of length \(n\), we wish to find all occurrences
of a given pattern \(P\) of length \(m\).
As fundamental as both this problem and its solutions are by today, as apparent are their
limitations: a single surplus or missing character in the pattern (or in a potential
occurrence) results in (potentially all) occurrences being missed.
Hence, a large body of work focuses on \emph{approximate pattern matching},
where we want to identify substrings of the text that are \emph{close} to the pattern.
In particular, in this paper, we consider a classic variant of approximate pattern matching where we allow for up to
\(k\) insertions, deletions, and substitutions (collectively: edits); that is, we consider approximate pattern
matching under the edit distance.

Formally, for two strings \(X\) and \(Y\), their edit distance (also known as
the Levenshtein distance) $\ed(X,Y)$, is the minimum number of insertions, deletions,
and substitutions of single characters required to transform $X$ into $Y$.
Now, in the \emph{pattern matching with edits} problem, for a given text $T$, pattern~$P$,
and an integer threshold $k>0$, the task is to find the starting positions
of all \emph{$k$-error} (or \emph{$k$-edit}) \emph{occurrences} of~$P$ in $T$.
Specifically, we wish to list all positions $v$ in $T$ such that the edit distance between
${T\fragmentco{v}{w}}=T\position{v}T\position{v+1}\cdots T\position{w-1}$ and $P$ is at most $k$
for some position $w$; we write \(\OccE_k(P,T)\) to denote the set of all such positions $v$.

Let us highlight the main prior results for pattern matching with edits;
for a thorough review of other (in particular) early results on pattern matching with edits,
we refer to the extensive survey of Navarro~\cite{N01}.
Back in 1980, Sellers~\cite{S80} demonstrated how the standard dynamic-programming
algorithm for computing $\ed(P,T)$ can be adapted to an $\cO(nm)$-time algorithm for the
pattern matching with edits problem.
Around the same time, Masek and Paterson~\cite{MP80} reduced the running time by a
poly-logarithmic factor using the Four-Russians technique.
Only several year later,
Landau and Vishkin~\cite{LV88} presented an $\cO(n k^2)$-time solution,
which they could then improve to the---by now---classic ``kangaroo jumping''
algorithm that solves this problem in $\Oh(nk)$ time~\cite{LandauV89}.
In search of even faster algorithms,
Sahinalp and Vishkin~\cite{SV96} developed an algorithm that runs in time $\cO(n +
nk^{8+1/3} (\log^* n)^{1/3} / m^{1/3})$---this algorithm was then improved
by Cole and Hariharan~\cite{ColeH98}, who gave an $\cO(n+k^4 n/m)$-time solution,
which is asymptotically faster than the aforementioned Landau--Vishkin algorithm
when $k = o(\sqrt[3]{m})$, and in that setting also the fastest known algorithm even
today.

From a lower-bound perspective, we can benefit from the discovery that the classic quadratic-time
algorithm for computing the edit distance of two strings is essentially optimal:
Backurs and Indyk~\cite{bi18} recently proved that any polynomial-factor improvement would
yield a major breakthrough for the satisfiability problem.
For pattern matching with edits, this means that there is no hope for an
algorithm running in time $\Oh(n + k^{2-\eps}n/m)$ for any constant $\eps>0$:
given an $\Oh(n + k^{2-\eps}n/m)$-time algorithm for pattern matching with edits,
we could compute the edit distance of any two
given strings $X$ and~$Y$ of total length $N$ over an alphabet $\Sigma$ in time $\cO(N^{2-\eps} \log N)$.
Specifically, we pad \(X\) and~\(Y\) to
$P \coloneqq \$^{2N} X \$^{2N}$ and $T \coloneqq \$^{2N} Y \$^{2N}$, where $\$ \notin
\Sigma$.
Now, as $\min_{v,w} \ed(P,T\fragmentco{v}{w})=\ed(P,T)=\ed(X,Y)$,
we can binary search for the smallest value of $k$ such that $\OccE_k(P,T)$ is not empty.

Despite the large gap between the quadratic and bi-quadratic dependency on $k$,
no further advancements have been made to settle the
running time of the pattern matching with edits problem.
In particular, there has not even been any
progress on resolving the 24-year-old conjecture of~Cole and
Hariharan~\cite{ColeH98} that an $\Oh(n + k^3n/m)$-time algorithm \emph{should be
possible}---until now.
We give the first algorithm that improves over the running time achieved by Cole
and Hariharan~\cite{ColeH98}:

\begin{restatable}{mtheorem}{stedalgmain}\label{thm:stedalgmain}
    Given a text $T$ of~length $n$, a pattern $P$ of~length $m$,
    and an integer threshold $k>0$,
    we can compute the set $\OccE_k(P, T)$
    in $\Oh(n + n/m \cdot k^{3.5}\sqrt{\log m \log k})$ time.
    \ifx\stedalgmaint\undefined\lipicsEnd\fi
\end{restatable}

Observe that if \(k\) is
roughly between \(n^{1/4}\) and \(n^{2/7}\), we obtain the first linear-time  algorithm
for the important special case where text and pattern are close in length. Further, we
still obtain polynomial improvements in the running time for values of \(k\) that are
roughly less than \(n^{2/5}\).
Consult \cref{fig:tradeoffs} for a graphical comparison of the running times of
our algorithm with the previous state-of-the-art and the conditional lower bound discussed above.

\paragraph*{The \modelname Model and Faster Algorithms in Other Settings}

Our approach is reasonably general and allows for an easy adaption to different settings
(where the text and the pattern are not given explicitly).
In particular, we follow the approach by
Charalampopoulos, Kociumaka, and Wellnitz~\cite{unified}
and implement the algorithm in the so-called \modelname model.
In that model, one bounds the running times of algorithms in terms of the number of
calls to a small set of very common operations (the \modelname operations) on
strings, such as computing the length of their longest common prefix.
Then, for any setting, an efficient implementation of the \modelname operations
yields a fast algorithm for approximate pattern matching.
For pattern matching with edits,~\cite{unified}
presented an algorithm that runs in $\cO(n/m \cdot k^4)$ time in the \modelname model.
We improve upon their algorithm.

\begin{restatable}{mtheorem}{edalgII}\label{thm:edalgII}
    Given a pattern $P$ of~length $m$, a text $T$ of~length $n$, and an integer
    threshold $k>0$, we can compute a representation of the set $\OccE_k(P,T)$
    as $\cO(n/m \cdot k^3)$ arithmetic progressions with the same difference
    in $\Oh(n/m \cdot k^{3.5} \sqrt{\log m \log k})$ time in the \modelname model.
    \ifx\edalgIt\undefined\lipicsEnd\fi
\end{restatable}

Consistently with~\cite{unified}, we represent the output set $\OccE_k(P,T)$ as $\Oh(k^3)$
disjoint arithmetic progressions with a common difference.
Unless $P$ is almost periodic, though, $\OccE_k(P,T)$ is of size $\Oh(k^2)$,
and we can report $\OccE_k(P,T)$ explicitly;
see~\cite{unified} for a structural characterization of $\OccE_k(P,T)$.


Now, in the standard setting, where the text and the pattern are both given explicitly,
after an $\cO(n)$-time preprocessing, we can perform
each primitive \modelname operation in constant time.
We thus instantly obtain \cref{thm:stedalgmain}.
The same \modelname implementation remains valid in the \emph{internal setting} introduced in~\cite{IPM}.
Specifically, after a linear-time preprocessing of an input string $X$,
the algorithm of \cref{thm:edalgII} can efficiently compute $\OccE_k(P,T)$ for any two fragments $P,T$ of the string $X$.


In \cref{sec:impl}, we show that existing implementations of the primitive operations of the \modelname model allow
us to also obtain efficient algorithms for pattern matching under edit distance in the fully-compressed setting (where the text and the pattern are
given as straight-line programs) and in the dynamic setting (where we maintain a
collection of strings under creation, splitting, and concatenation).
Our algorithms improve over the state-of-the-art algorithms of~\cite{unified} for these
settings: we trade a $\sqrt{k}$ factor for a factor that is asymptotically
upper-bounded by the logarithm of the length of the considered pattern.
Formally, we obtain the following results.
\begin{restatable}{mtheorem}{dynalgmain}\label{thm:dynalgmain}
    We can maintain a collection $\X$ of non-empty persistent strings of~total length $N$
    subject to
    $\makestring(U)$, $\concat(U,V)$, and $\splitOp(U,i)$ operations that require
    $\cO(\log N +|U|)$,
    $\cO(\log N)$, and $\cO(\log N)$ time, respectively, so that given two strings $P, T
    \in \X$, and an integer threshold $k>0$, we can compute a representation of $\OccE_k(P, T)$
    as $\cO(|T|/|P| \cdot k^3)$ arithmetic progressions with the same difference in
    time $\Oh(|T|/|P| \cdot k^{3.5} \sqrt{\log |P| \log k} \log^2
    N)$.\ifx\dynalgmaint\undefined{\footnote{All running time bounds hold with high
            probability (that is, $1-1/N^{\Omega(1)}$). A deterministic version can be obtained at the cost of a $\mathrm{poly}(\log \log N)$-factor overhead.}
    \lipicsEnd}\fi
\end{restatable}
\vspace*{-2ex}
\begin{restatable}{mtheorem}{gcedalgmain}\label{gc_ed_alg_intro}
    Let $\G_T$ denote a straight-line program of~size~$n$ generating a string~$T$,
    let~$\G_P$ denote a straight-line program  of~size~$m$ generating a string~$P$,
    let $k>0$ denote an integer threshold,
    and set $N \coloneqq |T|$ and $M \coloneqq |P|$.
    We can compute $|\OccE_k(P, T)|$ in time
    $\Oh(m\log N + n\, k^{3.5} \log^{2} N \sqrt{\log M \log k} \log\log N)$
    and we can report the elements of $\OccE_k(P, T)$ within $\Oh(|\OccE_k(P, T)|)$ extra time.
    \ifx\gcedalgmaint\undefined\lipicsEnd\fi
\end{restatable}

\subsection{Related Work}

\subparagraph*{Pattern Matching with Mismatches.}
The Hamming distance of~two
(equal-length) strings is the number of~positions where the strings differ.
This metric is more restrictive than edit distance since it allows substitutions but does not support insertions or deletions.

In the \emph{pattern matching with mismatches} problem,
we are given a text $T$ of~length $n$, a~pattern~$P$ of~length $m$, and an integer threshold $k>0$, and
we wish to compute the \emph{$k$-mismatch occurrences} of~$P$ in $T$,
that is, all length-$m$ substrings of~$T$ that are at Hamming distance at most
$k$ from $P$.
This problem has been extensively studied since the 1980s.
A long line of works~\cite{Abrahamson,Kosaraju,LandauV86,GG86,AmirLP04,CliffordFPSS16,GawrychowskiU18,cgkkp20} has culminated in an
$\tOh(n+ kn/\!\sqrt{m})$-time algorithm, presented by
Gawrychowski and Uznański~\cite{GawrychowskiU18}, who also showed that a significantly
faster ``combinatorial'' algorithm would
have (unexpected) consequences for the complexity of~Boolean matrix multiplication.
Pattern matching with mismatches on strings is thus well understood in the standard setting.

As shown in~\cite{unified}, pattern matching with mismatches admits an $\tOh(k^2 \cdot n/m)$-time algorithm in the \modelname model.
Analogously to pattern matching with edits, this solution constitutes the basis of the state-of-the-art algorithms in
the internal, fully-compressed, and dynamic settings.

\subparagraph*{Online Algorithms for Pattern Matching with Edits.}
The pattern matching with edits problem has also been considered in the online setting
where the text arrives character by character and, by the time $T\position{w}$ becomes available,
the algorithm needs to decide whether $\min_v \ed(P,T\fragmentco{v}{w})\le k$.
Landau, Myers, and Schmidt~\cite{LMS98} provided an online algorithm that runs in $\Oh(k)$ time per character.
Subsequent work focused on the streaming model, whether the main emphasis is on reducing the space complexity of an online algorithm,
usually at the cost of introducing Monte-Carlo randomization.
Starikovskaya~\cite{S17} presented an algorithm for this setting with both the space usage and the time required to process each character of the text
being proportional to $\sqrt{m} (k \log m)^{\cO(1)}$.
Very recently, Kociumaka, Porat, and Starikovskaya~\cite{KPS21}, improved upon this result, presenting an algorithm that uses $\cOtilde(k^5)$ space
and processes each character of the text in $\cOtilde(k^8)$ amortized time; here,
$\cOtilde(\star)$ hides $\log^{\cO(1)}m$ factors.

\subparagraph*{Approximating Pattern Matching with Edits.}
Chakraborty, Das, and Koucký~\cite{CDK19} presented an $\cOtilde(nm^{3/4})$-time algorithm that produces,
for each position $w$ of the text, a constant factor approximation of
$\min_v\ed(P,T\fragmentco{v}{w})$. They also provided an online algorithm with a weaker approximation guarantee.

\subsection{Open Problems}

The most important and obvious open problem is to close the gap between upper and lower
bounds for the pattern matching with edits problem; as is depicted in \cref{fig:tradeoffs}.
In the quest for faster algorithms, one could try to relax the problem in scope, for
instance, by considering its (easier) decision version
where we only need to check whether $\OccE_k(P,T)$ is empty,
or by allowing for some approximation by also reporting an arbitrary subset of the positions in $\OccE_{(1+\eps)k}(P,T) \setminus \OccE_k(P,T)$ for a small $\eps>0$.

Another research direction could be to devise an algorithm with an analogous running time as the
one presented here
that reports all fragments of $T$ that are at edit distance at most $k$ from $P$ (in
appropriate batches); recall that $\OccE_k(P,T)$ is only the set
of the starting positions of such fragments. While we think that the $\Oh(k^{4}\cdot n/m)$-time \modelname algorithm of \cite{unified}
can be generalized to report all such fragments, our $\tOh(k^{3.5}\cdot n/m)$-time solution
\emph{does not} seem to generalize.
We remark that Landau, Myers, and Schmidt~\cite{LMS98} showed that all the sought fragments can be listed in $\Oh(nk)$ time;
for this, they adapted the algorithm of~\cite{LandauV89}.

\subsection{Technical Overview}\label{sec:techov}

For a string $P$ (also called a \emph{pattern}), a string $T$ (also called a \emph{text}),
and an integer $k > 0$ (also called a \emph{threshold}), we say that $P$ has a
\emph{$k$-error occurrence} in $T$ at position $v$ if we have
$\ed(P, T\fragmentco{v}{w})\leq k$ for some $w \ge v$.
We write $\OccE_k(P,T)$ to denote the set of the starting positions of $k$-error occurrences of
$P$ in~$T$, that is, $\OccE_k (P,T)\coloneqq  \{v \mid \exists_{w \ge v} \; \ed(P,T\fragmentco{v}{w})\leq k\}$.
We now formally state the \emph{pattern matching with edits} problem.

\problembox{
    $\PMWE(P, T, k)$\\
    {\bf{Input:}} A pattern $P$ of length \(m\), a text $T$ of length \(n\), and a positive integer $k \leq m$.\\
    {\bf{Output:}} The set $\OccE_k(P,T)$.
}

\paragraph*{The {NewPeriodicMatches} Problem}

Let us start with a short exposition of parts of our notation.\footnote{See also \cref{sec:prel}, where we
provide a comprehensive exposition of the notation used throughout this paper, including
those we consider standard. Further, consider the notation tables at the very end of this
paper for a quick reference for the most important notations.}
A string $S$ is \emph{primitive} if it cannot be expressed as $U^y$ for a string $U$ and an integer $y>1$.
For two strings $U$ and $V$, we write $\edp{U}{V} \coloneqq \min
\{\ed(U,V^\infty\fragmentco{0}{j}) \mid j \in \mathbb{Z}_{\ge 0}\}$  to denote the minimum edit
distance between $U$ and any prefix of $V^\infty = V\cdot V \cdots$.
Further, we write $\edl{U}{V} \coloneqq \min\{\ed(U,V^\infty\fragmentco{i}{j}) \mid i, j \in \mathbb{Z}_{\ge 0}, i \le j\}$
to denote the minimum edit distance between $U$ and any substring of~$V^\infty$.

As we explain in \cref{sec:oldalgo,sec:redsm}, a recent algorithm of Charalampopoulos, Kociumaka, and Well\-nitz~\cite{unified}
reduces the \PMWE problem to several instances of the following restricted variant;
the reduction takes $\cO(n/m \cdot k^3)$ time in the \modelname model.

\def\SMproblem{\problembox{
    $\SM(P, T, k, d, Q, \A_P, \A_T)$\\
    {\bf{Input:}} A pattern $P$ of length $m$,
    an integer threshold $k \in \fragment{0}{m}$,
    a positive integer $d\ge 2k$,
    a text~$T$ of length $n \in \fragmentco{m-k}{\ceil{\threehalfs m} + k}$,
    a primitive string $Q$ of length $q \coloneqq |Q| \le {m}/{8d}$, an edit-distance alignment
    $\A_P: P \onto Q^\infty\fragmentco{0}{y_P}$
        of cost $d_P\coloneqq \edl{P}{Q}=\ed(P,Q^*)\le d$, and an edit-distance alignment
        $\A_T : T \onto Q^\infty{\fragmentco{x_T}{y_T}}$ of cost
        $d_T \coloneqq \edl{T}{Q}\le 3d$, where $x_T \in \fragmentco{0}{q}$.\\
    {\bf{Output:}} The set $\OccE_k(P,T)$ represented as $\cO(d^3)$ disjoint arithmetic progressions with difference $q$.
}}
\SMproblem

Specifically, \cite{unified} implies the following reduction.

\begin{restatable}{fact}{reduction}\label{fact:reduction}
    Let \(P\) denote a pattern of length \(m\), let \(T\) denote a text of length \(n\),
    and let \(k\leq m\) denote a positive integer.

    We can compute a representation of the set \(\OccE_k(P,T)\) as $\cO(n/m \cdot k^3)$ disjoint arithmetic progressions with the same
    difference in time
    \(\Oh(n/m \cdot k^3)\) in the \modelname model plus the time required for solving
    several instances $\SM(P_i, T_i, k_i, d_i, Q_i,\A_{P_i},\A_{T_i})$, where $\sum_i |P_i| = \cO(n)$ and, for each~$i$, we have
    $|P_i| \le m$ and $d_i=\ceil{8k/m\cdot |P_i|}$.\lipicsEnd
\end{restatable}
\def\redlend{1}

\begin{remark}
    In the case where \cref{fact:reduction} is applied to an
    instance of the \PMWE problem such that the pattern $P$ is approximately periodic,
    the input $(P_i, T_i, k_i, d_i, Q_i,\A_{P_i},\A_{T_i})$ to each produced instance of $\SM$
    satisfies the following conditions: $P_i=P$, $T_i$ is a fragment of $T$, $k_i = k$, and $d = \cO(k)$.
    For the purposes of this technical overview, one can focus solely on that case.
    \lipicsEnd
\end{remark}

\SetKwFunction{PerMat}{PeriodicMatches}

Using the algorithm \PerMat{$P$, $T$, $k$, $d$, $Q$} of {\cite[Lemma 6.11]{unified}}
to solve the \SM problem in $\cO(d^4)$ time in the \modelname model,
the time (in the \modelname model) required for solving all instances of the \SM problem
that are generated by \cref{fact:reduction} is
\[\sum_i \Oh(d_i^4) = \sum_i \Oh(k^4/m^4 \cdot |P_i|^4) = \sum_i
\Oh(k^4/m \cdot |P_i|) = \Oh(n/m \cdot k^4).\]
In particular, we can reinterpret the \(\Oh(n/m\cdot k^4)\)-time algorithm of~\cite{unified}
for the \PMWE problem as a combination of \cref{fact:reduction} and \PerMat{$P$, $T$, $k$,
$d$, $Q$}.
Our main contribution is the following faster algorithm for the \SM problem.

\begin{restatable}[\texttt{NewPeriodicMatches($P$, $T$, $k$, $d$, $Q$, $\A_P$, $\A_T$)}]{lemma}{newsm}\label{lem:sthold}
    We can solve the \SM problem in $\cO(d^{3.5} \sqrt{\log n \log d})$ time
    in the \modelname model.\lipicsEnd
\end{restatable}

By combining \cref{fact:reduction,lem:sthold}, we obtain \cref{thm:edalgII}.

\def\edalgIt{1}
\edalgII*
\begin{proof}
    By \cref{fact:reduction}, in \(\Oh(n/m \cdot k^3)\) time,
    we can reduce the \PMWE problem to several instances $\SM(P_i, T_i, k_i, d_i, Q_i,\A_{P_i},\A_{T_i})$, where $\sum_i |P_i| = \cO(n)$ and, for each~$i$, we have
    $|P_i| \le m$ and $d_i=\ceil{8k/m\cdot |P_i|}$.
    By \cref{lem:sthold}, the time required for solving all of the obtained instances
    (in the \modelname model) is
    \begin{align*}
        \sum_i \Oh(d_i^{3.5} \sqrt{\log |T_i| \log d_i}) & = \sum_i \Oh(k^{3.5}/m^{3.5} \cdot |P_i|^{3.5} \sqrt{\log m \log k}) \\
                                                         & = \sum_i \Oh(k^{3.5}/m \cdot |P_i| \sqrt{\log m \log k})\\
                                                         & = \Oh(n/m \cdot k^{3.5} \sqrt{\log m \log k}).\qedhere
    \end{align*}
\end{proof}

\paragraph*{A Fast Algorithm for the {NewPeriodicMatches} Problem}

We continue with a high-level description of the algorithm that underlies \cref{lem:sthold}.
In what follows, for simplicity, we assume that $m/k \gg q \gg k$ and
that both $Q^\infty\fragmentco{0}{y_P}$ and $Q^\infty{\fragmentco{x_T}{y_T}}$ are powers of $Q$.

\subparagraph*{A First Solution via the {DynamicPuzzleMatching} Problem.}

Let us first discuss how the (almost) periodicity of
\(P\) and \(T\) yields a simple way to \emph{filter out} many potential starting positions of $k$-error
occurrences.

As an introductory example, suppose that \(P\) and \(T\) are perfectly periodic with
period \(Q\), that is,
\(P = Q^{\infty}\fragmentco{0}{m}\) and
\(T = Q^{\infty}\fragmentco{0}{n}\).
Observe that in this special case, \(\A_P\) and \(\A_T\)
are cost-0 alignments (that is, \(d_T = d_P = 0\)), and we have \(m = y_P\), \(0 = x_T\),
and \(n = y_T\).
Next, we argue that all $k$-error occurrences of $P$ in $T$ start around the positions in \(T\) where an exact
occurrence of \(Q\) starts, that is, in the intervals \(\fragment{j q - k}{jq + k}\) for \(j \in \mathbb{Z}\).\footnote{Under our earlier assumption  that \(q \gg k\), this claim indeed allows for filtering out some positions where no occurrence may start
as we have \(jq + k \ll (j+1)q - k\) in that case.}
To that end, observe that, for any alignment of cost at most $k$ mapping $P$ to a fragment $T\fragmentco{v}{w}$ of $T$,
at least one of the copies of $Q$ that comprise $P$ must match exactly; otherwise the edit distance would be much larger than $k$.
Suppose that the $i$-th copy of $Q$, that is, $P\fragmentco{iq}{(i+1)q}$, is matched exactly.
As $Q$ is primitive and hence does not match any of its non-trivial rotations,
$P\fragmentco{iq}{(i+1)q}$ must be matched with
a fragment $T\fragmentco{i'q}{(i'+1)q}$ of $T$. As the entire alignment makes at most $k$ insertions and deletions,
this implies that $v \in \fragment{(i'-i)q-k}{(i'-i)q+k}$.

Now, the strings \(P\) and \(T\) are only \emph{almost} periodic---in particular, the
edits in \(\A_T\) and \(\A_P\) may widen the intervals of potential starting positions,
albeit only by a $d_T+d_P \leq 4 d$ additive term. Since $k<d$, we have
$\OccE_k \subseteq \bigcup_{j\in \mathbb{Z}} \fragment{j q - 5d}{j q + 5d}$.
Hence, for each $j$, we define a fragment $R_j = T\fragmentco{r_j}{r'_j}$ of $T$ that is of length $m+\cO(d)$
and, in the considered instance, is \emph{responsible} for capturing $k$-error occurrences of $P$ in $T$ that start in $\fragment{j q - 5d}{j q + 5d}$;
specifically, we have $r_j + \OccE_k(P,R_j) \supseteq \OccE_k(P,T) \cap \fragment{j q - 5d}{j q + 5d}$.
In addition, we identify a set $J \sub \mathbb{Z}$ of size $\cO(m/q)$ such that
$\OccE_k(P,T) = \bigcup_{j \in J} \big(r_j + \OccE_k(P,R_j) \big)$.

Our goal is to compute occurrences of $P$ in each $R_j$ separately.
To that end, observe that both  $P$ and all $R_j$s essentially decompose into
(possibly slightly ``edited'') copies of $Q$.
In particular, for $j,j+1 \in J$, we can obtain $R_{j+1}$ from \(R_{j}\)
by replacing $\cO(d)$ such ``edited'' copies.
As a first step toward capturing the notions of \(P\) and \(R_j\) decomposing into \emph{pieces} and our algorithm
\emph{replacing} pieces of \(R_j\), we define \emph{$\Delta$-puzzles};
consult \cref{fig:puzzles} for a visualization of an example of a \(\Delta\)-puzzle.

\begin{definition}
    For a $\Delta \in \Zz$, we say that $z\ge 2$ strings $S_1,\ldots,S_z$ form a
    \emph{$\Delta$-puzzle} if
    \begin{itemize}
        \item $|S_i|\ge \Delta$ for each $i\in \fragment{1}{z}$, and
        \item $S_i\fragmentco{|S_i|-\Delta}{|S_i|} = S_{i+1}\fragmentco{0}{\Delta}$ for
            each $i\in \fragmentco{1}{z}$.
    \end{itemize}
    The \emph{value} of the puzzle is  $\val_{\Delta}(S_1,\ldots,S_z)\coloneqq S_1
    \cdot S_2\fragmentco{\Delta}{|S_2|}\cdot S_3\fragmentco{\Delta}{|S_3|}\cdots
    S_z\fragmentco{\Delta}{|S_z|}$.
    \lipicsEnd
\end{definition}

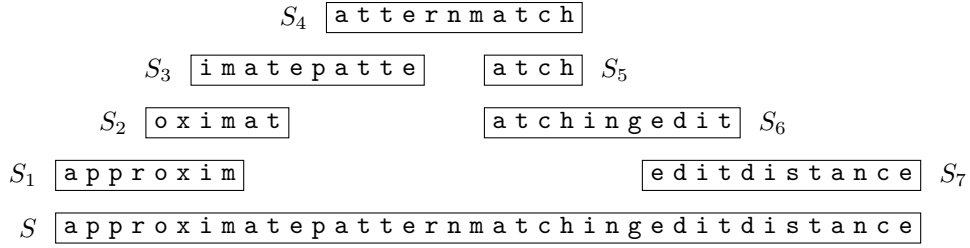
\begin{figure}[t!]
    \centering
    \begin{tikzpicture}

\node[label = {left: $S$}]  at (0,0.2) {};
\draw (0,0) rectangle (11.5,0.4);
\foreach \x/\y in {0.2/$\strut\texttt{a}$, 0.5/$\strut\texttt{p}$, 0.8/$\strut\texttt{p}$, 1.1/$\strut\texttt{r}$, 1.4/$\strut\texttt{o}$, 1.7/$\strut\texttt{x}$, 2/$\strut\texttt{i}$, 2.3/$\strut\texttt{m}$, 2.6/$\strut\texttt{a}$, 2.9/$\strut\texttt{t}$, 3.2/$\strut\texttt{e}$, 3.5/$\strut\texttt{p}$, 3.8/$\strut\texttt{a}$, 4.1/$\strut\texttt{t}$, 4.4/$\strut\texttt{t}$, 4.7/$\strut\texttt{e}$, 5/$\strut\texttt{r}$, 5.3/$\strut\texttt{n}$, 5.6/$\strut\texttt{m}$, 5.9/$\strut\texttt{a}$, 6.2/$\strut\texttt{t}$, 6.5/$\strut\texttt{c}$, 6.8/$\strut\texttt{h}$, 7.1/$\strut\texttt{i}$, 7.4/$\strut\texttt{n}$, 7.7/$\strut\texttt{g}$, 8/$\strut\texttt{e}$, 8.3/$\strut\texttt{d}$, 8.6/$\strut\texttt{i}$, 8.9/$\strut\texttt{t}$, 9.2/$\strut\texttt{d}$, 9.5/$\strut\texttt{i}$, 9.8/$\strut\texttt{s}$, 10.1/$\strut\texttt{t}$, 10.4/$\strut\texttt{a}$, 10.7/$\strut\texttt{n}$, 11/$\strut\texttt{c}$, 11.3/$\strut\texttt{e}$}{
                    	\node[anchor=base, label = {\y}]  at (\x,-.25) {};
                	}
                
\begin{scope}[yshift=0.7cm]
                    \foreach \x/\y in {0.2/$\strut\texttt{a}$, 0.5/$\strut\texttt{p}$, 0.8/$\strut\texttt{p}$, 1.1/$\strut\texttt{r}$, 1.4/$\strut\texttt{o}$, 1.7/$\strut\texttt{x}$, 2/$\strut\texttt{i}$, 2.3/$\strut\texttt{m}$}{
                    	\node[anchor=base, label = {above: \y}]  at (\x,-0.25) {};
                	}
                	\node[label = {left: $S_1$}]  at (0,0.2) {};
                    \draw (0,0) rectangle (2.5,0.4);                    
\end{scope}
                	
\begin{scope}[yshift=1.4cm]
					\foreach \x/\y in {1.4/$\strut\texttt{o}$, 1.7/$\strut\texttt{x}$, 2/$\strut\texttt{i}$, 2.3/$\strut\texttt{m}$, 2.6/$\strut\texttt{a}$, 2.9/$\strut\texttt{t}$}{
                    	\node[anchor=base, label = {above: \y}]  at (\x,-0.25) {};
                	}    
                	\node[label = {left: $S_2$}]  at (1.2,0.2) {};
                    \draw (1.2,0) rectangle (3.1,0.4);
\end{scope}            	

\begin{scope}[yshift=2.1cm]
					\foreach \x/\y in {2/$\strut\texttt{i}$, 2.3/$\strut\texttt{m}$, 2.6/$\strut\texttt{a}$, 2.9/$\strut\texttt{t}$, 3.2/$\strut\texttt{e}$, 3.5/$\strut\texttt{p}$, 3.8/$\strut\texttt{a}$, 4.1/$\strut\texttt{t}$, 4.4/$\strut\texttt{t}$, 4.7/$\strut\texttt{e}$}{
                    	\node[anchor=base, label = {above: \y}]  at (\x,-0.25) {};
                	}    
                	\node[label = {left: $S_3$}]  at (1.8,0.2) {};
                    \draw (1.8,0) rectangle (4.9,0.4);
\end{scope}            	

\begin{scope}[yshift=2.8cm]                	
                	\foreach \x/\y in {3.8/$\strut\texttt{a}$, 4.1/$\strut\texttt{t}$, 4.4/$\strut\texttt{t}$, 4.7/$\strut\texttt{e}$, 5/$\strut\texttt{r}$, 5.3/$\strut\texttt{n}$, 5.6/$\strut\texttt{m}$, 5.9/$\strut\texttt{a}$, 6.2/$\strut\texttt{t}$, 6.5/$\strut\texttt{c}$, 6.8/$\strut\texttt{h}$}{
                    	\node[anchor=base, label = {above: \y}]  at (\x,-0.25) {};
                	}     
                	\node[label = {left: $S_4$}]  at (3.6,0.2) {};
                    \draw (3.6,0) rectangle (7,0.4);
\end{scope}            	       	

\begin{scope}[yshift=2.1cm]                	
                	\foreach \x/\y in {5.9/$\strut\texttt{a}$, 6.2/$\strut\texttt{t}$, 6.5/$\strut\texttt{c}$, 6.8/$\strut\texttt{h}$}{
                    	\node[anchor=base, label = {above: \y}]  at (\x,-0.25) {};
                	}
                	\node[label = {right: $S_5$}]  at (7,0.2) {};
                    \draw (5.7,0) rectangle (7,0.4);
\end{scope}            	

\begin{scope}[yshift=1.4cm]                	
                	\foreach \x/\y in {5.9/$\strut\texttt{a}$, 6.2/$\strut\texttt{t}$, 6.5/$\strut\texttt{c}$, 6.8/$\strut\texttt{h}$, 7.1/$\strut\texttt{i}$, 7.4/$\strut\texttt{n}$, 7.7/$\strut\texttt{g}$, 8/$\strut\texttt{e}$, 8.3/$\strut\texttt{d}$, 8.6/$\strut\texttt{i}$, 8.9/$\strut\texttt{t}$}{
                    	\node[anchor=base, label = {above: \y}]  at (\x,-0.25) {};
                	}
                	\node[label = {right: $S_6$}]  at (9.1,0.2) {};
                    \draw (5.7,0) rectangle (9.1,0.4);
\end{scope}            	                	

\begin{scope}[yshift=0.7cm]
                	\foreach \x/\y in {8/$\strut\texttt{e}$, 8.3/$\strut\texttt{d}$, 8.6/$\strut\texttt{i}$, 8.9/$\strut\texttt{t}$, 9.2/$\strut\texttt{d}$, 9.5/$\strut\texttt{i}$, 9.8/$\strut\texttt{s}$, 10.1/$\strut\texttt{t}$, 10.4/$\strut\texttt{a}$, 10.7/$\strut\texttt{n}$, 11/$\strut\texttt{c}$, 11.3/$\strut\texttt{e}$}{
                    	\node[anchor=base, label = {above: \y}]  at (\x,-0.25) {};
                	}
                	\node[label = {right: $S_7$}]  at (11.5,0.2) {};
                    \draw (7.8,0) rectangle (11.5,0.4);
\end{scope}            	
\end{tikzpicture}
    \caption{$S_1, \ldots, S_7$ is a $4$-puzzle whose value is $S$.}
    \label{fig:puzzles}
\end{figure}

In \cref{sec:DPMtile}, we define \emph{pieces} $P_1, \ldots, P_z$ and $T_{j,1}, \ldots
,T_{j,z}$ (for each \(j \in J\)) that form $\Delta$-puzzles with values $P$ and $R_j$,
respectively, where $\Delta \coloneqq 6(d_P+d_T+k)$.
Let us intuitively describe these pieces.\footnote{This description provides an oversimplified definition of pieces. In particular, as defined in \cref{sec:DPMtile}, $P_1$ covers at least 2 tiles whereas $P_z$ covers 17 tiles. This is due to complications arising without the assumption $q \gg k$.}
First, let us partition both $P$ and $T$ into \emph{tiles}, that is,
maximal fragments that are aligned to different copies of $Q$ by $\A_P$ and $\A_T$, respectively.
Observe that all but $\cO(d)$ tiles are exact copies of $Q$.
Further, the endpoints $R_j$ are $\cO(d)$ positions apart from tile boundaries.
We then obtain an induced partition for $R_j$ by extending the first and last tiles that it fully contains by $\cO(d)$ positions.
Finally, we extend all tiles of the partition of $P$ and the induced partition of $R_j$, other than the trailing ones, by $\Delta$ characters to the right.
Consult \cref{fig:tiles} for a visualization of this setting.

\begin{figure}[t!]
\centering
\includegraphics[width=.95\textwidth]{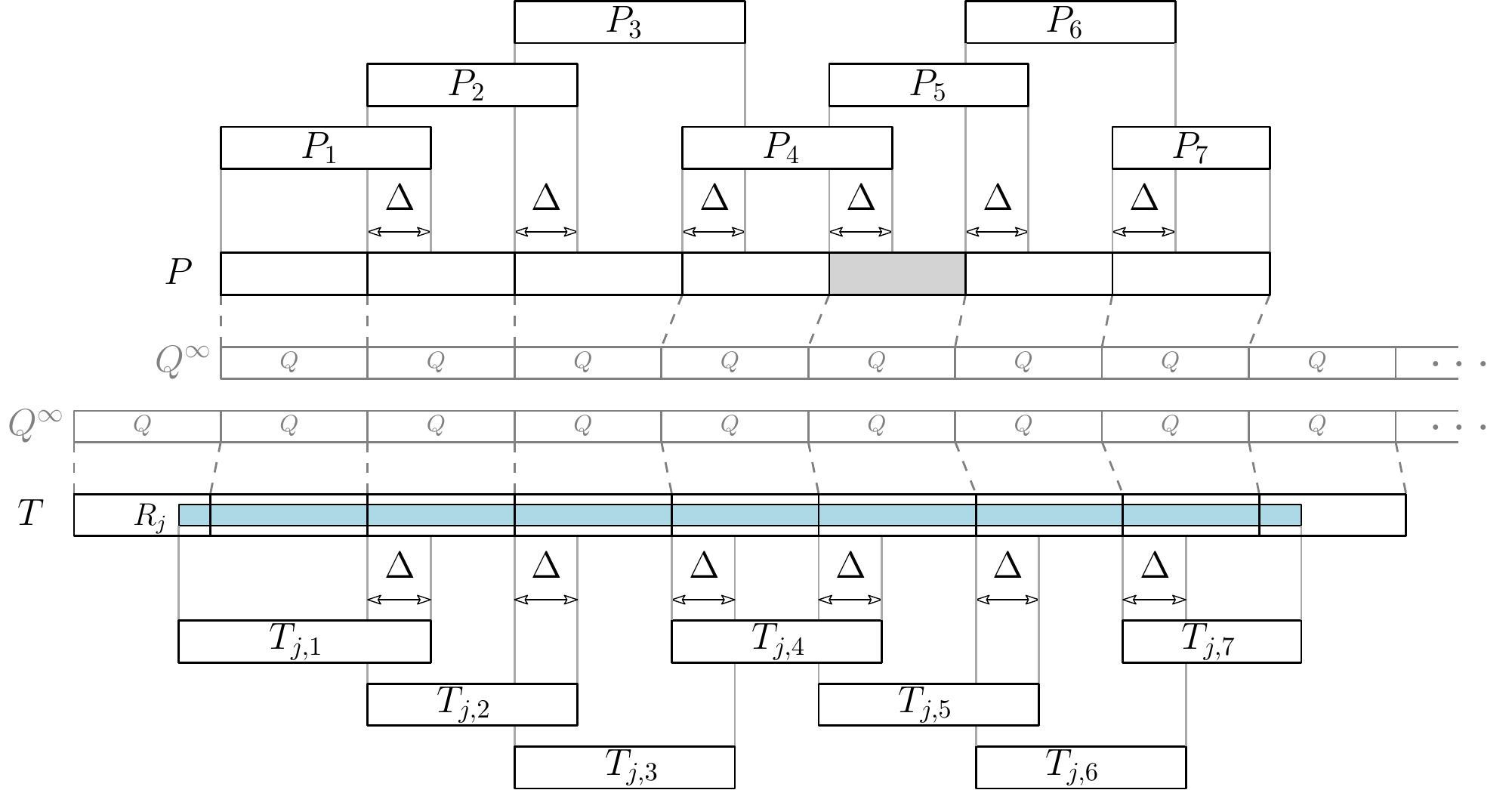}
\caption{
An alignment $\A_P: P \onto Q^7$ and an alignment $\A_T: P \onto Q^9$ are shown.
Both $P$ and $T$ are partitioned into tiles.
Specifically, dashed lines indicate the copy of $Q$ to which a tile of $P$ (or~$T$) is
aligned by $\A_P$ (or~$\A_T$).
For example, $\A_P$ aligns the shaded tile of $P$ with the fifth copy of $Q$.
The fragment $R_j$ starts $\cO(d)$ positions prior to the start of the second tile of $T$ and ends $\cO(d)$ positions after the end of the eighth tile.
The pieces $P_1, \ldots , P_7$ and $T_{j,1}, \ldots, T_{j,7}$ form $\Delta$-puzzles with values $P$ and $R_j$, respectively.}\label{fig:tiles}
\end{figure}

We call pieces $P_2, \ldots, P_{z-1}$ and  $T_{j,2}, \ldots, T_{j,z-1}$ (for \(j \in J\)) \emph{internal}.
Observe that, for each $i$, all internal pieces of the form $T_{j,i'}$ with \(j + i' = i\)
coincide; that is, overlapping parts of different \(R_j\)s share their internal pieces.
Hence, for each $i \in \fragmentoo{\min J +1}{\max J +z}$, we define $T_i \coloneqq T_{j,i'}$ for any $j \in J$ and
$i' \in \fragmentoo{1}{z}$ with $j+i'=i$.
This is an essential property for our approach to work: when
moving from \(R_j\) to \(R_{j + 1}\), we exploit that we need to only shift the pieces
\(T_i\), and not recompute them altogether.

Now, suppose that we can efficiently maintain a pair of $\Delta$-puzzles so that we can at any time
efficiently query for the $k$-error occurrences of the value of the first puzzle in the value of the second one.
Then, as a warm-up solution, we can initialize the two puzzles as $P_1, \ldots , P_z$ and $T_{\min J,1}, \ldots, T_{\min J, z}$
and then replace pieces of the second puzzle as necessary in order to iterate over puzzles $T_{j,1}, \ldots, T_{j, z}$ for all $j \in J$.
In fact, our final algorithm iterates over carefully trimmed versions of such puzzles,
where we omit \emph{plain} pieces that do not contribute to the solution set in an
interesting manner.
Formally, we capture the problem of maintaining such a pair of puzzles with the \DPM
problem.

\vspace{2mm}
\noindent\fbox{
    \begin{minipage}{0.96\textwidth}
        $\DPM(k, \Delta, \Sb, \Sm, \Sf)$\\ 
        {\bf{Input:}} Positive integers $k$ and $\Delta$,
        as well as string families $\Sb$, $\Sm$, and $\Sf$ of \emph{leading}, \emph{internal}, and \emph{trailing} pieces, respectively.\\
        {\bf{Maintained object:}} A sequence $\I=(U_1,V_1)(U_2,V_2)\cdots(U_z,V_z)$ of
        ordered pairs of strings (a \emph{DPM-sequence}), that additionally satisfies
        the following two conditions:
        \begin{enumerate}[(a)]
            \item $U_1,V_1 \in \Sb$, $U_z, V_z  \in \Sf$,
                and, for all $i\in \fragmentoo{1}{z}$, $U_i,V_i \in \Sm$,\label{it:orig}
            \item The \torn \(\tor(\I) \coloneqq \sum_{i=1}^z \big||U_i|-|V_i|\big|\)
                satisfies \(\tor(\I) \le \Delta/2 - k\).
                \label{it:unique}
        \end{enumerate}
        {\bf{Update operations:}}
        \begin{itemize}
            \item {\tt DPM-Delete(\(i\))}: Delete the $i$-th pair of strings.
            \item {\tt DPM-Insert($(U',V')$, \(i\))}:
                Insert the pair of strings $(U',V')$ after the $i$-th pair of strings.
            \item {\tt DPM-Substitute($(U',V')$, \(i\)):}
                Substitute the $i$-th pair of strings with the pair of strings $(U',V')$.
        \end{itemize}
        It is assumed that $\I$ satisfies conditions (\ref{it:orig}) and (\ref{it:unique})
        at initialization time and after each update.\\
        {\bf{Query ({\tt DPM-Query}):}} Return
        $\OccE_k(\I) \coloneqq
        \OccE_k(\val_{\Delta}(U_1,\ldots, U_z), \val_{\Delta}(V_1,\ldots, V_z))$
        under a promise that $U_1, \ldots, U_z$ and  $V_1, \ldots, V_z$ are $\Delta$-puzzles.
    \end{minipage}
}
\vspace{2mm}

\newcommand{\Sr}{\hat{S}}
We move on to our main result for the \DPM problem.
For a precise statement, we need to be able to
quantify the complexity of the input families of strings;
formally we define the \emph{median edit distance}
of a family $\S$ of strings over an alphabet $\Sigma$
as \(\ed(\S) \coloneqq \min_{\Sr\in \Sigma^*} \sum_{S\in \S}\ed(S,\Sr)\).
Now, our result reads as follows.


 \begin{restatable}{theorem}{dpm}\label{thm:dpm}
     There is a data structure for \(\DPM(k, \Delta, \Sb, \Sm, \Sf)\)
     with $\Oh(\Delta \log z \log \Delta)$-time updates and queries,
     $\Oh(\Delta z \log \Delta)$-time initialization,
     and $\Oh((d^3+\Delta^2 d)\log^2 (d+\Delta))$-time preprocessing,
     where $d=\ed(\Sb)+\ed(\Sm)+\ed(\Sf)$.\footnote{Recall that \(z\) is
         the length of the DPM-sequence that we maintain in the data
     structure.}
    \ifx\dpmlend\undefined\lipicsEnd\fi
\end{restatable}
\def\dpmlend{1}

Let us defer a detailed discussion of \cref{thm:dpm} (proved in \cref{sec:DPM_sol}) to the end of this overview.
Here, we discuss its application to the \SM problem with the following string families.
\begin{align*}
    \Sb &\coloneqq \{P_1\} \cup \{T_{j,1} \mid j\in J\},\\
    \Sm &\coloneqq \{P_{i} : i\in \fragmentoo{1}{z}\}\;\cup \{T_{i} : i \in \fragmentoo{\min J +1}{\max J +z} \},\\
    \Sf &\coloneqq \{P_z\} \cup \{T_{j,z} : j\in J\}.
\end{align*}
Next, we define multisets \(\sp{P}\), \(\sp{T}\), and \(\spb{T}\) of \emph{special pieces}---in this
overview, we focus on the two former multisets.
In our fixed instance, we have
\begin{align*}
    \sp{P} & = \{P_i : i \in \fragmentoo{1}{z} \text{ and } P_i \neq Q^\infty\fragmentco{0}{q+\Delta}\}, \\
    \sp{T} & = \{T_i :i \in \fragmentoo{1 + \min J}{z + \max J} \text{ and } T_i \neq Q^\infty\fragmentco{0}{q+\Delta}\}.
\end{align*}
As mentioned earlier, there are only very few special pieces---crucially, we show the
following lemma.

\begin{restatable}{lemma}{spbndsimpl}\label{lem:specialbound-simpl}
    The median edit distance of each of the families \(\Sb\), \(\Sm\), and $\Sf$ is bounded
    by \(\Oh(d)\).

    Further, each of the multisets \(\sp{P}\),
    \(\sp{T}\), and \(\spb{T}\) is of size \(\Oh(d)\) and can be computed
    in \(\Oh(d)\) time in the \modelname model.
    \lipicsEnd
\end{restatable}

For $j \in J$,
let $\I_j$ denote the DPM-sequence $(P_1,T_{j,1})(P_2,T_{j,2}) \cdots (P_z,T_{j,z})$.
Now for $j \in J \setminus \max J$,
each of $\I_j$ and $\I_{j+1}$ contains $\cO(d)$ \emph{special} pairs,
that is, pairs that contain special pieces.
Hence, we can naively iterate over all $\I_j$s
in an instance of the \DPM problem using $\cO(d\cdot|J|)$ updates;
in the considered instance we have $d\cdot|J|=\cO(d\cdot m/q)$.
See \cref{sec:warmup} for details of this reduction in the general case.
As a preliminary improvement, in
\cref{sec:druns}, we show how to reduce the number of updates to $\cO(d^3)$.
Let us give a brief sketch of this reduction.

We call a pair of pieces $(P_i, T_{j,i})$ \emph{plain} if $i \in \fragmentoo{1}{z}$ and neither $P_i$ nor $T_{j,i}$ is special;
in the restricted case that we are considering here,
the second condition is equivalent to $P_i = T_{j,i} = Q^{\infty}\fragmentco{0}{q+\Delta}$.
For $j\in J$, let $\I'_j$ denote the DPM-sequence obtained from $\I_j$
by trimming each run of plain pairs (that is, maximal contiguous subsequences that consist
of plain pairs) in $\I_j$ to length $k+1$ by deleting excess pairs.

The main idea is that we do not gain or lose any \(k\)-error occurrences by trimming the
DPM-sequences, that is, we have $\OccE_k(\I_j) = \OccE_k(\I'_j)$.
One direction is easy: removing the same substring from two strings \(P\) and \(T\) may
only decrease the edit distance between \(P\) and \(T\); this naturally translates to
DPM-sequences. For the other direction, observe that if a DPM-sequence contains a run of
at least \(k + 1\) plain pairs, then any cost-\(k\) alignment between the corresponding
strings has to perfectly match at least one copy of $Q$ pair in such a run---we can
hence duplicate said copy by adding more plain pairs in the DPM-sequence without increasing the cost of the alignment.
Induction then yields the claim.

With the aim of obtaining an $\cO(d^3)$ upper bound on the number of required updates for iterating over the $\I'_j$s,
let us think of the process of shifting $P$ along $T$.
For each $j$, each run of plain pairs in $\I_j$ can be attributed to a run of plain pieces in $P_2, \ldots , P_{z-1}$
that overlap a run of plain pieces in $T_{j,2} \ldots T_{j,z-1}$.
As we shift $P$, in the most general case, the length of the overlap first increases,
then it remains static, and, finally, it decreases.
Overall, as $j$ gets incremented, a run of plain pairs that is attributed to a specific pair
of runs of plain pieces may change length $\omega(d)$ times.
However, after trimming the lengths of all runs of plain pairs to $k+1$, the length of such a run
gets incremented/decremented $\cO(d)$ times.
As we have $\cO(d)$ special pieces in each of $P$ and $T$, we have $\cO(d^2)$ pairs of runs of plain pieces, and hence
we get the desired $\cO(d^3)$ upper bound, as we can bound the number of updates other than insertions/deletions
of plain pairs by $\cO(d^2)$.

Note that we cannot always iterate explicitly over all $\I'_j$s as this would require
$\Omega(m/q)$ calls to {\tt DPM-Query}.
We circumvent this problem by observing that if we have $\I'_{j-1} = \I'_{j}$
(for some $j \in J\setminus \min J$),
then $r_j + \OccE_k(P,R_{j}) = q + r_{j-1} + \OccE_k(P,R_{j-1})$.
Consequently, for any maximal interval $\fragment{j_1}{j_2} \subseteq J$ where $\I'_{j_1}= \cdots = \I'_{j_2}$, we only process $\I'_{j_1}$;
then, for each position $u \in \OccE_k(P,R_{j_1})$, we report an
arithmetic progression $\{r_{j_1}+u+iq: i \in \fragment{0}{j_2-j_1}\}$ of $k$-error occurrences of $P$ in $T$.
On a high level, we are offloading the computation of $\bigcup_{j=j_1}^{j_2} \big(r_j + \OccE_k(P,R_{j})\big)$
to the computation of $\OccE_k(\I'_{j_1})$.

\subparagraph*{A Faster Solution.}
To obtain a faster solution for the \SM problem, we intend to trim runs of plain pairs
even further, to a length of roughly \(\tOh(\sqrt{d})\).
Now, naively processing the obtained DPM-sequences, we may obtain ``false-positive''
occurrences, but---as we can prove---not too many. In
particular, we can extend existing tools to filter such ``false-positive'' occurrences.

For a slightly more detailed overview,
for any two positions $v<w$ of $T$, let us write $\Q\fragmentco{\rho(v)}{\rho(w)}$ for
the fragment of $\Q$ that $\A_T$ aligns with $T\fragmentco{v}{w}$.
Suppose that we have
\[\edl{P}{Q} = \ed(P,\Q\fragmentco{\rho(v)}{\rho(w)}) \leq \edl{T\fragmentco{v}{w}}{Q} =
\ed(T\fragmentco{v}{w},\Q\fragmentco{\rho(v)}{\rho(w)}).\]
Then, the triangle inequality yields
\begin{align*}
    \Lambda \coloneqq \edl{P}{Q} + \edl{T\fragmentco{v}{w}}{Q}
    & \geq \ed(P,T\fragmentco{v}{w}) \geq \edl{T\fragmentco{v}{w}}{Q} - \edl{P}{Q}.
\end{align*}
We see that, intuitively, the best case is when all the errors of $P$ with $\Q\fragmentco{\rho(v)}{\rho(w)}$
cancel out with errors of $T\fragmentco{v}{w}$ with $\Q\fragmentco{\rho(v)}{\rho(w)}$.
Now, roughly speaking, for each position $v$ of $T$, we quantify the ``potential
savings'' that an alignment $P\onto T\fragmentco{v}{w}$ of cost at most $k$ may yield
compared to $\min_x \ed(P, \Q\fragmentco{\rho(v)}{x}) + \min_y \ed(T\fragmentco{v}{w}, \Q\fragmentco{\rho(v)}{y}))$.
To this end, we use the notion of \emph{locked fragments} from \cite{unified} to mark
each position of the
text with a number of marks proportional to said ``potential
savings''. (A~similar notion was used in \cite{ColeH98}.)
Based on a threshold $\tradeoff = \widetilde{\Theta}(\sqrt{d})$ on the number of marks (and a few technical conditions),
we then classify each position as either \emph{heavy} or \emph{light}.
Details on locked fragments and our marking scheme can be found in \cref{sec:marking}.

We then present our solution for \SM in \cref{sec:faster}.
First, we show that the set of heavy positions intersects $\tOh(\sqrt{d})$ ranges, each of
size $\cO(d)$, where a $k$-error
occurrences of $P$ may start (recall that $\OccE_k \subseteq \bigcup_{j\in \mathbb{Z}} \fragment{j q - 5d}{j q + 5d}$).
We can then compute the intersection of $\OccE_k(P,T)$ with heavy positions efficiently,
that is, in $\cOtilde(d^{3.5})$ time, using known tools.

Having taken care of the heavy positions, we can return to \DPM for the light positions.
To that end, consider again $T\fragmentco{v}{w}$, supposing that $v$ is a
light position of~$T$.
We then have that $\Lambda \geq \ed(P,T\fragmentco{v}{w}) \geq \Lambda - \eta$.
Now, the optimal alignment $\A$ from $P$ to $T\fragmentco{v}{w}$
has to make $\Lambda - \eta$ edit operations just to align the locked fragments of the text and the pattern.
This means that the number of edit operations that $\A$ makes in aligning portions of $P$ disjoint from the locked fragments of $P$
to portions of $T$ disjoint from the locked fragments of $T$ is at most $\tradeoff$.

Now, we define a set $\rred(P) \supseteq \sp{P}$ that additionally contains
all pieces of $P$ that overlap some locked fragment of $P$;
we similarly define a set $\rred(T)$ of pieces of~$T$.
Importantly, both $\rred(P)$ and $\rred(T)$ are of size $\cO(d)$.
Redefining plain pairs to be those that contain no red piece,
we show that we can
trim each run of plain pairs
to have a length of $\cO(\tradeoff)=\widetilde{\Theta}(\sqrt{d})$.
This allows us to reduce our problem to an instance of the \DPM problem with $\widetilde{\Theta}(d^{2.5})$ updates in total;
as before, we can essentially charge all but $\cO(d^2)$ updates
to $\cO(d^2)$ pairs of runs of plain pairs,
so that each such pair gets charged with $\widetilde{\Theta}(\sqrt{d})$ updates.

\paragraph*{A Solution for the {DynamicPuzzleMatching} Problem}
For our solution to the \DPM problem
(which we present in \cref{sec:DPM_sol}),
we rely on a framework of Tiskin~\cite{abs-0707-3619,Tis08,Tis15}
(which we recall and extend in \cref{sec:seaweeds}).
A key observation behind this framework is that \emph{semi-local} alignments between strings $U$ and $V$
can be represented as paths between boundary vertices of a certain \emph{alignment graph}:
a grid on vertices $\fragment{0}{|V|} \times \fragment{0}{|U|}$, 
augmented with diagonal edges.
All horizontal and vertical edges have weight $1$ (they represent insertions and deletions),
whereas each diagonal edge $(u,v)\leftrightarrow (u+1,v+1)$ has weight $0$ (for
a match) or $1$ (for a substitution).
Then, $\ed(V\fragmentco{v}{w},U)$ corresponds to the distance
from $(v,0)$ to $(w,|U|)$.
As observed in~\cite{Tis08}, even though there are quadratically many such distances, they
can be encoded in linear space using a certain \emph{permutation matrix} that we denote by $P_{V,U}$.
Moreover, we can \emph{stitch} alignment graphs by computing
a certain \emph{seaweed product} of permutation matrices.
For example, $P_{V,UU'}$ can be expressed as the seaweed product of $P_{V,U}$ and $P_{V,U'}$ (shifted appropriately so that the characters of $U'$ are indexed from $|U|$ rather than from $0$). Tiskin~\cite{Tis15} provided an $\Oh(n\log n)$-time algorithm for computing the seaweed product of two $n\times n$ permutation matrices, but we cannot hope to compute $P_{V,U}$ in truly subquadratic time because it encodes $\ed(U,V)$.

In our setting, though, the strings $U$ and $V$ are of similar length
(that is, $\big||V|-|U|\big| \le \tor(\I)\le \Delta/2-k$) and we only care about alignments of cost at most $k$.
The underlying paths corresponding to such alignments are fully contained within a narrow \emph{diagonal band} of the alignment graph:
all of their vertices $(u,v)$ satisfy $u-v \in I\coloneqq \fragment{-k}{|V|-|U|+k}$ (in short, they belong to band $I$ of the alignment graph);
see \cref{fig:band} for an illustration.
In order to capture this scenario, we restrict the alignment graph to band~$I$, which corresponds to zeroing out the costs of all diagonal edges outside band $I$.
We prove that the permutation matrix $P_{U,V}|_I$ of the restricted graph can be encoded
in $\Oh(|I|)$ space and computed in $\tOh(|I|^2)$ time (in the \modelname model).
Moreover, we show that $P_{U,V}|_I$ can be expressed solely in terms of $P_{U,V}$, which
leads to a new operation of \emph{restricting} a permutation matrix $P$ to a
given interval $I$. We write $P|_I$ for the result of said operation and we present a linear-time algorithm that computes $P|_I$ directly from $P$ and $I$.

Let us now explain how these techniques are helpful in solving the \DPM problem.
Our high-level idea is to express $P_{V,U}|_I$ as the seaweed product of $z$ smaller permutation matrices $P_1,\ldots,P_z$, with
$P_i$ depending only on the $i$-th pair $(U_i,V_i)$. For a first attempt,
we could use $P_{V_i,U_i}$,
but the corresponding parts of the alignment overlap and thus cannot be stitched easily. Thus, we trim each piece $U_i$ to $U'_i$ so that $U=\val_{\Delta}(U_1,\ldots,U_z)=U'_1\cdots U'_z$. Now, the seaweed product of matrices $P_{V_i,U'_i}$ (shifted appropriately),
restricted \emph{a posteriori} to interval $I$, yields $P_{V,U}|_I$. However, the individual matrices $P_{V_i,U'_i}$ are still too large,
so we need to restrict them \emph{a priori} as well. Thus, we actually use $P_{V_i,U'_i}|_{I_i}$, for appropriate intervals $I_i$ of size at most $\Delta$;
see \cref{fig:stitch} for an illustration.
We build a balanced binary tree on top of the permutation matrices $P_{V_i,U'_i}|_{I_i}$ in order to maintain their seaweed product (so that every update requires recomputing $\Oh(\log z)$ partial products). For each query, we retrieve $P_{V,U}|_I$ and apply the SMAWK algorithm~\cite{SMAWK}
in order to check, for every $v\in \fragment{0}{|V|-|U|+k}$, whether $\ed(U,V\fragmentco{v}{w})\le k$ holds for some $w\in \fragment{|U|-k}{|W|}$
in $\cOtilde(\Delta)$ time in total.

The remaining challenge is to build the matrices $P_{V_i,U'_i}|_{I_i}$.
For this, we exploit the small median edit distance of the families $\Sb,\Sm,\Sf$
to show that all such matrices can be precomputed in $\tOh(d^3+\Delta^2 d)$ time.
If the puzzle pieces were of size $\Oh(d+\Delta)$, we could simply use an algorithm of Charalampopoulos, Kociumaka, and Mozes \cite{CKM20} that maintains $P_{X,Y}$
subject to edits of $X,Y$. In general, though, we decompose each piece into $\Oh(d)$ parts: perfect parts, which can be arbitrarily long but are kept intact among all the puzzle pieces, and imperfect parts, which can contain edits but are of size $\Oh(\Delta)$. For each perfect part, we compute a single restricted permutation matrix in $\tOh(\Delta^2)$ time. For imperfect parts, we use the dynamic algorithm of~\cite{CKM20}. Finally, the restricted permutation matrix of a pair of pieces
is obtained by stitching the matrices for pairs of parts similarly to how we obtain $P_{V,U}|_I$ from $P_{V_i,U'_i}|_{I_i}$s.

\begin{figure}[t]
\begin{center}
\begin{tikzpicture}[xscale =.45, yscale=-.45]

    \begin{scope}
    \clip (0,0) rectangle (8,10);
    \draw[line width = 2mm,blue!50] (0,0) -- (5,5) -- (5,6) -- (7,8) -- (7,9) -- (8,10);
    \end{scope}

    \begin{scope}
    \clip (3,0) rectangle (12,10);
    \draw[line width = 2mm,orange!50] (3,0) -- (11,8) -- (11,9) -- (12,10);
    \end{scope}

    \draw (0,0) grid (12, 10);

    \foreach \x in {0,...,11} {
        \foreach \y in {0,...,9} {
            \draw (\x,\y) -- (\x+1,\y+1);
    }
    }

    \foreach \y in {2,6,7,8} {
        \draw  (-0.25,\y+0.5) node[left] {$\mathtt{a\vphantom{b}}$};
    }

    \foreach \y in {0,1,3,4,5,9} {
        \draw (-0.25,\y+0.5) node[left] {$\mathtt{b\vphantom{a}}$};
    }

    \foreach \x in {2,5,6,9,10} {
        \draw  (\x+0.5, -0.25) node[above] {$\mathtt{a\vphantom{b}}$};
    }

    \foreach \x in {0,1,3,4,7,8,11} {
        \draw (\x+0.5, -0.25)  node[above] {$\mathtt{b\vphantom{a}}$};
    }

    \foreach \x in {2,5,6,9,10} {
    \foreach \y in {2,6,7,8} {
         \draw[line width = .5mm]  (\x, \y) -- (\x+1,\y+1);
    }
    }

    \foreach \x in {0,1,3,4,7,8,11} {
    \foreach \y in {0,1,3,4,5,9} {
         \draw[line width = .5mm]  (\x, \y) -- (\x+1,\y+1);
    }
    }

    \draw[blue, rounded corners=2mm, fill=blue, fill opacity = 0.1] (-.25, 2.25) -- (-.25,-.25) -- (4.25, -.25) -- (12.25, 7.75) -- (12.25, 10.25) -- (7.75, 10.25) -- cycle;

    \end{tikzpicture}
\end{center}
\caption{The alignment graph for $U=\mathtt{bbabbbaaab}$ and $V=\mathtt{bbabbaabbaab}$. Thin edges have cost $1$ whereas thick edges have cost $0$.
The blue and orange path represent cost-2 alignments $U\leadsto V\fragmentco{0}{8}$ and $U\leadsto V\fragmentco{3}{12}$, respectively.
The diagonal band $I=\fragment{-2}{4}=\fragment{-k}{|V|-|U|+k}$ corresponding to $k=2$ is
shaded in blue.}\label{fig:one}\label{fig:band}
\end{figure}
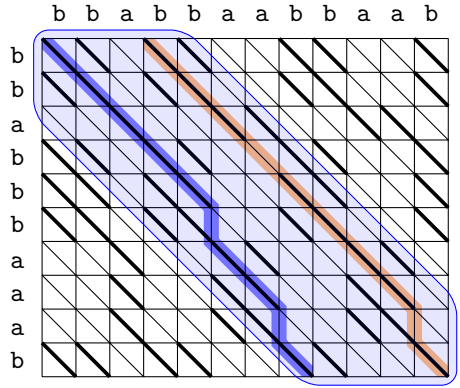

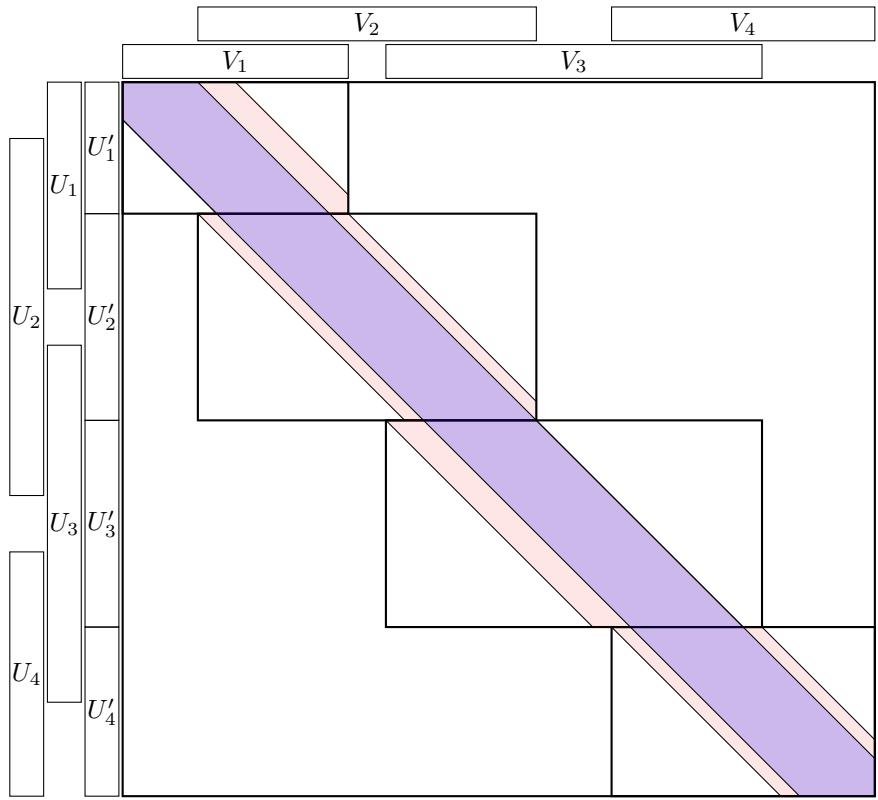
\begin{figure}[b!]
    \begin{center}
    \begin{tikzpicture}[xscale=.5, yscale=-.5]
    
        \draw[thick] (0,0) rectangle (20, 19);
        
        \draw (0, -.1) rectangle node{$V_1$} (6, -1);
        \draw (2, -1.1) rectangle node{$V_2$} (11, -2);
        \draw (7, -.1) rectangle node{$V_3$} (17, -1);
        \draw (13, -1.1) rectangle node{$V_4$} (20, -2);

        \draw (-1.1, 0) rectangle node{$U_1$} (-2, 5.5);
        \draw (-2.1, 1.5) rectangle node{$U_2$} (-3, 11);
        \draw (-1.1, 7) rectangle node{$U_3$} (-2, 16.5);
        \draw (-2.1, 12.5) rectangle node{$U_4$} (-3, 19);
        
        \draw (-.1, 0) rectangle node{$U'_1$} (-1, 3.5);
        \draw (-.1, 3.5) rectangle node{$U'_2$} (-1, 9);
        \draw (-.1, 9) rectangle node{$U'_3$} (-1, 14.5);
        \draw (-.1, 14.5) rectangle node{$U'_4$} (-1, 19);
        
        \draw[thick] (0,0) rectangle (6, 3.5);
        \draw[thick] (2,3.5) rectangle (11,9);
        \draw[thick] (7,9) rectangle (17, 14.5);
        \draw[thick] (13, 14.5) rectangle (20, 19);
        
        \draw[fill=red, fill opacity=0.1] (0,0) -- (3, 0) -- (6, 3)  -- (6, 3.5) -- (2.5, 3.5) -- (0, 1) -- cycle;
        \draw[fill=red, fill opacity=0.1] (2, 3.5) -- (6, 3.5) -- (11,8.5) -- (11,9) -- (7.5, 9) -- cycle;
        \draw[fill=red, fill opacity=0.1] (7,9) -- (11,9) -- (16.5,14.5) -- (12.5,14.5)  -- cycle;
        \draw[fill=red, fill opacity=0.1] (13,14.5) -- (17,14.5) -- (20, 17.5) -- (20, 19) -- (17.5, 19)  -- cycle;

        \draw[fill=blue, fill opacity=0.2] (0,0) -- (2,0) -- (20,18) -- (20,19) -- (18,19) -- (0,1) -- cycle;
        \end{tikzpicture}
    \end{center}
    \caption{A schematic illustration explaining why $P_{V,U}|_I$, which corresponds to the purple band, can be obtained by the seaweed products of matrices $P_{V_i,U'_i}|_{I_i}$, which correspond to the pink bands within the rectangles representing the alignment graphs of $U'_i$ and $V_i$ (subgraphs of the alignment graph of $U$ and $V$).}\label{fig:stitch}
\end{figure}

\newpage

\section{Preliminaries}\label{sec:prel}

\paragraph*{Sets and Arithmetic Progressions}
For $i,j\in \mathbb{Z}$,
we write $\fragment{i}{j}$ to denote $\{i, \dots, j\}$ and
$\fragmentco{i}{j}$ to denote $\{i ,\dots, j - 1\}$;
the sets $\fragmentoc{i}{j}$ and $\fragmentoo{i}{j}$ are defined similarly.

For integers \(a,d\), and \(\ell > 0\),
the set $\{ a + j \cdot d \mid j \in \fragmentco{0}{\ell}\}$
is an \emph{arithmetic progression} with starting value $a$, difference $d$, and
length~$\ell$.
Whenever we use arithmetic progressions in an algorithm, we store them
as triples $(a,d,\ell)$ consisting of their first value, their difference, and their length.

For a set $X\sub \mathbb{Z}$ and an integer $s \in \mathbb{Z}$, we write
$s+X$ and $X+s$ to denote the set $\{s+x : x\in X\}$ containing all elements of $X$ incremented by $s$.

\paragraph*{Strings}

We write $T=T\position{0}\, T\position{1}\cdots T\position{n-1}$ to denote a \textit{string} of
length $|T|=n$ over an alphabet $\Sigma$. The elements of~$\Sigma$ are called \textit{characters}.
We write $\varepsilon$ to denote the \emph{empty string}.

A string $P$ is a \emph{substring} of a string~$T$ (denoted by $P \substr T$)
 if for some integers $i,j$ with $0\le i \le j \le |T|$, we have
 $P=T\position{i}\cdots T\position{j-1}$.
In this case, we say that there is an \emph{exact occurrence} of~$P$ at position $i$
in~$T$, or, more simply, that $P$ \emph{exactly occurs in} $T$.
We write $T\fragmentco{i}{j}$ for this particular occurrence of \(P\) in \(T\),
which is formally a \emph{fragment} of $T$ specified by the two endpoints $i,j$.
For notational convenience, we may also refer to this fragment as $T\fragment{i}{j-1}$,
$T\fragmentoc{i-1}{j-1}$, or $T\fragmentoo{i-1}{j}$.
Two fragments (perhaps of different strings) \emph{match} if they are occurrences of the same strings.

A \emph{prefix} of~a string $T$ is a fragment that starts at position~$0$ (that is, a
prefix is a fragment of~the form $T\fragmentco{0}{j}$ for some $j\in \fragmentco{0}{|T|}$).
A \emph{suffix} of~a string $T$ is a fragment that ends at position ${|T|-1}$ (that is,
a suffix is a fragment of~the form $T\fragmentco{i}{|T|}$ for some $i\in \fragmentco{0}{|T|}$).
We write \(\lcp(U, V)\) for the length of the \emph{longest common prefix} of two strings
$U$ and $V$, that is \(\lcp(U, V)\) is the length of the longest string that occurs as a
prefix of~both $U$ and~$V$.
Similarly, we write \(\lcp^R(U, V)\) for the length of the \emph{longest common suffix} of two strings
$U$ and $V$, that is \(\lcp(U, V)\) is the length of the longest string that occurs as a
suffix of~both $U$ and~$V$.

For two strings $U$ and $V$, we write $UV$ or $U\cdot V$ to denote their concatenation.
We also write $U^k \coloneqq  U\cdots U$ to denote the concatenation of~$k$ copies of~the
string $U$. Furthermore, we write $U^\infty$ to denote an infinite string obtained by concatenating
infinitely many copies of~$U$. To simplify our exposition at certain points, we may access
such an infinite repetition of \(U\) also at negative positions; hence for an integer \(j \in
\fragmentco{0}{|Q|}\) and a (possibly negative) integer \(i\), we formally set \(U^{\infty}\position{i
\cdot |U| + j} \coloneqq U\position{j}\).
A string $T$ is \emph{primitive} if it cannot be expressed as $T=U^k$
for any string~$U$ and any integer~$k > 1$.

A positive integer $p$ is a \emph{period} of~a string $T$ if $T\position{i} =
T\position{i + p}$ for
all $i \in \fragmentco{0}{|T|-p}$. We refer to the smallest
period as \emph{the period} $\per(T)$ of~the string.
Further, we call $T\fragmentco{0}{\per(T)}$ the \emph{string period} of~$T$.
A string is \emph{periodic} if its period is at most half of~its length.

For a string $T$, we define the following \emph{rotation} operations:
The operation $\rot(\star)$ takes as input a string, and moves its last character to the
front; that is,~$\rot(T) \coloneqq  T\position{n-1}T\fragment{0}{n-2}$.
The inverse operation $\rot^{-1}(\star)$ takes as input a string and
moves its initial character to the end; that is,~$\rot^{-1}(T) \coloneqq  T\fragment{1}{n-1}T\position{0}$.
Observe that a primitive string $T$ does not match any of~its non-trivial rotations,
that is, we have $T=\rot^j(T)$ if and only if $j \equiv 0 \pmod{|T|}$.

For a string $T$, the \emph{reverse string} of~$T$ is
$T\position{n-1}T\position{n-2}\cdots T\position{0}$.

\paragraph*{Edit Distance and Pattern Matching with Edits}

The \emph{edit distance} (also known as \emph{Levenshtein distance}) between two
strings $X$ and $Y$, denoted by $\ed(X,Y)$, is the minimum number of
character insertions, deletions, and substitutions required to transform $X$ into~$Y$.
Similarly, the \emph{deletion distance} $\DD(X,Y)$ is the minimum number of character insertions and deletions required to transform $X$ into $Y$.

For a formal definition, we first define an \emph{alignment} between strings.
\begin{definition}\label{def:ali}
    A sequence $\A=(x_i,y_i)_{i=0}^{m}$ is an \emph{alignment}
    of $X\fragmentco{x}{x'}$ onto $Y\fragmentco{y}{y'}$, denoted by \(\A: X\fragmentco{x}{x'} \onto Y\fragmentco{y}{y'}\),
    if we have
    \begin{align*}
        (x_0,y_0) &= (x,y);\quad\text{and}\\
        \text{for all \(i \in \fragmentco{0}{m}\):}\quad (x_{i+1},y_{i+1}) &\in
        \{(x_i+1,y_i+1),(x_i+1,y_i),(x_i,y_i+1)\};\quad\text{and} \\
        (x_m,y_m) &=(x',y').
    \end{align*}
    \begin{itemize}
        \item If $(x_{i+1},y_{i+1})=(x_i+1,y_i)$, we say that $\A$ \emph{deletes}
            $X\position{x_i}$,
        \item If $(x_{i+1},y_{i+1})=(x_i,y_i+1)$, we say that $\A$ \emph{inserts} $Y\position{y_i}$,
        \item If $(x_{i+1},y_{i+1})=(x_i+1,y_i+1)$, we say that $\A$ \emph{aligns}
            $X\position{x_i}$ and $Y\position{y_i}$. If~additionally $X\position{x_i}=
            Y\position{y_i}$, we say that $\A$ \emph{matches} $X\position{x_i}$ and
            $Y\position{y_i}$.
            Otherwise, we say that $\A$ \emph{substitutes} $Y\position{y_i}$ for
            $X\position{x_i}$.\lipicsEnd
    \end{itemize}
\end{definition}


Further, for an alignment $\A:X\fragmentco{x}{x'}\onto Y\fragmentco{y}{y'}$
with \(\A = (x_i,y_i)_{i=0}^m\),
we define the \emph{inverse alignment} $\A^{-1} : Y\fragmentco{y}{y'}\onto
X\fragmentco{x}{x'}$ as $\A^{-1} \coloneqq (y_i,x_i)_{i=0}^m$.
The \emph{cost} of an alignment $\A$ of $X\fragmentco{x}{x'}$ onto $Y\fragmentco{y}{y'}$,
denoted by $\ed^{\A}(X\fragmentco{x}{x'},Y\fragmentco{y}{y'})$, is the total number of characters that $\A$ inserts, deletes, or substitutes.
Now, we define the edit distance $\ed(X,Y)$ as the minimum cost of an alignment
of $X\fragmentco{0}{|X|}$ onto~$Y\fragmentco{0}{|Y|}$.
An alignment of $X$ onto $Y$ is \emph{optimal} if its cost is equal to $\ed(X, Y)$.
The deletion distance $\DD(X,Y)$ is defined as the minimum cost of an alignment $\A : X \onto Y$
that aligns $X\position{x}$ to $Y\position{y}$ only if the two characters match.

Given an alignment $\A:X\fragmentco{x}{x'}\onto Y\fragmentco{y}{y'}$ and a fragment
$X\fragmentco{\bar{x}}{\bar{x}'}\substr X\fragmentco{x}{x'}$,
we write $\A(X\fragmentco{\bar{x}}{\bar{x}'})$ for the fragment
$Y\fragmentco{\bar{y}}{\bar{y}'}\substr Y\fragmentco{y}{y'}$ that $\A$ aligns against $X\fragmentco{\bar{x}}{\bar{x}'}$.
Given that insertions and deletions may render this definition ambiguous, we formally set
\[\bar{y} \coloneqq \min\{\hat{y} : (\bar{x},\hat{y})\in \A\}\quad\text{and}\quad
    \bar{y}' \coloneqq \left\{\begin{array}{c l}
            y' & \text{if }\bar{x}' = x',\\
            \min\{\hat{y}' : (\bar{x}',\hat{y}')\in \A\} & \text{otherwise}.
    \end{array}\right.\]
This particular choice satisfies the following decomposition property.
\begin{fact}\label{fct:ali}
    For any alignment $\A$ of $X$ onto $Y$ and a decomposition $X=X_1\cdots X_t$ into $t$
    fragments, $Y=\A(X_1)\cdots \A(X_t)$ is a decomposition into $t$ fragments with
    $\ed^\A(X,Y)\ge \sum_{i=1}^t \ed(X_i,\A(X_i))$.

    Further, if \(\A\) is an optimal alignment, then have equality:
    $\ed^\A(X,Y) = \sum_{i=1}^t \ed(X_i,\A(X_i))$.
    \lipicsEnd
\end{fact}

Consider \cref{ex:align} for a visualization of an example.

\begin{figure}[t]
    \begin{center}
        \begin{tabular}{lccccccccccccccc}
            $X=$&& \texttt{a} && \texttt{a} && \texttt{c} && $-$ && \texttt{b} && \texttt{c} && \texttt{d} \\
                &&   \textcolor{red}{\textbf{\rotatebox{90}{$*$}}}   &&  \textcolor{blue!50!green}{\textbf{\,\rotatebox{90}{$=$}}} &&
            \textcolor{blue!50!green}{\textbf{\,\rotatebox{90}{$=$}}}   && \textcolor{red}{\textbf{\rotatebox{90}{$*$}}}  &&
            \textcolor{blue!50!green}{\textbf{\rotatebox{90}{$=$}}} &&  \textcolor{blue!50!green}{\textbf{\,\rotatebox{90}{$=$}}}  &&  \textcolor{red}{\textbf{\rotatebox{90}{$*$}}} \\
            $Y=$&& \texttt{b} && \texttt{a} && \texttt{c} && \texttt{d} && \texttt{b} &&  \texttt{c} && $-$\\
            $\A=$ & $(0,0)$ && $(1,1)$ && $(2,2)$ && $(3,3)$ && $(3,4)$ && $(4,5)$ && $(5,6)$ && $(6,6)$
        \end{tabular}
    \end{center}
    \caption{Consider strings $X=\texttt{aacbcd}$ and $Y=\texttt{bacdbc}$.
        We have $\ed(X,Y)=3$, witnessed by an optimal alignment \(\A: X \onto Y\).
        The alignment $\A$ inserts, deletes, or substitutes three characters; these edit
        operations are denoted by red asterisks.
        The set of (augmented) breakpoints of $\A$ is
        $\brkp_{X,Y}(\A)=\{(-1,-1,\bot), (0,0, {\tt SUB}), (3,3, {\tt DEL}), (5,6, {\tt INS}),
        (6,6, \bot)\}$.
        Further, we have
        $\A(X\fragmentco{0}{3}) = Y\fragmentco{0}{3} = \texttt{bac}$,
        $\A(X\fragmentco{3}{5}) = Y\fragmentco{3}{5} = \texttt{dbc}$, and
    $\A(X\fragmentco{5}{6}) = Y\fragmentco{6}{6}=\eps$.}\label{ex:align}
\end{figure}

\begin{definition}
    For an alignment $\A:X\fragmentco{x}{x'} \onto Y\fragmentco{y}{y'}$, and a
    pair \((x_i,y_i)\in \A\) where \(\A\) does not match \(x_i\) and \(y_i\),
    we define the corresponding \emph{augmented breakpoint}~\(B\) as\[
        B \coloneqq \left\{
            \begin{array}{cl}
                (x_i, y_i, {\tt INS}), & \text{if \(\A\) inserts \(Y\position{y_i}\),}\\
                (x_i, y_i, {\tt DEL}), & \text{if \(\A\) deletes \(X\position{x_i}\),
                }\\
                (x_i, y_i, {\tt SUB}), & \text{if \(\A\) substitutes \(X\position{x_i}\),
                and}\\
                (x_i, y_i, \bot), & \text{if \(x_i = x'\) and \(y_i = y'\)
                or if \(x_i = x - 1\) and \(y_i = y - 1\).}
            \end{array}
    \right.\]
    We write $\brkp_{X,Y}(\A)$ for the set of augmented breakpoints.
    \lipicsEnd
\end{definition}

Observe that $\ed^{\A}(X,Y)=|\brkp_{X,Y}(\A)|-2$, with the $-2$ term corresponding to
$(x',y',\bot)\in \brkp_{X,Y}(\A)$ and
$(x - 1,y - 1,\bot)\in \brkp_{X,Y}(\A)$.

Moreover, observe that we can uniquely reconstruct $\A$ from $\brkp_{X,Y}(\A)$.
Our algorithms use $\brkp_{X,Y}(\A)$ (with elements stored in a sorted array) to represent the
alignment $\A$;
if the cost of $\A$ is $d$, then this \emph{breakpoint representation} takes $\Oh(d+1)$
space.\footnote{Observe that the breakpoint representation requires constant (non-zero)
space for alignments with cost \(d = 0\).}

\def\brpst#1#2{\sigma^{#1}_{#2}}
\def\brped#1#2{\lambda^{#1}_{#2}}

\begin{lemma}\label{fct:alignments}
    For an alignment $\A:X\fragmentco{x}{x'} \onto Y\fragmentco{y}{y'}$
    and a position $\bar{y}\in \fragment{y}{y'}$,
    write $\fragment{\brpst{\A}{\bar{y}}}{\brped{\A}{\bar{y}}} \coloneqq
    \{\bar{x} \in \fragment{x}{x'} \mid (\bar{x},\bar{y})\in \A\}$ for the corresponding
    positions in \(X\) under \(\A\).

    Given the breakpoint representation of a cost-$d$ alignment $\A$,
    in \(\Oh(d + 1)\) time, we can compute
    the sequences $(\brpst{\A}{\bar{y}})_{\bar{y}=y}^{y'}$
    and $(\brped{\A}{\bar{y}})_{\bar{y}=y}^{y'}$, represented as concatenations of $d+1$
    arithmetic progressions with difference $1$.
\end{lemma}
\begin{algorithm}[t]
    \SetKwBlock{Begin}{}{end}
    \SetKwFunction{revs}{Reverse}
    \revs{\(\brkp_{X,Y}(\A) = \{(x-1,y-1,\bot),\dots,(x',y', \bot)\}\)}\Begin{
        \(px \gets\) \(x\); \(py \gets\) \(y\)\;
        \ForEach{\(B \in \brkp_{X,Y}(\A)\)}{
            \Switch{\(B\)}{
                \uCase{\upshape (\(a\), \(b\), {\tt INS})}{
                    {\bf yield} \(\fragmentco{\brpst{\A}{px}}{\brpst{\A}{a}} \gets \fragmentco{px}{a}\)\;
                    {\bf yield} \(\fragment{\brped{\A}{py}}{\brped{\A}{b}} \gets \fragment{py}{b}\)\;
                    \(px \gets a\); \(py \gets b + 1\)\;
                }
                \uCase{\upshape (\(a\), \(b\), {\tt DEL})}{
                    {\bf yield} \(\fragment{\brpst{\A}{px}}{\brpst{\A}{a}} \gets \fragment{px}{a}\)\;
                    {\bf yield} \(\fragmentco{\brped{\A}{py}}{\brped{\A}{b}} \gets \fragmentco{py}{b}\)\;
                    \(px \gets a + 1\); \(py \gets b\)\;
                }
                \Other{
                    {\bf continue}\;
                }
            }
        }
        {\bf yield} \(\fragmentco{\brpst{\A}{px}}{\brpst{\A}{x'}} \gets \fragmentco{px}{x'}\)\;
        {\bf yield} \(\fragmentco{\brped{\A}{py}}{\brped{\A}{y'}} \gets \fragmentco{py}{y'}\)\;
    }
    \caption{Pseudo code for \cref{fct:alignments}.}\label{alg:reverse}
\end{algorithm}
\begin{proof}
    First, observe that for a fixed \(\bar{y}\in \fragment{y}{y'}\),
    we have
    \(\brpst{\A}{\bar{y}} \ne \brped{\A}{\bar{y}}\) if and only if \(\A\) deletes
    \(X\fragmentco{\brpst{\A}{\bar{y}}}{\brped{\A}{\bar{y}}}\).
    Similarly, multiple positions \(y_i,
    y_{i + 1}, \dots,y_{i + a}\) share a common value of \(\brpst{\A}{\bar{y_{\star}}} =
    \brped{\A}{\bar{y_{\star}}}\) if and only if \(\A\) inserts the fragment
    \(X\fragmentco{y_i}{y_{i + a - 1}}\).
    Hence, to compute the values \(\brpst{\A}{\bar{y}}\) and \(\brped{\A}{\bar{y}}\) for all
    \(\bar{y}\), it suffices to do a linear pass over the (augmented) breakpoints:
    we keep a pointer \(p_X\) in \(X\) and a pointer \(p_y\) in  \(Y\), each pointing to
    the character after the last
    insertion or deletion of \(\A\) (or to \(X\position{x}\) and \(Y\position{y}\) at the
    beginning). Now, when we encounter an insertion \((a, b, {\tt INS})\), deletion
    \((a, b, {\tt DEL})\) or the last breakpoint \((a = x', b = y', \bot)\)
    in \(\brkp_{X,Y}(\A)\), we set
    \begin{align*}
        \fragment{\brpst{\A}{p_x}}{\brpst{\A}{a}} \coloneqq \fragment{p_x}{a} \quad\text{and}\quad
        \fragment{\brped{\A}{p_y}}{\brped{\A}{b}} \coloneqq \fragment{p_y}{b}.
    \end{align*}
    Afterward, we update \(x_p\) and \(y_p\) accordingly.
    Consult \cref{alg:reverse} for a detailed pseudo-code.

    The correctness follows immediately from the preceding discussion; for the
    running time, observe that at each of the \(d + 1\) events, our algorithm requires
    constant time. In total, we hence obtain the claim.
\end{proof}

We write $\edp{S}{T} \coloneqq  \min
\{\ed(S,T^\infty\fragmentco{0}{j}) : j \in \mathbb{Z}_{\ge 0}\}$  to denote the minimum edit
distance between a string $S$ and any prefix of~a string $T^\infty$.
Further, we write $\edl{S}{T}\coloneqq  \min\{\ed(S,T^\infty\fragmentco{i}{j}) : i, j \in \mathbb{Z}_{\ge 0}, i \le j\}$
to denote the minimum edit distance between $S$ and any substring of~$T^\infty$,
and we set $\eds{S}{T} \coloneqq   \min\{\ed(S,T^\infty\fragmentco{i}{j|T|}) : i, j \in \mathbb{Z}_{\ge 0}, i \le j|T|\}$.

Next, it is easy to verify that the edit distance satisfies the triangle inequality.
\begin{fact}[Triangle Inequality]\label{Etria}
    Any strings $A$, $B$, and $C$ satisfy
    \[\ed(A, C) + \ed(C, B) \ge \ed(A, B) \ge |\ed(A, C) - \ed(C, B)|.
    \lipicsEnd\]
\end{fact}

Equally useful is the fact that we can easily remove prefixes (or suffixes) of the same length.
\begin{fact}\label{fct:ed-prefix}
    For any non-empty strings \(A\) and \(B\), we have \[
        \ed(A\fragmentco{1}{|A|}, B\fragmentco{1}{|B|}) \le \ed(A, B)
        \quad\text{and}\quad
        \ed(A\fragmentco{0}{|A|-1}, B\fragmentco{0}{|B|-1}) \le \ed(A, B).
        \tag*{\lipicsEnd}
    \]
\end{fact}

Combining \cref{fct:ed-prefix,fct:ali}, we see that we can ``cut''
fragments of equal length out of an optimal alignment.

\begin{corollary}\label{cor:al-cut-out}
    Let \(A = A_p \cdot A_s\) and \(B\) denote non-empty strings and write \(\A : A \onto B\)
    for an optimal alignment of \(A\) onto \(B\). For any \(u \in
    \fragmentco{0}{\min\{|A_s|, |\A(A_s)|\}}\), we have\[
        \ed(A_p, \A(A_p)) + \ed(A_s\fragmentco{u}{|A_s|},
        \A(A_s)\fragmentco{u}{|\A(A_s)|} )
        \le \ed(A, B).
    \tag*{\lipicsEnd}
    \]
\end{corollary}

We conclude with the easy, but useful, observation that no \(k\)-error occurrence of a pattern may start
in the final part of the text.

\begin{fact}\label{ft:no-occs-suffix}
    For any text $T$ of length $n$ and any pattern $P$ of length $m$ we have \(\OccE_{k}( P, T ) \cap \fragmentoo{n - m + k}{n} = \emptyset\).
    \lipicsEnd
\end{fact}

\subsection{The \modelname Model}\label{sec:pillar}

To unify the implementations of algorithms for (approximate) pattern matching problems in
different settings, \cite{unified} introduced the \modelname model. We use the same
\modelname model in this work.
In particular, we bound the running times of the algorithms in this work in terms of the number of
calls to a small set of very common operations (the primitive \modelname operations) on
strings.\footnote{If our algorithms require (asymptotically)
significant extra computations, we also specify the required extra running time.}
Together with the implementations of said primitive
operations (presented in detail in~\cite{unified}), we then obtain (fast) algorithms for
various different settings at once. This includes the standard setting, the fully compressed
setting, and a dynamic setting. We provide more details in \cref{sec:impl}.

To keep the \modelname model flexible, we do not directly work on specific representations
of (the input) strings. Instead, in the \modelname model, we maintain a collection of
strings~\(\X\); the operations in the \modelname model work on fragments
\(X\fragmentco{\ell}{r}\) of \(X \in \X\), which are represented via a
\emph{handle}.\footnote{The implementation details depend on the specific setting.
In the standard setting, a fragment $X\fragmentco{\ell}{r}$ is represented by a reference to $X$ and the endpoints $\ell,r$.}
At the start of the computation, the \modelname model provides a handle to each \(X \in
\X\), which represents \(X\fragmentco{0}{|X|}\). Using an \extractOpName operation, we can
obtain handles to other fragments of the strings in \(\X\)~\cite{unified}:
\begin{itemize}
    \item $\extractOpName(S,\ell,r)$: Given a fragment $S$ and positions $0 \le \ell \le r
        \le |S|$, extract the (sub)fragment $S\fragmentco{\ell}{r}$. If
        $S=X\fragmentco{\ell'}{r'}$ for $X\in \X$, then $S\fragmentco{\ell}{r}$ is defined
        as $X\fragmentco{\ell'+\ell}{\ell'+r}$.
\end{itemize}
The other primitive \modelname model operations read as follows~\cite{unified}:
\begin{itemize}
    \item $\lceOp{S}{T}$: Compute the length of~the longest common prefix of~$S$ and $T$.
    \item $\lcbOp{S}{T}$: Compute the length of~the longest common suffix of~$S$ and $T$.
    \item $\ipmOp{P}{T}$: Assuming that $|T|\le 2|P|$, compute $\OccEx(P,T)$ (represented
        as an arithmetic progression with difference $\per(P)$).
    \item $\accOpName(S,i)$: Assuming $i\in \fragmentco{0}{|S|}$, retrieve the character $\accOp{S}{i}$.
    \item $\lenOpName(S)$: Retrieve the length $|S|$ of~the string $S$.
\end{itemize}

As working just with the primitive \modelname model operations is rather bothersome,
\cite{unified} also provides a useful toolbox of operations already implemented.
For brevity, we only list these operations here and refer to \cite{unified} for
(pointers to) their implementation.
\begin{fact}[\modelname Toolbox,
    \cite{unified}]\label{fct:toolbox}\label{lm:streq}\label{lm:inflcp}\label{lm:inflcpold}\label{lm:inflcpR}\label{lm:emath}\label{lm:verifye}
    The \modelname model supports all of the following operations.
    \begin{itemize}
        \item Equality~\cite[Fact 2.5.2]{thesis}:
            For strings $S$ and $T$, we can  check whether $S$ and $T$ are equal
            in $\Oh(1)$ time in the \modelname model.
        \item $\perOp{S}$,~\cite{IPM,thesis}:
            For a string $S$, we can compute $\per(S)$ or declare that $\per(S) > |S| / 2$
            in $\Oh(1)$ time in the \modelname model.
        \item $\cycEqOp{S}{T}$,~\cite{IPM,thesis}:
            For strings $S$ and $T$, we can find all integers $j$ such that $T = \rot^j(S)$
            in $\Oh(1)$ time in the \modelname model. The output is represented as an arithmetic
            progression.
        \item $\lceOp{S}{Q^{\infty}\fragmentco{\ell}{r}}$,~\cite{unified}, see also~\cite[Fact 2.5.2]{thesis} and~\cite{BabenkoGKKS16}:
            For strings $S$ and $Q$ and integers $0 \le \ell \le r$, we can compute
            $\lceOp{S}{Q^{\infty}\fragmentco{\ell}{r}}$ in $\Oh(1)$ time in the \modelname
            model.
        \item $\lcbOp{S}{Q^{\infty}\fragmentco{\ell}{r}}$,~\cite{unified}:
            For strings $S$ and $Q$ and integers $0 \le \ell \le r$, we can compute
            $\lcbOp{S}{Q^{\infty}\fragmentco{\ell}{r}}$ in $\Oh(1)$ time in the \modelname
            model.
        \item {\tt ExactMatches}$(P, T)$,~\cite{unified}:
            Let $T$ denote a string of~length $n$ and let $P$ denote a string of~length~$m$.
            We can compute the set $\OccEx(P, T)$ using $\Oh(n/\per(P))$ time and $\Oh(n/m)$
            \modelname operations.
        \item {\tt Verify($P$, $T$, $k$, $I$)}, {\cite[Section 5]{ColeH98}; see
            also~\cite{unified}}:
            Let $P$ denote a string of~length $m$, let $T$ denote a string,
            and let $k \le m$ denote a positive integer.
            Further, let $I$ denote an interval of~positive integers. Using $\Oh(k(k + |I|))$
            \modelname operations, we can compute $\{(\ell,\min_r \ed(P,T\fragmentco{\ell}{r}))
            : \ell \in \OccE_k(P, T)\cap I\}$.
            \lipicsEnd
    \end{itemize}
\end{fact}

\SetKwFunction{verify}{Verify}
\SetKwFunction{alignment}{Alignment}

Finally, we observe that the adaptation of the Landau--Vishkin algorithm~\cite{LandauV89}
already provided in~\cite{unified} lets us efficiently compute an optimal alignment of a string $S$
onto a substring of $Q^\infty$ starting at a given position $x\in \Z$.

\begin{lemma}[{$\protect\alignment(S,Q,x)$}]\label{fct:alignment}
    Consider non-empty strings $S,Q$ and an integer $x\in \Z$.
    Using $\Oh(1+d^2)$ \modelname operations, we can construct (the breakpoint representation of) an alignment
    $\A: S\onto Q^\infty\fragmentco{x}{y}$ of optimal cost $d=\edp{S}{\rot^{-x}(Q)}$.
\end{lemma}
\begin{proof}
    Set $\hat{Q} \coloneqq \rot^{-x}(Q)$.
    \cite[Lemma 6.1]{unified} provides an online algorithm that, for subsequent integers $k\in \Zz$, computes the longest prefix $S_k$ of $S$ such that $\edp{S_k}{\hat{Q}}\le k$, as well as a witness length $\ell_k\in \Zz$ and alignment $\A_k : S_k \onto \hat{Q}\fragmentco{0}{\ell_k}$ of cost at most $k$.
    We run this algorithm until $S_k=S$ holds for the first time, which indicates that $k=\edp{S}{\hat{Q}}=d$, and return the final alignment $\A_d$, reinterpreted as an alignment $\A:S \onto Q\fragmentco{x}{y}$ for $y\coloneqq x+\ell_d$.

    A minor subtlety is that \cite[Lemma 6.1]{unified} provides a slightly weaker representation of the alignment $\A_d$, with $(i,j)$ for an substitution of $S\position{i}$ to $\hat{Q}\position{j}$, $(i,\bot)$ for a deletion of $S\position{i}$, and $(\bot,j)$ for an insertion of $\hat{Q}\position{j}$.
    Nevertheless, a left-to-right scan of this representation, keeping track of the shift $\delta$, equal to the number insertions processed so far minus number of deletions processed so far, lets us produce the breakpoints of $\A$, that is, $(-1,x-1,\bot)$ at the beginning, $(i,j+x,{\tt SUB})$ for each substitution $(i,j)$, $(i, i+\delta+x, {\tt DEL})$ for each deletion $(i,\bot)$, $(j-\delta, j+x, {\tt INS})$ for each insertion ($\bot, j)$, and $(|S|,y,\bot)$ at the end.

    In the \modelname model, the algorithm of \cite[Lemma 6.1]{unified} takes $\Oh(1)$ preprocessing time and $\Oh(1+k)$ time for the $k$th step (for $k\in \Zz$).
    Thus, running it for all $k\in \fragment{0}{d}$ costs $\Oh(1+d^2)$ time in total.
    The post-processing of the alignment takes $\Oh(1+d)$ extra time.
\end{proof}

\subsection{An Overview of an $\Oh(k^{4})$-Time Algorithm for Pattern Matching with Edits in the \modelname Model}
\label{sec:oldalgo}

Before we discuss the new algorithm for pattern matching with edits,
we briefly recall the algorithm presented
in~\cite{unified} that runs in \(\Oh(k^4)\) time in the \modelname model when $n = \cO(m)$.\footnote{This translates to a running time of \(\Oh(n+
    k^4)\) in the standard setting (under the assumption that $n = \cO(m)$).}
Thereby, we can naturally introduce required notations and concepts.
Further, we highlight the bottleneck in the aforementioned algorithm; all other parts
run in \(\Oh(k^3)\) time in the \modelname model, which translates to a running
time of $\Oh(n + k^3)$ in the standard setting.

Toward obtaining algorithms for \PMWE,
first observe that it suffices to focus on a bounded-ratio version
of said problem, where $n < \threehalfs m + k$ (consult~\cite{unified} for a rigorous
proof):
Given a text of arbitrary length $n$, we can reduce the problem to $\floor{2n/m}$ instances of the bounded-ratio version of the problem
using the so-called \emph{standard trick}~\cite{Abrahamson}. 
That is,
we consider the overlapping fragments
$T_i \coloneqq  T\fragmentco{\floor{i\cdot {m}/2}} {\min\{n, \floor{(i+3)\cdot {m}/2} + k - 1\}}$,
for \(i \in \fragmentco{0}{\floor{2n/m}-1}\),
and compute the
approximate occurrences of $P$ in each of $T_0, \dots, T_{\floor{2n/m}-1}$.
Afterward, we straightforwardly merge the obtained partial results.
Thus, an algorithm that solves the bounded-ratio version of the
approximate pattern matching problem in time $f(k)$ in the \modelname model,
directly yields an algorithm that solves an arbitrary instance of the
approximate pattern matching problem in time $\cO(n/m)\cdot f(k)$ in the \modelname model.

Next, as a first step in solving the bounded-ratio version of
\PMWE,
we analyze the pattern according to \cite[Lemma 6.4]{unified}.\footnote{When citing specific
statements of~\cite{unified}, we do so using their numbering in the full (arXiv) version of
the paper.} 

\begin{lemma}[{\tt Analyze($P$, $k$)},~{\cite[Lemma 6.4]{unified}}]\label{prp:EIalg}
    Let $P$ denote a string of~length $m$ and let $k \le m$ denote a positive integer.
    Then, there is an algorithm that computes one of~the following:
    \begin{enumerate}[(a)]
        \item $2k$ disjoint breaks $B_1,\ldots, B_{2k} \substr P$,
            each having period $\per(B_i)> m/\alphav k$ and length $|B_i| = \lfloor
            m/\betav k\rfloor$.
        \item Disjoint repetitive regions $H_1,\ldots, H_{r} \substr P$
            of~total length $\sum_{i=1}^r |H_i| \ge \deltavN/\deltavD \cdot m$ such
            that each region~$H_i$ satisfies
            $|H_i| \ge m/\betav k$ and is constructed along with a primitive approximate
            period $Q_i$
            such that $|Q_i| \le m/\alphav k$ and $\edl{H_i}{Q_i} = \ceil{\betav k/m\cdot
                |H_i|}$.\footnote{Observe that we renamed the repetitive regions to
            \(H_{\star}\), as we use \(R_{\star}\) later with a different meaning.}
        \item A primitive approximate period $Q$ of~$P$
            with $|Q|\le m/\alphav k$ and $\edl{P}{Q} < \betav k$.\label{item:alm-per}
    \end{enumerate}
    \noindent The algorithm uses $\Oh(k^2)$ time plus $\Oh(k^2)$ \modelname
    operations.\lipicsEnd
\end{lemma}

For the case where the analysis of the pattern using \cref{prp:EIalg} yields an approximate
period $Q$, we use the algorithm encapsulated in the following lemma.

\begin{fact}[{\tt PeriodicMatches($P$, $T$, $k$, $d$, $Q$)},~{\cite[Lemma 6.11]{unified}}]\label{lm:impEdC}
    Let $P$ denote a pattern of length $m$ and let $T$ denote a text of length $n$.
    Further, let $k \in \fragment{0}{m}$ denote a threshold, let $d\ge 2k$ denote a positive
    integer, and let $Q$ denote a primitive string that satisfies $|Q|\le m/8d$ and
    $\edl{P}{Q} \le d$.

    Then, we can compute a representation of the set $\OccE_k(P,T)$ as $\cO(n/m\cdot d^3)$ disjoint arithmetic progressions
    with difference $|Q|$
    using $\Oh(n/m\cdot d^4)$ time in the \modelname model.\footnote{The representation of $\OccE_k(P,T)$ is not stated explicitly in \cite[Lemma 6.11]{unified}.}\lipicsEnd
\end{fact}

As the main contribution of this work, we reduce the running time in \cref{lm:impEdC}
to $\tOh(n/m\cdot d^{3.5})$.

For the case where the analysis of the pattern using \cref{prp:EIalg} yields $2k$ disjoint breaks,
we use the efficient algorithm encapsulated in the following lemma.

\begin{fact}[{\tt BreakMatches($P$, $T$, $\{B_1,\dots,B_{2k}\}$, $k$)},~{\cite[Lemmas 5.21 and 6.12]{unified}}]\label{lm:impEdA}
    Let $k$ denote a threshold and
    let $P$ denote a pattern of~length $m$ having $2k$ disjoint breaks $B_1,\dots,B_{2k}
    \substr P$ each satisfying $\per(B_i) \ge m / \alphav k$.
    Further, let $T$ denote a text of~length $n < \threehalfs m + k$.
    Then, we can compute the set $\OccE_k(P, T)$, which is of size $\cO(k^2)$, using $\Oh(k^3)$
    time in the \modelname model.\lipicsEnd
\end{fact}


Finally, let us consider the case where the analysis of the pattern using~\cref{prp:EIalg}
returns disjoint repetitive regions $H_1,\ldots, H_{r} \substr P$
of~total length $\sum_{i=1}^r |H_i| \ge {}^{\deltavN}\!/\!{}_{\deltavD}\, m$ such
that each region~$H_i$ satisfies
$|H_i| \ge m/\betav k$ and is constructed along with a primitive approximate
period $Q_i$
such that $|Q_i| \le m/\alphav k$ and $\edl{H_i}{Q_i} = \ceil{\betav k/m\cdot
|H_i|}$.

In this case, we give a brief overview of the proof of~{\cite[Lemmas 5.24 and 6.13]{unified}}, stated below,
showing that the problem in scope reduces to several calls to the \texttt{PeriodicMatches}
procedure of \cref{lm:impEdC}.

\begin{fact}[{\tt RepetitiveMatches($P$,$T$,$\{ (H_1, Q_1), \dots, (H_{r},Q_r)\}$,$k$)},~{\cite[Lemmas 5.24 and 6.13]{unified}}]\label{lm:impEdB}
    Let $P$ denote a pattern of~length~$m$
    and let $k \le m$ denote a threshold.
    Further, let $T$ denote a string of~length~$n < \threehalfs m + k$.
    Suppose that $P$ contains disjoint repetitive regions $H_1,\ldots, H_{r}$
    of~total length at least $\sum_{i=1}^r |H_i| \ge {}^{\deltavN}\!/\!{}_{\deltavD}\, m$
    such that each region $H_i$ satisfies $|H_i| \ge m/\betav k$ and has a
    primitive approximate period~$Q_i$
    with $|Q_i| \le m/\alphav k$ and $\ed(H_i,Q_i^*) = \ceil{\betav k/m\cdot |H_i|}$.

    Then, we can compute the set $\OccE_k(P,T)$, which is of size $\cO(k^2)$, using $\Oh(k^4)$ time in the \modelname model.
\end{fact}
\begin{proof}[Proof sketch]
For each repetitive region $H_i$, set $k_i \coloneqq \floor{\betavh \cdot k/m \cdot
|H_i|}$
and $d_i \coloneqq \ceil{\betav\cdot k/m \cdot |H_i|}$.
We can compute $\OccE_{k_i}(H_i,T)$ using a call
$\texttt{PeriodicMatches}(H_i,T,k_i,d_i,Q)$; consult \cite{unified} for a rigorous proof
that the conditions in the statement of \cref{lm:impEdC} are indeed satisfied.

Next, using the sets $\OccE_{k_i}(H_i,T)$ we can identify in $\cO(k^2 \log \log
k)$ time
$\cO(k)$ length-$k$ intervals whose union is a superset of $\OccE_{k}(P,T)$.
Finally, we use {\tt Verify($P$, $T$, $k$, $J$)} for each such interval $J$ to filter out
false-positive positions; these calls to {\tt Verify} take $\cO(k^3)$ time in
total in the \modelname model.
\end{proof}

All in all, we see that the pattern matching with edits problem reduces to several calls
to the \texttt{PeriodicMatches} procedure:
\begin{itemize}
    \item If we can detect breaks in the pattern, no calls to \texttt{PeriodicMatches} are
        needed.
    \item If the pattern is close to being periodic, a single call
        {\tt PeriodicMatches($P$, $T$, $k$, $d$, $Q$)} suffices. In this case, we have $d=\cO(k)$
        and hence we obtain a representation of $\OccE_k(P,T)$ that consists of $\cO(k^3)$ disjoint arithmetic progressions
    		with difference $|Q|$.
    \item If $P$ contains disjoint repetitive regions $H_1,\ldots, H_{r}$
        of~total length at least $\sum_{i=1}^r |H_i| \ge \deltavN/\deltavD\cdot m$
        such that each region $H_i$ satisfies $|H_i| \ge m/\betav k$ and has a
        primitive approximate period~$Q_i$
        with $|Q_i| \le m/\alphav k$ and $d_i \coloneqq \ed(H_i,Q_i^*) = \ceil{\betav k/m\cdot
        |H_i|}$.
        Then, we make \(r\) calls ${\tt PeriodicMatches}(H_i,T,k_i,d_i,Q_i)$.
\end{itemize}
Not that if the pattern is not close to being periodic, the set $\OccE_k$ is of size $\cO(k^2)$ due to
\cite[Lemmas 5.21 and 5.24]{unified}.
Thus, in all cases, we obtain a representation of $\OccE_k(P,T)$ as $\cO(k^3)$ arithmetic progressions with the same difference.
Observe that the common difference of arithmetic progressions
is solely dependent on $P$ and $k$: it is $|Q|$ if \cref{prp:EIalg} returns an approximate period $Q$,
while it can be set to an arbitrary positive integer otherwise as the $\cO(k^2)$ positions of $\OccE_k(P,T)$ are interpreted
as arithmetic progressions of length one in this case.
We summarize the above discussion in the following statement.

\begin{fact}\label{cor:red_to_sync2}
    Let $P$ denote a pattern of~length~$m$
    and let $k \le m$ denote a threshold.
    Further, let $T$ denote a string of~length~$n < \threehalfs m + k$.

    We can compute a representation of $\OccE_k(P,T)$ as $\cO(k^3)$ disjoint arithmetic progressions with the same difference
    in $\cO(k^3)$ time in the \modelname model
    and additionally several calls $\texttt{PeriodicMatches}(P_i,T,k_i,d_i,Q_i)$
    such that $\sum_i |P_i| \le m$ and $d_i=\ceil{8k/m\cdot |P_i|}$.
	The common difference of said arithmetic progressions can be computed given only $P$ and $k$.\lipicsEnd
\end{fact}

\clearpage
\partn{From {\SM} to {\DPM}}

\section{The {\SM} Problem}\label{sec:smnew}

As a first step toward an improved implementation of the \texttt{PeriodicMatches}
procedure, we recall useful tools from \cite{unified}---this also allows us to slightly
simplify the problem statement of the problem we need to solve.
Then, we also give a first algorithm that solves said problem fast in a (very restricted)
special case; however, this algorithm turns out to be useful in general as well.

\subsection{Computing Occurrences in the Periodic Case: Preprocessing and Simplifications}\label{sec:redsm}

In this (sub-)section, we exploit \cref{cor:red_to_sync2} to reduce an instance of the \PMWE
to several instances of the \SM problem defined in \cref{sec:techov},
which we restate here for convenience.

\SMproblem

Sepcifically, we prove the following reduction.

\reduction*

As before, we use the standard trick to reduce the problem to $\cO(n/m)$ bounded-ratio instances:
specifically, for each $j \in \fragmentco{0}{\floor{2n/m}-1}$, an instance
$\PMWE(P,T_j,k)$ with $|T_j| < \threehalfs m + k$
and, without loss of generality, $|T_j|\ge m-k$
(otherwise, $\OccE_k(P,T_j)=\emptyset$).

Let us fix some $j \in \fragmentco{0}{\floor{2n/m}-1}$.
We apply \cref{cor:red_to_sync2} to reduce $\PMWE(P,T_j,k)$, in $\cO(k^3)$ time in the \modelname model,
to several calls $\texttt{PeriodicMatches}(P_{j,i},T_j,k_{j,i},d_{j,i},Q_{j,i})$
with $\sum_i |P_{j,i}| \le m$ and $d_{j,i}=\ceil{8k/m\cdot |P_{j,i}|}$.
Over all $j$, the total length of the patterns in the constructed instances of \SM is thus
$\sum_{j=0}^{\floor{2n/m}-2} \sum_i |P_{j,i}|\leq \sum_{j=0}^{\floor{2n/m}-2} m = \cO(n)$.
As shown in the remainder of \cref{sec:redsm}, answering a call $\texttt{PeriodicMatches}(P_{j,i},T_j,k_{j,i},d_{j,i},Q_{j,i})$
reduces, in $\cO(d_{j,i}^2)$ time in the \modelname model, to solving an instance $\SM(P_{j,i},T_j,k_{j,i},d_{j,i},Q_{j,i},\A_{P_{j,i}},\A_{T_{j}})$.
Observe that then, for a fixed $j \in \fragmentco{0}{\floor{2n/m}-1}$, we have
\[\sum_i \cO(d_{j,i}^2)= \sum_i \cO(k^2/m^2\cdot |P_{j,i}|^2)=\sum_i \cO(k^2/m \cdot |P_{j,i}|)=\cO(k^2).\]
Hence, over all $j \in \fragmentco{0}{\floor{2n/m}-1}$, the total time required for the reduction in the \modelname model is $\cO(n/m \cdot k^3)$,
dominated by the time required for the calls to the algorithm underlying \cref{cor:red_to_sync2}.
Finally, observe that for each $T_j$, we obtain a representation of $\OccE_k(P,T_j)$ as $\cO(k^3)$ arithmetic progressions with the same difference,
which is computable from $P$ and $k$.
As $P$ and $k$ are common in all instances $\PMWE(P,T_j,k)$, all computed arithmetic progressions have the same difference.
In linear time in their number, we can linearly scan them in order to merge overlapping ones, thus ensuring that the final output consists of
$\cO(n/m \cdot k^3)$ disjoint arithmetic progressions with the same difference.

In the remainder of \cref{sec:redsm} we show that answering a call $\texttt{PeriodicMatches}(P,T,k,d,Q)$
reduces, in $\cO(d^2)$ time in the \modelname model, to solving an instance $\SM(P,T,k,d,Q,\A_P,\A_T)$.

We first use the following (simplified) version of \cite[Lemma 6.8]{unified};
it is readily verified that all conditions of \cref{lem:erelevant} are satisfied in this call.

\begin{fact}[\texttt{FindRelevantFragment($P$, $T$, $k$, $d$, $Q$)},~{Compare \cite[Lemma 6.8]{unified}}]\label{lem:erelevant}
    Let $P$ denote a pattern of~length $m$, let $T$ denote a text of~length $n$,
    and let $0 \le k\le m$ denote a threshold such that $n<\threehalfs m+k$.
    Further, let $d\ge 2k$ denote a positive integer and let
    $Q$ denote a primitive string that satisfies $|Q|\le m/8d$ and $\edl{P}{Q}\le d$.

    Then, there is an algorithm that computes a fragment $T'$ of $T$ such that
    $\edl{T'}{Q}\le 3d$ and $|\OccE_k(P,T)|=|\OccE_k(P,T')|$.
    The algorithm runs in $\Oh(d^2)$ time in the \modelname model.
    \lipicsEnd
\end{fact}

A simplified version of \cite[Lemma 6.5]{unified} lets us pick $Q$ so that $\edl{P}{Q} = \edp{P}{Q}$.
\SetKwFunction{witness}{FindAWitnes}

\begin{fact}[\texttt{FindAWitness($k$, $Q$, $S$)},~{Compare \cite[Lemma 6.5]{unified}}]\label{lem:witness}
    Let $k$ denote a positive integer,
    let $S$ denote a string,
    and let $Q$ denote a primitive string that satisfies $|S|\ge (2k+1)|Q|$.

    Then, we can compute a \emph{witness} $x\in \Zz$
    such that $\edp{S}{\rot^{-x}(Q)}=\edl{S}{Q}\le k$,
    or report that $\edl{S}{Q}>k$.
    The algorithm takes $\Oh(k^2)$ time in the \modelname model.
    \lipicsEnd
\end{fact}
We use \cref{lem:witness} as follows. First, we make a call $\texttt{FindAWitness}(d, Q, P)$
to derive $x\in \Zz$ such that $\edp{S}{\rot^{-x}(Q)}=\edl{P}{Q}$.
Then, we make a call $\alignment(P,Q,x)$ to the function of \cref{fct:alignment},
which yields an optimum alignment $\A:P\onto Q^\infty\fragmentco{x}{y}$ of cost $\edl{P}{Q}$.
Next, we extract a fragment $Q'$ of $P$
that $\A$ matches without any edits against $Q^\infty\fragmentco{x'}{x'+|Q|}$
for some $x'\equiv x \bmod {|Q|}$. Finally, we replace $Q\coloneqq Q'$.

We conclude that it suffices to implement the $\texttt{PeriodicMatches}$ procedure under the following assumptions:
\begin{itemize}
  \item $m-k \le n < \threehalfs m + k$,
  \item $\edl{T}{Q}\le 3d$,
  \item $\edl{P}{Q}=\edp{P}{Q}$.
\end{itemize}
These are exactly the conditions of the \SM problem.

Next, we show how to efficiently obtain optimal alignments from \(P\) and \(T\) to \(\Q\),
thus completing the proof that of the fact that answering a call $\texttt{PeriodicMatches}(P,T,k,d,Q)$
reduces, in $\cO(d^2)$ time in the \modelname model, to solving an instance $\SM(P,T,k,d,Q,\A_P,\A_T)$.

\begin{lemma}\label{lem:apat}
    For any instance of the $\SM$ problem, we can construct
    optimal alignments $\A_P: P \onto Q^\infty\fragmentco{0}{y_P}$ of cost $d_P$
    (for some $y_P\in \Zz$)
    and $\A_T : T \onto Q^\infty{\fragmentco{x_T}{y_T}}$ of cost $d_T$
    (for some $x_T\in \fragmentco{0}{|Q|}$ and $y_T\in \Zz$)
    in $\Oh(d^2)$ time in the \modelname model.
\end{lemma}
\begin{proof}
    As for \(\A_P\), we call $\alignment(P,Q,0)$ of \cref{fct:alignment},
    which directly yields the alignment $\A_P: P \onto Q^\infty\fragmentco{0}{y_P}$
    of cost $d_P = \edp{P}{Q}$.

    As for $\A_T$, we first call $\witness(3d,Q,T)$ of \cref{lem:witness} to obtain a
    positive integer $x_T$
    such that $\edp{T}{\rot^{-x_T}(Q)}=d_T$.
    This call is valid as the properties of the \SM instance yield \[
    n \ge m-k \ge   8d|Q|-k \ge 7d|Q| \ge (6d+1)|Q|.\]
    Moreover, due to $\rot^{-x_T}(Q)=\rot^{-x_T \bmod |Q|}(Q)$, we may set
    $x_T \coloneqq x_T \bmod {|Q|}$
    and thus assume $x_T\in \fragmentco{0}{|Q|}$ without loss of generality.
    Next, we call $\alignment(T,Q,x_T)$ to obtain
    an alignment $T\onto Q^\infty\fragmentco{x_T}{y_T}$
    of cost $\edp{T}{\rot^{-x_T}(Q)} = d_T$.

    The overall running time in the \modelname model is $\Oh(1+d_P^2+d^2+d_T^2)=\Oh(d^2)$.
\end{proof}

\subsection{A First Algorithm for {\SM}}\label{3.2}
For the remainder of this section, we fix an instance of the {$\SM(P, T, d, k,
Q, \A_P, \A_T)$} problem and set {$\kappa \coloneqq k+d_P+d_T$},
and {$\tau \coloneqq q\ceil{{\kappa}/{2q}}$}.

In a first step toward algorithms for \SM, we discuss how the (almost-)periodicity of
\(P\) and \(T\) yield simple ways to \emph{filter out} many potential starting positions of
occurrences. In particular, this allows us to obtain a fast algorithm for {\SM} when \(q\)
is very large.

Let us briefly recall an example from \cref{sec:techov}.
Suppose that \(P\) and \(T\) are perfectly periodic with
period \(Q\), that is
\(P = Q^{\infty}\fragment{0}{m}\) and
\(T = Q^{\infty}\fragment{0}{n}\).
In particular, in this special case \(\A_P\) and \(\A_T\)
are identity maps and we have \(m = y_P\), \(0 = x_T\),
 \(n = y_T\), and \(d_T = d_P = 0\).
Now, clearly all occurrences start around the positions in \(T\) where an
exact occurrence of \(Q\) starts, that is, in the intervals \(\fragment{j q - k}{jq + k}\) for
\(j \in \mathbb{Z}\).
Now, if we have \(q > 2k\), then indeed \(jq + k < (j+1)q - k\) and we can
thus filter out positions where no occurrence may start.

Now, as we are dealing with \emph{almost} periodic strings \(P\) and \(T\), potential
edits in \(\A_T\) and \(\A_P\) widen the intervals of potential starting positions.
Returning to the more general \(T\) and \(P\) fixed at the beginning of this
(sub-)section and given an alignment $\A : P \onto T\fragmentco{v}{w}$ of
cost at most $k$, observe that $\A_T\circ \A \circ \A_P^{-1}$
induces an alignment $Q^\infty\fragmentco{0}{y_P} \onto U \coloneqq
\A^{-1}_T(T\fragmentco{v}{w})$
of cost at most $\kappa = k + d_T + d_P$.
Now, as~$Q$ is primitive and $Q^\infty\fragmentco{0}{y_P}$ is a long substring of $Q^\infty$,
we can conclude that $U$ must be a prefix of $(\rot^{r}(Q))^\infty$ for
$r \in \fragment{-\kappa}{\kappa}$.
Hence, $U$ starts at some position
$u \in \bigcup_{j\in \mathbb{Z}} \fragment{j q-x_T-\kappa}{j q-x_T+\kappa}$
of $Q^\infty\fragmentco{x_T}{y_T}$.
The fact that $\A_T$ is of cost $d_T$, yields that $v \in \bigcup_{j\in \mathbb{Z}} \fragment{j q-x_T-\kappa-d_T}{j q-x_T+\kappa+d_T}$.
Now, if we have \(q > 2\kappa + 2 d_T\), then indeed \(j q-x_T + \kappa + d_T
< (j+1) q-x_T-\kappa-d_T\) and we can
thus filter out positions where no occurrence may start.

\begin{lemma}\label{lm-tmp-4.2-2}
    Consider an alignment \(\A : P \onto T\fragmentco{v}{w}\) of cost at most \(k\)
    and write \(\Q\fragmentco{v'}{w'} \coloneqq
    \A_T(T\fragmentco{v}{w})\).
    Further, for an integer \(j\),
    we have \begin{align*}
        v  &\in \fragment{j q-x_T-\kappa-d_T}{j q-x_T+\kappa+d_T},\\
        v' &\in \fragment{\max\{x_T, jq - \kappa\}}{jq + \kappa},
        \quad\text{and}\\
        w' &\in \fragment{jq + y_P - \kappa}{\min\{y_T, jq + y_P + \kappa\}}.
    \end{align*}
    Further, we have \(
        |\Q\fragmentco{v'}{w'}| \in \fragment{y_P - \kappa}{y_P + \kappa}.
        \)
\end{lemma}
\begin{proof}
    We formalize the example from before. To that end, write
    \(\B \coloneqq \A_T\circ \A\circ\A_P^{-1}\) for the alignment of
    \(\Q\fragmentco{0}{y_P}\) to \(\Q\fragmentco{v'}{w'}\) induced by \(\A\);
    observe that \(\B\) has a cost of at most \(\kappa\). Observe that this immediately
    yields the claimed bound on the size of \(\Q\fragmentco{v'}{w'}\).

    Now, observe that the fragments of \(\Q\) that \(\B\) aligns to each other are
    both long; hence at
    least one full occurrence of \(Q\) is matched exactly under \(\B\).
    \begin{claim}\label{cl-tmp-4.2-1}
        There is a fragment \(Q_i \coloneqq \Q\fragmentco{iq}{(i+1)q}\) of
        \(\Q\fragmentco{0}{y_P}\) with \(\B(Q_i) = Q_i = Q\).
    \end{claim}
    \begin{claimproof}
        As \(P\), \(T\), \(d\), \(k\), and \(Q\) stem from our fixed instance of \(\SM\),
        we have \[
            8dq \le m\quad\text{and}\quad d_P \le d\quad\text{and}\quad \kappa \le
            5d\quad\text{and}\quad 0 < d.
        \] Combined, we hence also have \[
            |\Q\fragmentco{0}{y_P}| = y_P \ge m - d_P \ge 8dq - d \ge 7dq \ge (\kappa + 2)q.
        \]
        In particular, the fragment \(\Q\fragmentco{0}{y_P}\) contains at least \((\kappa +
        1)\) full repetitions of \(Q\). As the alignment \(\B\) makes at most \(\kappa\)
        edits, at least one full occurrence of \(Q\) is hence aligned without edits.
    \end{claimproof}

    Now, observe that the fragment \(\B(Q_i)\) from \cref{cl-tmp-4.2-1} starts at some
    position \(i'q\) in \(\Q\)---as \(\B\) allows for at most \(\kappa\) insertions and
    deletions, we hence obtain that \(v' \in \fragment{(i'-i)q - \kappa}{(i'-i)q + \kappa}\).
    Symmetrically, we obtain \(w' \in \fragment{(i'-i)q + y_P - \kappa}{(i'-i)q + y_p
    + \kappa}\). Further, we observe that \(\B\) cannot extend beyond \(\A_T\).
    Finally, as $\A_T(T\fragmentco{v}{w}) = \Q\fragmentco{v'}{w'}$ and $\A_T$ is of cost $d_T$,
    we also have that $v \in \fragment{(i'-i)q - \kappa - d_T}{(i'-i)q + \kappa +d_T}$,
    thus completing the proof.
\end{proof}

Observe that expanding the minimum and maximum expressions
from \cref{lm-tmp-4.2-2} easily yields bounds on the possible values of \(j\); we have
\begin{equation*}
    j \in \bar{J} \coloneqq  \fragment{\ceil{{(x_T-\kappa)}/{q}}}{\floor{{(y_T-y_P+\kappa)}/{q}}}
    =\{j\in \mathbb{Z} \mid x_T \le jq+\kappa  \text{ and } jq + y_P - \kappa \le y_T \}.
\end{equation*}

Now, intuitively, we wish to process each fragment
\(T\fragment{jq - x_T-\kappa-d_T}{j q-x_T+\kappa+d_T}\) separately (each corresponding to
a single \(j \in \bar{J}\)).
However, if \(q\) is very small (that is, if \(q < \kappa / 2\)), this results in
too many fragments to process.
Hence we group them into longer fragments of length at
least \(\tau \coloneqq q\ceil{\kappa/2q} \ge q(\kappa/2q) = \kappa / 2\);
observe that \(\tau\) is an integer multiple of \(q\) and that we have
\begin{alignat}{2}\label{eq:tau-is-q}
    \tau &= q\ceil{{\kappa}/{2q}} = q, &&\quad\text{if \(q \ge \kappa / 2\) and}\\
    \label{eq:tau-is-less-k}
    \tau &\le q( 1 + \kappa / 2q) < \kappa,
         &&\quad\text{if \(q < \kappa / 2\).}
\end{alignat}
As discussed in the introductory example,
for two consecutive integers \(j\) and \(j + 1\),
the intervals \(\fragment{jq - \kappa}{jq + \kappa}\) and \(\fragment{(j+1)q -
\kappa}{(j+1)q + \kappa}\) are disjoint only if \(q > 2\kappa\)---in which case \(\tau =
q\). Hence, the definition of \(\tau\) agrees with \cref{lm-tmp-4.2-2}.

In total, this leads to the following definition.

\begin{definition}\label{def:jdef}
    Partition \(\Q\) into blocks of length \(\tau\) and number them starting from \(0\).
    For the \(j\)-th block \(Q_j \coloneqq \Q\fragmentco{j\tau}{(j + 1)\tau}\),
    we define the \emph{interesting region \(\bar{Q}_j\) of \(Q_j\)} as the fragment
    \(\bar{Q}_j \coloneqq \Q\fragment{j\tau - \kappa}{j\tau + \kappa}\).
    Now, write \(J\) for the set of all numbers of blocks whose interesting regions
    overlap \(\Q\fragmentco{x_T}{y_T - y_P + \kappa}\).
    Formally, we set
    \[J \coloneqq \fragment{\ceil{{(x_T-\kappa)}/{\tau}}}{\floor{{(y_T-y_P+\kappa)}/{\tau}}}
    =\{j\in \mathbb{Z} \mid x_T \le j\tau+\kappa  \text{ and } j\tau + y_P -\kappa \le y_T\}.\]
    Now, for each $j\in J$, we write \(R_j \coloneqq T\fragmentco{r_j}{r'_j}\) for the
    largest fragment of \(T\) that corresponds to the interesting region \(\bar{Q}_j\);
    that is, we set
    \begin{align*}
        r_j &\coloneqq \min\{a_T \mid (a_T,a_Q)\in \A_T
        \text{ for }  j\tau-\kappa \le a_Q\}
        \quad\text{and}\\
            r'_j &\coloneqq \max\{a_T \mid (a_T,a_Q)\in \A_T\text{ for } a_Q \le
             j\tau+y_P + \kappa\}.
    \end{align*}
    For convenience, we also set \(r_{\max J + 1} \coloneqq n\).
        \lipicsEnd
\end{definition}
\begin{remark}\label{rem:rj-corner}
    For convenience, we write
    \begin{align*}
        \A_T(r_j) &\coloneqq \min\{ a_Q \mid (r_j, a_Q) \in \A_T \text{ and }
        j\tau-\kappa \le a_Q\}\quad \text{and }\\
        \A_T(r'_j) &\coloneqq \max\{ a_Q \mid (r'_j, a_Q) \in \A_T \text{ and }  a_Q
        \le j\tau+y_P + \kappa \}.
    \end{align*}

    Observe that for \emph{most} of \(j \in J\), we have
    \(\A_T(r_j) = j\tau - \kappa\) and \(\A_T(r'_j) = j\tau + y_P + \kappa\).
    In particular, whenever \(\A_T(r_j) > j\tau - \kappa\), that is, if
    \(j\in \fragment{\ceil{({x_T-\kappa})/{\tau}}}{\floor{({x_T+\kappa})/{\tau}}}\),
    then \(j\tau - \kappa\) is a position before the first position that \(\A_T\) maps to:
    we have \(j\tau - \kappa \le x_T\) and hence \(\A_T(r_j) = x_T\),
    which in turn implies \(r_j = 0\).

    Similarly,
    whenever \(\A_T(r'_j) < j\tau + y_P + \kappa\),
    that is, if \(j\in
    \fragment{\ceil{(y_T-y_P-\kappa)/\tau}}{\floor{(y_T-y_P+\kappa)/{\tau}}}\),
    we have \(y_T < j\tau + y_P + \kappa\) and hence \(\A_T(r'_j) = y_T\),
    which in turn implies \(r'_j = n\).
    \lipicsEnd
\end{remark}
\begin{remark}\label{rem:negj}
    It is useful to number the blocks before the zeroth block with negative
    indices: observe that for \(x_T < \kappa\), we may have \((\min J) < 0\).
    In particular, we see that \(
    (\min J) \ge \ceil{{-\kappa}/{\tau}} \ge -2\), so we may indeed need to address a
    \(-1\)-st and \(-2\)-nd block.

    Similarly, \((\max J)\) may exceed the natural barrier of \((y_T - y_P) / \tau\).
    In particular, we see that
    \( (\max J) \le \ceil{y_T / \tau} - \ceil{y_P / \tau} + \kappa / \tau \le \ceil{y_T /
    \tau} - \ceil{y_P / \tau} + 2\).
    \lipicsEnd
\end{remark}
\begin{remark}\label{lem:rj-length}
    It is useful to upper bound the length of each fragment \(R_j\):
    fix a \(j \in J\) and consider the fragment \(\A_T( R_j )\) of \(\Q\).
    By construction, \(\A_T( R_j )\) has a length of\[
        |\A_T( R_j )| = (y_P + j\tau + \kappa) - (j\tau - \kappa) = y_P +  2\kappa.
    \] Now, as \(\A_T\) has a cost of \(d_T\), we conclude that the fragment \(R_j\) has a
    length of at most
    \begin{align*}
        |R_j| = r'_j - r_j \leq |\A_T( R_j )| + d_T  = y_P + 2\kappa + d_T \le m + 3
        \kappa - k.\tag*{\lipicsEnd}
    \end{align*}
\end{remark}

Intuitively, we want to think of \(R_j\) as the fragment of \(T\) that contains
all \(k\)-error occurrences of \(P\) that start at a
position in \(T\) that corresponds to a position in \(\bar{Q}_j\).
Formally, this requires a proof as a fragment~\(R_j\) could (in theory) be too short to
fully contain such an occurrence.

\begin{lemma}\label{lem:aligned-proto}
    Fix a \(j \in J\) and a position \(v \in \fragmentco{r_j}{r_{j + 1}}\).
    For any position \(w > r_j'\), any
    alignment \(\A: P \onto T\fragmentco{v}{w}\)
    has a cost of at least \(k + 1\).
\end{lemma}
\begin{proof}
    We prove the contraposition.
    To that end, fix an alignment
    \(\A : P \onto T\fragmentco{v}{w}\)
    with a cost of at most \(k\).
    We intend to show that \(w \le r_j'\).

    To that end, write \(\Q\fragmentco{\A_T( v )}{\A_T( w )} \coloneqq
    \A_T(T\fragmentco{v}{w})\) and consider the position \(\A_T( v )\).
    First, suppose that \(q \ge \kappa / 2\) and, in particular, \(\tau = q\).
    Now, the definition of \(r_j\) yields\footnote{Observe that if \(j\) is equal to
        \(\max J\), then by \cref{lm-tmp-4.2-2}, for \(\A_T( v ) \in \fragment{(j+1)\tau
        - \kappa}{y_T}\) the interval of possible positions for \(w'\) is empty. Hence,
    also in this corner case, we can safely assume \(\A_T( v ) \le (j+1)\tau\).}\[
        \A_T( v ) \in
        \fragmentco{ j\tau - \kappa }{\min\{y_T, (j + 1)\tau - \kappa\}}
        \subseteq \fragmentco{ jq - \kappa }{(j + 1)q - \kappa}
    \]
    Applying \cref{lm-tmp-4.2-2},
    we see that \(\A_T( v ) \in \fragment{jq- \kappa}{jq + \kappa}\) and
    \[
        \A_T(w) \le \min\{y_T, jq + y_P + \kappa\} \le jq + y_P + \kappa
        = j\tau + y_P + \kappa
    \] Hence, by definition \(w \le r_j'\).

    Next, suppose that \(q < \kappa/2\). Now, using \cref{eq:tau-is-q}, we see
    that\footnote{Again, if \(j\) is \(\max J\), then by \cref{lm-tmp-4.2-2} with \(\tau
        \ge q\), we may assume \(\A_T( v ) < (j+1)\tau\).}\[
        \A_T( v ) \le (j + 1)\tau - \kappa \le j\tau + \kappa - \kappa.
    \] Now, using \cref{lm-tmp-4.2-2}, we see that \(|\A_T(T\fragmentco{v}{w})|
    \le y_P + \kappa\) and hence\[
        \A_T(w) \le \A_T(v) + |\A_T(T\fragmentco{v}{w})|
        \le j\tau + y_P + \kappa.
    \] Thus, we have \(w \le r_j'\); completing the proof.
\end{proof}

In particular, \cref{lem:aligned-proto} confirms our earlier intuition: indeed, the
fragment \(R_j\) fully contains all \(k\)-error occurrences of \(P\) that start at a
position in  \(T\) that corresponds to a position in \(\bar{Q}_j\).

As another consequence of \cref{lem:aligned-proto}, we see that the set \(\OccE_k(P,T)\)
decomposes over~\(J\)~and~\(R_j\).
\begin{corollary}\label{lem:aligned}
    For any position \(v \in \OccE_k(P, T)\) and a corresponding alignment
    \(\A : P \onto T\fragmentco{v}{w}\), there is a  (unique) \(j \in J\)
    such that \(v \in \fragmentco{r_j}{r_{j + 1}}\) and \(w \le r_j'\);
    that is, we have
    $\OccE_k(P, T)=\bigcup_{j\in J} \big(\OccE_k(P, R_j) + r_j \big)$.
\end{corollary}
\begin{proof}
    Recall that \(R_j = T\fragmentco{r_j}{r_j'}\).
    Now, using \cref{lem:aligned-proto}, we see that for each \(j \in J\),
    the set \(\OccE_k(P, R_j)\) (shifted by \(r_j\)) contains all \(k\)-error
    occurrences of \(P\) that start in \(T\fragmentco{r_j}{r_{j+1}}\).

    Finally, we need to argue that no \(k\)-error occurrence may start at a position in
    \(T\) not covered by the fragments \(R_j\). To that end, first observe that \(r_{\max
    J + 1} = n\), so we may lose an occurrence only if it starts before \(r_{\min J}\) in
    \(T\). Set \(j^* \coloneqq \min J\) and observe that we are done if \(j^* = 0\), as
    \(r_0 = 0\). Now, for a positive~\(j^*\), if \(q \ge \kappa / 2\) (and hence \(\tau =
    q\)), we see that for any positive integer \(\alpha\),
    we have \((j^* - \alpha)q + \kappa < x_T\)---hence (by \cref{lm-tmp-4.2-2})
    no \(k\)-error occurrence starts in \(T\) before \(r_{j^*}\).

    Finally, if \(q < \kappa / 2\) (and hence \(\tau < \kappa\) by
    \cref{eq:tau-is-less-k}); we have\[
        j^* \tau - \kappa \le \left(1 + \frac{x_T - \kappa}{\tau}\right) \tau - \kappa
                  < x_T - \kappa \le x_T.
    \] Hence, we have \(\A_T( r_{j^*} ) = x_T\) and hence \(r_{j^*} = 0\); completing the
    proof.
\end{proof}

Next, we discuss how to compute the sequences \((r_j)\) and \((r'_j)\).

\begin{lemma}\label{fct:rj}
    Given the alignment $\A_T$, in \(\Oh(d_T + 1)\) time, we can construct the sequences
    $(r_j)_{j\in J}$ and $(r'_j)_{j\in J}$, represented as concatenations of $\Oh(d_T+1)$
    arithmetic progressions with difference $\tau$.
\end{lemma}
\begin{proof}
    We intend to employ \cref{fct:alignments} to compute the first positions in \(T\) that
    \(\A_T\) maps to a position \(j\tau - \kappa\) and the last positions in \(T\) that
    \(\A_T\) maps to a position \(j\tau + y_P + \kappa\). Thereby, we immediately obtain
    the desired entries of \(r_j\) when \(\A_T(r_j) = j\tau - \kappa\) (by setting
    \(r_j \coloneqq \brpst{\A_T}{j\tau - \kappa}\)) and
    the desired entries of \(r'_j\) when \(\A_T(r'_j) = j\tau + y_P - \kappa\)
    (by setting \(r'_j \coloneqq \brped{\A_T}{j\tau + y_P - \kappa}\)).
    Recalling \cref{rem:rj-corner}, we see that the remaining entries of \((r_j)\) are those
    that correspond to a (short) prefix of \(T\) (for which \(r_j = 0\)) and we see that the
    remaining entries of \((r'_j)\) are those that correspond to a (short) suffix of \(T\)
    (for which \(r'_j = n\)).

    Formally, we use \cref{fct:alignments} to compute the sequences
    \((\brpst{\A_T}{v})_{v = x_T}^{y_T}\) and
    \((\brped{\A_T}{w})_{w = x_T}^{y_T}\)---which we obtain represented
    as \(\Oh(d_T + 1)\) arithmetic progressions with difference 1 (that is intervals) each.
    Now, for the sequence $(r_j)_{j\in J}$; for
    $j\in \fragment{\ceil{({x_T-\kappa})/{\tau}}}{\floor{({x_T+\kappa})/{\tau}}}$,
    we set \(r_j = 0\) (which we output as single-element arithmetic progressions).

    For the remaining values of \((r_j)\), that is for
    \(j \in \fragment{\ceil{({x_T+\kappa})/{\tau}}}{\floor{({y_T -y_P+\kappa})/{\tau}}}\),
    we process the intervals that represent \((\brpst{\A_T}{v})\) as follows: for
    each interval \(I\) representing \((\brpst{\A_T}{v})_{v = x_T}^{x_T + |I| - 1}\),
    we compute the intersection \(\fragmentco{x_T}{x_T + |I|} \cap (j\tau - \kappa)_j\) and
    return the corresponding entries of \(I\) as an arithmetic progression with difference
    \(\tau\).
    For \((r'_j)_{j \in J}\) we proceed analogously.

    For the running time, computing \((\brpst{\A_T}{v})\) and \((\brpst{\A_T}{w})\) takes time
    \(\Oh(d_T + 1)\); post-processing the obtained intervals takes constant time per
    interval for a total running time of \(\Oh(d_T + 1)\).
    Next, observe that as \(\tau \ge \kappa / 2\), we have
    \[|\fragment{\ceil{({x_T-\kappa})/{\tau}}}{\floor{({x_T+\kappa})/{\tau}}}| \le 2\kappa
    / \tau \le 4.\]
    In particular we can handle the initial values of \((r_j)\) in
    constant time and we output a constant number of extra arithmetic progressions.
    Similarly, we have
    \[|\fragment{\ceil{(y_T-y_P-\kappa)/\tau}}{\floor{(y_T-y_P+\kappa)/{\tau}}}| \le
    2\kappa/ \tau \le 4.\]
    In particular we can handle the final values of \((r'_j)\) in
    constant time  and we output a constant number of extra arithmetic progressions.
    Hence in total, the algorithm runs in time \(\Oh(d_T + 1)\) and produces \(\Oh(d_T +
    1)\) arithmetic progressions, each with difference \(\tau\);\footnote{Observe that
    this is also true for single-element arithmetic progressions.} thus completing the
    proof.
\end{proof}

Finally, we present a first algorithm for \SM, which is particularly useful for small sets
\(J\), that is, when \(q \gg k\).
Let us note that this algorithm is also used in the general case for the efficient computation of
$k$-error occurrences that start in heavy positions; see \cref{sus:heavy}.

\begin{lemma}\label{fct:verifyRj}
    For each $j\in J$, we can compute the set $\OccE_k(P,R_j)$
    in $\Oh(d^2)$ time in the \modelname model.
    In particular, we can solve the \SM problem in $\cO(d^2|J|)$ time in the \modelname model.
\end{lemma}
\begin{proof}
    Recall from \cref{lem:rj-length} that for each $j \in J$ we have $|R_j| \leq m+3\kappa-k$.
    Thus, by \cref{fct:toolbox,fct:rj}, the construction of $\OccE_k(P,R_j)$ using
    \texttt{Verify}$(P,R_j,k,\fragmentco{0}{|R_j|-m+k})$
    takes $\cO(\log(1+d_T)+k(k+\kappa))=\cO(d^2)$ time in the \modelname model.
    \Cref{lem:aligned} shows how the sets $\OccE_k(P,R_j)$ can be combined into $\OccE_k(P,T)$;
    this takes time $\Oh(\kappa|J|)=\Oh(d|J|)$.
    In total, we see that the overall running time in the \modelname model is
    dominated by the calls to $\texttt{Verify}$, which take time
     $\Oh(d^2|J|)$ in the \modelname model.
\end{proof}

\section{Using {\DPM} for Algorithms for {\SM}}\label{sec:DPMtile}

In this section, we formalize the connection between the \DPM problem and the \SM problem.
In particular, we define suitable puzzles for the strings \(P\) and \(T\); then we
show that in the setting of \SM, we can compute \(\OccE_k(P, T)\) with using very few {\tt
DPM-*} operations. To that end, we first define partitions of \(P\) and \(T\) (resulting
in \emph{tiles}) from which we then construct appropriate puzzles.

\begin{definition}[$\tau$-tile partition]\label{def:can_part}
    Consider a string $S$, a primitive string $Q$ of length $q$, an integer $\tau\in \Zp$
    divisible by $q$, and an alignment $\A_S : S \onto Q^\infty\fragmentco{x_S}{y_S}$, where
    $x_S\in \fragmentco{0}{q}$ and $y_S\ge x_S$.

    Partition \(\Q\) into blocks of length \(\tau\) and number them starting from \(1\).
    For the \(j\)-th block \(Q_j \coloneqq \Q\fragmentco{\max\{x_S, (j-1)\tau\}}{\min\{y_S,
    j\tau\}}\), we define the \emph{\(j\)-th tile of \(S\) (with respect to \(\A_S\))} as\[
        S\fragmentco{s_{i-1}}{s_i} \coloneqq \A_S^{-1}(Q_j).
    \]
    Further, we define the \emph{$\tau$-tile partition of $S$ with respect to $\A_S$}
    as the partition \[
    S = \bigodot_{i\in\Zp} S\fragmentco{s_{i-1}}{s_i}.\tag*{\lipicsEnd}\]
\end{definition}
\begin{remark}\label{rem:can_part}
    Observe that the \(\ceil{{y_S}/{\tau}}\)-th tile is the last non-empty tile of \(S\).
    Hence, writing \(\beta_S \coloneqq \ceil{{y_S}/{\tau}}\), we have
    \[S = \bigodot_{i=1}^{\beta_S} S\fragmentco{s_{i-1}}{s_i}.\tag*{\lipicsEnd}\]
\end{remark}
\begin{remark}\label{rem:can_part2}
    We observe that \(s_0 = 0\) and \(s_{\beta_S} = |S|\).
    We also see that
    that for all other \(j \in \fragmentoo{0}{\beta_S}\), we have \[
        s_i = \min\{a_S\in \fragment{0}{|S|} \mid (a_S,i\tau)\in \A_S\}.
    \]
    Recalling \cref{fct:alignments}, we hence have \(s_i = \brpst{\A_S}{i\tau}\).
    \lipicsEnd
\end{remark}

It is easy to verify that we can indeed efficiently
compute the \(\tau\)-tile partition of a string.

\begin{lemma}\label{fct:tile}
    Let $S=\bigodot_{i=1}^\beta S\fragmentco{s_{i-1}}{s_i}$ denote the $\tau$-tile partition
    of a string $S$ with respect to a cost-$d$
    alignment $\A_S : S \onto Q^\infty\fragmentco{x_S}{y_S}$.

    Given (the breakpoint representation of) $\A_S$, in \(\Oh(d + 1)\) time,
    we can construct the sequence $(s_i)_{i=0}^\beta$, represented as $\Oh(d+1)$
    arithmetic progressions with difference $\tau$.
\end{lemma}
\begin{proof}
    We use \cref{fct:alignments} to obtain the sequence \((\brpst{\A_S}{v})_{v=x_S}^{y_S}\) as a
    concatenation of \(\Oh(d+1)\) arithmetic progressions with difference 1.
    Next, in view of \cref{rem:can_part2},
    we extract and return the subsequence \((\brpst{\A_S}{j\tau})_{j = 1}^{\beta_S - 1} = s_j\);
    together with the boundary values of \(s_0 = 0\) and \(s_{\beta_S} = |S|\).

    For the running time, observe that it suffices to process each arithmetic progression in
    $(\brpst{\A_S}{v})_{v=x_S}^{y_S}$ in $\Oh(1)$ time.
    Further, each such arithmetic progression contributes a single (perhaps empty) arithmetic
    progression of difference $\tau$ to the output; to which we add exactly two further
    arithmetic progressions. In total, this yields the claim.
\end{proof}

As in \cref{3.2},
we fix an instance of the {$\SM(P, T, d, k, Q, \A_P, \A_T)$} problem
and we set {$\kappa \coloneqq k+d_P+d_T$},
and {$\tau \coloneqq q\ceil{{\kappa}/{2q}}$}.

Now additionally, let $P = \bigodot_{i=1}^{\beta_P} P\fragmentco{p_{i-1}}{p_i}$
denote the $\tau$-tile partition of $P$ with respect to $\A_P$,
and let $T=\bigodot_{i=1}^{\beta_T} T\fragmentco{t_{i-1}}{t_i}$ denote
the $\tau$-tile partition of $T$ with respect to $\A_T$.

Toward our goal of creating a suitable puzzle for \DPM, we first observe that
if $\beta_P$ is very small (that is if the tiles are very long),
then we can already efficiently solve the \SM problem.

\begin{lemma}\label{lem:smallbeta}
    We can solve \SM in time $\cO(d^2 \beta_P)$ in the \modelname model.
\end{lemma}
\begin{proof}
    In light of \cref{fct:verifyRj}, it suffices to prove that $|J|=\Oh(\beta_P)$.

    To that end, observe that by the definition of \(J\) (\cref{def:jdef}), we have
    \begin{align*}
        |J|-1 &= \floor{\frac{y_T-y_P+\kappa}{\tau}}-\ceil{\frac{x_T-\kappa}{\tau}}
              \le \frac{y_T-y_P-x_T+2\kappa}{\tau} =
              \frac{(y_T-x_T)-(y_P-0)+2\kappa}{\tau}\\
              \intertext{Rephrasing in terms of \(n\) and \(m\) yields}
        |J| - 1
              &\le \frac{n+d_T - m+d_P+2\kappa}{\tau}.\\
              \intertext{Using \(n - m \le m/2 + k\) (from using the Standard
              Trick) and, thereafter, \(\kappa \coloneqq k + d_P + d_T\) yield}
        |J| - 1
              &\le
              \frac{m/2+k+d_P+d_T+2\kappa}{\tau}
              =\frac{m+6\kappa}{2\tau}. \\
              \intertext{Rephrasing in terms of \(y_P\) and, thereafter, using
                \(\tau \ge \kappa / 2\), and the definition of \(\beta_P\) yield}
        |J| - 1
              &\le \frac{y_P+7\kappa}{2\tau}
              \le \frac{y_P}{2\tau}+7 \le \beta_P/2+7.
    \end{align*}
    Rearranging yields $|J|\le \beta_P/2+8 \le {}^{17}\!/\!{}_{2}\,\beta_P$, completing the
    proof.
\end{proof}

In particular, \cref{lem:smallbeta} allows us to assume (without loss of generality) that
\(\beta_P \ge 20\).
Further, we set $\Delta \coloneqq  6\kappa$ and $z \coloneqq \beta_P - 17$.

Now, we (essentially) extend the tiles from the \(\tau\)-tile partition of \(P\) by an additional
\(\Delta\) characters to obtain a \(\Delta\)-puzzle with value \(P\).

\begin{lemma}\label{fct:puzzleP}
    The following sequence $P_1,\ldots,P_z$ forms a $\Delta$-puzzle with value $P$:
    \begin{itemize}
        \item $P_1 \coloneqq P\fragmentco{p_0}{p_2+\Delta}$;
        \item $P_{i} \coloneqq P\fragmentco{p_{i}}{p_{i+1}+\Delta}$ for $i\in \fragmentoo{1}{z}$;
        \item $P_z \coloneqq P\fragmentco{p_{z}}{|P|}$.
    \end{itemize}
\end{lemma}
\begin{proof}
    The only non-trivial claim is that each internal piece is well-defined, that is that
    \(0\le p_i\) and that \(p_{i+1}+\Delta \le |P|\).
    To verify this claim, observe that we have for every \(i \in \fragmentoo{1}{z}\)
    \begin{align*}
        0 &= p_0 \le p_i \quad\text{and}\\
        p_{i+1}+\Delta&\le p_z+\Delta
            = p_z +6\kappa
            \intertext{Using \(\tau \ge \kappa/2\) and, thereafter,
            \(\big|(p_{z+14}-p_z)-14\tau\big|
            \le \ed(P\fragmentco{p_z}{p_{z+14}},\Q\fragment{z\tau}{(z+14)\tau})
            \le d_P\), we obtain}
        p_{i+1}+\Delta
          &\le p_z+14\tau-d_P
            \le p_{z+14}
            \le |P|. \qedhere
    \end{align*}
\end{proof}

Similarly, we obtain \(\Delta\)-puzzles for \(T\).
In particular, we use the $\tau$-tile partition of $T$ in order to represent the fragments
$R_j$ as $\Delta$-puzzles.

\begin{lemma}\label{fct:puzzleT}
    For each $j\in J$, the following sequence $T_{j,1},\ldots,T_{j,z}$ forms a
    $\Delta$-puzzle with value $R_j$:
    \begin{itemize}
        \item $T_{j,1} \coloneqq T\fragmentco{r_j}{t_{j+2}+\Delta}$;
        \item $T_{j,i} \coloneqq T\fragmentco{t_{j+i}}{t_{j+i+1}+\Delta}$ for $i\in \fragmentoo{1}{z}$;
        \item $T_{j,z} \coloneqq T\fragmentco{t_{j+z}}{r'_j}$.
    \end{itemize}
\end{lemma}
\begin{proof}
    The only non-trivial claim is that each internal piece is contained within $R_j$,
    that is that \(r_j \le t_{j+i}\) and that \(t_{j+i+1}+\Delta \le r'_j\).
    To verify this claim,
    first observe that \(\tau \ge \kappa / 2\) yields for every \(j \in J = \fragment{\ceil{{(x_T-\kappa)}/{\tau}}}{\floor{{(y_T-y_P+\kappa)}/{\tau}}}\)
    and every \(i \in \fragmentoo1z\) that\begin{equation*}
        x_T \le x_T-\kappa+2\tau  \le \left(\ceil{\frac{x_T-\kappa}{\tau}}+2\right)\tau
        \le (j + i) \tau.
    \end{equation*}
    In particular, we have $(t_{j+i},(j+i)\tau)\in \A_T$.
    Similarly, we obtain $j\tau-\kappa \le (j+2)\tau \le (j+i)\tau$,
    so $r_j \le t_{j+i}$ holds by definition of $r_j$.

    For the other bound, we first observe that for any \(j \in J\) and \(i \in
    \fragmentoo{1}{z}\), we have \[
        (j+i+15)\tau
        \le \left(\floor{\frac{y_T-y_P+\kappa}{\tau}}+\beta_P-3\right)\tau
        \le y_T-y_P+\kappa + (\beta_P-3)\tau
        \le y_T-y_P+\kappa + y_P-2\tau
        \le y_T.
    \]
    In particular, we have $(t_{j+i+15},(j+i+15)\tau)\in \A_T$.
    Similarly, we obtain $(j+i+15)\tau \le (j+\beta_P-3)\tau \le  y_P+j\tau-\kappa$,
    so $t_{j+i+15} \le r'_j$ holds by definition of $r'_j$.
    Now, as $\Delta = 7\kappa-\kappa \le 14\tau-d_T \le  t_{j+i+15}-t_{j+i+1}$,
    we obtain $t_{j+i+1}+\Delta \le r'_j$; completing the proof.
\end{proof}

Taken together, we obtain the families of puzzle pieces that we use in the remainder of this
work.

\begin{definition}\label{def:puzzleset}
    \begin{itemize}
        \item We write \(\Sb \coloneqq \{P_1\} \cup \{T_{j,1} \mid j\in J\}\) for the
            family of
        \emph{leading} puzzle pieces.
    \item We write \(\Sm \coloneqq \{P_{i} : i\in \fragmentoo{1}{z}\} \cup \{T_{i} : i \in
        \fragmentoo{\min J +1}{\max J +z} \}\) for the family of
        \emph{internal} puzzle pieces.
        \item We write \(\Sf \coloneqq \{P_z\} \cup \{T_{j,z} : j\in J\}\) for the family of
    \emph{trailing} puzzle pieces.
    \lipicsEnd
    \end{itemize}
\end{definition}
\begin{remark}
    For convenience, for $i\in \fragmentoo{\min J + 1}{\max J +z}$, we write
    $T_i\coloneqq T\fragmentco{t_i}{t_{i+1}+\Delta}$.
    Observe that by construction, we have \(T_i = T_{j,i'}\) for all \(j + i' = i\) with
    \(i' \in \fragmentoo{1}{z}\); that is
    overlapping parts of different \(R_j\)'s share their internal pieces.
    Observe further that this is an essential property for our approach to work: when
    moving from \(R_j\) to \(R_{j + 1}\), we exploit that we need to only shift the pieces
    \(T_i\), and not recompute them altogether.
    \lipicsEnd
\end{remark}

Finally, we wish to use the sets from \cref{def:puzzleset} for \DPM.
To that end, we need to convince ourselves that pairs of pieces \(P_i\) and \(T_{j,i}\)
are indeed roughly of the same length on average. In particular, we verify that our choice
of \(\Delta \coloneqq 6\kappa\) is indeed sufficient.

\begin{lemma}\label{fact:lengths_diff}
    For each $j\in J$, we have $\sum_{i=1}^z \big| |T_{j,i}|-|P_i|\big| \le 3\kappa-k$.
\end{lemma}
\begin{proof}
    Fix a \(j \in J\).
    In a first step, we obtain bounds on the length differences of each individual pair of
    pieces. To that end, consider a pair of internal pieces \(P_i\) and \(T_{j,i}\) for
    an $i\in \fragmentoo{1}{z}$; the triangle inequality yields
    \begin{align}
        \big||T_{j,i}|-|P_i|\big| &\le
        \ed(T\fragmentco{t_{j+i}}{t_{j+i+1}},P\fragmentco{p_i}{p_{i+1}})\nonumber\\
        &\le
        \ed(T\fragmentco{t_{j+i}}{t_{j+i+1}},\Q\fragmentco{(j+i)\tau}{(j+i+1)\tau})\nonumber\\\label{eq:4-8-1}
        &\quad+\ed(P\fragmentco{p_i}{p_{i+1}},\Q\fragmentco{i\tau}{(i+1)\tau}).
    \end{align}
    As for $i=1$, recall from \cref{rem:rj-corner} that
    $\A_T(r_j) = \max\{x_T, j\tau-\kappa\} \le j\tau+\kappa$.
    Hence, the triangle inequality yields
    \begin{align}
        \big||T_{j,1}|-|P_1|\big|
        &\le \ed(T\fragmentco{r_j}{t_{j+2}},P\fragmentco{0}{p_{2}})\nonumber\\
        &\le
        \ed(T\fragmentco{r_j}{t_{j+2}},\Q\fragmentco{\A_T(r_j)}{(j+2)\tau})\nonumber\\
        &\quad+ \ed(\Q\fragmentco{\A_T(r_j)}{(j+2)\tau},\Q\fragmentco{0}{2\tau})
        \nonumber\\
        &\quad+
        \ed(P\fragmentco{0}{p_2},\Q\fragmentco{0}{2\tau})
        \nonumber\\
        &\le \ed(T\fragmentco{0}{t_{j+2}},Q^\infty\fragmentco{x_T}{(j+2)\tau})+\kappa +
        \ed(P\fragmentco{0}{p_2},Q^\infty\fragmentco{0}{2\tau}).\label{eq:4-8-2}
    \end{align}
    As for $i=z$, recall from \cref{rem:rj-corner} that $\A_T(r'_j) = \min\{y_T,
    y_P+j\tau+\kappa\} \ge y_P+j\tau-\kappa$.
    Hence, the triangle inequality yields
    \begin{align}
        \big||T_{j,z}|-|P_z|\big|
        &\le \ed(T\fragmentco{t_{j+z}}{r'_j},P\fragmentco{p_z}{|P|})\nonumber\\
        &\le \ed(T\fragmentco{t_{j+z}}{r'_j},\Q\fragmentco{(j+z)\tau}{\A_T(r'_j)})\nonumber\\
        &\quad + \ed(\Q\fragmentco{(j+z)\tau}{\A_T(r'_j)},\Q\fragmentco{z\tau}{y_P})
        \nonumber\\
        &\quad+
        \ed(P\fragmentco{p_z}{|P|},\Q\fragmentco{z\tau}{y_P})\nonumber\\
        &\le \ed(T\fragmentco{t_{j+z}}{|T|},\Q\fragmentco{(j+z)\tau}{y_T}) + \kappa +
        \ed(P\fragmentco{p_z}{|P|},\Q\fragmentco{z\tau}{y_P}).\label{eq:4-8-3}
    \end{align}
    Now, observe that we have
    \begin{align*}
        d_P&\ge\ed(P\fragmentco{0}{p_2},Q^\infty\fragmentco{0}{2\tau})\\
        &\quad+\sum_{i=2}^{z-1}\ed(P\fragmentco{p_i}{p_{i+1}},\Q\fragmentco{i\tau}{(i+1)\tau})\\
        &\quad+\ed(P\fragmentco{p_z}{|P|},Q^\infty\fragmentco{z\tau}{y_P}).
    \end{align*}
    Symmetrically,
    \begin{align*}
        d_T&\ge\ed(T\fragmentco{0}{t_{j+2}},Q^\infty\fragmentco{x_T}{(j+2)\tau})\\
           &\quad+ \sum_{i=2}^{z-1}\ed(T\fragmentco{t_{j+i}}{t_{j+i+1}},
        Q^\infty\fragmentco{(j+i)\tau}{(j+i+1)\tau})\\
           &\quad+ \ed(T\fragmentco{t_{j+z}}{|T|},Q^\infty\fragmentco{(j+z)\tau}{y_T}).
    \end{align*}
    Adding \cref{eq:4-8-1,eq:4-8-2,eq:4-8-3} hence yields the claimed
    $\sum_{i=1}^z \big| |T_{j,i}|-|P_i|\big|\le 2\kappa +d_P+d_T = 3\kappa-k$.
\end{proof}

\subsection{Special Puzzle Pieces and How to Compute Them Efficiently}

In order for puzzle pieces to be useful to us, we need to be able to efficiently
compute the families \(\Sb\), \(\Sm\), and \(\Sf\) from \cref{def:puzzleset}.
As it turns out, when considering instances of \SM, most puzzle pieces are substrings of
\(\Q\)---this in turn means that it suffices to locate and compute the \emph{special} pieces (that
are different from some specific substrings of \(\Q\)).
We start with a formal definition of special puzzle pieces.

\begin{definition}\label{def:specialpiece}
    We say that an internal piece is \emph{special} if and only if it is different from
    \(\Q\fragmentco{0}{\tau+\Delta}\). We write \(\sp{P}\) for the set of \emph{special
    internal pieces} of $P$ and we \(\sp{T}\) for set of \emph{special internal
    pieces} of $T$; that is, we set
    \begin{align*}
        \sp{P} &\coloneqq \{P_i : i \in \fragmentoo{1}{z}
        \text{ and } P_i \neq Q^\infty\fragmentco{0}{\tau+\Delta}\} \quad\text{and}\\
        \sp{T} &\coloneqq \{T_i :i \in \fragmentoo{1 + \min J}{z + \max J}
        \text{ and } T_i \neq Q^\infty\fragmentco{0}{\tau+\Delta}\}.
    \end{align*}
    Further, we say that a leading piece \(T_{j,1}\) is \emph{special} if and only if it is different from
    \(Q^\infty\fragmentco{-\kappa}{2\tau+\Delta}\) and that a trailing piece \(T_{j,z}\)
    is \emph{special} if and only if it is different from \(Q^\infty\fragmentco{z\tau}{y_P+\kappa}\).
    Similar to before, we write \(\spb{T}\) for the set of special leading and trailing
    pieces of \(T\).
    \lipicsEnd
\end{definition}
\begin{remark}
    When we store a special internal piece \(P_i\) (or \(T_i\)), we store it as a pair
    \((i,P_i)\) (or \((i, T_i)\)) together with its index \(i\) in the sequence of pieces;
    special leading and trailing pieces are stored similarly.\lipicsEnd
\end{remark}

As a main result of this (sub-)section, we prove that there are only very few special
pieces and that we can compute them efficiently.

\spbndsimpl*

We proceed to prove \cref{lem:specialbound-simpl} by first bounding the edit
distances of the puzzle pieces to a common substring of \(\Q\); we start with the internal
pieces.

\begin{lemma}\label{lem:puzzle-proto}
    For any \(i \in \fragmentoo1z\), the internal piece \(P_i\) from \cref{fct:puzzleP} satisfies\[
        \ed(P_i,\Q\fragmentco{0}{\tau+\Delta}) \le
            \sum_{j=1}^{13} 2 \ed(P\fragmentco{p_{i+j-1}}{p_{i+j}},\Q\fragmentco{(i+j-1)\tau}{(i+j)\tau}).
    \]
    In particular, the internal piece \(P_i\) is special only if
    any of the fragments \(P\fragmentco{p_{i+j-1}}{p_{i+j}}\) for \(j \in \fragment1{13}\)
    differs from \(\Q\fragmentco{0}{\tau}\).

    For any \(i \in \fragmentoo{\min J + 1}{\max J + z}\), the internal piece \(T_i\) from \cref{fct:puzzleT}
    satisfies\[
        \ed(T_i,\Q\fragmentco{0}{\tau+\Delta}) \le
            \sum_{j=1}^{13}2 \ed(T\fragmentco{t_{i+j-1}}{t_{i+j}},\Q\fragmentco{(i+j-1)\tau}{(i+j)\tau}).
    \]
    In particular, the internal piece \(T_i\) is special only if
    any of the fragments \(T\fragmentco{t_{i+j-1}}{t_{i+j}}\) for \(j \in \fragment1{13}\)
    differs from \(\Q\fragmentco{0}{\tau}\).
\end{lemma}
\begin{proof}
    Fix an \(i \in \fragmentoo1z\).
    Recall that we defined
    \(P_i = P\fragmentco{p_i}{p_{i+1} + \Delta}\) and \(\Delta = 6 \kappa \le 12 \tau\).
    Observe that we have \(0 \le i\tau\)
    and
    \((i + 15) \tau \le (\beta_P - 3)\tau
    \le y_P\); and thus indeed \(( p_{i + j},(i + j)\tau) \in \A_P \) for all \(j
    \in\fragment{0}{15}\).
    Now, exploiting \(d_P \le \kappa \le 2\tau\) and \((i + 15)\tau \le
    y_P\), we see that \(
    p_{i+1}+12\tau
    \le p_{i+1}+14\tau-d_P
    \le p_{i + 15} \le |P|.
    \)
    Hence, we obtain (with the additional help of \cref{fct:ed-prefix})
    \begin{align*}
    &\ed(P_i,Q^\infty\fragmentco{0}{\tau+\Delta})\\
    &\quad= \ed(P\fragmentco{p_{i}}{p_{i+1}+\Delta},Q^\infty\fragmentco{i\tau}{(i+1)\tau+\Delta})\\
    &\quad\le \ed(P\fragmentco{p_{i}}{p_{i+1}+12\tau},Q^\infty\fragmentco{i\tau}{(i+13)\tau})
    \intertext{
        Extending the fragments \(P\fragmentco{p_{i}}{p_{i+1}+12\tau}\) and
        \(Q^\infty\fragmentco{i\tau}{(i+13)\tau}\) and correcting for potential length
        differences, we obtain
    }
    &\ed(P_i,Q^\infty\fragmentco{0}{\tau+\Delta})\\
    &\quad\le \ed(P\fragmentco{p_{i}}{p_{i+13}},Q^\infty\fragmentco{i\tau}{(i+13)\tau}) + |p_{i+13} - (p_{i+1} + 12\tau)|\\
    &\quad= \ed(P\fragmentco{p_{i}}{p_{i+13}},Q^\infty\fragmentco{i\tau}{(i+13)\tau}) + |(p_{i+13} - p_{i+1}) - 12\tau|\\
    &\quad\le \ed(P\fragmentco{p_{i}}{p_{i+13}},Q^\infty\fragmentco{i\tau}{(i+13)\tau}) + \ed(P\fragmentco{p_{i+1}}{p_{i+13}},Q^\infty\fragmentco{(i+1)\tau}{(i+13)\tau})\\
    &\quad\le \sum_{j=1}^{13}2\ed(P\fragmentco{p_{i+j-1}}{p_{i+j}},Q^\infty\fragmentco{(i+j-1)\tau}{(i+j)\tau}).
    \end{align*}

    We proceed similarly for the internal pieces of \(T\).
    To that end, fix an \(i \in \fragmentoo{\min J + 1}{\max J + z}\) and
    recall that we set \(T_i = T\fragmentco{t_i}{t_{i+1} + \Delta}\).
    As before, observe that we have
    \begin{align*}
        x_T &\le x_T-\kappa+2\tau \le (\ceil{{(x_T-\kappa)}/{\tau}}+2)\tau \le  i\tau
        \quad\text{and}\\
        (i+15)\tau &\le
        (\floor{{(y_T-y_P+\kappa)}/{\tau}}+z+14)\tau\\
        &\le y_T-y_P+\kappa + (\beta_P-3)\tau \le y_T+\kappa -2\tau \le y_T.
    \end{align*}
    Thus, indeed, we have
    $( t_{i+j}, (i+j)\tau) \in \A_T$ for all $j\in \fragment{0}{15}$.
    Now, exploiting \(d_T \le \kappa \le 2\tau\) and \((i+15) \tau \le y_T\), we see that
    $t_{i+1}+\Delta \le t_{i+1}+14\tau - d_T \le t_{i+15}\le |T|$.
    Hence, analogously to before, we conclude that
    \[
        \ed(T_i, Q^\infty\fragmentco{0}{\tau+\Delta})
        \le \sum_{j=1}^{13} 2 \ed(T\fragmentco{t_{i+j-1}}{t_{i+j}}, Q^\infty\fragmentco{(i+j-1)\tau}{(i+j)\tau});
    \]
    completing the proof.
\end{proof}

As an immediate corollary, we obtain the desired bound for \(\ed(\Sm)\).

\begin{corollary}
    We have \( \ed(\Sm) \le 26 \kappa\) and \(|\sp{T}| + |\sp{P}| \le 26 \kappa\).
\end{corollary}
\begin{proof}
    Recall that \(\ed(\Sm) = \min_{\hat{S}} \sum_{S \in \Sm} \ed(S, \hat{S})\).
    Now, by \cref{lem:puzzle-proto}, choosing \(\Q\fragmentco{0}{\tau+\Delta}\)
    yields the upper bound
    \begin{align*}
        \ed(\Sm) &\le \sum_{S \in \Sm} \ed(S, \Q\fragmentco{0}{\tau+\Delta})\\
                 &= \sum_{i = 2}^{z-1} \ed(P_i, \Q\fragmentco{0}{\tau+\Delta})
                 + \sum_{i = \min J + 2}^{\max J + z - 1} \ed(T_i,
             \Q\fragmentco{0}{\tau+\Delta})\\
                 &\le 2\cdot \sum_{j=1}^{13}
                    \sum_{i=2}^{z-1}
                    \ed(P\fragmentco{p_{i+j-1}}{p_{i+j}},Q^\infty\fragmentco{(i+j-1)\tau}{(i+j)\tau})\\
                 &\quad+ \sum_{i=\min J + 2}^{\max J + z - 1} \ed(T\fragmentco{t_{i+j-1}}{t_{i+j}},Q^\infty\fragmentco{(i+j-1)\tau}{(i+j)\tau})
                    \\
                 &\le 2\cdot 13\cdot (d_P + d_T) = 26\kappa.
    \end{align*}
    Finally observe that each special piece has an edit distance of at least 1 to
    \(\Q\fragmentco{0}{\tau+\Delta}\); hence there may be at most \(26 \kappa\) special
    pieces in total.
\end{proof}

Finally, \cref{lem:puzzle-proto} easily yields an algorithm for computing the special
internal pieces of \(P\) and \(T\).

\begin{corollary}\label{lem:puzzleP-algo}\label{lem:puzzlemu-algo}
    Given the alignment $\A_P$, we can compute \(\sp{P}\) in $\Oh(d_P+1)$ time in the
    \modelname model; given the alignment $\A_T$, we can compute \(\sp{T}\) in $\Oh(d_T+1)$ time in the \modelname model.
\end{corollary}
\begin{proof}
    By \cref{lem:puzzle-proto}, for each special internal piece \(P_i\) (or \(T_i\)),
    there is a \(j \in \fragment1{15}\) such that
    $P\fragmentco{p_{i+j-1}}{p_{i+j}}\ne Q^\infty\fragmentco{(i+j-1)\tau}{(i+j)\tau}$
    (or $T\fragmentco{t_{i+j-1}}{t_{i+j}}\ne
    Q^\infty\fragmentco{(i+j-1)\tau}{(i+j)\tau}$).
    This, in turn, requires a breakpoint $(a_P,a_Q)\in \A_P$ (or \((a_T,a_Q)\in
    \A_T\)) with $a_Q\in \fragmentco{i\tau}{(i+13)\tau}$.

    Hence, for each breakpoint $(a_P,a_Q)\in \A_P$,
    we identify pieces \(P_i\) with $i\in
    \fragmentoc{\floor{{a_Q}/{\tau}}-15}{\floor{{a_Q}/{\tau}}}\cap \fragmentoo{1}{z}$ as
    \emph{candidates};
    for each breakpoint $(a_T,a_Q)\in \A_T$,
    we identify pieces \(T_i\) with $i\in
    \fragmentoc{\floor{{a_Q}/{\tau}}-15}{\floor{{a_Q}/{\tau}}}\cap \fragmentoo{\min J +
    1}{\max J + z}$ as \emph{candidates}.
    For each of these candidates, we check $P_i = Q^\infty\fragmentco{0}{\tau+\Delta}$ (or
    \(T_i = \Q\fragmentco{0}{\tau+\Delta}\) using an \lceOpName operation of the \modelname model.

    By processing the breakpoints in left-to-right order and keeping track of the rightmost candidate generated so far,
    we can list the indices $i$ of the candidates $P_i$ (or \(T_i\)) in left-to-right order without duplicates.
    Then, a linear-time scan over the representation of $(p_i)_{i=0}^{\beta_P}$ (or
    $(p_i)_{i=0}^{\beta_T}$) obtained from \cref{fct:tile} lets us determine the endpoints
    of each candidate $P_i$ (or \(T_i\)).
    Each breakpoint yields up to $15$ candidates, so the total running time in the
    \modelname model is $\Oh(d_P+1)$ (or \(\Oh(d_T + 1)\)).
\end{proof}

We proceed to special leading and trailing pieces---as there is at most one
special leading piece of \(P\) and at most one special trailing piece of \(P\), we can
trivially compute them fast. Hence, we first focus on special leading and trailing peaces
of \(T\). To that end, we proceed as with the internal pieces; however, the calculations
become slightly more involved.


\begin{lemma}\label{lem:puzzlebeta}
    For any \(j \in J\), the leading piece \(T_{j,1}\) from \cref{fct:puzzleT} satisfies\[
        \ed(T_{j,1},Q^\infty\fragmentco{-\kappa}{2\tau+\Delta}) \le 2\kappa + 2d_T.
    \] Further, for any \(j \in \fragment{3}{\max J}\), the leading piece
    \(T_{j,1}\) additionally satisfies\[
        \ed(T_{j,1},Q^\infty\fragmentco{-\kappa}{2\tau+\Delta})
        \le 2\sum_{i=-1}^{14}
        \ed(T\fragmentco{t_{i+j-1}}{t_{i+j}},Q^\infty\fragmentco{(i+j-1)\tau}{(i+j)\tau}).
    \]
    In particular, the leading piece \(T_{j, 1}\) is special only if either \(j < 3\) or any
    of the fragments \(T\fragmentco{t_{i+j-1}}{t_{i+j}}\) differs from
    \(Q^\infty\fragmentco{(i+j-1)\tau}{(i+j)\tau}\) for \(i \in \fragment{-1}{14}\).
\end{lemma}
\begin{proof}
    Fix a $j\in J$ and recall from \cref{rem:rj-corner}
    that $\A_T( r_j ) = \max\{x_T, j\tau-\kappa\} \le j\tau+\kappa$.
    Observe that we have
    \begin{align*}
        x_T &\le x_T - \kappa + 2\tau \left(\ceil{\frac{x_T-\kappa}{\tau}}+2\right)\tau
        \le (j+2)\tau
        \quad\text{and}\\
        (j+16)\tau &\le
        \left(\floor{\frac{y_T-y_P+\kappa}{\tau}}+16\right)\tau
                            \le y_T-y_P+\kappa +16\tau\\
                            &\le y_T - (\beta_P-1)\tau + 2\tau + 16\tau
                            = y_T - (\beta_P-19)\tau \le y_T.
    \end{align*}
    Thus, indeed, we have $(t_{j+i}, (j+i)\tau) \in \A_T$ for all $i\in \fragment{2}{16}$.
    Now exploiting \(d_T \le \kappa \le 2\tau\) and \((j + 16)\tau \le y_T\), we see
    that  $t_{j+2} + \Delta \le t_{j+2}+12\tau \le t_{j+2}+14\tau-d_T \le t_{j+16} \le |T|$.
    Hence, we conclude (with the additional help of \cref{fct:ed-prefix})
    \begin{align*}
        &\ed(T_{j,1},Q^\infty\fragmentco{-\kappa}{2\tau+\Delta})\\
        &\quad = \ed(T\fragmentco{r_j}{t_{j+2}+\Delta}, Q^\infty\fragmentco{j\tau-\kappa}{(j+2)\tau+\Delta})\\
        &\quad \le \ed(T\fragmentco{r_j}{t_{j+2}+12\tau},Q^\infty\fragmentco{j\tau-\kappa}{(j+14)\tau})
        \intertext{Replacing strings by superstrings and accounting for potential length
        differences, we obtain}
        &\ed(T_{j,1},Q^\infty\fragmentco{-\kappa}{2\tau+\Delta})\\
        &\quad \le \ed(T\fragmentco{r_j}{t_{j+14}},Q^\infty\fragmentco{j\tau-\kappa}{(j+14)\tau}) + |t_{j+14} - (t_{j+2}+12\tau)|\\
        &\quad \le \ed(T\fragmentco{r_j}{t_{j+2}},Q^\infty\fragmentco{j\tau-\kappa}{(j+2)\tau}) + 2\ed(T\fragmentco{t_{j+2}}{t_{j+14}},Q^\infty\fragmentco{(j+2)\tau}{(j+14)\tau})\\
        &\quad =
        \ed(T\fragmentco{r_j}{t_{j+2}},Q^\infty\fragmentco{j\tau-\kappa}{(j+2)\tau})\\
        &\qquad+ 2\sum_{i=3}^{14}\ed(T\fragmentco{t_{i+j-1}}{t_{i+j}},Q^\infty\fragmentco{(i+j-1)\tau}{(i+j)\tau})\\
        &\quad \le (\A_T(r_j)-(j\tau-\kappa))+2d_T \le 2\kappa + 2d_T.
    \end{align*}
    Now, for $j\ge 3$, we have \(x_T \le \tau \le 3\tau - \kappa \le j\tau - \kappa\)
    and hence $\A_T( r_j )  = j\tau-\kappa$.
    Moreover, we have $x_T \le \tau \le (j-2)\tau$. 
    Thus, in particular, we have $(t_{j+i}, (j+i)\tau) \in \A_T$ even for all
    $i\in \fragment{-2}{16}$.
    Hence, we conclude similarly to before
    \begin{align*}
        &\ed(T\fragmentco{r_j}{t_{j+2}},Q^\infty\fragmentco{j\tau-\kappa}{(j+2)\tau})\\
        &\quad \le \ed(T\fragmentco{t_{j-2}}{t_{j+2}},Q^\infty\fragmentco{(j-2)\tau}{(j+2)\tau}) + |(r_j-t_{j-2})-((j\tau-\kappa)-(j-2)\tau)|\\
        &\quad \leq 2\ed(T\fragmentco{t_{j-2}}{t_{j+2}},Q^\infty\fragmentco{(j-2)\tau}{(j+2)\tau})\\
        &\quad = 2\sum_{i=-1}^{2}\ed(T\fragmentco{t_{i+j-1}}{t_{i+j}},Q^\infty\fragmentco{(i+j-1)\tau}{(i+j)\tau}).
    \end{align*}
    In total, we obtain the claimed \[
        \ed(T_{j,1},Q^\infty\fragmentco{-\kappa}{2\tau+\Delta})
        \le 2\sum_{i=-1}^{14}
        \ed(T\fragmentco{t_{i+j-1}}{t_{i+j}},Q^\infty\fragmentco{(i+j-1)\tau}{(i+j)\tau}).
        \qedhere
    \]
\end{proof}

Similarly to \cref{lem:puzzlebeta}, we analyze the (special) trailing pieces of \(T\).

\begin{lemma}\label{lem:puzzleeps}
    For any \(j \in J\), the trailing piece \(T_{j,z}\) from \cref{fct:puzzleT} satisfies\[
        \ed(T_{j,z},Q^\infty\fragmentco{z\tau}{y_P + \kappa}) \le 2\kappa + d_T.
    \] Further, for any \(j \in \fragment{\min J}{\beta_T - \beta_P - 3}\), the trailing piece
    \(T_{j,z}\) additionally satisfies\[
        \ed(T_{j,z},Q^\infty\fragmentco{z\tau}{y_P + \kappa})
        \le \sum_{i=z + 1}^{z + 19}
        \ed(T\fragmentco{t_{i+j-1}}{t_{i+j}},Q^\infty\fragmentco{(i+j-1)\tau}{(i+j)\tau}).
    \]
    In particular, the trailing piece \(T_{j, z}\) is special only if either \(j > \beta_T -
    \beta_P - 3\) or any of the fragments \(T\fragmentco{t_{i+j-1}}{t_{i+j}}\) differs from
    \(Q^\infty\fragmentco{(i+j-1)\tau}{(i+j)\tau}\) for \(i \in\fragment1{z+19}\).
\end{lemma}
\begin{proof}
    Fix a \(j \in J\) and recall from \cref{rem:rj-corner} that
    \(\A_T(r'_j) = \min\{ y_T, y_P + j\tau + \kappa\} \ge y_p + j\tau - \kappa\).
    Observe that we have
    \begin{align*}
        x_T &\le x_T - \kappa + 2\tau \le x_T - \kappa + z\tau \le
        \left(\ceil{\frac{x_T-\kappa}{\tau}}+z\right)\tau \le (j + z)\tau
        )\quad\text{and}\\
        (j+z)\tau &\le \left(\floor{\frac{y_T-y_P+\kappa}{\tau}}+z\right)\tau
        \le y_T-y_P+\kappa + (\beta_P-17)\tau\\
                      &\le y_T-(\beta_P-1)\tau +2\tau +(\beta_P-17)\tau
        \le y_T-14\tau \le y_T
    \end{align*}
    Thus, indeed, we have  $(t_{j+z}, (j+z)\tau) \in \A_T$. Hence, we conclude
    \begin{align*}
        \ed(T_{j,z},Q^\infty\fragmentco{z\tau}{y_P+\kappa})
        &= \ed(T\fragmentco{t_{j+z}}{r'_j}, Q^\infty\fragmentco{(j+z)\tau}{j\tau+y_P+\kappa})\\
        & \le 2\kappa
        +\ed(T\fragmentco{t_{j+z}}{r'_j},Q^\infty\fragmentco{(j+z)\tau}{\A_T(r'_j)})\\
        & \le 2\kappa + d_T.
    \end{align*}

    Now, for $j\le \beta_T-\beta_P-3$,
    we have \(y_P + j\tau + \kappa \le (\beta_P + j + 2)\tau \le (\beta_T - 1)\tau \le
    y_T\) and hence $\A_T(r'_j) = y_P+j\tau+\kappa$.
    Moreover, we have $y_P+j\tau+\kappa < (\beta_P+j+2)\tau = (j+z+19)\tau \le y_T$.
    Thus, in particular, we have $(t_{j+i}, (j+i)\tau) \in \A_T$ even
    for all $i\in \fragment{z}{z+19}$.
    Hence,
    \begin{align*}
        \ed(T_{j,z},Q^\infty\fragmentco{z\tau}{y_P+\kappa}) & =
        \ed(T\fragmentco{t_{j+z}}{r'_j}, Q^\infty\fragmentco{(j+z)\tau}{j\tau+y_P+\kappa})\\
        & \le \ed(T\fragmentco{t_{j+z}}{t_{j+z+19}}, Q^\infty\fragmentco{(j+z)\tau}{(j+z+19)\tau})\\
        & =  \sum_{i=z+1}^{z+19}
        \ed(T\fragmentco{t_{i+j-1}}{t_{i+j}},Q^\infty\fragmentco{(i+j-1)\tau}{(i+j)\tau});
    \end{align*}
    completing the proof.
\end{proof}

As with the internal pieces, we proceed to prove the desired bounds on \(\ed(\Sb)\) and
\(\ed(\Sf)\).

\begin{corollary}
    We have \( \ed(\Sb) \le 53 \kappa\), \(\ed(\Sf) \le 35 \kappa\),
    and \(|\spb{T}| \le 88 \kappa\).
\end{corollary}
\begin{proof}
    We start with the leading pieces. To that end, first observe that
    \begin{align*}
        \ed(P_1,  Q^\infty\fragmentco{-\kappa}{2\tau+\Delta})
            &\le \kappa + \ed(P\fragmentco{0}{p_2+\Delta},Q^\infty\fragmentco{0}{2\tau+\Delta})\\
            &\le \kappa + \ed(P\fragmentco{0}{p_2+12\tau},Q^\infty\fragmentco{0}{14\tau})
            \intertext{Replacing strings by superstrings and accounting for potential
            length differences, we obtain}
        \ed(P_1,  Q^\infty\fragmentco{-\kappa}{2\tau+\Delta})
            &\le \kappa + \ed(P\fragmentco{0}{p_{14}},Q^\infty\fragmentco{0}{14\tau}) + |p_{14} - (p_2- 12\tau)|\\
            &\le \kappa + \ed(P\fragmentco{0}{p_{14}},Q^\infty\fragmentco{0}{14\tau}) + \ed(P\fragmentco{p_2}{p_{14}},Q^\infty\fragmentco{2\tau}{14\tau})\\
            &\le \kappa + 2\ed(P, Q^\infty\fragmentco{0}{y_P}) \\
            &\le \kappa + 2d_P.
    \end{align*}
    Next, we recall \cref{rem:negj} and in particular that \(\min J \ge -2\).
    Therefore, \cref{lem:puzzlebeta} yields
    \begin{align*}
       &\sum_{j\in J} \ed(T_{j,1},  Q^\infty\fragmentco{-\kappa}{2\tau+\Delta}) \\
       &\quad \le 5(2\kappa + 2d_T)+
       2\sum_{i=-1}^{14} \sum_{j=3}^{\max J}\ed(T\fragmentco{t_{i+j-1}}{t_{i+j}},Q^\infty\fragmentco{(i+j-1)\tau}{(i+j)\tau}) \\
       &\quad \le 10\kappa +10d_T + 32d_T = 10\kappa +42d_T.
    \end{align*}
    Combined, we obtain the desired bound
    \begin{align*}
        \ed(\Sb)&\le \sum_{S\in \Sb} \ed(S, Q^\infty\fragmentco{-\kappa}{2\tau+\Delta})\\
                &\le \ed(P_1,  Q^\infty\fragmentco{-\kappa}{2\tau+\Delta}) + \sum_{j\in J} \ed(T_{j,1},  Q^\infty\fragmentco{-\kappa}{2\tau+\Delta})\\
                &\le \kappa + 2d_P + 10\kappa + 42d_T \le 53 \kappa.
    \end{align*}
    Now observe that each special leading piece has an edit distance of at least 1 to
    \(\Q\fragmentco{-\kappa}{2\tau+\Delta})\); hence there may be at most \(53 \kappa\) special
    leading pieces in total.

    We proceed to the trailing pieces. To that end, first observe that
    \[
        \ed(P_z,  Q^\infty\fragmentco{z\tau}{y_P+\kappa})
        \le \kappa + \ed(P_z,  Q^\infty\fragmentco{z\tau}{y_P}) \le \kappa + d_P.
    \]
    Next, we recall \cref{rem:negj} and in particular that \(\max J \le \ceil{y_T / \tau}
    - \ceil{y_P / \tau} + 2 = \beta_T - \beta_P + 2\).
    Therefore, \cref{lem:puzzleeps} yields
    \begin{align*}
        &\sum_{j\in J} \ed(T_{j,z}, Q^\infty\fragmentco{z\tau}{y_P})\\
        &\quad \le 5(2\kappa+d_T)+\sum_{i=z+1}^{z+19}\sum_{j=\min J}^{\beta_T-\beta_P-3} \ed(T\fragmentco{t_{i+j-1}}{t_{i+j}},Q^\infty\fragmentco{(i+j-1)\tau}{(i+j)\tau})\\
        &\quad \le 10\kappa + 5d_T + 19d_T = 10\kappa + 24d_T.
    \end{align*}
    Combined, we obtain the desired bound
    \begin{align*}
        \ed(\Sf)&\le \sum_{S\in \Sf} \ed(S, Q^\infty\fragmentco{z\tau}{y_P+\kappa})\\
                &\le \ed(P_z,  Q^\infty\fragmentco{z\tau}{y_P+\kappa}) + \sum_{j\in J}
                \ed(T_{j,z},  Q^\infty\fragmentco{z\tau}{y_P+\kappa})\\
                &\le \kappa + d_P + 10\kappa + 24d_T \le 35 \kappa.
    \end{align*}

    Finally observe that each special trailing piece has an edit distance of at least 1 to
    \(\Q\fragmentco{z\tau}{y_P+\kappa})\); hence there may be at most \(35 \kappa\) special
    trailing pieces in total. In total, we obtain the claimed bound of
    \(|\spb{T}| \le 88 \kappa\), completing the proof.
\end{proof}

Finally, we discuss how to compute the special leading and trailing pieces of
\(T\)---again, the algorithm strongly resembles \cref{lem:puzzleP-algo} for the special
internal pieces.

\begin{corollary}\label{lem:puzzleT-algo}
    Given the alignment $\A_T$, we can compute \(\spb{T}\) in $\Oh(d_T+1)$ time in the
    \modelname model.
\end{corollary}
\begin{proof}
    For computing the special leading pieces, \cref{rem:negj,lem:puzzlebeta}
    let us focus on the following candidates: at most 5 pieces $T_{j,1}$ with $j\le 2$,
    as well as all pieces $T_{j,1}$ such that
    $T\fragmentco{t_{i+j-1}}{t_{i+j}}\ne Q^\infty\fragmentco{(i+j-1)\tau}{(i+j)\tau})$
    holds for some $i\in \fragment{-1}{14}$.
    Similarly for computing the special trailing peaces, \cref{rem:negj,lem:puzzleeps}
    let us focus on the following candidates: at most 5 pieces $T_{j,z}$ with
    $j\ge \beta_T-\beta_P-2$,
    as well as all pieces $T_{j,z}$ such that $T\fragmentco{t_{i+j-1}}{t_{i+j}}\ne
    Q^\infty\fragmentco{(i+j-1)\tau}{(i+j)\tau})$ holds for some $i\in
    \fragment{z+1}{z+19}$.

    To satisfy $T\fragmentco{t_{i+j-1}}{t_{i+j}}\ne
    Q^\infty\fragmentco{(i+j-1)\tau}{(i+j)\tau})$, a breakpoint $(a_T,a_Q)\in \A_T$
    with $a_Q\in \fragmentco{(i-2)\tau}{(i+14)\tau}$ is needed.
    Hence, by scanning the breakpoints in the left-to-right order we can generate all
    candidates in $\Oh(d_T+1)$ time.
    Next, we retrieve their endpoints using \cref{fct:rj,fct:tile},
    and we verify each candidate using an \lceOpName operation of the \modelname model.
    The overall running time is $\Oh(d_T+1)$.
\end{proof}

Recalling and combining the various lemmas of this subsection, we obtain \cref{lem:specialbound-simpl}, which we
restate here for convenience.

\spbndsimpl*

\subsection{Solving {\SM} via {\DPM}: A Warm-up Algorithm}\label{sec:warmup}

In this (sub-)section we obtain a first reduction from \SM to \DPM. While overly naive in
nature, it serves as an overview of the general structure of the more involved variants
that follow in later (sub-)sections. In particular, in this (sub-)section, we discuss the
following (easy) result.

\begin{restatable}{lemma}{redwarmup}\label{lem:red-warmup}
    Given an instance {\(\SM(P, T, k, d, Q, \A_P, \A_T)\)},
    after a preprocessing step that takes \(\Oh(d + m/\max\{q, d\} )\) time
    in the \modelname model, we can compute \(\OccE_k(P, T)\) using
    a \DPM\!\!\((k, \Delta, \Sb, \Sm, \Sf)\) data structure,
    with \(\Delta + \ed(\Sb) + \ed(\Sm) + \ed(\Sf) = \Oh( d )\),
    that maintains a DPM-sequence of length \(z = \Oh( m/\max\{q, d\} )\) under
    \begin{itemize}
        \item \(\Oh(d \cdot m/\max\{q, d\})\) calls to {\tt DPM-Substitute} and
        \item \(\Oh(m/\max\{q, d\})\) calls to {\tt DPM-Query}.\lipicsEnd
    \end{itemize}
\end{restatable}

Consider a given instance \SM(\(P\), \(T\), \(k\), \(d\), \(Q\), \(\A_P\), \(\A_T\))
and write \(\kappa \coloneqq k + d_P + d_T\) and \(\tau \coloneqq q \ceil{\kappa/2q}\).
Further, as before, let $P = \bigodot_{i=1}^{\beta_P} P\fragmentco{p_{i-1}}{p_i}$
denote the $\tau$-tile partition of $P$ with respect to $\A_P$;
let $T=\bigodot_{i=1}^{\beta_T} T\fragmentco{t_{i-1}}{t_i}$ denote
the $\tau$-tile partition of $T$ with respect to $\A_T$; and set \(\Delta \coloneqq
6\kappa\) and \(z \coloneqq \beta_P - 17\).
Finally, we define the families of puzzle pieces according to \cref{def:puzzleset}.

Recalling \cref{rem:can_part,lem:specialbound-simpl},
we readily confirm that our choices
for \(\Delta\), \(\Sb\), \(\Sm\), and \(\Sf\) indeed satisfy
\(z = \Oh( m/\max\{q, d\} )\),
\(\Delta = \Oh( d )\), and \(\ed(\Sb) + \ed(\Sm) + \ed(\Sf) = \Oh( d )\).

Intuitively, we proceed as follows:
we initialize a \(\DPM(k, \Delta, \Sb, \Sm, \Sf)\) data structure with a DPM-sequence
\(\I\) of length \(z\) that represents \(P\) and \(R_{\min J}\), and call {\tt DPM-Query} to obtain
\(\OccE_k(\I_{\min J}) = \OccE_k( P, R_{\min J} )\), that is,
the \(k\)-error occurrences of \(P\) in \(R_{\min J}\). Then, iterating over \(J\),
for a \(j \in \fragmentoc{\min J}{\max J}\), we use {\tt DPM-Substitute} operations to transform
\(\I\) into a DPM-sequence that represents \(P\) and \(R_j\); to then call {\tt DPM-Query} to
\(\OccE_k(\I_{j}) = \OccE_k( P, R_{j} )\), that is,
obtain the \(k\)-error occurrences of \(P\) in \(R_j\).

Observe that by \cref{lem:aligned}, we indeed obtain all \(k\)-error occurrences of \(P\) in \(T\).
As we have \(|J| = \Oh(\beta_P) = \Oh(m/\max\{q,d\})\) (see \cref{lem:smallbeta}),
we can already see that
the number of calls to {\tt DPM-Query} is just as promised in \cref{lem:red-warmup}.
For a bound on the number of calls to {\tt DPM-Substitute}, we need to be more precise on
how we transform the DPM-sequence \(\I\).
We start with formally defining the DPM-sequence(s) \(\I\).

\begin{definition}\label{def:seq-i}
    For each $j \in J$, we set \(
    \I_j \coloneqq (P_1, T_{j,1}) (P_2, T_{j,2}) \cdots (P_z, T_{j,z}).
    \)
    \lipicsEnd
\end{definition}
Recall that by \cref{fct:puzzleP,fct:puzzleT} we have \(\val_{\Delta}(P_1,\ldots,P_z)=P\)
and \(\val_{\Delta}(T_{j,1},\ldots,T_{j,z})=R_j\).

Let us write \(\I = (U_1, V_1)(U_2,V_2)\dots(U_z,V_z)\) for the DPM-sequence maintained in the \DPM
data structure. We initialize \(\I\) with the DPM-sequence $\I \gets \I_{\min J}$.
Then, we iterate over \(J\) (starting from \(\min J\)) as follows.
\begin{itemize}
    \item If $j>\min J$, for each $i \in \fragment{1}{z}$ with $T_{{j-1},i} \neq T_{j,i}$,
        we substitute the $i$-th pair of strings with $(P_i, T_{j,i})$.
        We perform these updates in non-decreasing order
        with respect to $\big| |T_{j,i}|-|P_i|\big| - \big||T_{j-1,i}|-|P_i|\big|$.
    \item We call {\tt DPM-Query} to obtain $\OccE_k( \I_j) = \OccE_k(P, R_j)$; and
        we set $\OccE_k(P, T) \coloneqq \OccE_k(P, T) \cup (r_j+\OccE_k(P, R_j))$.
\end{itemize}

\paragraph*{Verifying the Conditions of \DPM}

Before we discuss the number of required calls to {\tt DPM-Substitute}, we have to
convince ourselves that \(\I\) fulfills the constraints of \DPM on the maintained DPM-sequence.
In particular, we need to show that the following conditions are satisfied at all times.
\begin{enumerate}[(a)]
    \item $U_1,V_{1} \in \Sb$ and $U_z, V_{z}  \in \Sf$,
        and, for all $i\in \fragmentoo{1}{z}$, $U_i,V_{j} \in \Sm$\label{it:orig-r}
    \item \(\tor(\I) = \sum_{i=1}^z \big||U_i|-|V_i|\big| \le \Delta/2 - k\).\label{it:unique-r}
\end{enumerate}

Observe that before any {\tt DPM-Query}, the DPM-sequence \(\I\) is equal to the DPM-sequence
\(\I_j\).
Now, while Condition~\eqref{it:orig-r} is satisfied by construction, we need to work slightly
harder for Condition~\eqref{it:unique-r}.
First, if we have \(\I = \I_j\) for some \(j \in J\), then \cref{fact:lengths_diff}
yields\[
    \tor(\I) = \tor(\I_j) =
    \sum_{i=1}^z \big||P_i|-|T_{j,i}|\big| \le 3\kappa - k = \Delta/2 - k.
\]
Next, consider the sequence of updates \(\mathcal{U}\) between two sequences \(\I_{j}\)
and \(\I_{j + 1}\).
Observe that by construction, we can split \(\mathcal{U}\) into two
disjoint parts: first the updates \(\mathcal{U}_{\le}\) that do not increase
the \torn \(\tor(\I_j)\), followed by the updates \(\mathcal{U}_{>}\) that do.
In particular, we can apply the following useful fact, which directly yields the desired
properties.

\begin{fact}\label{fact:interm}
    Let \(C\) denote a non-negative integer and
    let \(\I\) and \(\I'\) denote sequences of length \(z\) that satisfy
    \(\tor(\I) \le C\) and \(\tor(\I') \le C\).
    Consider a sequence \(\mathcal{U}\) of updates that
    \begin{itemize}
        \item transforms \(\I\) into \(\I'\), and
        \item is ordered such that all updates that do not increase
            the \torn of the DPM-sequence precede all updates that increase the \torn of the
            DPM-sequence.
    \end{itemize}
    Then, any DPM-sequence \(\I''\) obtained from \(\I\) by applying a prefix of \(\mathcal{U}\)
    satisfies \(\tor(\I'') \le C\).\lipicsEnd
\end{fact}

\paragraph*{Bounding the Number of Calls to {\tt DPM-Substitute}}

We return to bounding the number of calls to {\tt DPM-*} operations.
In particular, our next goal is to upper-bound the number of pairs that our algorithm substitutes
to transform $\I_{j-1}$ to $\I_j$ for some $j\in \fragmentoc{\min J}{\max J}$.

To that end, observe that we might need to substitute the head or the tail of the DPM-sequence.
Further, for each $i \in \fragmentoo{1}{z}$, the $i$-th pair is substituted only if at least one
of $T_{{j-1},i}$ and $T_{j,i}$ differs from $Q^\infty\fragmentco{0}{\tau+\Delta}$; that
is, if either \(T_{{j-1},i}\) or \(T_{j,i}\) is special.
Now, recall from \cref{lem:specialbound-simpl} that the set \(\sp{T}\cup\sp{P}\) is of
size \(\Oh(d)\)---thus, for each $j \in J$, we perform $\cO(d)$ substitutions.

In total, we hence perform $\Oh( \kappa |J|) = \cO(\kappa \cdot m/\tau) = \Oh(d \cdot
m/\max\{q, d\})$ calls to {\tt DPM-Substitute}; recall that earlier we already
bounded the number of calls to {\tt DPM-Update} by \(\Oh(|J|) = \Oh(m /\max\{q, d\})\).

\paragraph*{Analyzing the Preprocessing Time}

Now for bounding the preprocessing time, we need to argue about the time required for constructing the
initial DPM-sequence \(\I_{\min J}\), as well as the time required to compute the substitutions that are to be performed.

For computing $\I_{\min J}$ we turn to \cref{fct:rj,fct:tile}, which yield
the sequences $(r_j)_{j\in J}$, $(r'_j)_{j\in J}$, $(p_i)_{i=0}^{\beta_P}$,
and $(t_i)_{i=0}^{\beta_T}$, from which we can easily obtain \(\I_{\min J}\)
in time \(\Oh(z + \kappa) = \cO(m/\max\{q,d\} + d)\).

For computing which pieces to update, \cref{lem:specialbound-simpl} yields the special
pieces of \(P\) and \(T\) in a suitable representation in time \(\Oh(d)\)
(in the \modelname model).

To summarize our result, let us recall \cref{lem:red-warmup}, which we have just proved.

\redwarmup*

Observe that combining \cref{thm:dpm,lem:red-warmup}, we obtain an algorithm for \SM
that requires
$\cOtilde(d^3+\Delta^2 d)=\cOtilde(d^3)$ time for preprocessing,
$\cOtilde(\Delta z)=\cOtilde(d \cdot m/\max\{q, d\})$ time for initialization,
and $\cOtilde(\Delta \kappa \cdot m/\tau)=\cO(d^2 \cdot m/\max\{q,d\})$
time for processing updates and queries.
This yields an overall running time of $\cOtilde(d^3+d^2 \cdot m/\max\{q,d\})$,
which is by far too much.
In particular, in the remainder of this section (and, by extension, this part), we aim to
replace the factor \(m/\max\{q, d\}\) with a \emph{small} power of \(d\).

\subsection{Solving {\SM} via {\DPM}, Improvement 0: Replacing Pair Substitutions with
Pair Insertions and Pair Deletions}
\label{sec:warmup2}

In our quest to replace the factor $\Oh(m/\tau) = \Oh(m/\max\{q,d\})$ in
\cref{lem:red-warmup}  by a small power of $d$, we take a small detour to give an
alternative algorithm for obtaining \(\I_{j}\) from \(\I_{j-1}\).
In particular, observe that up until now, we have not used the {\tt DPM-Delete} and {\tt
DPM-Insert} operations, as we could make do with just calling {\tt DPM-Substitute}.
While we do not decrease the overall number of calls to {\tt DPM-*},
in this (sub-)section, we decrease the number of such calls that involve special pieces to
\(\Oh(\kappa^2) = \Oh(d^2)\).
Formally, we call an internal pair of pieces in \(\I\) \emph{canonical} if it does not contain a special
piece.

\begin{definition}\label{def:plain}
    For an instance  {\(\SM(P, T, k, d, Q, \A_P, \A_T)\)}, write \(\Delta \coloneqq 6(k +
    d_P + d_T)\) and \(\tau \coloneqq q\ceil{\kappa/2q}\).
    We call a pair of pieces \emph{canonical} if it equals \[
        \mathcal{Q}\coloneqq (\Q\fragmentco0{\tau+\Delta},\Q\fragmentco0{\tau+\Delta}).
        \tag*{\lipicsEnd}
    \]
\end{definition}

Now, we obtain the following variant of
\cref{lem:red-warmup}.

\begin{lemma}\label{lem:red-warmup2}
    Given an instance {\(\SM(P, T, k, d, Q, \A_P, \A_T)\)},
    after a preprocessing step that takes \(\Oh(d + m/\max\{q, d\} )\) time
    in the \modelname model, we can compute \(\OccE_k(P, T)\) using
    a \DPM\!\!\((k, \Delta, \Sb, \Sm, \Sf)\) data structure,
    with \(\Delta + \ed(\Sb) + \ed(\Sm) + \ed(\Sf) = \Oh( d )\),
    that maintains a DPM-sequence of length \(z = \Oh( m/\max\{q, d\} )\)
    under
    \begin{itemize}
        \item \(\Oh(d \cdot m/\max\{q, d\})\) calls to {\tt DPM-Delete} and {\tt DPM-Insert} that delete and insert, respectively, a canonical pair,
        \item \(\Oh(d^2)\) calls to {\tt DPM-Substitute}, and
        \item \(\Oh(m/\max\{q, d\})\) calls to {\tt DPM-Query}.
    \end{itemize}
\end{lemma}
\begin{proof}
    In general, we follow the algorithm from \cref{lem:red-warmup}; however,
    we use a different method for transforming \(\I_{j-1}\) to \(\I_j\) (for \(j \in
    \fragmentoc{\min J}{\max J}\)).

    First, observe that, by \cref{lem:puzzlebeta,lem:puzzleeps}, we need to
    perform \(\Oh(\kappa)=\cO(d)\) calls to {\tt DPM-Substitute} for the head and the tail
    of the maintained DPM-sequence \(\I\).
    We can thus focus on updates involving pairs of internal pieces.

    Now, as an illustrative example, consider an internal pair
    \(\mathcal{P}_i^{j-1} \coloneqq (P_i, T_{j-1,i})\),
    and suppose that only \(T_{j-1,i}\) is special (that is, different from
    \(\Q\fragmentco{0}{\Delta  + \tau}\)). Further, suppose that the pairs
    \(\mathcal{P}_{i-1}^{j-1} \coloneqq (P_{i-1}, T_{j-1,i-1})\) and
    \(\mathcal{P}_{i+1}^{j-1}
    \coloneqq (P_{i + 1}, T_{j-1,i+1})\) are both canonical.
    Now, observe that in \(\I_j\), all pieces of \(P\) get aligned to a piece of \(T\)
    that is one piece to the right compared to \(\I_{j-1}\), that is, we need to construct
    the pairs
    \begin{align*}
        \mathcal{P}_{i-1}^{j} \coloneqq
        (P_{i-1}, T_{j,i-1}) =
        (P_{i-1}, T_{j-1,i}) = \mathcal{P}_{i}^{j-1}
        \quad\text{and}\quad
        \mathcal{P}_{i}^{j} \coloneqq  (P_{i}, T_{j,i}) =
        (P_{i}, T_{j-1,i + 1}) = \mathcal{Q}.
    \end{align*}
    Now, instead of calling {\tt DPM-Substitute} to directly create
    \(\mathcal{P}_{i-1}^{j}\) and \(\mathcal{P}_{i}^{j}\) (as we did in
    \cref{lem:red-warmup}), we {\tt DPM-Delete} the canonical pair \(\mathcal{P}_{i-1}^{j-1}\) and
    {\tt DPM-Insert} it after (the unchanged) \(\mathcal{P}_{i-1}^{j} =
    \mathcal{P}_{i}^{j-1}\).
    Observe that by doing so, we replaced the previous calls to {\tt DPM-Substitute} with
    an equal number of calls to {\tt DPM-Delete} and {\tt DPM-Insert} involving copies of
    \(\mathcal{Q}\).

    Generalizing the above example,
    unless both pieces \(P_i\) and \(T_{j-1,i}\) (or both
    \(P_i\) and \(T_{j,i}\)) are special (and if \(i\in\fragmentoo{2}{z-1}\)),
    we can replace the calls to {\tt DPM-Substitute}
    with an equal number of calls to {\tt DPM-Delete} and {\tt DPM-Insert} that delete and insert a copy of \(\mathcal{Q}\), respectively.
    Now, by \cref{lem:specialbound-simpl}, there are at most \(\Oh(d)\)
    special pieces in \(P\) and \(T\) each---hence, over all \(\I_j, j\in J\) in total,
    only at most \(\Oh(d^2)\)
    pairs contain two special pieces.
    Replacing all other calls to {\tt DPM-Substitute}, we obtain the claimed result.
\end{proof}

\begin{remark}\label{rem:upd-lb}
    Observe that in the worst case,
    no algorithm can iterate over all $\I_j$ for $j \in J$ with
    $o(\kappa \cdot m/\tau)$ updates.
    Consider the example, where for some constant~$c$,
    every $\lfloor c z /\kappa \rfloor$-th piece
    within each of the sequences $P_2,\ldots, P_{z-1}$ and
    $T_{\min J+2}, \ldots, T_{\max J+z-1}$
    is special and no two special pieces are equal.
    We can easily convince ourselves that for each pair $j_1,j_2 \in J$,
    we have $\ed(\I_{j_1},\I_{j_2})=\Omega(\kappa)$
    and hence $\Omega(\kappa \cdot m/\tau)$ updates are required for exactly computing
    each \(\I_j\).
    \lipicsEnd
\end{remark}

\begin{remark}\label{rem:trim-proto-idea}
    Observe that for each $j\in J$, the DPM-sequence $\I_j$ without its head and tail
    consists of $\cO(\kappa)$ pairs that contain at least one special piece and
    $\cO(\kappa)$ runs of canonical pairs.
    Then, roughly speaking the following claims hold (which we prove---in a more general
    setting---in~\cref{sec:druns}):
    each run of canonical pairs in $\I_{j-1}$ remains intact in $\I_j$ up to the potential insertion
    or deletion of a canonical pair.
    In addition, we can decide whether we need any of these two potential updates
    by inspecting the pairs adjacent to the run in scope.

    Thus, in the transformation of $\I_{j-1}$ to $\I_j$ we have to
    insert or delete $\cO(\kappa)$ canonical pairs.
    The crucial observation in \cref{sec:druns} now is that for very long runs of
    \(\mathcal{Q}\), inserting or deleting a single copy of \(\mathcal{Q}\) \emph{does not
    matter}, that is, the resulting \(k\)-error occurrences are essentially the same.
    This allows us to skip many of the calls to {\tt DPM-Insert}
    and {\tt DPM-Delete}, decreasing their number to a power~of~\(\Oh(d)\).
    Observe that this allows us to circumvent \cref{rem:upd-lb}, as we compute the
    sequences \(\I_j\) only approximately.
    \lipicsEnd
\end{remark}

\subsection{Solving {\SM} via {\DPM}, Improvement 1: Trimming Long Perfectly Periodic
Segments}
\label{sec:druns}

In this (sub-)section, we extend and formalize the idea that \cref{rem:trim-proto-idea}
hinted at. In particular,
we exploit that insertions and deletions of canonical pairs \(\mathcal{Q}\)
that extend or shrink a ``long'' run of canonical pairs
do not alter the answer to a \DPM query.
Based on this, we iterate over ``compressed'' versions of the sequences in scope,
while still being able to compute all sought approximate occurrences.
To simplify our exposition, let us define a \(\Comp\) operator for DPM-sequences.

\begin{definition}
    Consider a DPM-sequence \(\I\),
    where the elements of some subset \(\rpl(\I)\) of the internal canonical DPM-pairs are labeled as \emph{plain}.
    (We call all other pieces \emph{non-plain}.)

    For a positive integer \(\comp\), we write \(\Comps(\I, \rpl(\I), \comp)\) for
    the DPM-sequence obtained from \(\I\) by removing exactly one plain DPM-pair from any one
    contiguous subsequence of plain DPM-pairs of length at least \(\comp + 1\);
    if \(\I\) does not contain any such subsequence, we set \(\Comps(\I, \rpl(\I),
    \comp) \coloneqq \I\).
    (That is, \(\Comps(\I, \rpl(\I), \comp)\) is a DPM-sequence of length at least \(z - 1\).)

    Further, we write \(\Comp(\I, \rpl(\I), \comp) \coloneqq \Comps^{\star}(\I,
    \rpl(\I), \comp)\) for an iterated application of \(\Comps\) until the DPM-sequence
    remains unchanged. (That is, in \(\Comp(\I, \rpl(\I), \comp)\) every contiguous
    subsequence of plain DPM-pairs is of length at most \(\comp\).)
    \lipicsEnd
\end{definition}
\begin{remark}
    For technical reasons, we may not want to include all canonical pairs
    \(\mathcal{Q}\) in the set \(\rpl(\I)\). In particular, write \(\rred(P) \supseteq \sp{P}\) for
    a set of \emph{red} pieces of \(P\) and write \(\rred(T) \supseteq \sp{T}\) for a set
    of red pieces of \(T\). Now, for a \(j \in J\), we set \[
        \rpl(\I_j) \coloneqq \{ (P_i, T_{j,i}) \mid i \in \fragmentoo{1}{z} \text{ and }
            P_i \notin \rred(P) \text{ and } T_{j,i} \notin \rred(T).
        \] Abusing notation, we write \(\Comp(\I_j, \rred(P), \rred(T), \comp) \coloneqq
    \Comp(\I_j, \rpl(\I_j), \comp)\).
    \lipicsEnd
\end{remark}
\begin{remark}\label{rem:trimmed-length}
    It is easy to see that for every \(j \in J\),  we have \(|\Comp(\I_j, \rred(P),
    \rred(T), \comp)| = \Oh(\comp |\rred(P)|\, |\rred{T}|)\):  in the DPM-sequence \(\Comp(\I_j, \rred(P), \rred(T),
    \comp)\), the set \(\rpl(\I_j)\)
    contains at most \(\Oh(|\rred(P)| |\rred(T)|)\) sequences of contiguous pairs, each of
    length of at most \(\comp\). In particular, we have \(|\Comp(\I_j, \sp{P}, \sp{T}, \comp)| = \Oh( d^2 \comp)\) by \cref{lem:specialbound-simpl}.
    \lipicsEnd
\end{remark}

As a key result of this (sub-)section, we proceed to
show that we can trim the DPM-sequences \(\I_j\)
with \(\alpha = k + 1\) and still recover all \(k\)-error occurrences of \(P\) in \(R_j\).

\begin{restatable}{lemma}{druns}\label{fact:druns}
    For any \(j \in J\), we have
    \[
        \OccE_k(P, R_j) = \OccE_k(\I_j)
        = \OccE_k( \Comp(\I_j, \sp{P}, \sp{T}, k + 1) ).
        \tag*{\lipicsEnd}
    \]
\end{restatable}

Now, we can concisely state the main algorithmic result of this (sub-)section.

\begin{lemma}\label{lem:swaps-proto}
    Given an instance {\(\SM(P, T, k, d, Q, \A_P, \A_T)\)},
    after a preprocessing step that takes \(\Oh( d^3 \comp \log n )\) time
    in the \modelname model
    model, we can compute \[
        \mathcal{G} \coloneqq \bigcup_{j \in J}
        (r_j+\OccE_k(\Comp(\I_j, \sp{P}, \sp{T}, k + 1)))
    \] as \(\Oh(d^4)\) arithmetic
    progressions with difference \(\tau = \Oh(\max\{q,d\})\), using
    a \DPM\!\!\((k, \Delta, \Sb, \Sm, \Sf)\) data structure,
    with \(\Delta + \ed(\Sb) + \ed(\Sm) + \ed(\Sf) = \Oh( d )\),
    that maintains a DPM-sequence of length \(z = \Oh( d \comp )\) under
    \begin{itemize}
        \item \(\Oh(d^3)\) calls to {\tt DPM-Delete} and {\tt DPM-Insert} that delete and insert, respectively, a canonical pair,
        \item \(\Oh(d^2)\) calls to {\tt DPM-Substitute}, and
        \item \(\Oh(d^2) \) calls to {\tt DPM-Query}.
    \end{itemize}
    \lipicsEnd
\end{lemma}

We proceed with a proof of \cref{fact:druns};
formally, we prove the following, slightly stronger statement.

\begin{lemma}\label{lem:redundantk}
    Fix a string \(\hat{Q}\) and integers \(k \geq 0\) and \(\Delta > 0\).
    Further, consider a DPM-sequence \(\I\) whose pieces form \(\Delta\)-puzzles and
    that has a \torn of \(\tor(\I) \le \Delta/2 - k\), and a set \(\rpl(\I)\) of
    DPM-pairs labeled as plain, where \(\rpl(\I) \subseteq \{ (U_i, V_i) \in \I \mid
        i\in\fragmentoo1z \text{ and } U_i =
    V_i = \hat{Q}\} \).

    For any \(\alpha \ge 1\), we have
    \(\OccE_k(\Comps(\I, \rpl(\I), \alpha)) \supseteq \OccE_k(\I)\) and
    \(\OccE_k(\Comps(\I, \rpl(\I), k + \alpha)) = \OccE_k(\I)\).
\end{lemma}
\begin{proof}
    Let us assume that the \(\Comps\)-operation indeed removes a plain piece of \(\I\),
    otherwise there is nothing to prove. Further, let us fix an \(\alpha \ge 1\).
    Now, first, we convince ourselves that the \(\Comps\) operation produces valid
    \(\Delta\)-puzzles.

    \begin{claim}\label{clm:cutout}
        Given a \(\Delta\)-puzzle \(\mathcal{X} \coloneqq X_1,\dots,X_i, X_{i + 1},\dots,X_z\)
        with \(X_i = X_{i + 1}\) and value \(X \coloneqq \val_{\Delta}(\mathcal{X})\),
        the sequence \[
            \mathcal{X}' \coloneqq
            X_1,\dots,X_{i-1},X_{i+1},\dots,X_z
            = X_1,\dots,X_{i},X_{i+2},\dots,X_z
        \] forms a \(\Delta\)-puzzle with value \[
        \val_{\Delta}(\mathcal{X}') =
        X\fragmentco{0}{\xi_{i} + \floor{\Delta/2}} X\fragmentco{\xi_{i} + |X_{i}| -
        \ceil{\Delta/2}}{|X|} =
        X\fragmentco{0}{\xi_i + p} X\fragmentco{\xi_{i+1} + p}{|X|},
        \] where
        \(\xi_i \coloneqq \sum_{t=1}^{i-1} |X_t| - \Delta\) and \(p \in \fragment{0}{|X_{i}|
        - \Delta}\).
    \end{claim}
    \begin{claimproof}
        By definition, we have
        $X_i\fragmentco{|X_i|-\Delta}{|X_i|} = X_{i+1}\fragmentco{0}{\Delta}$ and
        $X_{i+1}\fragmentco{|X_{i+1}|-\Delta}{|X_{i+1}|} =
        X_{i+2}\fragmentco{0}{\Delta}$.
        As \(X_i = X_{i+1}\), we also have
        $X_i\fragmentco{|X_i|-\Delta}{|X_i|} = X_{i+2}\fragmentco{0}{\Delta}$,
        yielding the claim.
    \end{claimproof}

    Now, let us write \(\I = (U_1, V_1)(U_2, V_2)\cdots(U_z, V_z)\) and
    \(U \coloneqq \val_{\Delta}(U_1,\dots, U_z)\) and
    \(V \coloneqq \val_{\Delta}(V_1,\dots, V_z)\).
    Further, for each \(j \in \fragmentoo{1}{z}\), write $U\fragmentco{u_j}{u'_j} = U_j$ and
    $V\fragmentco{v_j}{v'_j} = V_j$;
    that is $u_j = \sum_{t=1}^{j-1}(|U_t|-\Delta)$ and
    $v_j = \sum_{t=1}^{j-1}(|V_t|-\Delta)$.
    With \cref{clm:cutout} in mind,
    we also write
    \(\hat{U}_j \coloneqq U\fragmentco{\hat{u}_j}{\hat{u}'_j}\), where
    \begin{align*}
        \hat{u}_j &\coloneqq u_j + \floor{\Delta/2}
        = \sum_{t=1}^{j-1}(|U_t|-\Delta) + \floor{\Delta/2} \qquad\text{and}\\
        \hat{u}'_j &\coloneqq \hat{u}_j + |U_j| - \Delta
        = \sum_{t=1}^{j-1}(|U_t|-\Delta) + |U_j|
        - \ceil{\Delta/2}
        = u'_j - \ceil{\Delta/2}.
    \end{align*}
    For convenience, we define \(\hat{u}'_1\) and \(\hat{u}_z\) analogously.
    Observe that we have \(\hat{u}'_j = \hat{u}_{j+1}\),
    that is, we obtain a partition of
    \(U = U\fragmentco{0}{\hat{u}'_1} \hat{U}_2 \cdots \hat{U}_{z-1}
    U\fragmentco{\hat{u}_z}{|U|}\).

    We proceed to show \(\OccE_k(\Comps(\I, \rpl(\I), \alpha)) \supseteq \OccE_k(\I)\).
    To that end, suppose that the \(\Comps\)-operation removes the DPM-pair \((U_i, V_i)\)
    from \(\I\). We may assume that \(U_i = U_{i+1} = V_i = V_{i+1}\), as the
    \(\Comps\)-operation may not fully remove a subsequence of plain DPM-pairs and we can
    relabel the remaining plain pair, if necessary.
    In particular, we can interpret \(\hat{U}_i\) as the fragment cut out of \(U\) due to
    the \(\Comps\)-operation.

    Now, write
    \begin{align*}
        U' &\coloneqq \val_{\Delta}(U_1,\dots,U_{i-1},U_{i+1}, \dots,
        U_z)\quad\text{and}\quad
        V' \coloneqq \val_{\Delta}(V_1,\dots,V_{i-1},V_{i+1},\dots, V_z),
    \end{align*}
    and consider an optimal alignment $\A : U \onto V\fragmentco{a}{b}$ of cost at most $k$.
    As \(\I\) has a bounded \torn, we can see that \(\A\) indeed aligns \(\hat{U}_i\)
    within \(V_i\)---in fact, it does so for every fragment \(\hat{U}_i\).
    \begin{claim}\label{clm:cutout2}
        For any \(j \in \fragmentoo1z\), the fragment $\hat{V}_j
        \coloneqq V\fragmentco{\hat{v}_j}{\hat{v}'_j}
        \coloneqq \A(\hat{U}_j)$ is contained in $V_j$.
    \end{claim}
    \begin{claimproof}
        Observe that we have by construction
        \begin{align*}
            v_j &= u_j + \sum_{t=1}^{j-1} (|V_t|-|U_t|)
            \le u_j + \sum_{t=1}^z \big||U_t|-|V_t|\big| = u_j + \tor(\I)
                \le u_j + \floor{\Delta/2} - k
                \le a+\hat{u}_j - k \le \hat{v}_j.
        \end{align*}
        Symmetrically, we obtain
        \begin{align*}
            \hat{v}'_j \le b+\hat{u}'_j-|U| +k
            &\le |V|-|U|+u'_j-\ceil{\Delta/2}+k\\
            &\le |V|-|U|+u'_j-\tor(\I) = |V|-|U|+u'_j-\sum_{t=1}^z \big||U_t|-|V_t|\big|)\\
            &\le |V|-|U|+u'_j - \sum_{t=j+1}^z (|V_t|-|U_t|)= v'_j,
        \end{align*}
        thus completing the proof.
    \end{claimproof}

    Next, we use \cref{cor:al-cut-out} to cut out the pair \((U_i, V_i)\)
    from \(\A\); we have
    \begin{align*}
        \ed(U, V\fragmentco{a}{b})
        &\ge \ed(U\fragmentco{0}{\hat{u}_i}, V\fragmentco{a}{\hat{v}_i})
        + \ed(U\fragmentco{\hat{u}_i+|U_i|-\Delta}{|U|},V\fragmentco{\hat{v}_i+|U_i|+\Delta}{b})
    \end{align*}
    Now, \cref{clm:cutout2} ensures that \(\hat{V}_i\) is indeed contained
    in \(V_i\). Adding that \(V_i\fragmentco{0}{|U_i-\Delta}V_{i+1}\)
    has a period \(|U_i| - \Delta\),
    we see that \(V\fragmentco{\hat{v}_i}{\hat{v}_i+|U_i| - \Delta}\) is a rotation
    of \(V_i\fragmentco{0}{|U_i|-\Delta}\). Thus, we indeed obtain \(V'\):
    \begin{align*}
        &\ed(U\fragmentco{0}{\hat{u}_i}, V\fragmentco{a}{\hat{v}_i})
        +
        \ed(U\fragmentco{\hat{u}_i+|U_i|-\Delta}{|U|},V\fragmentco{\hat{v}_i+|U_i|+\Delta}{b})\\
        &\quad= \ed(U'\fragmentco{0}{\hat{u}_i}, V'\fragmentco{a}{\hat{v}_i})
        + \ed(U'\fragmentco{\hat{u}_i}{|U'|},V'\fragmentco{\hat{v}_i}{b-|U_i|+\Delta})\\
        &\quad\ge \ed(U', V'\fragmentco{a}{b-|U_i|+\Delta}).
    \end{align*}
    In total, we constructed an alignment \(\A' : U' \onto
    V'\fragmentco{a}{b-|U_i|+\Delta}\) of cost at most \(k\)---and hence completed the
    proof of \(\OccE_k(\Comps(\I, \rpl(\I), \alpha)) \supseteq \OccE_k(\I)\).

    Now, to prove \(\OccE_k(\Comps(\I, \rpl(\I), k + \alpha)) = \OccE_k(\I)\),
    and more specifically
    \(\OccE_k(\Comps(\I, \rpl(\I), k + \alpha)) \subseteq \OccE_k(\I)\),
    we first observe that in \(\I\), the DPM-pair \((U_i, V_i)\) is part of a contiguous
    subsequence
    \(\X \coloneqq (U_{i-k'},V_{i-k'})\cdots(U_i,V_i)\cdots(U_{i + k''}, V_{i + k''})\) of
    length (at least) \(k' + k'' + 1\ge k + 2\),
    where all DPM-pairs that are equal to \((U_i, V_i)\).
    Again, write
    \begin{align*}
        U' &\coloneqq \val_{\Delta}(U_1,\dots,U_{i-1},U_{i+1}, \dots,
        U_z)\quad\text{and}\quad
        V' \coloneqq \val_{\Delta}(V_1,\dots,V_{i-1},V_{i+1},\dots, V_z),
    \end{align*}
    and consider an optimal alignment $\A : U' \onto V'\fragmentco{a}{b}$ of cost at most $k$.
    As the alignment \(\A\) can make at most \(k\) edits, at least one DPM-pair of the
    contiguous subsequence \[
        \X' \coloneqq (U_{i-k'},V_{i-k'})\cdots(U_{i-1},V_{i-1})
        (U_{i+1},V_{i+1})\cdots(U_{i + k''}, V_{i + k''})
    \]
    gets aligned without any edits; as all pairs of \(\X\) are the same, we may thus
    assume without loss of generality that \(\hat{V}_{i-1} = \hat{U}_{i-1}\).
    In particular, we can insert \((U_i, V_i)\) into \(\X'\) after \((U_{i-1}, V_{i-1})\)
    such that \(\hat{U}_i = \hat{V}_i\); that is (using \cref{fct:ali}), we have
    \begin{align*}
        \ed(U', V'\fragmentco{a}{b})
        &=\ed(U'\fragmentco{0}{\hat{u}_i}, V'\fragmentco{a}{\hat{v}_i})
        + 0 + \ed(U'\fragmentco{\hat{u}'_i}{|U|},V'\fragmentco{\hat{v}'_i}{b})\\
        &=\ed(U\fragmentco{0}{\hat{u}_i}, V\fragmentco{a}{\hat{v}_i})
        + 0 \\
        &\quad+
        \ed(U\fragmentco{\hat{u}'_i+|U_i|-\Delta}{|U|},V\fragmentco{\hat{v}'_i+|U_i|-\Delta}{b+|U_i|-\Delta})\\
        &\ge \ed(U, V\fragmentco{a}{b+|U_i| - \Delta}).
    \end{align*}
    In total, we constructed an alignment \(\A' : U \onto
    V\fragmentco{a}{b+|U_i|-\Delta}\) of cost at most \(k\)---and hence completed the
    proof of \(\OccE_k(\Comps(\I, \rpl(\I), k + \alpha)) \subseteq \OccE_k(\I)\).
\end{proof}

Combining \cref{fact:lengths_diff,lem:redundantk}, we obtain \cref{fact:druns}, which we
restate here for convenience.

\druns*

\begin{remark}\label{rem:overtrim}
    Observe that we could trim sequences of plain pairs to be even shorter than \(k + 1\).
    By \cref{lem:redundantk}, processing such a ``over-trimmed'' sequence naively
    might result in false-positive reports of \(k\)-error occurrences, though.
    In essence, the remainder of this part shows when and how we are able to avoid such
    false-positive reports when trimming with \(\comp = \widetilde{\Theta}(\sqrt{d})\).
    \lipicsEnd
\end{remark}

Observe that trimming the DPM-sequences alone is not enough to obtain
\cref{lem:swaps-proto}: in the interesting case where $\kappa^3=o(m/\tau)$,
we cannot afford to consider each of the $\cO(m/\tau)$ DPM-sequences
\(\I'_j \coloneqq \Comp(\I_j, \rred(P), \rred(T), k + 1)\)
separately.
However, for most values \(j \in J\), we have $\I'_{j-1}=\I'_j$.
Further, if $r_j-r_{j-1}=\tau$, then we have
$r_{j} + \OccE_k(P, R_{j}) = r_{j-1} + \OccE_k(P,R_{j-1}) + \tau$.
Consequently, we can return all occurrences compactly as arithmetic progressions.
(Also recall from \cref{fct:rj} that the sequence $(r_j)_{j\in J}$ is the concatenation of
$\Oh(d_T+1)$ arithmetic progressions with difference $\tau$.)
While these observations alone \emph{still} are not enough to obtain a faster algorithm
for \SM, they do yield an alternative algorithm with (roughly) the same running time as
\cref{lm:impEdC}.

With \cref{rem:overtrim} in mind, we give a slightly more general algorithm that allows
for an arbitrary threshold \(\comp\) up to which we may trim the runs of plain pairs in the
sequences \(\I_j\).
In particular, we prove the following more general variant of \cref{lem:swaps-proto}.

\begin{lemma}[$\protect\swaps(T, P, k, Q, \A_P, \A_T, \rred(P), \rred(T),
    \comp)$]\label{lem:swaps}
    Suppose we are given
    an instance {\(\SM(P, T, k, d, Q, \A_P, \A_T)\)}, a positive integer \(\comp\),
    as well as sets of puzzle
    pieces \(\rred(P) \supseteq \sp{P}\) and \(\rred(T) \supseteq \sp{T}\) of size \(\Oh(d)\)
    each.

    After a preprocessing step that takes \(\Oh( d^2 \comp \log n )\) time in the \modelname
    model, we can compute \[
        \mathcal{G} \coloneqq \bigcup_{j \in J}
        (r_j+\OccE_k(\Comp(\I_j, \rred(P), \rred(T), \comp)))
    \] as \(\Oh(d^3 \comp)\) arithmetic
    progressions with difference \(\tau = \Oh(\max\{q,d\})\), using
    a \DPM\!\!\((k, \Delta, \Sb, \Sm, \Sf)\) data structure,
    with \(\Delta + \ed(\Sb) + \ed(\Sm) + \ed(\Sf) = \Oh( d )\),
    that maintains a DPM-sequence of length \(z = \Oh( d \comp )\)
    under
    \begin{itemize}
        \item \(\Oh(d^2 \comp)\) calls to {\tt DPM-Delete} and {\tt
            DPM-Insert} that delete and insert, respectively, a canonical pair,
        \item \(\Oh(d^2)\) calls to {\tt DPM-Substitute}, and
        \item \(\Oh(d^2\comp) \) calls to {\tt DPM-Query}.
    \end{itemize}
    \lipicsEnd
\end{lemma}
\begin{proof}
    As in \cref{sec:warmup,sec:warmup2},
    we set $\mathcal{S}_\beta$, $\mathcal{S}_\mu$, $\mathcal{S}_{\varphi}$
    according to \cref{def:puzzleset} and choose $\Delta \coloneqq 6\kappa = 6(k + d_P +
    d_T)$.
    Further, for each \(j\in J\), set \(\I'_j \coloneqq \Comp(\I_j, \rred(P), \rred(T), \comp)\).

    In contrast to the warm-up algorithms, we initialize the \DPM data structure
    with \(\I \gets \I'_{\min J}\) (instead of \(\I_{\min J}\)).
    Now, in general we follow the approach of the algorithms from
    \cref{sec:warmup,sec:warmup2}: we transform \(\I\) to be equal to \(\I'_{j}\) and then
    call {\tt DPM-Query} to obtain the contained \(k\)-error occurrences.
    However, we skip over an \(\I'_j\) if it is equal to \(\I'_{j + 1}\) and instead
    extend the arithmetic progressions obtained from the last call to {\tt DPM-Query} by
    \(\tau\).

    \begin{claim}\label{claim:kaupdates}
        In $\cO(\kappa^2\comp \log n)$ time in total, we can compute,
        for all $j \in J\setminus \{\min J\}$ with
        $\I'_{j-1} \neq \I'_j$,
        a sequence $\textsf{update}_j$ of \DPM update operations that transforms
        $\I'_{j-1}$
        to
        $\I'_j$,
        such that if we perform these updates one by one,
        all intermediate sequences satisfy conditions (\ref{it:orig}) and (\ref{it:unique}) of \DPM.
        Further, these sequences of updates are returned in increasing order with respect to $j$
        and contain in total
        \begin{itemize}
            \item \(\Oh(d^2 \comp)\) calls to {\tt DPM-Delete} and {\tt
                DPM-Insert} that delete and insert, respectively, a canonical pair and
            \item \(\Oh(d^2)\) calls to {\tt DPM-Substitute}.
        \end{itemize}
    \end{claim}
    \begin{claimproof}
        We store $\I$ as a doubly-linked list $I$.
        We compute sequences that consist of the following types of updates:
        \begin{itemize}
            \item $\textsf{list-sub}(\textsf{pointer}, (P_i,T_{j,i}))$:
                given a handle $\textsf{pointer}$ to an element of $I$,
                substitute $(P_i,T_{j,i})$ for this element.
            \item $\textsf{list-ins}(\textsf{pointer}, (P_i,T_{j,i}))$:
                given a handle $\textsf{pointer}$ to an element of $I$,
                insert $(P_i,T_{j,i})$ before this element.
            \item $\textsf{list-del}(\textsf{pointer})$:
                given a pointer $\textsf{pointer}$ to an element of $I$,
                delete its successor in $I$.
        \end{itemize}
        While this interface is different from that of \DPM, each update can be mapped to a
        \DPM update as long as we are able to perform the following operation:
        given a handle to an element of $I$, return its rank, that is,
        the number of elements that precede it in $I$.
        This operation can be implemented in $\cO(\log n)$ time by maintaining a balanced binary
        search tree over the elements of $I$
        at the cost of an $\cO(\log n)$-time additive overhead for each update operation.

        In $\cO(\kappa \log \kappa)$ time, we store the elements of $\rred(P) \cup \{P_1,
        P_z\}$ in a doubly linked list $L_P$ and we store the elements of
        \(\rred(T)\) in a doubly linked list $L_X$, each
        sorted with respect to their starting positions.
        At all times, we store a bidirectional pointer between each element of each
        of $L_P$ and $L_T$ and the DPM-pair that contains it in $I$.
        (We do not explicitly mention when such pointers need to be updated.)
        For both \(P\) and \(T\), we write \(\textsf{pointer}(P_i)\) (or
        \(\textsf{pointer}(T_i)\))
        for the pointer from the element of $L_P$ (or \(L_T\)) corresponding to $P_i$ (or
        \(T_i\))  to the element of $I$ that contains $P_i$ (or \(T_i\)).

        We insert all updates in a global min-heap with keys in
        $\fragmentoc{\min J}{\max J} \times \mathbb{Z} \cup \{-\infty\}$.
        When we are done generating updates, we pop them from the heap, inserting
        an update with key $(j, p)$ in the min-heap $\update_j$ with priority $p$.
        These priorities ensure that the conditions on the maintained sequence are satisfied at
        all times.
        We can easily convince ourselves that this process indeed returns sequences
        $\textsf{update}_j$ in increasing order with respect to~$j$.

        Similar to the algorithm that was sketched in~\cref{sec:warmup2}, we consider a few types
        of updates to $\I$:
        \begin{itemize}
            \item $\cO(\kappa)$ substitutions of the head/tail of the sequence,
            \item $\cO(\kappa^2)$ updates that involve a pair that contains a red
                fragment, and
            \item $\cO(\kappa^2 \comp)$ insertions/deletions of a plain pair, effectively
                shrinking/expanding a run of plain pairs.
        \end{itemize}

        In order to develop some more intuition, let us examine how the length of the run
        of plain pairs that succeeds the
        head of $\I$ changes in the course of the algorithm.
        This length may decrease as a DPM pair~$A$ containing some $T_t \in \rred(T)$
        approaches the head of $\I$ while we slide $P$ on $T$.
        When such a DPM-pair disappears in the process of transforming
        $\I_{t-1}$ to $\I_{t}$,
        the length of the considered run might increase depending on where the leftmost
        pair with a red piece in $\I_t$ is;
        roughly speaking, the run of plain pairs after $A$ in $\I_{t-1}$ becomes the run of
        plain pairs that succeeds the head of $\I_t$.

        We say that a run of plain pairs is \emph{enclosed} by the two DPM-pairs
        that are adjacent to it.
        Suppose that in each of $\I_{j-1}$ and $\I_j$ there is a run of plain pairs that is
        enclosed by a pair that contains a fragment $P_i \in \rred(P)\cup \{P_1, P_z\}$
        and a pair that contains a fragment $T_t \in \rred(T)$.
        The lengths of these two runs differ by at most one; details are provided below.
        Observe that any run of plain pairs in $\I_{j-1}$ that is not enclosed by a pair
        that contains a fragment $P_i \in \rred(P)\cup \{P_1, P_z\}$ and a pair that contains
        a fragment $T_t \in \rred(T)$
        either remains intact in $\I_j$ or is shifted to the left by one position.
        In particular, each change to the length of a run as we transform $\I'_{j-1}$ to
        $\I'_j$ can
        be attributed to a single pair of fragments $P_i \in \rred(P)\cup \{P_1, P_z\}$
        and $T_t \in \rred(T)$
        getting closer to (or farther from) each other.

        We first issue updates for the head and the tail of the sequence.\footnote{Some
        of the issued updates might be redundant or duplicate, but this is not a problem.}
        To this end, we first compute the union $\textsf{HeadTail}$ of the sets
        \[
            \{ j \in J : T_{j,1} \neq Q^\infty\fragmentco{-\kappa}{2\tau+\Delta} \}
            \quad\text{and}\quad
            \{ j \in J : T_{j,z} \neq Q^\infty\fragmentco{z\tau}{y_P+\kappa} \},
        \]
        with its elements sorted in increasing order,
        in $\cO(\kappa)$ time using \cref{lem:puzzlebeta,lem:puzzleeps}.
        For each $j \in \textsf{HeadTail} \setminus \min{J}$, for each $e \in \{1,z\}$
        we issue an update
        $\textsf{list-sub}(\textsf{pointer}(P_e), (P_e,T_{j,e}))$
        with key $(j, \big| |T_{j,e}|-|P_e|\big| - \big||T_{j-1,e}|-|P_e|\big|)$.
        Further, for each $j \in \textsf{HeadTail} \setminus \max{J}$, for each $e \in \{1,z\}$,
        issue update $\textsf{list-sub}(\textsf{pointer}(P_e), (P_e,T_{j+1,e}))$
        with key $(j+1, \big| |T_{j+1,e}|-|P_e|\big| - \big||T_{j,e}|-|P_e|\big|)$.

        Let us now issue all updates that involve some red fragment.
        For each pair of fragments $P_i \in \rred(P)$ and $T_t \in \rred(T)$ we do the following.
        \begin{itemize}
            \item If $t-i \in \fragmentoc{\min J}{\max J}$, we issue the following updates:
                \begin{itemize}
                    \item $\textsf{list-sub}(\textsf{pointer}(P_i), (P_i,T_t))$ 
                        with key $(t-i, \big| |T_t|-|P_i|\big| - \big||T_{t-1}|-|P_i|\big|)$, and
                    \item if $i+1<z$ and $P_{i+1}\not\in \rred(P)$,	$\textsf{list-del}(\textsf{pointer}(P_i))$ 
                        with key $(t-i, -\infty)$.
                \end{itemize}
            \item If $t-i+1 \in \fragmentoc{\min J}{\max J}$, we issue the following updates:
                \begin{itemize}
                    \item if $i>2$ and $P_{i-1} \not\in \rred(P)$, $\textsf{list-ins}(\textsf{pointer}(P_i), (P_{i-1},T_t))$ 
                        with key $(t-i+1,\big||T_t|-|P_{i-1}|\big|)$, and
                    \item if $T_{t+1} \not \in \rred(T)$,  $\textsf{list-sub}(\textsf{pointer}(P_i), (P_i,T_{t+1}))$ 
                        with key $(t-i+1, \big| |T_{t+1}|-|P_i|\big| - \big||T_{t}|-|P_i|\big|)$.
                \end{itemize}
        \end{itemize}

        Next, we show how to compute updates that allow the algorithm to maintain the
        trimmed length of the run of plain pairs that is enclosed by a pair
        containing some $P_i \in \rred(P)\cup\{P_1, P_z\}$
        and a pair containing some $T_t \in \rred(T)$, when such a run exists.

        First, we consider the case where the pair containing $P_i$ precedes the pair containing $T_t=T_{j-1,t-(j-1)}$ in some $\I'_{j-1}$
        and these two pairs enclose a non-empty run of plain pairs.
        In this case, the run would either shrink in $\I'_j$ or retain its length $\comp$.
        Let the successor of $P_i$ in $L_P$ be $P_v$.
        Further, if $T_t$ is not the first element in $L_T$ let its predecessor in $L_T$ be $T_u = T_{j-1,u-(j-1)}$; otherwise, set $u\coloneqq -\infty$.
        The following conditions must be satisfied:
        \begin{enumerate}[(i)]
            \item $t - (j - 1) \leq v$, which is equivalent to $j > t - v$, so that $P_v$ is not in a pair strictly between the two pairs and hence none of the elements of $\rred(P)$ is,
            \item $u - (j - 1) \leq i $, which is equivalent to $ j > u - i$, so that
                there is no non-plain pair containing some element of $\rred(T)$ between the two pairs, and
            \item $t - (j - 1) - i > 1 $, which is equivalent to $ j \leq t - i - 1$, so that the precedence condition is satisfied and the two pairs are not adjacent.
        \end{enumerate}
        Further, the run should shrink in $\I'_j$ only if the two pairs are close enough,
        that is, if
        $t - (j - 1) - i -1 \leq \comp$, which is equivalent to $j > t - i - \comp - 1$.
        Hence, for each \[
            j \in \fragmentoc{\max\{ \min J, t - v, u - i, t - i - \comp - 1  \}}{\min \{
        \max J, t - i - 1 \}},\] we issue an update
        $\textsf{list-del}(\textsf{pointer}(P_i))$ 
        with key $(j,-\infty)$.
        Observe that the total number of such issued updates is at most $(t - i - 1)-(t - i - \comp - 1) = \comp$.

        We now consider the complementary case where in each of $\I'_{j-1}$ and $\I'_j$,
        the pair containing $P_i$ succeeds the pair containing $T_t$
        and these two pairs either enclose a non-empty run of plain pairs or are
        adjacent.\footnote{Observe that the pairs can only be adjacent in $\I'_{j-1}$.}
        In this case, the run would either be expanded by one plain pair in $\I'_j$ or retain its length $\comp$.
        Let the predecessor of $P_i$ in $L_P$ be $P_y$.
        Further, if $T_t$ is not the last element in $L_T$ let its successor in $L_T$ be $T_w = T_{j-1,w-(j-1)}$; otherwise, set $w\coloneqq \infty$.
        The following conditions must be satisfied:
        \begin{enumerate}[(i)]
            \item $t - (j - 1) > y $, which is equivalent to $ j \leq t - y$, so that $P_w$ is strictly to the left of both pairs in $\I'_{j-1}$, as if this is not the case then
                the condition on $\I'_j$ would not be satisfied,
            \item $w - (j - 1) > i $, which is equivalent to $ j \leq w - i$, so that $T_y$ (if it exists) is strictly to the right of both pairs in $\I'_{j-1}$, as if this is not the case then
                the condition on $\I'_j$ would not be satisfied,
            \item $t - (j - 1) < i $, which is equivalent to $ j > t - i + 1$, so that the precedence condition is satisfied.
        \end{enumerate}
        Further, the run should be expanded in $\I'_j$ only if the two pairs are close enough,
        that is, if
        $i - (t - (j - 1)) -1 < \comp $, which is equivalent to $ j \leq \comp + t - i + 1$.
        Hence, for each \[
            j \in \fragmentoc{\max\{ \min J, t - i + 1 \}}{\min \{ \max J, t - y, w - i,
        \comp + t - i + 1 \}},\] we issue an update
        $\textsf{list-ins}(\textsf{pointer}(P_i),
        (Q^\infty\fragmentco{0}{\tau+\Delta},Q^\infty\fragmentco{0}{\tau+\Delta}))$ 
        plain pair in the (potentially empty) run preceding the pair that contains $P_i$
        with key $(j,-\infty)$.
        Observe that the total number of such issued updates is at most $(\comp + t - i + 1)-(t - i + 1) = \comp$.

        Clearly, we can compute all updates for a given pair (in the intermediate interface)
        in $\cO(\comp)$ time.

        The conditions of \DPM are clearly satisfied at initialization.
        Condition (\ref{it:orig}) is satisfied at all times as all the updates involving the head or tail of
        only involve pairs in $\mathcal{S}_\beta \times \mathcal{S}_\beta$ and
        $\mathcal{S}_{\varphi} \times \mathcal{S}_{\varphi}$, respectively,
        while all remaining updates involve pairs in $\mathcal{S}_\mu \times \mathcal{S}_\mu$.
        Condition (\ref{it:unique}) is satisfied at all times; it is satisfied for each constructed $\I'_j$ by~\cref{fact:lengths_diff}
        as $\I'_j$ is a subsequence of $\I_j$, and for all other constructed sequences by
        a direct application of~\cref{fact:interm}, which is applicable because of how the updates in each $\textsf{update}_j$ are sorted.
    \end{claimproof}

    We are now ready to complete the reduction to \DPM; see~\cref{alg:arprogr} for a pseudocode implementation.

    Using \cref{claim:kaupdates}, we initialize the
    set \[
        \mathcal{U}=\{\textsf{update}_j : j\in \fragmentoc{\min J}{\max J} \text{ and }
        \I'_{j-1} \neq \I'_j\}.
    \]
    To simplify the construction of certain arithmetic progressions later in the proof,
    we insert to $\mathcal{U}$ the following single-element (void) sequences of updates.
    For each $j \in \fragmentoc{\min J}{\max J}$ such that $r_j - r_{j-1}\neq \tau$ and $\I'_{j-1} = \I'_j$,
    we insert to $\mathcal{U}$ the sequence $\textsf{update}_j$ that consists of a single
    element: {\tt DPM-Substitute(\((P_1,T_{j,1})\),\(1\))}, that is,
    the substitution of the first pair of the maintained sequence with $(P_1,T_{j,1})$.
    By \cref{fct:rj}, we only generate $\cO(d_T+1)$ new updates.

    Now, we initialize the sequence $\I$ with $\I'_{\min J}$.
    We consider the sequences of updates $\textsf{update}_j$ in $\mathcal{U}$
    in increasing order with respect to~$j$
    and apply them to the maintained sequence~$\I$.
    Prior to performing the sequence of updates $\textsf{update}_j$ we do the following.
    Suppose that the maintained sequence $\I$ corresponds to $\I'_{t}$, that is,
    either $\textsf{update}_j$ is the first element of $\mathcal{U}$ and $t=\min J$
    or the previously applied sequence of updates was $\textsf{update}_{t}$.
    First, we compute $\OccE_k(\I'_t)$ using a \texttt{DPM-Query}.
    Observe that we have $\I'_t = \cdots = \I'_{j-1}$
    and hence $\OccE_k(\I'_t) = \cdots = \OccE_k(\I'_{j-1})$.
    Further, for each
    $i \in \fragmentoo{t}{j}$, we have $r_i-r_{i-1} = \tau$ and hence
    $r_i + \OccE_k(\I'_i) = r_t + \OccE_k(\I'_t) + (i-t)\tau$.
    We can thus efficiently return
    $\bigcup_{i = t}^{j-1} (r_i+\OccE_k(\I'_t))$
    as the union of $\cO(|\OccE_k(\I'_i)|)$ arithmetic progressions
    \[
        \bigcup_{a \in \OccE_k(\I'_t)} \{r_t+a + i \cdot \tau : i \in \fragmentco{0}{j-t}\}.
    \]
    (Observe that in the case where $t=j-1$, each of the constructed arithmetic progressions consists of a single element.)
    Finally, we apply the sequence of updates $\textsf{update}_j$.
    The case where there are no more updates to be processed is treated analogously (with $j$ set to $\max J+1$).
    The upper bound on the number of returned arithmetic progressions is embedded in the analysis of the time complexity of the algorithm.

\begin{algorithm}[t!]
    \SetKwBlock{Begin}{}{end}
    \SetKwFunction{DPMquery}{DPM-Query}
    \swaps{$T$, $P$, $k$, $Q$, $\A_P$, $\A_T$, $\rred(P)$, $\rred(T)$, $\comp$}\Begin{
		$\mathcal{G} \gets \emptyset$\;
		Construct set $\mathcal{U}$ of sequences $\textsf{update}_j$ specified in \cref{claim:kaupdates}\tcp*{$\cO(\kappa^2 \comp \log n)$ time}
		\ForEach{$j \in \fragmentoc{\min J}{\max J}$ such that $r_j - r_{j-1}\neq \tau$ and $\I'_{j-1} = \I'_j$}{\label{line:void1}
			$\textsf{update}_j \gets \{${\tt DPM-Substitute(\((P_1,T_{j,1})\),\(1\))}\(\}\)\;
			$\mathcal{U} \gets \mathcal{U} \cup \{\textsf{update}_j\}$\;\label{line:void2}
		}
    	$\I \gets \I'_{\min J}$\tcp*{$\cO(\kappa \comp + \kappa \log \kappa)$ time}
		$t \gets \min J$\;
		\ForEach{$j \in \{ i \in \fragmentoc{\min J}{\max J + 1} : \textsf{update}_i \neq \emptyset \text{ or } i = \max J + 1\}$ in incr.~order}{\label{line:for1}
			$\OccE_k(\I'_j) \gets \DPMquery$\;\label{line:dpmq}
			\ForEach{$a \in \OccE_k(\I'_j)$}{
           		$\mathcal{G} \gets \mathcal{G} \cup \{r_{t} + a + i \cdot \tau : i \in \fragmentco{0}{j-t}\}$\;\label{line:arprogr}
           	}
			\If {$j \leq \max J$}{
				Apply the sequence of updates $\textsf{update}_{j}$ on $\I$\;
				$t \gets j$\;\label{line:for2}
			}
		}
        \Return{$\mathcal{G}$}\;
    }
    \caption{Reduction of computing $\mathcal{G}$ to an instance of \DPM{}.}\label{alg:arprogr}
\end{algorithm}

We now proceed to analyze the time required by \cref{alg:arprogr}.

First, recall that the elements of $\mathcal{U}$ as returned by \cref{claim:kaupdates} in $\cO(\kappa^2 \comp \log n)$ time are sorted in increasing order with respect to $j$ and a representation of $(r_j)_{j\in J}$
as the concatenation of $\Oh(d_T+1)$ arithmetic progressions with difference $\tau$
can be computed in $\Oh(d_T+1)$ time in the \modelname model due to \cref{fct:rj}.
Thus, Lines~\ref{line:void1}--\ref{line:void2} can be implemented in $\cO(\kappa^2 \comp)$ time in a parallel left-to-right scan of $\mathcal{U}$ and the aforementioned representation of $(r_j)_{j\in J}$,
maintaining the sortedness of $\mathcal{U}$.

Next, observe that $\I'_{\min J}$ is of length $\cO(\kappa \comp)$ as it consists of a head,
a tail, $\cO(\kappa)$ non-plain pairs and $\cO(\kappa)$ runs of plain pairs, each
of length at most $\comp$.
This sequence can be computed in $\cO(\kappa \comp + \kappa \log \kappa)$
time: the head and the tail are computed in $\cO(\log (\kappa +1))$ time due to
\cref{lem:puzzleP-algo,lem:puzzlebeta,lem:puzzleeps}, while the pairs consisting of internal pieces can be computed
by processing the elements of $\rred(P)$, $\rred(T)$ in a left-to-right manner, after sorting them
with respect to their starting positions in $\cO(\kappa \log \kappa)$ time.

The for-loop of Lines~\ref{line:for1}--\ref{line:for2} is executed $\cO(\kappa^2 \comp)$ times
as this is an upper bound on the number of sequences of updates that are processed.
For each $j\in J$, due to \cref{fact:lengths_diff}, we have
$\sum_{i=1}^{z_j}\big||F_{j,i}| - |G_{j,i}|\big| \leq \sum_{i=1}^{z}\big||T_{j,i}| - |P_{i}|\big| \leq 3\kappa-k$
and hence $|\OccE_k(F_j,G_j)|=\cO(\kappa)$.
Thus, each query in \cref{line:dpmq} returns a set of size $\cO(\kappa)$.
Hence, apart from the time required for \DPM updates and queries, each iteration of the for-loop takes $\cO(\kappa)$ time,
for a total of $\cO(\kappa^3 \comp)$ time.
It readily follows that the output of the algorithm consists in $\cO(\kappa^3 \comp)$ arithmetic progressions.

All in all, in $\cO(\kappa^3 \comp + \kappa^2 \comp \log n)=\cO(d^3\comp+d^2\comp \log n)$ time, the problem in scope
reduces to an instance of \DPM
with $\mathcal{S}_\beta+\mathcal{S}_\mu+\mathcal{S}_{\varphi}=\cO(d)$ (due to
\cref{lem:specialbound-simpl}), $\Delta=6\kappa = \Oh(d)$,
$\I$ initialized as a sequence of length $\cO(d \comp)$, and
$\cO(d^2 \comp)$ updates and queries.
\end{proof}

As a direct consequence, we obtain an alternative algorithm for the \SM problem with a
competitive running time, that is, we obtain an algorithm for \SM, which is slower
than the algorithm implied by \cref{lm:impEdC} only by a poly-logarithmic factor.

\begin{corollary}
    We can solve the \SM problem in $\cO(d^4 \log n \log d)$ time in the \modelname model.
\end{corollary}
\begin{proof}
    By \cref{fact:druns}, in order to solve an instance of \SM,
    it suffices to construct sets
    $\textsf{Special}(P)$ and $\textsf{Special}(T)$ in $\cO(d)$ time in the \modelname model using
    \cref{lem:specialbound-simpl}  and
    call the algorithm underlying~\cref{lem:swaps} with $\rred(X) \coloneqq
    \textsf{Special}(X)$ for each $X \in \{P,T\}$ and $\comp \coloneqq k+1$.

    Using \cref{thm:dpm} for the \DPM problem, we need
    $\cO(d^3\log^2 d)$ time for preprocessing,
    $\Oh(d^4 \log d)$ time for initialization,
    and $\cO(d^3 \cdot d \log n \log d)$ time for processing updates and queries.
    Overall, the required time is thus $\cO(d^4 \log n \log d)$.
\end{proof}

In the following sections, our goal is to reduce the \SM problem to an instance
$\swaps(T, P, k, Q, \A_T, \A_P, \rred(P), \rred(T), O(\sqrt{d}))$ with $|\rred(P)| + |\rred(T)|= \cO(d)$.

\section{Faster {\SM}: Additional Combinatorial Insights}\label{sec:marking}

While the improved algorithms in \cite{unified} crucially relied on analyzing and
understanding the structure of the \emph{pattern}, we obtain the improvements in this work
by additionally analyzing and understanding the structure of the \emph{text}.
Throughout this section, we fix a text \(T\), a pattern \(P\), a primitive string \(Q\), and parameters $d_P$, $d_T$, and $d$
stemming from an instance of \SM.
Our goal is to classify each position of the text where a $k$-error occurrence of $P$ can start
(cf.~\cref{ft:no-occs-suffix} in this regard) as either \emph{heavy} or \emph{light}, using the notion of \emph{locked fragments},
depending on some parameter $\tradeoff$.
The heavy positions are few and can be covered by a few \emph{heavy ranges} of small total length.
The set of light positions may be large, but, as explained and exploited in \cref{sec:faster},
for each light position $v \in \OccE_k(P,T)$,
any alignment $\A: P \onto T\fragmentco{v}{w}$ of cost \(\min_t \ed(P,T\fragmentco{v}{t})\)
is quite restricted: it does not makes edit operations outside the vicinity (interpreted as $\cO(\tradeoff \tau)$ positions)
of locked fragments of the text and the pattern.

\subsection{Locked Fragments and their Properties}

We intend to understand and analyze the structure of the text by using a more elaborate
version of the marking schemes used in \cite{unified}. In particular, we heavily rely on
the notion of \emph{locked fragments} in text and pattern from \cite{unified}.

\begin{definition}[{\cite[Definition~5.5]{unified}}]
    Let $S$ denote a string and let $Q$ denote a primitive string.
    We say that a fragment $L$ of~$S$ is \emph{locked} (with respect to $Q$)
    if at least one of~the following holds:
    \begin{itemize}
        \item For some integer $\alpha$, we have $\edl{L}{Q}=\ed(L,Q^\alpha)$.
        \item The fragment $L$ is a suffix of~$S$ and $\edl{L}{Q}=\ed(L,Q^*)$.
        \item The fragment $L$ is a prefix of~$S$ and $\edl{L}{Q}=\eds{L}{Q}$.
        \item We have $L = S$.\lipicsEnd
    \end{itemize}
\end{definition}

Let us also recall an intuitive example from \cite{unified}.

\begin{example}[\cite{unified}]
Fix a primitive string \(Q\) and consider a string \(U=Q^{k+1} S Q^{k+1}\) with
$\edl{U}{Q} \leq k$.
In any optimal alignment of~$U$ with a substring of~$Q^\infty$ with at most \(k\) edits,
at least one of~the leading $k+1$ occurrences of~$Q$ in $U$ is matched exactly and
at least one of the trailing \(k + 1\) occurrences of~$Q$ in $U$ is matched exactly.
Hence, all occurrences preceding (or succeeding) said exactly matched occurrence of \(Q\)
are also matched exactly. Thus, $U$ is locked with respect to $Q$.
\lipicsEnd
\end{example}

As in \cite{unified}, we also need the slightly stronger notion of an \emph{$h$-locked prefix} of a string.

\begin{definition}[{\cite[Definition~5.10]{unified}}]\label{def:klocked}
    Let $S$ denote a string, let $Q$ denote a primitive string,
    and let $h\ge 0$ denote an integer.
    We say that a prefix $L$ of~$S$ is \emph{$h$-locked} (with respect to $Q$)
    if at least one of~the following holds:
    \begin{itemize}
        \item For every $p\in \fragmentco{0}{|Q|}$, if $\ed(L,\rot^p(Q)^*)\le h$, then
            $\ed(L,\rot^p(Q)^*)=\ed(L,Q^\infty\fragmentco{-p}{j|Q|})$ for some \(j\in\Z\).
        \item We have $L=S$.\lipicsEnd
    \end{itemize}
\end{definition}

Given a string \(S\) (that is either the pattern or the text),
we intend to construct locked fragments
covering all errors of $S$ with respect to $\!{}^*\!Q^*\!\!$,
such that the total length of these locked fragments is roughly proportional to the product
of $|Q|$ and $\edl{S}{Q}$.

\SetKwFunction{locked}{Locked}

\begin{lemma}[{\tt Locked($S$, $Q$, $d$, $h$)}, {\cite[Lemma 6.9]{unified}}]\label{lem:klocked}
    Let $S$ denote a string, let $Q$ denote a primitive string,
    let $d_S$ denote a positive integer such that $\edl{S}{Q}\le d_S$
    and $|S| \ge (2d_S+1)|Q|$, and let $h\in \Zz$.

    Then, there is an algorithm that computes
    disjoint locked fragments $L_1,\ldots,L_{\ell} \preceq S$
    such that
        \begin{itemize}
        \item $S=L_1 \cdot \bigodot_{i=1}^{\ell-1} (Q^{\alpha_i} L_{i+1})$ for
            positive integers $\alpha_1,\ldots,\alpha_{\ell-1}$;\footnote{This item is not
            stated in \cite[Lemma 6.9]{unified}. However, it readily follows from the
        construction algorithm underlying that lemma and it is already used in
    \cite{unified} (for instance, in the proof of \cite[Claim 5.17]{unified}). We believe
that adding a proof of this item here would not be instructive.}
        \item $L_1$ is an $h$-locked prefix of $S$ and $L_{\ell}$ is a suffix of~$S$;
        \item
            \(\edl{S}{Q}=\sum_{i=1}^{\ell} \edl{L_i}{Q}\) and  $\edl{L_i}{Q} > 0$ for $i\in \fragmentoo{1}{\ell}$; and
        \item \(\displaystyle
                \sum_{i=1}^\ell |L_i| \le (5|Q|+1)\edl{S}{Q} + 2(h+1)|Q|.
            \)
    \end{itemize}

    The algorithm takes $\Oh(d_S^2+h)$ time in the \modelname model.
    \lipicsEnd
\end{lemma}

\subsection{Analyzing the Text Using Locked Fragments}

Our marking scheme is similar to the marking scheme used in
the proof of \cite[Theorem~5.2]{unified}. As both \(T\) and \(P\) are close to being
periodic, we can compute locked fragments with respect to a common string \(Q\) (which is
given as a parameter in the call of \SM).

We start with some intuition for how different locked fragments of \(T\) and \(P\)
influence the edit distance between \(T\) and \(P\).
To that end, let \(\mathcal{L}^P\) denote the set of locked fragments computed for \(P\),
let \(\mathcal{L}^T \) denote the set of locked fragments computed for \(T\),
and consider an alignment $\A: P \onto T\fragmentco{t}{t'}$.
If \(\A\) aligns a locked fragment \(L^P \in \mathcal{L}^P\) to a fragment of
\(Q^{\infty}\), that is, $\A(L^P)$ does not overlap any locked fragment of
\(T\fragmentco{t}{t'}\), we obtain \(\ed^{\A}(L^P,\A(L^P))\geq \edl{L^P}{Q}\).
Symmetrically, any locked fragment $L^T \in \mathcal{L}^T$ of $T\fragmentco{t}{t'}$,
for which $\A^{-1}(L^T)$ does not overlap any locked fragment of $P$
satisfies $\ed^{\A}(\A^{-1}(L^T),L^T)\geq \edl{L^T}{Q}$.

However, if \(\A\) aligns a substring \(U\) of a locked fragment \(L^P \in \mathcal{L}^P\)
to a substring \(V\) of a locked fragment $L^T \in \mathcal{L}^T$,
the situation is not as straightforward. (One can consider the cleaner case where $U$ is $L^P$ and $V$ is $L^T$.)
Suppose for simplicity that there exist integers \(x\) and \(y\) such that $\edl{U}{Q}=\ed(U,Q^\infty\fragmentco{x}{y})$
and $\edl{V}{Q}=\ed(V,Q^\infty\fragmentco{x}{y})$.
In this case, by the triangle inequality, we have that
\[\edl{U}{Q} + \edl{V}{Q} - \min\{\edl{U}{Q},\edl{V}{Q}\} \leq \ed(U,V)\leq \edl{U}{Q} + \edl{V}{Q}.\]
In other words, by aligning these substrings of locked fragments, we can hope to ``save'' at
most \(\min\{\edl{U}{Q},\) \(\edl{V}{Q}\}\) edits compared to the $\edl{U}{Q} + \edl{V}{Q}$ upper
bound (when aligning the same locked fragments to substrings of \(Q^{\infty}\) instead).

To quantify said potential ``savings'', we give marks to each position $t$ in the text
corresponding to the total number of marks potentially saved in alignment
$\A: P \onto T\fragmentco{t}{t'}$.
Roughly speaking, for each pair of fragments
\(L^T\) and \(L^P\), we place \(\min\{\edl{L^T}Q, \edl{L^P}Q\}\) marks at position \(t\),
if they may overlap in any alignment $\A: P \onto T\fragmentco{t}{t'}$ with at most \(\ktotm\) insertions and deletions;
\(\ktotm\) can be thought to be $\cO(d)$.
In what follows, when we check whether two fragments overlap, we thus allow for a (small) slack~\(\ktotm\).
Formally, the marking scheme is captured in \cref{def:marking}.

\begin{definition}\label{def:marking}
    For a text \(T\), a pattern \(P\), a primitive string \(Q\) with \(\edl{P}{Q} =
    d_P\) and \(\edl{T}{Q} = d_T\), and corresponding sets of locked
    fragments \(\mathcal{L}^P \coloneqq \locked(P,Q,d_P,\lpref)\) and \(\mathcal{L}^T \coloneqq \locked(T,Q,d_T,0)\),
    write \(\markf\) for the function that
    maps an integer $v$ to a (weighted) number of
    locked fragments in \(\mathcal{L}^P\) that
    (almost) overlap locked fragments in \(\mathcal{L}^T\) when aligning \(P\) to
    position $v$.

    Formally, we first define a function
    \(\markf: \Z \times \Zz \times \mathcal{L}^P\times  \mathcal{L}^T \to \Zz \) by
    \begin{multline*}
        \markf(v,\ktotm, L^P = P\fragmentco{\ell_P}{r_P}, L^T = T\fragmentco{\ell_T}{r_T} )\\
         \coloneqq \begin{cases}
            \edl{L^T}{Q},&\text{if } \ell_P=0 \text{ and }  v\in \fragmentoo{\ell_T - \ktotm-r_P}{r_T + \ktotm-\ell_P}\text{;}\\
            \min\{ \edl{L^T}{Q}, \edl{L^P}{Q} \}, &\text{if } \ell_P \neq 0 \text{ and }  v\in \fragmentoo{\ell_T - \ktotm-r_P}{r_T + \ktotm-\ell_P} \text{;}\\
        0,& \text{otherwise.}
            \end{cases}
        \end{multline*}
     Now, set
        \[\markf(v,\ktotm,\mathcal{L}^P,\mathcal{L}^T)
        \coloneqq \sum_{L^P \in \mathcal{L}^P} \sum_{L^T \in \mathcal{L}^T}
        \markf(v,\ktotm, L^P, L^T ).\tag*{\lipicsEnd}\]
\end{definition}

When $\ktotm$, $\mathcal{L}^P$, and $\mathcal{L}^T$ are clear from context we may just write
$\markf(v)$ for $\markf(v,\ktotm,\mathcal{L}^P,\mathcal{L}^T)$.

We continue with a set of useful observations about our marking scheme.
Fix sets of locked
fragments \(\mathcal{L}^P \coloneqq \locked(P,Q,d_P,\lpref)\) and \(\mathcal{L}^T \coloneqq \locked(T,Q,d_T,0)\),
with $\lpref = \cO(d)$.

\begin{lemma}\label{lem:mark-pos}
For every $\ktotm\in \Zz$, we have
\[\sum_{v\in \Z} \markf(v,\ktotm,\mathcal{L}^P,\mathcal{L}^T) \le 2d_T(d_P+2)\ktotm + 2d_T(6d_P+3d_T+\lpref+2)|Q| + 2d_P|Q| = \cO(d^2(\ktotm + |Q|)).\]
\end{lemma}
\begin{proof}
    Fix a locked fragment
    \(L^P = P\fragmentco{\ell_P}{r_P}\in \mathcal{L}^P\) and
    a locked fragment
    \(L^T = P\fragmentco{\ell_T}{r_T}\in \mathcal{L}^T\).
    If $\ell_P\ne 0$, then the definition of \(\markf\) yields
    \begin{align*}
        \sum_{v\in \Z}
        \markf(v,\ktotm, L^P, L^T)
            & \le \sum_{v\in \fragmentoo{\ell_T-\ktotm -r_P}{r_T+\ktotm-\ell_P}}
            \min\{ \edl{L^T}{Q}, \edl{L^P}{Q} \} \\
            & \le \min\{ \edl{L^T}{Q}, \edl{L^P}{Q} \} (|L^T| + |L^P| + 2\ktotm)\\
            & \le \edl{L^T}{Q}(|L^P| + 2\ktotm) + \edl{L^P}{Q}|L^T|.
    \end{align*}
    Similarly, if $\ell_P =0$, then the definition of \(\markf\) yields
    \begin{align*}
        \sum_{v\in \Z}
        \markf(v,\ktotm, L^P, L^T)
            & \le \sum_{v\in \fragmentoo{\ell_T-\ktotm -r_P}{r_T+\ktotm-\ell_P}}
            \edl{L^T}{Q} \\
            & \le \edl{L^T}{Q}(|L^T| + |L^P| + 2\ktotm)\\
            & \le \edl{L^T}{Q}(|L^P| + 2\ktotm) + d_T|L^T|.
    \end{align*}
    Overall, we have
    \begin{align*}
        \sum_{a\in \Z} \markf(v,\ktotm,\mathcal{L}^P,\mathcal{L}^T)& \le \sum_{L^P \in \mathcal{L}^P} \sum_{L^T \in \mathcal{L}^T} (\edl{L^T}{Q}(|L^P| + 2\ktotm)+\edl{L^P}{Q}|L^T|) + d_T\sum_{L^T \in \mathcal{L}^T}|L^T|\\
                                                                   & \le d_T\sum_{L^P \in \mathcal{L}^P}(|L^P|+2\ktotm) + (d_P+d_T)\sum_{L^T \in \mathcal{L}^T}|L^T|\\
                                                                   & \le d_T(d_P+2)2\ktotm + d_T\sum_{L^P \in \mathcal{L}^P}|L^P| + (d_P+d_T)\sum_{L^T \in \mathcal{L}^T}|L^T|\\
                                                                   & \le 2d_T(d_P+2)\ktotm + d_T(5|Q|+1)d_P+2d_T(\lpref+1)|Q|+(d_P+d_T)(5|Q|+1)d_T \\
                                                                   & \qquad\qquad\qquad\qquad\qquad\qquad\qquad\qquad\qquad\qquad\qquad\quad + (d_P+d_T)2|Q|\\
                                                                   & \le 2d_T(d_P+2)\ktotm + 2d_T(6d_P+3d_T+\lpref+2)|Q|+2d_P|Q|.\qedhere
    \end{align*}
\end{proof}

\begin{definition}\label{def:restrset}
For a set $\mathcal{U}$ of fragments of a string $S$ and an interval $I\sub \Z$,
we write
$\res{\mathcal{U}}{I}=\{S\fragmentco{x}{y} \in \mathcal{U} : \fragmentco{x}{y} \cap I \neq \emptyset\}$.
\lipicsEnd
\end{definition}

Recall that, as stated in \cref{ft:no-occs-suffix}, $\OccE_k(P,T) \cap \fragmentoo{n-m+k}{n} = \emptyset$.
We partition the remaining positions of the text into two groups.
Intuitively, a position $t \in \fragment{0}{n-m+k}$ is \emph{light} if it has a few marks,
and there cannot be an alignment $\A: P \onto T\fragmentco{t}{t'}$ with at most $\ktotm$ insertions and deletions
that aligns the first position of $P$, the last position of $L^P_1$, or the last position of $P$ against a portion of a locked fragment of $T$.

\begin{definition}\label{def:light}
For any fixed thresholds $\ktotm\in \Zz$ and $\tradeoff\in \Zp$, we say that a position \(v\in \fragment{0}{n-m+k}\) is
\emph{light} if the following conditions are simultaneously satisfied:
\begin{itemize}
    \item $\markf(v, \ktotm,\mathcal{L}^P,\mathcal{L}^T) < \tradeoff$,
    \item $\res{\mathcal{L}^T}{\fragmentco{v-\ktotm}{v+\ktotm}}=\emptyset$, 
    \item $\res{\mathcal{L}^T}{\fragmentco{v+m-\ktotm}{v+m+\ktotm}}=\emptyset$, and 
    \item $\res{\mathcal{L}^T}{\fragmentco{v+|L_1^P|-\ktotm}{v+|L_1^P|+\ktotm}}=\emptyset$. 
\end{itemize}
Otherwise, the position $v\in \fragment{0}{n-m+k}$ is called \emph{heavy}.
We denote the sets of heavy and light positions by $\Hv$ and $\light$, respectively. \lipicsEnd
\end{definition}

\SetKwFunction{heavy}{Heavy}
\begin{lemma}[$\protect\heavy(P,T,k,d,Q,\LP,\LT,\ktotm,\tradeoff)$]\label{lem:heavy-alg}
Consider an instance of the \SM problem, families $\LP = \locked(P,Q,d_P,\lpref)$ and $\LT = \locked(T,Q,d_T,0)$, as well as thresholds $\ktot\in \Zz$ and $\tradeoff\in \Zp$,

The set $\Hv$ of heavy positions, represented as the union of $\Oh(d^2)$ disjoint integer ranges,
can be computed in $\Oh(d^2 \log \log d)$ time.
\end{lemma}
\begin{proof}
We implement the marking process according to \cref{def:marking}, assigning $\tradeoff$ extra marks to all positions made heavy due to the last three items in \cref{def:light}.
Formally, we produce the following $\Oh(d^2)$ weighted intervals:
\begin{itemize}
    \item $\fragmentoo{\ell_T-\ktotm-r_P}{r_T+\ktotm-\ell_P}$ of weight $\min\{ \edl{L^T}{Q}, \edl{L^P}{Q} \}$ for each locked fragment $L^P=P\fragmentco{\ell_P}{r_P}\in \LP$ with $\ell_P\ne 0$ and $L_T=T\fragmentco{\ell_T}{r_T}\in \LT$,
    \item $\fragmentoo{\ell_T-\ktotm-r_P}{r_T+\ktotm-\ell_P}$ of weight $\edl{L^T}{Q}$ for each locked fragment $L^P=P\fragmentco{\ell_P}{r_P}\in \LP$ with $\ell_P=0$ and $L_T=T\fragmentco{\ell_T}{r_T}\in \LT$,
    \item $\fragmentoo{\ell_T-\ktotm}{r_T+\ktotm}$ of weight $\tradeoff$ for each locked fragment $L_T=T\fragmentco{\ell_T}{r_T}\in \LT$,
    \item $\fragmentoo{\ell_T-\ktotm-m}{r_T+\ktotm-m}$ of weight $\tradeoff$ for each locked fragment $L_T=T\fragmentco{\ell_T}{r_T}\in \LT$,
    \item $\fragmentoo{\ell_T-\ktotm-|L_1^P|}{r_T+\ktotm-|L_1^P|}$ of weight $\tradeoff$ for each locked fragment $L_T=T\fragmentco{\ell_T}{r_T}\in \LT$.
\end{itemize}
The heavy positions are exactly those positions in $\fragment{0}{n-m+k}$ that are contained in intervals of total weight at least $\tradeoff$;
they can be computed using a sweep-line procedure with events corresponding to interval endpoints.
The number of events is $\Oh(d^2)$, so the output consists of $\Oh(d^2)$ disjoint ranges.
In terms of the running time, the bottleneck is sorting the events using the algorithm of \cite{DBLP:journals/jcss/AnderssonHNR98}.
\end{proof}

\begin{lemma}\label{lem:heavy-bound}
The total number of heavy positions does not exceed
\begin{align*}
|\Hv| & \le \frac{2d_T(d_P+2+3\tradeoff)\ktotm + 2d_T(6d_P+3d_T+\lpref+2+9\tradeoff)|Q| + 2d_P|Q| + 6|Q|{\tradeoff} + 2\ktotm \tradeoff}{\tradeoff}\\
	  &  =\cO((d^2/\tradeoff + d)(\ktotm+|Q|)).
	  \end{align*}
Moreover, for any integer $b\in \Zp$,
\begin{align*}
& \left|\bigcup_{v\in \Hv}\fragment{v-b}{v+b} \right| \\
& \quad \le \frac{2d_T(d_P+2+3\tradeoff)(\ktotm+b) + 2d_T(6d_P+3d_T+\lpref+2+9\tradeoff)|Q| + 2d_P|Q| + 6|Q|{\tradeoff} + 2(\ktotm+2b) \tradeoff}{\tradeoff}\\
& \quad = \cO((d^2/\tradeoff + d)(\ktotm+b+|Q|)).
\end{align*}
\end{lemma}
\begin{proof}
By \cref{lem:mark-pos}, the number of positions violating the first condition of \cref{def:light}
does not exceed $\big({2d_T(d_P+2)\ktotm + 2d_T(6d_P+3d_T+\lpref+2)|Q| + 2d_P|Q|}\big)/{\tradeoff}=\cO(d^2(\ktotm+|Q|)/\tradeoff)$.
As for the remaining conditions, each locked fragment $L^T\in \mathcal{L}^T$ may yield at most $3(|L^T|+2\ktotm)$
heavy positions, for a total of:
\[\sum_{L^T\in \mathcal{L}^T} 3(|L^T|+2\ktotm)
\le 3(5|Q|+1)d_T +3\cdot 2|Q|+6(d_T+2) \ktotm \le 18d_T|Q|+6|Q|+6d_T\ktotm+2\ktotm = \cO(d(\ktotm+|Q|)).\]

As for the second claim, observe that if $v$ is heavy, then all positions in $\fragment{v-b}{v+b} \cap \fragment{0}{n-m+k}$ would
be heavy if we increased the threshold $\ktotm$ to $\ktotm+b$.
This is because the left endpoint of all intervals considered in the proof of \cref{lem:heavy-alg} contains a $-\ktotm$ term whereas the right endpoint contains a $+\ktotm$ term (and there is no other dependency on $\ktotm$).
Further, $\left|\bigcup_{v\in \Hv}\fragment{v-b}{v+b} \setminus \fragment{0}{n-m+k} \right| \leq 2b$, since $\Hv \subseteq \fragment{0}{n-m+k}$.
\end{proof}

\section{A Faster Algorithm for {\SM}}\label{sec:faster}

In this section, we present the core of our improvements: a faster algorithm for \SM.

\newsm*

Let us fix an instance of the \SM problem, sets of locked fragments $\LP \coloneqq \locked(P,Q,d_P,2d)$ and $\LT \coloneqq \locked(T,Q,d_T,0)$,
that is, we have $\lpref \coloneqq 2d$,
an integer
$\ktotm \coloneqq 7d \geq \lpref+2k+d_P+d_T > \ktot$
and an integer threshold
$\tradeoff \coloneqq \max\{1 , \lfloor{\sqrt{d}/\sqrt{\log(n+1) \log(d+1)}}\rfloor\} \leq d$.
Further, consider a partition of the positions of $T$ in $\fragment{0}{n-m+k}$ into heavy and light using \cref{lem:heavy-alg}.

First, we show how to compute $k$-error occurrences that start at heavy positions; this is a straightforward application of \cref{fct:rj,fct:verifyRj}.
Then, we reduce the problem of computing $k$-error occurrences that start at light positions
to an instance $\swaps(T, P, k, Q, \A_T, \A_P, \rred(P), \rred(T), \alpha))$, where  $\rred(P)$ and $\rred(T)$
are both of size $\cO(d)$ and contain all the internal pieces that (almost) overlap locked fragments,
and $\alpha=o(\sqrt{d})$.


\subsection{Computing Occurrences Starting at Heavy Positions}
\label{sus:heavy}

\begin{lemma}[\texttt{HeavyMatches($P$, $T$, $k$, $d$, $Q$, $\A_P$, $\A_T$, $\Hv$)}]\label{lem:heavy-total}
    Given an instance of the \SM problem
    and the set $\Hv$ of heavy positions constructed using \cref{lem:heavy-alg},
    $\OccE_k(P,T)\cap \Hv$ can be computed in $\Oh(d^3+d^4/\tradeoff)$ time in the \modelname model.
\end{lemma}
\begin{proof}
    Recall that the set $\Hv$ is represented as the union of $\Oh(d^2)$ disjoint heavy ranges
    (listed in the left-to-right order).

    Our algorithm starts with an application \cref{fct:rj} to construct the sequence $(r_j)_{j\in J}$, represented as a concatenation of $\Oh(d)$ arithmetic progressions.
    Our first goal is to enumerate elements of the set $J' \coloneqq \{j\in J : \fragmentco{r_j}{r_{j+1}}\cap \Hv \ne \emptyset\}$.
    For this, we simultaneously traverse the sequence $(r_j)_{j\in J}$ along with the heavy ranges constituting $\Hv$.
    For each heavy range $H \sub \Hv$, we list all $j\in J$ such that $\fragmentco{r_j}{r_{j+1}}\cap H \ne \emptyset$
    (only the smallest such $j$ might have already been listed for an earlier heavy range).
    In the second phase, we compute $O\coloneqq\bigcup_{j\in J'} (r_j+\OccE_k(P,R_j))\cap \fragmentco{r_j}{r_{j+1}}$
    using \cref{fct:verifyRj}, making sure that the positions are listed in the left-to-right order.
    Finally, we simultaneously scan $O$ and the heavy ranges constituting $\Hv$, reporting all positions $v\in O\cap \Hv$.

    As for correctness, observe that $R_j=T\fragmentco{r_j}{r'_j}$,
    so $O \sub \OccE_k(P,T)$ and, due to the final filtering step, we output a subset of $\OccE_k(P,T)\cap \Hv$.
    To prove the converse inclusion, consider a position $v\in \OccE_k(P,T)\cap \Hv$.
    By \cref{lem:aligned}, there exists $j\in J$ such that $v\in (r_j+\OccE_k(P,R_j))\cap \fragmentco{r_j}{r_{j+1}}$
    and, by the definition of $J'$, we also have $j\in J'$. Consequently, $v$ is indeed reported.

    As for the complexity analysis, let us first compute the running time in terms of $|J'|$.
    Constructing the sequence $(r_j)_{j\in J}$ costs $\Oh(d)$ time.
    The set $J'$ can be computed in $\Oh(|J'|+d+d^2)=\Oh(|J'|+d^2)$ time.
    The applications of \cref{fct:verifyRj} cost $\Oh(d^2)$ time each (in the \modelname model),
    for a total of $\Oh(|J'|\cdot d^2)$ time in the \modelname model.
    This is also a (crude) upper bound on the output size, so the final filtering step works in $\Oh((1+|J'|)d^2)$ time.
    Overall, the running time in the \modelname model is $\Oh((1+|J'|)d^2)$.

    It remains to bound $|J'|$. Observe that each interval $\fragmentco{r_j}{r_{j+1}}$ (possibly except for the last one with $j=\max J$) has length at most $\tau+d_T$.
    Moreover, the total length of any $s$ intervals $\fragmentco{r_j}{r_{j+1}}$ with $0<r_j < r_{j+1} < n$
    is at least $s\tau-d_T$. On top of that, there might be one non-empty interval with $r_j=0$ and one non-empty interval with $r_{j+1}=n$.
    We conclude that
    \[ |J'| \le 2 + \frac{\sum_{j\in J'\sm \{\max J\}} |\fragmentco{r_j}{r_{j+1}}|+d_T}{\tau}
    \le 2+\frac{\left|\bigcup_{v\in \Hv} \fragment{v-\tau-d_T}{v+\tau+d_T}\right|+d_T}{\tau}.\]
    Since $\tau = \Theta(\max(|Q|,d))$ and $d_T,d_P,k,\ktotm=\Oh(d)$,
    the bound of \cref{lem:heavy-bound} yields that $|J'|=\Oh(d+{d^2}/{\tradeoff})$.
    Hence, the total running time is $\Oh(d^3+{d^4}/{\tradeoff})$ just as claimed.
\end{proof}

\subsection{Computing Occurrences Starting at Light Positions}
\label{sus:light}
\subsubsection{Combinatorial Insights}

The following fact follows directly from~\cref{def:light}.
\begin{fact}\label{fact:restrlckd}
    For any light position $v$ in $T$ and for any $w\in \fragmentoo{v+m-\ktotm}{v+m+\ktotm}$, we have $\res{\mathcal{L}^T}{\fragmentco{v}{w}}=\{T\fragmentco{\ell}{r} \in \mathcal{L}^T : \fragmentco{\ell}{r} \subseteq \fragmentco{v}{v+m}\}$.
    \lipicsEnd
\end{fact}

For each position $v$ of $T$, set $\rho(v)\in \fragment{x_T}{y_T}$ such that
$\A_T(T\fragmentco{v}{n})=\Q\fragmentco{\rho(v)}{y_T}$.
Observe that
\begin{align*}\edl{T}{Q}&=\ed(T\fragmentco{0}{v},\Q\fragmentco{x_T}{\rho(v)})+\ed(T\fragmentco{v}{n},\Q\fragmentco{\rho(v)}{y_T})\\
    &=\eds{T\fragmentco{0}{v}}{\rot^{-\rho(v)}(Q)}+\ed(T\fragmentco{v}{n},\rot^{-\rho(v)}(Q)^*).\end{align*}

\newcommand{\Lck}{\mathcal{L}}
\begin{lemma}\label{fact:rot}
    For any two positions $v<w$ of $T$ such that $\res{\mathcal{L}^T}{\position{v}}=\res{\mathcal{L}^T}{\position{w}}=\emptyset$,
    we have
    \[\edl{T\fragmentco{v}{w}}{Q}=\ed(T\fragmentco{v}{w},\Q\fragmentco{\rho(v)}{\rho(w)}) = \sum_{L \in \Lck^T_{\fragmentco{v}{w}}} \edl{L}{Q}.\] 
\end{lemma}
\begin{proof}
    Set $\res{\mathcal{L}^T}{\fragmentco{v}{w}} = \{L^T_j \in \mathcal{L}^T: j \in
    \fragment{j_1}{j_2}\}$, and observe
    that all elements of this set are fragments of $T\fragmentco{v}{w}$.
    Now, we have
    \begin{align*}
& \ed(T\fragmentco{v}{w},\Q\fragmentco{\rho(v)}{\rho(w)})\\
& \quad = \ed(T,\Q\fragmentco{x_T}{y_T}) - \ed(T\fragmentco{0}{v},\Q\fragmentco{x_T}{\rho(v)}) - \ed(T\fragmentco{w}{n},\Q\fragmentco{\rho(w)}{y_T}))\\
& \quad \leq d_T - \sum_{j=1}^{j_1-1} \edl{L^T_j}{Q} - \sum_{j=j_2+1}^{\ell^T} \edl{L^T_j}{Q}\\
& \quad = \sum_{j=j_1}^{j_2} \edl{L^T_j}{Q}\\
& \quad \leq \edl{T\fragmentco{v}{w}}{Q}\\
& \quad \leq \ed(T\fragmentco{v}{w},\Q\fragmentco{\rho(v)}{\rho(w)}).\qedhere
    \end{align*}
\end{proof}

\begin{lemma}\label{lem:ub}
    Consider a light position $v$ of $T$.
    If $\ed(L_1^P,\rot^{-\rho(v)}(Q)^*) < \lpref$, then there is a $w\in \fragment{v}{n}$
    such that
    \[\ed(P,T\fragmentco{v}{w})\leq \ed(L_1^P,\rot^{-\rho(v)}(Q)^*)+\sum_{i=2}^{\ell^P}\edl{L^P_i}{Q}+\sum_{L\in\res{\mathcal{L}^T}{\fragmentco{v}{v+m}} }\edl{L}{Q}.\]
\end{lemma}
\begin{proof}
    We first prove two auxiliary claims.

    \begin{claim}\label{claim:rotate}
        We have
        $\ed(P,\rot^{-\rho(v)}(Q)^*) = \ed(L_1^P,\rot^{-\rho(v)}(Q)^*)+\sum_{i=2}^{\ell^P}\edl{L^P_i}{Q} < \lpref+d_P$.
    \end{claim}
    \begin{claimproof}
        The inequality $\ed(P,\rot^{-\rho(v)}(Q)^*) \geq \ed(L_1^P,\rot^{-\rho(v)}(Q)^*)+\sum_{i=2}^{\ell^P}\edl{L^P_i}{Q}$ holds trivially.

        By~\cref{lem:klocked}, we have $P=L_1^P Q^{\alpha_1} L_2^P Q^{\alpha_2} \cdots Q^{\alpha_{\ell^P-1}} L_{\ell^P}^P$ for some non-negative integers $\alpha_i$.
        Since $L_1^P$ is $\lpref$-locked, we have
        $\ed(L_1^P,\rot^{-\rho(v)}(Q)^*)=\ed(L_1^P,Q^\infty\fragmentco{\rho(v)}{jq})$ for some non-negative integer $j$.
        Further, for each $i \in \fragmentoo{1}{\ell^P}$,
        we have $\edl{L^P_i}{Q}=\ed(L^P_i,Q^{\beta_i})$ for some non-negative integer $\beta_i$.
        Finally, we have $\edl{L^P_{\ell^P}}{Q}=\ed(L^P_{\ell^P},Q^\infty\fragmentco{0}{x})$ for some non-negative integer $x$.

        Set $\gamma = \alpha_1+\sum_{i=2}^{\ell^P-1}(\alpha_i+\beta_i)$.
        The above discussion implies that there is an alignment of $P$ with the prefix $Q\fragmentco{\rho(v)}{jq}Q^\gamma Q^\infty\fragmentco{0}{x}$
        of $Q\fragmentco{\rho(v)}{jq}Q^\infty$ that costs $\ed(L_1^P,\rot^{-\rho(v)}(Q)^*)+\sum_{i=2}^{\ell^P}\edl{L^P_i}{Q}$, thus proving
        that $\ed(P,\rot^{-\rho(v)}(Q)^*) \leq \ed(L_1^P,\rot^{-\rho(v)}(Q)^*)+\sum_{i=2}^{\ell^P}\edl{L^P_i}{Q}$.
        This concludes the proof of the claimed equality.

        The claimed inequality holds since $\ed(L_1^P,\rot^{-\rho(a)}(Q)^*)< \lpref$ and $\sum_{i=2}^{\ell^P}\edl{L^P_i}{Q}\leq d_P$.
    \end{claimproof}

    \begin{claim}\label{claim:rhow}
        There is a $w\in \fragmentoo{v+m-\ktotm}{v+m+\ktotm}\cap \fragmentco{0}{n}$ such that $\ed(P,\rot^{-\rho(v)}(Q)^*)=\ed(P,Q^\infty \fragmentco{\rho(v)}{\rho(w)})$.
    \end{claim}
    \begin{claimproof}
        Recall that $\Lck^T_{\fragmentco{v+m-\ktotm}{v+m+\ktotm}} = \emptyset$
        holds because $v$ is a light position of $T$.
        Since $\Lck^T$ contains a suffix of $T$, we conclude that either $n \le v+m-\ktotm$ or $v+m+\ktot \le n$.
        The former case contradicts $v \le n-m+k < n-m+\ktotm$, so $T\fragmentco{v+m-\ktotm}{v+m+\ktotm}$ must be a fragment of $T$ disjoint with all locked fragments in $\Lck^T$.
        This means that $\A_T$ matches $T\fragmentco{v+m-\ktotm}{v+m+\ktotm}$ without any edits to $\A_T(T\fragmentco{v+m-\ktotm}{v+m+\ktotm})$.
        Since $(v,\rho(v))\in \A_T$ and the cost of $\A_T$ is $d_T$, the fragment $\A_T(T\fragmentco{v+m-\ktotm}{v+m+\ktotm})$ must contain $\Q\fragmentco{\rho(v)+m-\ktotm+d_T}{\rho(v)+m+\ktotm-d_T}$.
        Consequently, for every $y\in \fragmentoo{\rho(v)+m-\ktotm+d_T}{\rho(v)+m+\ktotm-d_T}$, there exists $w\in \fragmentoo{v+m-\ktotm}{v+m+\ktotm}\cap \fragmentco{0}{n}$ such that $y=\rho(w)$.
        Since $\ed(P,\rot^{-\rho(v)}(Q)^*)\le \lpref + d_P < \ktotm - d_T$ holds by \cref{claim:rotate},
        such $w$ exist in particular for $y$ chosen so that $\ed(P,\rot^{-\rho(v)}(Q)^*)=\ed(P,Q^\infty \fragmentco{\rho(v)}{y})$.
    \end{claimproof}

    We get the desired bound by combining \cref{claim:rhow,claim:rotate,fact:restrlckd,fact:rot} via the triangle inequality;
    observe that \cref{fact:rot} is applicable because $\Lck^T_{\position{w}}\sub \Lck^T_{\fragmentco{v+m-\ktotm}{v+m+\ktotm}}=\emptyset$.
    \begin{align*}
    &\ed(P,T\fragmentco{v}{w})\\
    &\quad \leq \ed(P,Q^\infty \fragmentco{\rho(v)}{\rho(w)}) + \ed(T\fragmentco{v}{w},Q^\infty \fragmentco{\rho(v)}{\rho(w)})   & (\text{triangle inequality})\\
    &\quad = \ed(P,\rot^{-\rho(v)}(Q)^*) + \ed(T\fragmentco{v}{w},Q^\infty \fragmentco{\rho(v)}{\rho(w)}) 					& (\text{\cref{claim:rhow}})\\
    &\quad = \ed(L_1^P,\rot^{-\rho(v)}(Q)^*)+\sum_{i=2}^{\ell^P}\edl{L^P_i}{Q} + \sum_{L\in \res{\mathcal{L}^T}{\fragmentco{v}{w}}}\edl{L}{Q}			& (\text{\cref{claim:rotate,fact:rot}})\\
    &\quad = \ed(L_1^P,\rot^{-\rho(v)}(Q)^*)+\sum_{i=2}^{\ell^P}\edl{L^P_i}{Q} + \sum_{L\in \res{\mathcal{L}^T}{\fragmentco{v}{v+m}}}\edl{L}{Q}.			& (\text{\cref{fact:restrlckd}})
    \end{align*}
    This completes the proof of the lemma.
\end{proof}

For a position $v$ and a locked fragment $L \in \mathcal{L}^P \cup \mathcal{L}^T$,
we write $\markf(v,L)$ for the number of marks placed in
$v$ due to pairs of locked fragments that contain $L$; formally:
\begin{align*}
\markf(v,L)= \begin{cases}
    \sum_{L^T \in \Lck^T} \markf(v,\ktotm ,L,L^T) & \text{if }L\in \Lck^P;\\
    \sum_{L^P \in \Lck^P} \markf(v,\ktotm ,L^P,L) & \text{if }L\in \Lck^T.\\
\end{cases}
\end{align*}

\begin{definition}
    For a light position $v$ of $T$, set 
    \[\mathcal{D}(v) \coloneqq \{L^P_1\} \cup \{L \in \mathcal{L}^P\cup \mathcal{L}^T_{\fragmentco{v}{v+m}} : \markf(v,L) < \edl{L}{Q}\}.
    \tag*{\lipicsEnd}\]
\end{definition}

Let us now provide some intuition on what follows.
Consider an alignment $\mathcal{A}: P \onto T\fragmentco{v}{w}$, where $v$ is a light position of $T$
and the cost of $\mathcal{A}$ is not larger than $k$.
In the next lemma, we essentially lower bound the cost of the restriction of such an alignment $\A$
to each locked fragment. For instance, we lower bound $\ed(L,\A(L))$ for each $L\in \mathcal{L}^P$.
Our lower bound is positive only for elements of $\mathcal{D}(v)$.
Consider some locked fragment $L \in \mathcal{D}(v) \cap \mathcal{L}^P$ other than $L^P_1$.
Roughly speaking, at most $\markf(v,L)$ errors of $L$ with a fragment $Q'$ of $Q^\infty$
cancel out with errors between $\mathcal{A}(L)$ and $Q'$, yielding a lower bound $\ed(L,\A(L)) \ge \edl{L}{Q} - \markf(v,L)$.
Then, the definition of $\mathcal{D}(v)$ guarantees that,
for any $L \in \mathcal{D}(v) \cap \mathcal{L}^P$, $\A(L)$ is disjoint from all $L' \in \mathcal{D}(v) \cap \mathcal{L}^T$.
One can exploit this property to obtain a lower bound for the cost of~$\A$ by showing that we can
sum over the lower bounds for individual locked fragments in $\mathcal{D}(v)$.
In fact, we use this reasoning to lower bound the cost of an alignment between two other strings,
obtained from $P$ and $T$, respectively, via the deletion of some fragments;
this happens in the proof of \cref{lem:nofakelight} in \cref{sec:redundantsqrtk}.

\begin{lemma}\label{lem:lb_local}
    Consider a light position $v$ of $T$
    and a locked fragment $L \in \mathcal{L}^P \cup \mathcal{L}^T_{\fragmentco{v}{v+m}}$.
    \begin{enumerate}
        \item If $L = T\fragmentco{v+\ell}{v+r} \in \mathcal{L}^T$, then every $U\in \{P\fragmentco{i}{j} : i \geq \ell-k \text{ and } j \leq r+k \}$ satisfies
            \[ \ed(L,U)\geq \edl{L}{Q} - \markf(v,L).\]
        \item If $L = P\fragmentco{\ell}{r} \in \mathcal{L}^P\setminus \{L^P_1\}$, then every $U\in \{T\fragmentco{i}{j} : i \geq v+\ell-k \text{ and } j \leq v+r+k \}$ satisfies
            \[\ed(L,U) \geq \edl{L}{Q} - \markf(v,L).\]
        \item If $L$ is $L^P_1$, then for every $U\in  \{ T\fragmentco{v}{v+j} : j \in \fragment{|L|-k}{|L|+k}\}$ satisfies
            \[\ed(L,U)\geq \ed(L,\rot^{-\rho(v)}(Q)^*) - \markf(v,L).\]
    \end{enumerate}
\end{lemma}
\begin{proof}
    Consider some $L \in \mathcal{L}^P\cup \mathcal{L}^T_{\fragmentco{v}{v+m}}$.
    If $L\in \mathcal{L}^T$, let $Y=P$; otherwise, let $Y=T$.
    Further, let $U \coloneqq Y\fragmentco{u}{w}$ denote
    any of the fragments specified in the statement of the lemma, and let
    \begin{align*}
        \zeta \coloneqq
        \begin{cases}
            \ed(L_1^P,\rot^{-\rho(v)}(Q)^*) & \text{if } L = L^P_1,\\
            \edl{L}{Q} & \text{otherwise.}
        \end{cases}
    \end{align*}
    Consult~\cref{fig:lbed} for an illustration of the setting.

    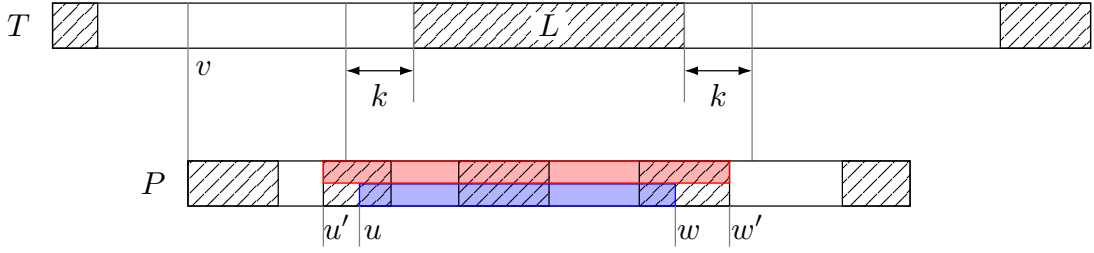
\begin{figure}[t!]
        \centering
        \scalebox{1.2}{
        \begin{tikzpicture}

\draw [pattern={Lines[angle=45,distance=4pt]},pattern color=black]  (5,0) rectangle (8,0.5);
\draw [pattern={Lines[angle=45,distance=4pt]},pattern color=black]  (1,0) rectangle (1.5,0.5);
\draw [pattern={Lines[angle=45,distance=4pt]},pattern color=black]  (11.5,0) rectangle (12.5,0.5);
\fill[white] (6.3,0.1) rectangle (6.7,0.4);
\draw (1,0) rectangle (12.5,0.5);

\node[label = {left: $T$}]  at (1,0.25) {};
\node[label = {right: \hphantom{$T$}}]  at (12.5,0.25) {};
\node[label = {below: $v$}]  at (2.675,0.1) {};
\foreach \x/\y in {6.5/$\strut L$}{
	\node[anchor=base, label = {\y}]  at (\x,-.25) {};
}

\draw[thin, gray] (2.5,0.5) --  (2.5,-1.75);

\draw[thin, gray] (4.25,0.5) --  (4.25,-1.25);
\draw[thin, gray] (8.75,0.5) --  (8.75,-1.25);

\draw[thin, gray] (5,0.5) --  (5,-.6);
\draw[thin, gray] (8,0.5) --  (8,-.6);

\begin{scope}[yshift=-1.75cm]
	\fill[blue!30!white]
    (4.4,0) rectangle (7.9,0.25);
	\fill[red!30!white]
    (4,0.25) rectangle (8.5,0.5);

	\draw [pattern={Lines[angle=45,distance=4pt]},pattern color=black]  (2.5,0) rectangle (3.5,0.5);
	\draw [pattern={Lines[angle=45,distance=4pt]},pattern color=black]  (4,0) rectangle (4.75,0.5);
	\draw [pattern={Lines[angle=45,distance=4pt]},pattern color=black]  (5.5,0) rectangle (6.5,0.5);
	\draw [pattern={Lines[angle=45,distance=4pt]},pattern color=black]  (7.5,0) rectangle (8.5,0.5);
	\draw [pattern={Lines[angle=45,distance=4pt]},pattern color=black]  (9.75,0) rectangle (10.5,0.5);
	\draw (2.5,0) rectangle (10.5,0.5);
    \path[draw=blue]
    (4.4,0) rectangle (7.9,0.245);
    \path[draw=red]
    (4,0.255) rectangle (8.5,0.5);

	\node[label = {left: $P$}]  at (2.5,0.25) {};
	\foreach \x/\y in {4.15/$\strut u'$, 4.55/$\strut u$, 8.05/$\strut w$, 8.65/$\;\strut w'$}{
		\node[anchor=base, label = {\y}]  at (\x,-.75) {};
        \pgfmathsetmacro{\nx}{\x - .15};
        \draw[thin, gray] (\nx,0) -- (\nx,-.45);
	}

	\draw[{Latex[length=1.5mm, width=1mm]}-{Latex[length=1.5mm, width=1mm]}] (4.25,1.5) -- (5,1.5);
	\draw[{Latex[length=1.5mm, width=1mm]}-{Latex[length=1.5mm, width=1mm]}] (8,1.5) -- (8.75,1.5);
	\node[label = {below: $k$}]  at (4.625,1.6) {};
	\node[label = {below: $k$}]  at (8.375,1.6) {};


\end{scope}

\end{tikzpicture}}
        \caption{An illustration of the setting in the proof of~\cref{lem:lb_local}.
            Locked fragments are distinguished by diagonal stripes.
            In this example, $L$ is the sole element of $\mathcal{L}^T_{\fragmentco{v}{v+m}}$ and $Y=P$.
            The fragment $U=P\fragmentco{u}{w}$ is shaded in blue,
        the fragment $P\fragmentco{u'}{w'}$ is shaded in red.}
        \label{fig:lbed}
    \end{figure}


    Let $x$ and $y$ denote integers that satisfy
    $\ed(U, Q^\infty\fragmentco{x}{y})=\edl{U}{Q}$.
    In the case where $L=L^P_1$,
    we choose $x=\rho(v)$ and $y=\rho(w)$ so that $Q^\infty\fragmentco{x}{y}$ is a prefix of $\rot^{-\rho(v)}(Q)^\infty$.
    This is allowed by~\cref{fact:rot} because $\res{\mathcal{L}^T}{\fragmentco{v-\ktotm}{v+\ktotm}}=\res{\mathcal{L}^T}{\fragmentco{v+|L_1^P|-\ktotm}{v+|L_1^P|+\ktotm}}=\emptyset$, and it ensures that the following inequality holds in all three cases (the inequality holds trivially if $L\ne L^P_1$):
    \begin{equation}\label{eq:rot}
        \ed(L, Q^\infty\fragmentco{x}{y}) \geq \zeta.
    \end{equation}
    Further, our marking scheme
    implies
    \begin{equation}\label{eq:dom}
        \markf(v,L) \geq \sum_{K \in \mathcal{L}^Y_{\fragmentco{u}{w}}} \edl{K}{Q}.
    \end{equation}
    Let $Y\fragmentco{u'}{w'}$ denote the fragment of $Y$ that is covered by
    $Y\fragmentco{u}{w}$ and the elements of $\mathcal{L}^Y_{\fragmentco{u}{w}}$.
    We have
    \begin{equation}\label{eq:extlock}
        \ed(U,Q^\infty\fragmentco{x}{y}) = \edl{U}{Q} \leq  \edl{Y\fragmentco{u'}{w'}}{Q} =\sum_{K \in \mathcal{L}^Y_{\fragmentco{u}{w}}} \edl{K}{Q},
    \end{equation}
    where the last equality follows from the properties of locked fragments as computed by \cref{lem:klocked}.

    We are now ready to prove the claimed inequality.
    \begin{align*}
        \ed(L,U) & \geq |\ed(L, Q^\infty\fragmentco{x}{y})-\ed(U,Q^\infty\fragmentco{x}{y})| & \text{\hspace{2cm}(triangle inequality)}\\
                 &  \geq    \ed(L, Q^\infty\fragmentco{x}{y})-\ed(U,Q^\infty\fragmentco{x}{y})  & \\
                 &  \geq  \zeta -\ed(U,Q^\infty\fragmentco{x}{y})& \text{(due to~\eqref{eq:rot})}\\
                 &  \geq  \zeta - \sum_{K \in \mathcal{L}^Y_{\fragmentco{i}{j}}} \edl{K}{Q}
                 & \text{(due to~\eqref{eq:extlock})}\\
                 &  \geq  \zeta - \markf(v,L).                   & \text{(due to \eqref{eq:dom})}
    \end{align*}
    This completes the proof of the lemma.
\end{proof}

\subsubsection{Shrinking Runs of Plain Pairs to Length $o(\sqrt{d})$}\label{sec:redundantsqrtk}

\newcommand{\tbd}{13}
\newcommand{\slack}{30}

We set $\textsf{Red}(P)$ to be equal to
\[\textsf{Special}(P)\cup \{P_i : i\in\fragmentoo{1}{z} \text{ and } \exists_{P\fragmentco{\ell}{r} \in \mathcal{L}^P} \fragmentco{\ell-\tbd\ktotm}{r+\tbd\ktotm} \cap \fragmentco{p_i}{p_{i+1}+\Delta} \neq \emptyset\}.\]
Similarly, we set $\textsf{Red}(T)$ to be equal to
\[\textsf{Special}(T)\cup \{T_i : i\in\fragmentoo{\min J + 1}{\max J + z} \text{ and } \exists_{T\fragmentco{\ell}{r} \in \mathcal{L}^T} \fragmentco{\ell-\tbd\ktotm}{r+\tbd\ktotm} \cap \fragmentco{t_i}{t_{i+1}+\Delta} \neq \emptyset\}.\]
Next, we upper-bound the sizes of these two sets and show how to construct them efficiently.

\begin{lemma}\label{lem:fewreds}
    Given $P$, $T$, $\A_P$, $\A_T$, $\mathcal{L}^P$, and $\mathcal{L}^T$,
    the sets $\textsf{Red}(P)$ and $\textsf{Red}(T)$ are of size $\cO(d)$ and can be
    constructed in time $\cO(d \log\log d)$ in the \modelname model.
\end{lemma}
\begin{proof}
    We start with upper-bounding the size of each of the sets $\textsf{Red}(P)$ and $\textsf{Red}(T)$.
    $\textsf{Special}(P)$ and $\textsf{Special}(T)$ are of size $\cO(d)$ due to \cref{lem:specialbound-simpl}.
    The task is therefore to bound, for each of $P$ and $T$, the number of internal pieces that are within $\tbd \ktotm$ positions of a locked fragment.

    Let $(S,\beta_S)$ denote either of $(P,\beta_P)$ or $(T,\beta_T)$ and $\mathcal{L}^S=\{L\fragmentco{\ell_j}{h_j} : j \in \fragment{1}{\ell^S}\}$.
    It suffices to upper bound the number of tiles $S\fragmentco{s_{i-1}}{s_i}$ with $i \in \fragmentoo{1}{\beta_S}$
    in the $\tau$-tile partition of~$S$ (with respect to $\A_S$) such that $\fragmentco{s_{i-1}}{s_i}$ overlaps $\fragmentco{\ell_j-\Delta - \tbd \ktotm}{h_j+ \tbd \ktotm}$ for some $j\in \fragment{1}{\ell^S}$;
    that is,
    \[|\{i \in \fragmentoo{1}{\beta_S} : \exists_{j \in \fragment{1}{\ell^S}} \fragmentco{\ell_j-\Delta -\tbd \ktotm}{h_j + \tbd \ktotm} \cap \fragmentco{s_{i-1}}{s_{i}} \neq \emptyset\}|.\]

\def\tbc{32}
    Intuitively, we extend each locked fragment in $\mathcal{L}^S$ by $\Delta + \tbd \ktotm$ characters to the left and by $\tbd \ktotm$ characters to the right,
    thus covering at most $\|\mathcal{L}^S\|+ \ell^S \cdot (2\cdot\tbd+6)\ktotm=\|\mathcal{L}^S\|+ \tbc\ell^S \ktotm$ positions of $S$, since $\Delta=6 \ktot \leq 6 \ktotm$.
    Let us first upper bound the size $\mu$ of the set $M$ of tiles (other than the first and last ones) that are fully covered by these ``extended'' locked fragments,
    that is, $M = \{i \in \fragmentoo{0}{\beta_S} : \fragmentco{s_{i-1}}{s_{i}} \subseteq \bigcup_{j  \in \fragment{1}{\ell^S}} \fragmentco{\ell_j-\Delta-\tbd \ktotm}{h_j+\tbd \ktotm}\}$.
    We have
    \[\sum_{i \in M} (\tau - |S\fragmentco{s_{i-1}}{s_i}|) \leq \sum_{i \in M}\ed(S\fragmentco{s_{i-1}}{s_i},Q\fragmentco{(i-1)\tau}{i\tau}) \leq \edl{S}{Q}\]
    and hence
    \[\mu \cdot \tau - \edl{S}{Q} \leq \sum_{i \in M} |S\fragmentco{s_{i-1}}{s_i}| \leq \|\mathcal{L}^S\|+\tbc\ell^S \ktotm,\]
    which is equivalent to
    \[\mu \leq \frac{\|\mathcal{L}^S\|+\tbc\ell^S\ktotm +\edl{S}{Q}}{\tau}.\]
    Finally, we have to account for the at most $2\ell^S$ tiles that overlap ``extended'' locked fragments but are not fully contained in them; we have at most two such tiles for each $L \in \mathcal{L}^S$.

    Since $\ell^S = \cO(\edl{S}{Q})$ and $\|\mathcal{L}^S\|=\cO(\edl{S}{Q} \cdot q)$ by~\cref{lem:klocked},
    $\edl{S}{Q} = \cO(\ktot)$, and $\tau = \Theta(\max\{\ktot,q\})$,
    we have
    \[2\ell^S + \frac{\|\mathcal{L}^S\|+\tbc \ell^S\ktotm + \edl{S}{Q}}{\tau}  = \cO\left(\ktot + \frac{\ktot q + \ktot \ktotm + \ktot}{\tau}\right)
    																		= \cO(\ktotm)=\cO(d).\]

    Let us now show how to efficiently construct the sets in scope.
    First, recall that $\textsf{Special}(P)$ and $\textsf{Special}(T)$ can be constructed in $\cO(d)$ time in the \modelname model due to \cref{lem:specialbound-simpl}.
    The remaining elements of $\rred(P)$ (resp.~$\rred(T)$) can be computed in a simultaneous left-to-right scan of:
    \begin{itemize}
        \item the representation of starting positions of internal pieces $p_i$ (resp.~$t_i$) as $\cO(d)$ arithmetic progressions, which can be computed in $\cO(d)$ time due 				to \cref{fct:tile};
        \item the $\cO(d)$ locked fragments of $P$ (resp.~$T$) sorted with respect to their starting positions. 
    \end{itemize}
    The $\cO(d \log\log d)$ time required for sorting the starting positions of the locked fragments, using the algorithm of \cite{DBLP:journals/jcss/AnderssonHNR98}, is the bottleneck of the algorithm in the \modelname model.
\end{proof}

\begin{definition}\label{def:fjgj}
    For $j \in J$, set $\I'_j
    \coloneqq (F_{j,1},G_{j,1})\cdots(F_{j,z_j}, G_{j,z_j})
    \coloneqq \Comp(\I_j, \rred(P), \rred(T), 2\tradeoff + \slack)$,
    and \(F_j \coloneqq \val_{\Delta}( F_{j,1},\ldots, F_{j,z_j} )\) and
    \(G_j \coloneqq \val_{\Delta}( G_{j,1},\ldots, G_{j,z_j} )\).

    Further, let $F_{j,i}$ and $G_{j,i}$ correspond to the fragments
    $F_j\fragmentco{f_{j,i}}{f_{j,i+1}+\Delta}$ and $G_j\fragmentco{g_{j,i}}{g_{j,i+1}+\Delta}$,
    respectively, that is, we set $f_{j,1}=g_{j,1}=0$ and, for $i\in \fragment{2}{z_j+1}$, $f_{j,i}=\left(\sum_{x<i} |F_{j,x}| \right)-\Delta$ and $g_{j,i}=\left(\sum_{x<i} |G_{j,x}| \right)-\Delta$.
    \lipicsEnd
\end{definition}

The focus of the remainder of~\cref{sec:redundantsqrtk} is to prove the following lemma.

\begin{restatable}{lemma}{correctlight}\label{lem:correctlight}
    $\OccE_k(P,T) \cap \light
    = \big(\bigcup_{j \in J} (r_j +  \OccE_k(\I'_j))\big) \cap \light
    =
    \big(\bigcup_{j \in J} (r_j +  \OccE_k(F_j,G_j))\big) \cap \light$.
    \ifx\correctlightlend\undefined\lipicsEnd\fi
\end{restatable}
\def\correctlightlend{1}

The combination of \cref{lem:aligned} and \cref{lem:redundantk} directly yields the following.

\begin{corollary}\label{cor:supset}
$\big(\bigcup_{j \in J} (r_j +  \OccE_k(F_j,G_j))\big) \supseteq \big(\bigcup_{j \in J} (r_j +  \OccE_k(P,R_j))\big)=\OccE_k(P,T)$.
\lipicsEnd
\end{corollary}

The following---skippable---example illustrates that,
for some $j\in J$ and $p\in \OccE_k(F_j,G_j)$,
we might have $r_j + p \not\in \OccE_k(P,T)$ if $r_j+p\in \Hv$.

\begin{example}
    Let $k>6$ denote an integer and set $\Sigma = \{\texttt{a},\texttt{b}\}$.
    Further, set $U=\texttt{a}^{i} \texttt{b}$ and $V=\texttt{a}^{i+1} \texttt{b}$
    for $i = \lfloor 2k /3 \rfloor$.
    Set $Q=UVU$ and
    $P=Q^{y} (UVV)^{k+1} Q^{y}$
    and $T=Q^{w} (VVU)^{k+1} Q^{w}$
    for some integers $y$ and $w$ that satisfy $w>y>k$, ensuring that $d\coloneqq 2k \leq |P|/8|Q|$.
    We have a valid instance of the \SM problem with $d_P=d_T=k+1$ and $\tau = \Theta(q)$.

    Consider a set of locked fragments (with respect to $Q$) for each of $P$ and $T$ such that
    the only locked fragment of $P$ (resp.~$T$) that is not its prefix or suffix is $L^P_2 \coloneqq P\fragmentco{y|Q|}{y|Q|+(k+1)(3i+5)}=(UVV)^{k+1}$
    (or~$L^T_2 \coloneqq \fragmentco{w|Q|}{w|Q|+(k+1)(3i+5)}=(VVU)^{k+1}$).\footnote{These are not precisely the locked fragments that would be computed by the algorithm underlying \cref{lem:klocked}, but they are consistent with the properties that need to be satisfied and are easier to work with for the sake of this example.}
    Let $\tradeoff = \cO(\sqrt{k})$ denote the threshold used in the marking and assume that $k$ is large enough so that $\tradeoff <k$.
    Consider some $j \in J$ and a position $p$ of $R_j$ such that $Q^{y} (VVU)^{k+1} Q^{y}$
    is a prefix of $R_{j}\fragmentco{p}{r'_j}$,
    noting that $r_j+p$ has at least $\min\{\edl{L_2^P}{Q}, \edl{L_2^T}{Q}\} = k+1 > \tradeoff$ marks, and is thus heavy.

    We next argue that $\min_t(\ed(P, R_{j}\fragmentco{p}{t})) >k$.
    Toward a contradiction, suppose that there exists an integer $t'$ and an alignment $\A: P \onto R_{j}\fragmentco{p}{t'}$
    of cost at most $k$.
    Then, at least one of the first/last $y$ copies of $Q$ in $P$ is matched exactly by $\A$ since $y>k$.
    In addition, at least one of the copies of $UVV$ is matched exactly by $\A$.
    Thus, $\A$ makes at least $|U|$ edits in order to ``synchronize'' a copy of $VVU$ in $P$ to a copy of $VVU$ in $T$, and then at least
    $|U|$ more edits to ``synchronize'' copies of $Q$.
    Hence, we have
    $\min_t(\ed(P, R_{j}\fragmentco{p}{t}))\geq 2|U| = 2  (\lfloor 2k /3 \rfloor +1) \geq 4k/3 > k$.

    On the other hand, for some integers $y_1\leq y$ and $y_2 \leq c \tradeoff$, where $c$ is a constant independent from $k$, we have
    \begin{align*}
        F_{j} & = Q^{y_1} (UVV)^{k+1} Q^{y_2} = Q^{y_1} U (VVU)^{k+1} VU Q^{y_2-1} \text{ and }\\
        G_{j}\fragmentco{p}{|G_{j}|} & = Q^{y_1} (VVU)^{k+1} Q^{y_2} W, \text{ for some string } W.
    \end{align*}
    It is easy to observe that
    \begin{align*}
        &\min_t\ed(F_{j}, G_{j}\fragmentco{p}{t}) \\
        &\quad \leq \ed(Q^{y_1},Q^{y_1}) + \ed(U, \varepsilon) + \ed((VVU)^{k+1},(VVU)^{k+1})+ \ed(VU Q^{y_2-1}, Q^{y_2-1}UV)\\
        &\quad \leq |U| + \ed(VU Q^{y_2-1},U^{3y_2-1})+ \ed(U^{3y_2-1},Q^{y_2-1}UV)\\
        &\quad = i+1 + 2y_2.
    \end{align*}
    We can assume that $k$ is large enough so that $k>12 c \tradeoff$, in which case
    $i + 1 + 2y_2 < 2k/3 + k/6 + k/6 = k$. The point is that the implied alignment pays $|U|$ to ``synchronize'' copies
    of $VVU$, but it can then afford to pay for the (fewer than $k$) misaligned copies of $Q$ without needing to ``synchronize'' again.
    \lipicsEnd
\end{example}

\newcommand{\D}{\mathcal{D}}

It remains to show that, for any light position $p \in \big(\bigcup_{j \in J} (r_j +  \OccE_k(F_j,G_j))\big)$,
we have $p \in \OccE_k(P,T)$.
The following lemma demonstrates that, for each $j \in J$, we can restrict our attention to a subset of the positions of $R_j$.

\begin{lemma}\label{lem:focus3k}
    Consider some $j \in J$ and a position $p$ of $R_j$ such that $r_j+p$ is a light position of $T$.
    If $\ed(L_1^P,\rot^{-\rho(r_j+p)}(Q)^*) \geq \lpref$, then $p \not\in \OccE_k(F_j,G_j)$.
\end{lemma}
\begin{proof}
    By the definition of $\rred(T)$, it follows that $|F_j| \geq |L^P_1|$ and
    any prefix of $F_j$ of length at most $|L^P_1|+\tbd \ktotm$ is also a prefix of $P$.
    Combined with the upper bound on the sum of length-differences from \cref{fact:lengths_diff},
    this also implies that any prefix of $G_j$ of length at most $|L^P_1| + \tbd \ktotm - 3\ktot$
    is also a prefix of $R_j$.
    Further, recall that we have $|R_j|\leq |P| + 3\ktot - k$ due to \cref{fact:lengths_diff}.
    Consequently, since $|R_j|-|F_j| = |P| - |G_j|$,
    $\OccE_k(F_j,G_j) \subseteq \fragment{0}{|G_j|-|F_j|+k} \subseteq \fragment{0}{3\ktot}$.

    It thus suffices to consider the case where $p \in \fragment{0}{3\ktot}$.
    Let $W$ denote a prefix of $G_j\fragmentco{p}{|G_j|}$ that satisfies $\ed(F_j\fragmentco{0}{|L^P_1|}, W) = \min_t \ed(F_j\fragmentco{0}{|L^P_1|}, G_j\fragmentco{p}{t})$.
    We distinguish between two cases:
    \begin{itemize}
        \item If $|W| \in \fragment{|L^P_1|-k}{|L^P_1|+k}$, then
            $W$ is a prefix of $T\fragmentco{r_j+p}{n}$ since $p+|W| \leq 3\ktot+|L^P_1|+k \leq |L^P_1| + \tbd \ktotm - 3\ktot$.
            Thus, by~\cref{lem:lb_local}, we have
            \[\ed(L^P_1, W) \geq \ed(L_1^P,\rot^{-\rho(r_j+p)}(Q)^*) - \markf(r_j+p,L^P_1) \geq \lpref - \markf(r_j+p,L^P_1) \geq 2d - \tradeoff \geq d > k.\]
        \item Otherwise, we have $\ed(F_j\fragmentco{0}{|L^P_1|}, W)\geq \big||W|-|L^P_1|\big| > k$.
    \end{itemize}

    To conclude the proof, it suffices to observe that $\min_{t} \ed(F_j,G_j\fragmentco{p}{t}) \geq \min_{t} \ed(L^P_1,G_j\fragmentco{p}{t}) > k$,
    and hence $p \not\in \OccE_k(F_j, G_j)$.
\end{proof}

\newcommand{\orig}{\textsf{orig}}
In what follows, for convenience, we assume that, for each run of plain pairs of $\I_j$
that has been trimmed, the deleted pairs correspond to a suffix of this run.
This yields a natural mapping from pairs
of $\I'_j$ 
to pairs of $\I_j$. 

\begin{definition}
    For each $j\in J$ and each $i \in \fragment{1}{z_j}$, let $\orig(j,i)$ denote the number of pairs to the left of pair $(F_{j,i}, G_{j,i})$
    that were  deleted in the process of obtaining $\I'_j$ from $\I_j$.
    We say that pair $(F_{j,i}, G_{j,i})$ \emph{originates} from pair $(P_{i+\orig(j,i)}, T_{j,i+\orig(j,i)})$.
    \lipicsEnd
\end{definition}

Let each internal pair of pieces $(F_{j,i}, G_{j,i})$ inherit the color of
$(P_{i+\orig(j,i)}, T_{j, i+\orig(j,i)})$. In addition, mark the first and the last pairs
as not plain.

\begin{definition}
    For $j\in J$ and $i\in \fragment{1}{z_j}$, we say that $F_j\fragmentco{x_1}{x_2}$
    (or~$G_j\fragmentco{x_1}{x_2}$) has an \emph{overlap} with a pair $(F_{j,i},G_{j,i})$
    if and only if
    $\fragmentco{x_1}{x_2} \cap \fragmentco{f_{j,i}}{f_{j,i+1}+\Delta} \neq \emptyset$
    (or ~$\fragmentco{x_1}{x_2} \cap \fragmentco{g_{j,i}}{g_{j,i+1}+\Delta} \neq \emptyset$).
    \lipicsEnd
\end{definition}


\begin{definition}
    Consider some $j \in J$ and a contiguous sequence
    $\mathcal{M}=(F_{j,i_1},G_{j,i_1})\cdots (F_{j,i_2},G_{j,i_2})$ of pairs in $\I'_j$
    that are either all plain or all not plain (that is, \(\mathcal{M}\) is
    \emph{monochromatic}).
    Fix an $X \in \{F_j,G_j\}$.
    For $i \in \fragment{1}{z_j}$, if $X=F$, we set
    $x_{j,i}=f_{j,i}$; otherwise, we set
    $x_{j,i}=g_{j,i}$.

    For a non-negative integer $\gamma$, we say that a fragment
    $X\fragmentco{x_1}{x_2}$ of $X \in \{F_j,G_j\}$ is
    $\gamma$-\emph{contained} by $\mathcal{M}$ when
    the following two conditions are satisfied: (a)
    $x_1 \geq x_{j,i_1} + \gamma$ or $i_1 = 1$
    and (b)
    $x_2 \leq x_{j,i_2+1}+\Delta - \gamma$ or $i_2 = z_j$.
    \lipicsEnd
\end{definition}

\begin{fact}\label{obs:Dcont}
If a fragment $U$ of $F$ or $G$ is $\Delta$-contained by a monochromatic sequence $\mathcal{M}$ of contiguous pairs in $\I'_j$,
then $U$ only overlaps pairs in $\mathcal{M}$.
\lipicsEnd
\end{fact}

\begin{lemma}\label{lem:monochrom}
For $j\in J$, consider a monochromatic sequence $\mathcal{M}=(F_{j,i_1},G_{j,i_1})\cdots (F_{j,i_2},G_{j,i_2})$ of contiguous pairs in $\I'_j$.
Let $\{X, Y\} = \{F_j, G_j\}$ and $X\fragmentco{x_1}{x_2}$ be a fragment of $X$ that is $\gamma$-contained by $\mathcal{M}$ for $\gamma \geq 13 \ktot$.
All the pairs of $\I'_j$ that overlap $Y\fragmentco{y_1}{y_2} \coloneqq Y\fragmentco{\max\{0,x_1-4\ktot\}}{\min\{|Y|,x_2+4\ktot\}}$ belong to $\mathcal{M}$.
\end{lemma}
\begin{proof}
For $i \in \fragment{1}{z_j}$, if $X=F_j$ and $Y=G_j$, let
$x_{j,i}=f_{j,i}$ and $y_{j,i}=g_{j,i}$; otherwise, let
$x_{j,i}=g_{j,i}$ and $y_{j,i}=f_{j,i}$.

Recall that the sum of length-differences of pairs in $\I'_j$ is at most $3\ktot - k$,
that is, $|x_{j,i}-y_{j,i}|\leq 3\ktot -k$ for all $i$.
Let us first show that $Y\fragmentco{y_1}{y_2}$ does not overlap any of the first $i_1-1$ pairs.
The case where $i_1=1$ is trivial.
In the remaining case,
\[y_1 \geq x_1 - 4\ktot \geq x_{j,i_1} + \gamma - 4\ktot \geq y_{j,i_1}+\tbd \ktot - 4\ktot - (3\ktot - k) \geq y_{j,i_1} + (\tbd-7) \ktot = y_{j,i_1}+\Delta.\]
We next show that $Y\fragmentco{y_1}{y_2}$ does not overlap any of the last $z_j - i_2$ pairs, thus concluding the proof.
The case where $i_2=z_j$ is trivial.
In the remaining case,
\[y_2 \leq x_2 +4\ktot \leq x_{j,i_2+1}+\Delta - \gamma + 4\ktot \leq y_{j,i_2+1}+\Delta - \tbd \ktot + 4\ktot + (3\ktot - k) \leq y_{j,i_2+1} + \Delta - (\tbd-7) \ktot = y_{j,i_2+1}.\]
We are done by a direct application of \cref{obs:Dcont}.
\end{proof}

\begin{lemma}\label{lem:nofakelight}
    Consider some $j\in J$ and a position $p \in \OccE_k(F_j,G_j)$ such that $r_j+p$ is a light position of $T$.
    Then, $r_j + p \in \OccE_k(P,T)$.
\end{lemma}
\begin{proof}
    First, observe that we may assume that $\I'_j \neq \I_j$; otherwise, the statement follows trivially.
    To avoid clutter, we drop the subscript $j$ when referring to $F_j$ and $G_j$ and simply call them $F$ and~$G$, respectively.

    Let $b$ denote a position of $G$ and let
    $\mathcal{B}: F \onto G\fragmentco{p}{b}$ denote an alignment of cost $\min_t \ed(F, G\fragmentco{p}{t}) \leq k$.
    We intend to show that, in this case, there exist a position $c$ of $R_j$ and an alignment $\mathcal{C}: P \onto R_j\fragmentco{p}{c}$ of the same cost.


    Set \[
        \res{\mathcal{L}^T}{\fragmentco{r_j+p}{r_j+p+m}}=\{L^T_i : i \in \fragment{i_1}{i_2}\}.
    \]
    Observe that there is a natural mapping of each locked fragment $L^P_i \in \mathcal{L}^P$ to a fragment of $F$, which we denote by~$L^{F}_i$;
    we denote the set of fragments in the image of this mapping by~$\mathcal{L}^F$.
    Similarly, there is a natural mapping of each locked fragment $L^T_i \in \{L^T_i : i \in \fragment{i_1}{i_2}\}$ to a fragment of~$G$, which we denote by~$L^{G}_i$;
    we denote the set of fragments in the image of this mapping by~$\mathcal{L}^G$.
    Let $\mathcal{D}'(p)$ consist of the images of the locked fragments in $\D(r_j+p)$ under these mappings.
    Observe that each fragment $L \in \mathcal{L}^F \cup \mathcal{L}^G$ only overlaps
    non-plain pairs.
    For a fragment $L^X_y \in \mathcal{L}^X$, where $X \in \{P,T,F,G\}$, let $L^X_y=X\fragmentco{\ell^X_y}{r^X_y}$.
    The following claim follows instantly.

    \begin{claim}\label{claim:lockcont}
        Each fragment $L \in \mathcal{L}^F \cup \mathcal{L}^G$ is $\tbd \ktotm$-contained by a sequence
        of contiguous non-plain pairs in~$\I'_j$.
        \lipicsClaimEnd
    \end{claim}

    We next essentially show that, if we were to mark positions of $G_j$
    based on overlaps of pairs of fragments in $\mathcal{L}^F \times \mathcal{L}^G$, consistently with \cref{def:marking},
    position $p$ of $G_j$ would get the same number of marks as position $r_j+p$ of $T$.

    \begin{claim}\label{claim:wip}
        For all $x \in \fragment{1}{|\mathcal{L}^P|}$ and $y \in \fragment{i_1}{i_2}$, we have
        \[\big|\fragmentco{\ell^G_y-\ktotm}{r^G_y+\ktotm} \cap \fragmentco{p + \ell^F_x}{p + r^F_x}\big| =
        \big|\fragmentco{\ell^T_y-\ktotm}{r^T_y+\ktotm} \cap \fragmentco{r_j + p + \ell^P_x}{r_j + p + r^P_x}\big|.\]
    \end{claim}
    \begin{claimproof}
        Consider sequences $\mathcal{M}_x = (F_{j,w_1}, G_{j,w_1}) \cdots (F_{j,w_2}, G_{j,w_2})$
        and $\mathcal{M}_y = (F_{j,w_3}, G_{j,w_3}) \cdots (F_{j,w_4}, G_{j,w_4})$ of
        contiguous non-plain pairs of $\I'_j$
        that $\tbd \ktotm$-contain $L^F_x$ and $L^G_y$, respectively,
        and are maximal in the sense that they cannot be extended and remain monochromatic.
        (Such sequences exist by \cref{claim:lockcont}.)

        First, consider the case where $\mathcal{M}_x$ and $\mathcal{M}_y$ do not coincide.
        We treat the case where $\mathcal{M}_x$ lies to the left of~$\mathcal{M}_y$; the other case can be handled analogously.
        In this case, $w_2<z_j$ and $w_3>1$.
        Recall that $p \in \fragment{0}{3\ktot}$.
        By \cref{lem:monochrom}, we have
        that $G\fragmentco{p + \ell^F_x}{p + r^F_x}$
        only overlaps pairs in $\mathcal{M}_x$ and hence it is disjoint
        from $\fragmentco{\ell^G_y-\ktotm}{r^G_y+\ktotm}$ since
        $p + r^F_x \leq g_{j,w_3} \leq \ell^G_y-\ktotm$.
        Now, observe that $\orig(j,w_1) \leq \orig(j,w_3)$
        and hence
        \[r_j + p + r^P_x = r_j + \orig(j,w_1)\tau + p + r^F_x < r_j + \orig(j,w_3)\tau + \ell^G_y-\ktotm = \ell^T_y - \ktotm.\]
        Thus, in this case, both considered intersections are empty.

        Otherwise, $\mathcal{M}_x$ and $\mathcal{M}_y$ coincide.
        We then have
        \[r_j + \orig(j,w_1)\tau + \fragmentco{p+\ell^F_x}{p+r^F_x} = r_j + \fragmentco{p+\ell^P_x}{p+r^P_x} = \fragmentco{r_j + p+\ell^P_x}{r_j + p+r^P_x}\]
        and
        \[r_j + \orig(j,w_1)\tau + \fragmentco{\ell^G_y - \ktotm}{r^G_y + \ktotm} = r_j + \fragmentco{\ell^T_y - r_j - \ktotm}{r^T_y - r_j + \ktotm} = \fragmentco{\ell^T_y - \ktotm}{r^T_y + \ktotm}.\]
        The statement readily follows in the considered case.
    \end{claimproof}

    Let $\mathcal{B} = (f_v,g_v)_{v=0}^u$.
    For each $L_i^F\in \D'(p)$, let $\fragmentco{a^F_i}{b^F_i}\sub \fragment{0}{u}$ so that $L^F_i = F\fragmentco{f_{a^F_i}}{f_{b^F_i}}$ and $\mathcal{B}(L^F_i) = G\fragmentco{g_{a^F_i}}{g_{b^F_i}}$.
    Symmetrically, for each  $L^G_i\in \D'(p)$, let $\fragmentco{a^G_i}{b^G_i}\sub \fragment{0}{u}$ so that
    $L^G_i = G\fragmentco{g_{a^G_i}}{g_{b^G_i}}$ and $\mathcal{B}^{-1}(L^G_i) = F\fragmentco{f_{a^G_i}}{f_{b^G_i}}$.
    Further, let $\mathcal{E}$ denote the multiset union of the multisets
    \[\mathcal{E}^F = \{\fragmentco{a^F_i}{b^F_i} : L_i^F\in \mathcal{L}^F \cap \mathcal{D}'(p)\} \quad \text{and} \quad \mathcal{E}^G = \{\fragmentco{a^G_i}{b^G_i} : L_i^G\in \mathcal{L}^G\cap \D'(p)\}.\]
    In fact, the following claim implies that the multiplicity of each element of $\mathcal{E}$ is one.

    \begin{claim}\label{claim:disjoint}
        The intervals in $\mathcal{E}$ are pairwise disjoint.
    \end{claim}
    \begin{claimproof}
        First, observe that the elements of each of $\mathcal{L}^F$ and $\mathcal{L}^G$
        are pairwise disjoint and hence the elements of each of $\mathcal{E}^F$ and $\mathcal{E}^F$
        are pairwise disjoint.

        Now, consider any two fragments $L^F_x \in \D'(p)$ and
        $L^G_y\in \D'(p)$. Observe that since $L^P_x, L^T_y \in \D(r_j+p)$, we have
        $\markf(r_j+p,\ktotm,L^P_x,L^T_y)=0$.
        Hence, $\fragmentco{\ell^T_y-\ktotm}{r^T_y+\ktotm} \cap \fragmentco{r_j + p + \ell^P_x}{r_j + p + r^P_x}=\emptyset$.
        By \cref{claim:wip}, we also have $\fragmentco{g_{a^G_y}-\ktotm}{g_{b^G_y}+\ktotm} \cap \fragmentco{p + f_{a^F_x}}{p + f_{b^F_x}}=\emptyset$.

        Since the cost of $\mathcal{B}$ is no more than $k$,
        $g_{i}\in \fragment{p+f_i-k}{p+f_i+k}$ holds for all $i\in \fragment{0}{u}$.
        In particular, we have $\fragment{g_{a^G_y}-k-1}{g_{b^G_y}+k} \supseteq \fragmentco{p+f_{a^G_y}-1}{p+f_{b^G_y}+1}$.
        Since $\fragmentco{g_{a^G_y}-\ktotm}{g_{b^G_y}+\ktotm} \supseteq \fragment{g_{a^G_y}-k-1}{g_{b^G_y}+k}$,
        we then have $\fragmentco{f_{a^G_y}-1}{f_{b^G_y}+1}\cap \fragmentco{f_{a^F_x}}{f_{b^F_x}}=\emptyset$,
        that is, $\fragment{f_{a^G_y}}{f_{b^G_y}}\cap \fragment{f_{a^F_x}}{f_{b^F_x}}=\emptyset$.
        This implies $\fragmentco{a^G_y}{b^G_y}\cap \fragmentco{a^F_x}{b^F_x}=\emptyset$;
        consequently, the intervals in $\mathcal{E}$ are pairwise disjoint.
    \end{claimproof}

    \[\text{Set } \Lambda \coloneqq \ed(L_1^P,\rot^{-\rho(r_j+p)}(Q)^*)+\sum_{i=2}^{\ell^P}\edl{L^P_i}{Q}+\sum_{i=i_1}^{i_2}\edl{L^T_i}{Q}.\]

    \begin{claim}\label{claim:lb}
        \[\sum_{L^{F}_i\in \mathcal{D}'_j(p)}\ed(L^{F}_i,\mathcal{B}(L^{F}_i)) + \sum_{L^{G}_i \in \mathcal{D}'_j(p)}\ed(L^{G}_i, \B^{-1}(L^{G}_i)) \geq \Lambda - 2\markf(r_j+p).\]
    \end{claim}
    \begin{claimproof}
        Using~\cref{lem:lb_local}, we lower bound each individual term of the left-hand side of the proved inequality.
        We consider three cases.
        \begin{enumerate}
            \item Consider some $L^{G}_i=G\fragmentco{p+\ell}{p+r} \in \mathcal{D}'_j(p)$ and let $L^T_i=\fragmentco{r_j + p +\ell'}{r_j + p + r'}$.
                Since $\ed(F,G\fragmentco{p}{b}) \leq k$, we have $\mathcal{B}^{-1}(L^{G}_i) \substr F\fragmentco{\max\{0,\ell-k\}}{\min\{|F|,r+k\}}$.
                By \cref{claim:lockcont}, $L^G_i$ is $\tbd \ktotm$-contained by a sequence
                $\mathcal{M}$ of contiguous non-plain pairs in $\I'_j$.
                Then, by \cref{lem:monochrom}, $F\fragmentco{\max\{0,\ell-k\}}{\min\{|F|,r+k\}}$ only overlaps pairs of $\I'_j$ that are in $\mathcal{M}$ and
                hence it is a substring of
                $P\fragmentco{(\ell' - \ell) + \max\{0,\ell-k\}}{(r' - r) + \min\{|F|,r+k\}}$, which in turn is a substring of
                $P\fragmentco{(\max\{0,\ell'-k\}}{\min\{|P|,r'+k\}}$.
                Thus, $\ed(L^G_i, \mathcal{B}^{-1}(L^{G}_i)) \geq \edl{L^T_j}{Q} - \markf(r_j+p,L^T_j)$
                holds by~\cref{lem:lb_local}.
            \item Consider some $L^F_i = F\fragmentco{\ell}{r} \in \mathcal{D}'_j(p)\setminus \{L^P_1\}$ and let $L^P_i = P\fragmentco{\ell'}{r'}$.
                Since $\ed(F,G\fragmentco{p}{b}) \leq k$, we have $\mathcal{B}(L^{F}_i) \substr G\fragmentco{\max\{p,p+\ell-k\}}{\min\{|G|,p+r+k\}}$.
                As before, by combining \cref{claim:lockcont} and \cref{lem:monochrom} we get that $G\fragmentco{\max\{p,p+\ell-k\}}{\min\{|G|,p+r+k\}}$
                is a substring of
                $T\fragmentco{r_j + (\ell' - \ell) + \max\{p,p+\ell-k\}}{r_j + (r' - r) + \min\{|G|,p+r+k\}}$, which in turn is a substring of
                $T\fragmentco{r_j + p + \max\{0,\ell'-k\}}{r_j + p + \min\{|T|,r'+k\}}$.
                Thus,
                $\ed(L^F_i,\mathcal{B}(L^F_i))\geq \edl{L^P_i}{Q} - \markf(r_j+p,L^P_i)$ holds by~\cref{lem:lb_local}.
        \item Lastly, since $\ed(F,G\fragmentco{p}{b}) \leq k$, we have that $\mathcal{B}(L^F_1)$ is a prefix of $T\fragmentco{r_j+p}{r_j+p+|L^P_1|+k}\}$, and hence
            $\ed(L^F_1,\mathcal{B}(L^F_1))\geq \ed(L_1^P,\rot^{-\rho(r_j+p)}(Q)^*) - \markf(r_j+p,L_1^P)$ holds by~\cref{lem:lb_local}.
    \end{enumerate}

    We now put everything together.
    In the following inequalities, we use the fact that,
    for each $L \in (\mathcal{L}^P \cup \mathcal{L}^T_{\fragmentco{r_j+p}{r_j+p+m}})\setminus{\mathcal{D}(r_j+p,B)}$, we have $\edl{L}{Q} \leq \markf(r_j+p,L)$,
    and hence the sum of $\edl{L}{Q} - \markf(r_j+p,L)$ over all such $L$ is at most zero.
    \begin{align*}
    &		\sum_{L^{F}_i\in \mathcal{D}'_j(p)}\ed(L^{F}_i,\mathcal{B}(L^{F}_i)) + \sum_{L^{G}_i \in \mathcal{D}'_j(p)}\ed(\B^{-1}(L^{G}_i),L^{G}_i)\\
    &\quad  \geq \ed(L^{P}_1,\rot^{-\rho(r_j+p)}(Q)^*) - \markf(r_j+p,L^P_1) + \sum_{L^P_i \in \mathcal{D}(r_j+p) \setminus \{L^P_1\}} (\edl{L^P_i}{Q} - \markf(r_j+p,L^P_i))\\
    &\qquad\qquad\qquad\qquad\qquad\qquad\qquad\qquad\qquad\qquad + \sum_{L^T_i \in \mathcal{D}(r_j+p)} (\edl{L^T_i}{Q} - \markf(r_j+p,L^T_i))\\
    &\quad  \geq \ed(L_1^P,\rot^{-\rho(r_j+p)}(Q)^*) + \sum_{i=2}^{\ell^P}\edl{L^P_i}{Q}+\sum_{i=i_1}^{i_2}\edl{L^T_i}{Q} - \sum_{i=1}^{\ell^P}\markf(r_j+p,L^P_i) \\
    &\qquad\qquad\qquad\qquad\qquad\qquad\qquad\qquad\qquad\qquad\qquad\qquad\qquad\qquad\;\, - \sum_{i=i_1}^{i_2}\markf(r_j+p,L^T_i)\\
    &\quad  = \Lambda - \sum_{i=1}^{\ell^P}\sum_{v=1}^{\ell^T} \markf(r_j+p,\ktotm,L^P_v,L^T_i) - \sum_{i=i_1}^{i_2}\sum_{v=1}^{\ell^P} \markf(r_j+p,\ktotm,L^P_v,L^T_i) \\
    &\quad  \geq \Lambda - 2\markf(r_j+p).
    \end{align*}
    This concludes the proof of the claim.
\end{claimproof}

Next, due to \cref{lem:focus3k}, we have
$\ed(L_1^P,\rot^{-\rho(r_j+p)}(Q)^*) < \lpref$.
By a direct application of \cref{lem:ub}, we then have
that $\min_t \ed(P,T\fragmentco{r_j+p}{t})\leq \Lambda$.
If $\Lambda\leq k$, then we are done. For the remainder of the proof we thus consider the case where $\Lambda >k$, which, combined with the fact
that $r_j + p$ is a light position of $T$, means that
$\Lambda-2\markf(r_j+p) > k-2\markf(r_j+p) \geq k-2 \tradeoff$.

\begin{claim}\label{claim:plain_exact}
    For any run $M$ of $2\tradeoff + \slack$ consecutive plain pairs $(F_{j,i_1},G_{j,i_1}), \ldots, (F_{j,i_2},G_{j,i_2})$ in $\I'_j$,
    there exists some $i\in \fragmentoo{i_1}{i_2}$ for which $\ed(F\fragmentco{f_{j,i}}{f_{j,i+1}}, \mathcal{B}(F\fragmentco{f_{j,i}}{f_{j,i+1}})) = 0$.
\end{claim}
\begin{claimproof}
    We have $|f_{j,i_1+14} - f_{j,i_1}| \geq 14 \tau - d_T \geq 14 \ktot/2 - d_T \geq 6 \ktot = \Delta$. Similarly,
    $|f_{j,i_2-13} - f_{j,i_2+1}| \geq \Delta$.
    Thus, $U=F\fragmentco{f_{j,i_1+14}}{f_{j,i_2-13}}$ is $\Delta$-contained by $M$ and hence only overlaps pairs in $M$ by \cref{obs:Dcont}.

    Now, each fragment $L_i^F \in \mathcal{L}^F$ is $13\ktotm$-contained by a sequence of
    contiguous non-plain pairs and hence only overlaps non-plain pairs.
    Further, for each fragment $L_i^G \in \mathcal{L}^G$, as shown in the proof of
    \cref{claim:lb}, $\mathcal{B}^{-1}(L_i^G)$ only overlaps non-plain pairs.
    Observe that two fragments of $F$ are necessarily disjoint if the sets of pairs that they overlap are disjoint, and hence $\fragmentco{f_{j,i_1+14}}{f_{j,i_2-13}}$
    is disjoint from all elements of $\mathcal{E}$.
    As the intervals in $\mathcal{E}$ are pairwise disjoint by \cref{claim:disjoint}, using \cref{claim:lb}, we obtain
    \begin{align*}
        k &\geq \ed(F,G\fragmentco{p}{b})\\
          &\geq \ed(U,\mathcal{B}(U)) + \sum_{L^{F}_i\in \mathcal{D}'_j(p)}\ed(L^{F}_i,\mathcal{B}(L^{F}_i)) + \sum_{L^{G}_i \in \mathcal{D}'_j(p)}\ed(\B^{-1}(L^{G}_i),L^{G}_i)\\
          &\geq \ed(U,\mathcal{B}(U)) + \Lambda - 2\markf(r_j+p) \\
          &> \ed(U,\mathcal{B}(U)) + k - 2\tradeoff.
    \end{align*}
    Hence, we have $\ed(U,\mathcal{B}(U)) < 2\tradeoff$.
    This implies that $\sum_{i=i_1+14}^{i_2-14} \ed(F\fragmentco{f_{j,i}}{f_{j,i+1}},\mathcal{B}(F\fragmentco{f_{j,i}}{f_{j,i+1}})) < 2\tradeoff$,
    and hence, since the number of summands in in the left-hand side of the inequality is $i_2 - 14 - (i_1+14) + 1 = 2\tradeoff + \slack - 28 > \tradeoff$
    and each of these summands is a non-negative integer, the claim follows.
\end{claimproof}

We are now ready to conclude the proof of the lemma.
Let $\mathcal{M}_1, \ldots, \mathcal{M}_w$ denote the trimmed runs of $2\tradeoff + \slack$
consecutive plain pairs in $\I'_j$ such that, for all $i$, $\mathcal{M}_i$
originated from a run with $\chi(i)$ more plain pairs. For $i \in \fragment{1}{w}$, let $\mathcal{M}_i=(F_{j,s_i},G_{j,s_i}) \cdots (F_{j,e_i},G_{j,e_i})$
let $\psi(i) \in \fragmentoo{s_i}{e_i}$ be such that
and $\ed(F\fragmentco{f_{j,\psi(i)}}{f_{j,\psi(i)+1}},
\mathcal{B}(F\fragmentco{f_{j,\psi(i)}}{f_{j,\psi(i)+1}})) = 0$; observe that $\psi(i)$
exists due to \cref{claim:plain_exact}.
Let us now show how to construct $\mathcal{C}$ given $\mathcal{B} = (f_v,g_v)_{v=0}^u$.
We intuitively achieve this by inserting $\chi(i)$ copies of $Q^{\infty}\fragmentco{0}{\tau}$ after each of $F\fragmentco{f_{j,\psi(i)}}{f_{j,\psi(i)+1}}$
and $\mathcal{B}(F\fragmentco{f_{j,\psi(i)}}{f_{j,\psi(i)+1}})$ and aligning them without
errors, thus restoring the original length of each trimmed plain run,
without changing the cost of the alignment.
Initially, set $\mathcal{C} \coloneqq \mathcal{B}$.
Then, for each $\psi(i)$, in decreasing order
\begin{itemize}
    \item replace each pair $(f_v,g_v)$ of $\mathcal{C}$ that satisfies $f_v \geq f_{j,\psi(i)+1}$
        with $(\chi(i)\cdot\tau + f_v, \chi(i)\cdot\tau + g_v)$, and
    \item insert $(f_{j,\psi(i)+1}+\mu, g_{j,\psi(i)+1}+\mu)_{\mu=0}^{\chi(i)\cdot\tau-1}$ after $(f_{j,\psi(i)+1}-1,g_{j,\psi(i)+1}-1)$.\qedhere
\end{itemize}
\end{proof}

We conclude with the proof of \cref{lem:correctlight} as promised.

\correctlight*
\begin{proof}
$(\subseteq)$: This direction is an immediate consequence of \cref{cor:supset}.

$(\supseteq)$: This direction is an immediate consequence of \cref{lem:nofakelight}.
\end{proof}

\subsection{Combining the Partial Results: Faster \SM}

We are ready to prove our headline result---\cref{lem:sthold}---which we restate here for
convenience.

\newsm*
\begin{proof}
\newcommand{\lmore}{\mathcal{R}}
\newcommand{\mq}{\mathcal{F}}
    Consistently with all previous sections, set \[
        \tradeoff \coloneqq \max\{1 , \lfloor{\sqrt{d}/\sqrt{\log(n+1) \log(d+1)}}\rfloor\}
        \leq d\quad\text{and}\quad \ktot \coloneqq k+d_P+d_T\quad\text{and}\quad
        \ktotm \coloneqq 7d.
    \]

    First, observe that we may assume
    \(\tradeoff = \Theta (\sqrt{d}/\sqrt{(\log(n+1) \log(d+1)})\)---otherwise,
    in the case where $\sqrt{d} \leq \sqrt{\log(n+1) \log(d+1)}$,
    {\tt PeriodicMatches($P$, $T$, $k$, $d$, $Q$)} from \cref{lm:impEdC} already
    runs in time
    $\cO(d^4) = \cO(d^{3.5}\sqrt{\log(n+1) \log(d+1)})$.\footnote{All running times in this proof are in the \modelname model.}

    We proceed in roughly four steps.
    \begin{itemize}
        \item First, in a preprocessing step, we identify the heavy positions \(\Hv\) and the light
            positions \(\light\) in \(T\), as well as a filter \(\mq\) (according to
            \cref{lm-tmp-4.2-2}) for where potential \(k\)-error occurrences may start.
        \item Next, we compute all \(k\)-error occurrences starting at a position in \(\Hv
            \), using the {\tt Verify}-based algorithm from \cref{sus:heavy}.
            In particular, we obtain \(\Hv \cap \OccE_k(P,T)\) as a set of positions.
        \item Next, we use \(\swaps\) from \cref{lem:swaps} to obtain a candidate set
            \(\lmore\) of potential starting positions of \(k\)-error occurrences
            (represented as \(\Oh(d^3\tradeoff)\) disjoint arithmetic progressions with difference
            \(\tau\)).
            From \cref{sus:light}, we have
            \((\lmore \cap \mq) \setminus (\Hv \cap \mq) = \OccE_k(P, T) \cap \light\);
            hence we proceed to compute \(\Hv \cap \mq\) (represented as a set), and
            \(\lmore \cap \mq\) (represented as \(\Oh(d^3\tradeoff)\) disjoint
            arithmetic progressions with difference~\(\tau\)), and thereafter compute
            their set difference to obtain \(\OccE_k(P, T) \cap \light\)
            (represented as \(\Oh(d^3\tradeoff)\) disjoint arithmetic progressions with
            difference~\(\tau\)).
        \item In a post-processing step, we union the two sets \(\Hv \cap \OccE_k(P,T)\) and
            \(\OccE_k(P, T) \cap \light\) and compute a representation as \(\Oh(d^3)\)
            arithmetic progressions with difference \(q\).
    \end{itemize}

    \subparagraph{Preprocessing.}
    We compute the sets of locked fragments $\LP = \locked(P,Q,d_P,\lpref)$ and
    $\LT = \locked(T,Q,d_T,0)$ in $\cO(d^2)$ time using \cref{lem:klocked} and call
    $\protect\heavy(P,T,d,k,Q,\LP,\LT,\ktotm,\tradeoff)$ to obtain the set $\Hv$ of
    heavy positions,
    represented as the union of $\Oh(d^2)$ disjoint integer ranges,
    in $\Oh(d^2 \log \log d)$ time (cf.~\cref{lem:heavy-alg}).

    Define $J$ and each of $r_j$ and $\I'_j$, for $j\in J$, as
    in \cref{sec:smnew,sec:DPMtile} and \cref{def:fjgj}, and set \[
        \lmore \coloneqq \bigcup_{j \in J} (r_j +  \OccE_k( \I'_j)).
    \]
    Observe that by \cref{lem:correctlight}, $\OccE_k(P,T)$ is equal to the union
    \[
        \OccE_k(P,T) = (\OccE_k(P,T) \cap \Hv) \cup (\OccE_k(P,T) \cap \light) =
         (\OccE_k(P,T) \cap \Hv) \cup (\lmore \cap \light).
    \]
    Finally, for our filter, we recall that \cref{lm-tmp-4.2-2} yields
    \[
        \mq \coloneqq \bigcup_{j\in \mathbb{Z}} \fragment{j q-x_T-\kappa-d_T}{j
        q-x_T+\kappa+d_T} \supseteq \OccE_k(P,T).
    \]
    Observe that we thus have
    \begin{equation}\label{eq:mq}
        \OccE_k(P,T) \cap \light = (\lmore \cap \mq) \setminus (\Hv \cap \mq).
    \end{equation}

    \subparagraph{Heavy occurrences.}

    By \cref{lem:heavy-total}, we can compute $\OccE_k(P,T) \cap \Hv$ in time
    \[\cO(d^3 + d^{4} / \tradeoff)=\cO(d^{3.5}\sqrt{\log n \log d}).\]

    \subparagraph{Light occurrences.}

    Let us now proceed to computing the remaining $k$-error occurrences, namely those that
    start at light positions.
    To this end, we first compute the $\cO(d)$-size sets $\rred(P)$ and $\rred(T)$,
    defined in the beginning of \cref{sec:redundantsqrtk},
    in $\cO(d \log\log d)$ time (cf.~\cref{lem:fewreds}).
    Then, we call \(\swaps(T, P, k, Q, \A_P, \A_T, \rred(P), \rred(T), 2\tradeoff +
    \slack)\) from \cref{lem:swaps}, which returns
    a representation of \(\lmore\)
    as $\cO(d^3 \cdot \tradeoff)$ arithmetic progressions with difference
    $\tau \coloneqq q\ceil{{\kappa}/{2q}}$.
    Plugging in \cref{thm:dpm} for the \DPM data structure, \swaps runs in time \[
        \cO(d^3 \cdot \tradeoff \cdot \log n \log d)=\cO(d^{3.5}\sqrt{\log n \log d}).
    \]

    We move on to remove the surplus positions from the candidate set \(\lmore\).
    We start by computing \(\Hv \cap \mq\).

    \begin{claim}\label{claim:hmq}
        The set $\Hv \cap \mq$ is of size $\cO(d^3/\tradeoff)$ and can be computed in time $\cO(d^3/\tradeoff)$.
    \end{claim}
    \begin{claimproof}
        The proof is similar to a part of the proof of~\cref{lem:heavy-total}.

        Due to \cref{lem:heavy-alg}, our representation of $\Hv$ consists of $\Oh(d^2)$ disjoint integer ranges.
        In addition, due to \cref{lem:heavy-bound}, we have \[
            |\Hv|=\cO((d^2/\tradeoff + d)(\ktotm+q)) = \cO(d^2/\tradeoff \cdot (d+q)).
        \]

        Let us first upper-bound the size of \(|\Hv \cap \mq|\).
        We distinguish between two cases.
        \begin{itemize}
            \item First, if $q \leq d$, we have $|\Hv \cap \mq| \leq |\Hv|=\cO(d^3/\tradeoff)$.
            \item As for the complementary case where $q > d$,
                observe that only $\cO(d)$ out of any $\cO(q)$ consecutive integers are in $\mq$.
                Hence, $\Hv \cap \mq$ is of size $\cO(d^2 + |\Hv|\cdot d/q) = \cO(d^3/\tradeoff)$.
        \end{itemize}

        As for computing $\Hv \cap \mq$, we first
        sort the $\Oh(d^2)$ integer ranges that comprise $\Hv$ with respect to their starting positions in $\cO(d^2 \log d)$ time
        and then scan them from left to right, skipping positions in $\mathbb{Z} \setminus \mq$.
        The running time of the scan is proportional to the total number of input ranges and output positions, and hence we are done.
    \end{claimproof}

    We move on to compute \(\lmore \cap \mq\).
    In what follows, for any integer $x$, let us say that the \emph{residue modulo $x$} of a non-empty arithmetic progression whose difference is a multiple of $x$
    is the residue of any element of this arithmetic progression modulo $x$.

    \begin{claim}\label{claim:fmq}
        We can compute a representation of
        $\lmore \cap \mq$ as $\cO(d^3 \tradeoff)$ disjoint arithmetic progressions with difference~$\tau$,
        sorted according to their starting positions,
        in $\cO(d^{3.5}\sqrt{\log n \log d})$ time.
    \end{claim}
    \begin{claimproof}
        In a linear scan of the $\cO(d^3 \tradeoff)$ arithmetic progressions that comprise $\lmore$ as returned by the call to the algorithm $\swaps$,
        we delete any arithmetic progression whose elements are not in $\mq$.
        Then, we sort all arithmetic progressions according to their starting positions in time $\cO(d^3 \tradeoff \log d)=\cO(d^{3.5})$
        and distribute them among $\cO(d/q\cdot \tau)=\cO(d^2)$ buckets according to their residues modulo~$\tau$.
        Finally, we scan linearly the arithmetic progressions in each bucket, greedily
        merging progressions that overlap (that is, we merge progressions if their union
        is also a valid arithmetic progression with difference \(\tau\)).
        The number of the resulting arithmetic progressions is clearly upper-bounded by the size
        of the input, that is, $\cO(d^3 \tradeoff)$;
        we sort them according to their starting positions in $\cO(d^{3.5})$ time.
    \end{claimproof}

    Finally, we compute the set difference \(
        (\lmore \cap \mq) \setminus (\Hv \cap \mq) =
        \OccE_k(P,T) \cap \light.
    \)

    \begin{claim}\label{claim:apdt}
        A representation of $\OccE_k(P,T) \cap \light$ as $\cO(d^3 \tradeoff)$ disjoint arithmetic progressions with difference $\tau$,
        sorted according to their starting positions,
        can be computed in time $\cO(d^{3.5}\sqrt{\log n \log d})$ time the \modelname model.
    \end{claim}
    \begin{claimproof}
        We intend to use equation \eqref{eq:mq}, relying on \cref{claim:hmq,claim:fmq}
        to compute $\Hv \cap \mq$ and a representation of $\lmore \cap \mq$ as $\cO(d^3
        \tradeoff)$ disjoint arithmetic progressions with difference $\tau$ in
        $\cO(d^{3.5})$ time in total.
        Then, we process the elements of these two sets in $\cO(d/q\cdot \tau)=\cO(d^2)$
        batches, where each batch contains all elements with a specific residue modulo
        $\tau$.
        For such a fixed residue, we scan in parallel the arithmetic progressions in
        $\lmore \cap \mq$ and
        the positions in $\Hv \cap \mq$, both of which are sorted in increasing order.
        When some element of $\Hv \cap \mq$ is contained in some arithmetic progression in $\lmore \cap \mq$,
        this arithmetic progression is split into two parts, either of which may be empty.
        The number of the resulting arithmetic progressions is upper-bounded by the total
        size of the input, that is, $\cO(d^3 \tradeoff)$;
        we sort them according to their starting positions in $\cO(d^{3.5})$ time.
    \end{claimproof}

    \subparagraph{Post-processing.}

    Recall that the set $\OccE_k(P,T) \cap \Hv$ is computed explicitly using \cref{lem:heavy-total} in $\cO(d^{3.5}\sqrt{\log n \log d})$ time;
    its size is $\cO(d^3/\tradeoff)$ due to \cref{lm-tmp-4.2-2,claim:hmq}.
    In addition, due to \cref{claim:apdt}, a representation of $\OccE_k(P,T) \cap \light$ as $\cO(d^3 \tradeoff)$ disjoint arithmetic progressions with difference $\tau$
    can be computed in $\cO(d^{3.5}\sqrt{\log n \log d})$ time.
    Overall, by taking the union of the disjoint sets $\OccE_k(P,T) \cap \Hv$ and $\OccE_k(P,T) \cap \light$,
    we obtain a representation of $\OccE_k(P,T)$ as $\cO(d^3 \tradeoff)$ disjoint arithmetic progressions with difference $\tau$.
    In what follows, we show how to efficiently replace these arithmetic progressions with difference $\tau$ with arithmetic progressions with difference $q$.
    This allows us to benefit from the fact that $\OccE_k(P,T)$ can be decomposed to $\cO(d^3)$ arithmetic progressions with difference~$q$
    (see~\cite[Main Theorem 7]{unified}) to decrease the size of the output.

    First, we repeatedly merge any two arithmetic progressions that contain elements that are~$\tau$ positions apart.
    Then, we process each of the $\cO(d)$ relevant residues modulo~$q$ separately.
    Let us fix such a residue~$r$.
    We maintain arrays $B_r$ and $S_r$, each of size $\tau/q$, during a left-to-right scan of the text.
    When some position $\mu$ of $T$ is processed,
    for $i \in \fragmentco{0}{\tau/q}$,
    \begin{itemize}
        \item $B_r\position{i}$ stores a boolean variable indicating whether the successor of $\mu$ in $\mathbb{Z}$ with residue $iq + r$ modulo $\tau$ is in $\OccE_k(P,T)$.
        \item $S_r\position{i}$ stores the successor of $\mu$ in $\mathbb{Z}$ with residue $iq + r$ modulo $\tau$ that is \emph{not} in $\OccE_k(P,T)$ if $B_r\position{i}=1$ and $\infty$ otherwise.
    \end{itemize}
    We maintain balanced binary trees over arrays $B_r$ and $S_r$ so that we can efficiently query for
    the leftmost non-zero element in any subarray of $B_r$
    and the minimum element of $S_r$,
    with an $\cO(\log d)$-time additive overhead per update of $B_r$ and $S_r$.
    All updates to $B_r$ and $S_r$ can be stored in a priority queue after a linear-time preprocessing of the arithmetic progressions in scope,
    prioritized by the value of $\mu$ that triggers them.

    In our scan of $T$, we maintain the minimum position $\mu \in \OccE_k(P,T)$ of residue $r$ modulo $q$ that we have not reported so far.
    Given such a position $\mu$, we query in $\cO(\log d)$ time for the successor $\nu$ of $\mu$ in $\mathbb{Z}$ with residue $iq + r$ that is \emph{not} in $\OccE_k(P,T)$:
    \begin{itemize}
        \item If there is any false entry in $B_r\fragmentoo{\mu-r \mod q}{\tau/q}$, then the sought integer can be retrieved in $\cO(1)$ time given the leftmost such entry;
        \item else, if there is any false entry in $B_r\fragmentco{0}{\mu-r \mod q}$, then the sought integer can be retrieved in $\cO(1)$ time given the leftmost such entry;
        \item else, all entries of $B_r$ are set to true and hence the sought integer corresponds to the smallest integer stored in $S_r$.
    \end{itemize}
    Given $\nu$, we report the arithmetic progression $\{\mu + iq : i \in \fragmentco{0}{(\nu-\mu)/q)}\}$.
    Then, we implicitly continue our scan of $T$ by performing the necessary (precomputed) updates in $B_r$ and $S_r$
    until we reach either the successor of $\nu$ in $\mathbb{Z}$ with residue $r$ modulo $q$ that is in $\OccE_k(P,T)$ or the end of the text.
    The former condition can be checked using array $B_r$ in $\cO(\log d)$ time prior to each update to either $B_r$ or~$S_r$.
    Then, we set said position as $\mu$ and repeat the above process.

    Over all residues, the total time taken is $\cO(d^2 + d^3 \tradeoff \log d) = \cO(d^{3.5})$: we pay $\cO(d\tau/q) = \cO(d^2)$ time to initialize all arrays,
    while all other operations take total time proportional to the product of $\log d$ with the total size of the input and the output arithmetic progressions.
    We can benefit from the upper bound on the number of arithmetic progressions to which $\OccE_k(P,T)$ can be decomposed due to the greedy nature of the
    algorithm that computes them.
    This concludes the proof of this lemma.
\end{proof}

\section{Faster Approximate Pattern Matching in Important Settings}\label{sec:impl}

In this section, we rely on known implementations (see \cite{unified})
of the \modelname model in the static, dynamic, and fully compressed
settings, thereby lifting \cref{thm:edalgII} to these settings.

\subsection{An Algorithm for the Standard Setting}

In the standard setting, we implement a handle to $S=X\fragmentco{\ell}{r}$
as a pointer to $X\in \X$ (which is stored explicitly) along with the indices $\ell$ and $r$.
As was argued in detail in~\cite{unified}, for the implementation of the \modelname model
in the standard setting it suffices to combine well-known results. Namely,
\begin{itemize}
    \item $\extractOpName$, $\accOpName$, and $\lenOpName$ admit trivial implementations;
    \item $\lceOpName$ queries can be efficiently implemented
        by constructing a generalized suffix tree for the strings in the collection~\cite{F97} and preprocessing it
        for $\cO(1)$-time lowest common ancestor queries~\cite{Bender2000}, while for \lcbOpName we can use an analogous
        construction over the reverse strings of the strings in the collection;
    \item $\ipmOpName$ queries can be answered in $\cO(1)$ time by the linear-size data structure of Kociumaka et al.~\cite{IPM,thesis}.
\end{itemize}
The above discussion is summarized in the following statement.

\begin{theorem}[see {\cite[Theorem 7.2]{unified}}]\label{thm:pilis}
    After an $\cO(n)$-time preprocessing of~a collection of~strings of~total length $n$,
    each \modelname operation can be performed in $\Oh(1)$ time.\lipicsEnd
\end{theorem}

Combining~\cref{thm:pilis,thm:edalgII}, we obtain an algorithm for pattern
matching with edits that is faster than the algorithm of Cole and Hariharan \cite{ColeH98},
unless $k/\log(k+1) = \cO(\log m)$ or both algorithms run in $\cO(n)$ time.
\stedalgmain*

Observe that the algorithm encapsulated in \cref{thm:pilis}
is faster than the classical $\cO(nk)$-time
algorithm of Landau and Vishkin~\cite{LandauV89} when $k=\omega(1)$ and $k^{2.5}\sqrt{\log k}=o(m/ \sqrt{\log (m+1)})$.

\begin{remark}
    Our algorithm also applies to the internal setting.
    That is, a string $S$ of~length $n$ can be preprocessed in $\cO(n)$ time, so that
    given fragments $P$~and~$T$ of~$S$, and a threshold $k$, we can compute
    $\OccE_k(P, T)$ in time $\Oh(|T|/|P| \cdot k^{3.5} \sqrt{\log m \log k})$.\lipicsEnd
\end{remark}

\subsection{An Algorithm for the Dynamic Setting}

Next, we consider the dynamic setting. In particular, we consider the dynamic
maintenance of a collection of non-empty persistent strings $\X$ that is initially empty
and undergoes updates specified by the following operations:
\begin{itemize}
    \item $\makestring(U)$: Insert a non-empty string $U$ to $\X$.
    \item $\concat(U,V)$: Insert $UV$ to $\X$, for $U,V \in \X$.
    \item $\splitOp(U,i)$: Insert $U\fragmentco{0}{i}$ and $U\fragmentco{i}{|U|}$ in $\X$,
        for $U \in \X$ and $i \in \fragmentco{0}{|U|}$.
\end{itemize}

Let $N$ denote an upper bound on the total length of~all strings in $\X$
throughout the execution of~the algorithm.
Gawrychowski et al.~\cite{ods} presented a data structure that
efficiently maintains such
a collection and allows for efficient longest common prefix queries.
In~\cite{unified}, it was (a) argued in detail that that the aforementioned data structure
readily supports all \modelname operations other than \ipmOpName queries,
and (b) shown that it can be augmented to efficiently answer \ipmOpName queries;
see also~\cite[Section 4]{panosthesis} for a more direct proof of the latter claim.
The above discussion is formalized in the following statement.

\begin{theorem}[\cite{ods,unified}]\label{thm:pildyn}
    A collection $\X$ of non-empty persistent strings of~total length $N$ can be
    dynamically maintained with operations $\makestring(U)$, $\concat(U,V)$,
    $\splitOp(U,i)$ requiring time $\cO(\log N +|U|)$,
    $\cO(\log N)$ and $\cO(\log N)$, respectively, so that \modelname operations can
    be performed in time $\Oh(\log^2 N)$.\footnote{All running time bounds hold w.h.p.}\lipicsEnd
\end{theorem}

A very recent work~\cite[Section 8]{KK22} provides an alternative deterministic implementation of dynamic strings,
supporting operations $\makestring(U)$, $\concat(U,V)$,
$\splitOp(U,i)$ in $\Oh(|U|\log^{\Oh(1)}\log N)$, $\Oh(\log|UV|\log^{\Oh(1)}\log N)$, and $\Oh(\log|U|\log^{\Oh(1)}\log N)$,
respectively, so that \modelname operations can
be performed in time $\Oh(\log N \log^{\Oh(1)}\log N)$ time.

Combining \cref{thm:pildyn,thm:edalgII}, we obtain the following result for approximate pattern matching
under the edit distance for dynamic strings.
\dynalgmain*

The result encapsulated in \cref{thm:dynalgmain}
should be compared to~\cite[Main Theorem 3]{unified},
which has the same complexity guarantees for processing updates, but answers approximate pattern matching queries under
the edit distance in $\Oh(|T|/|P| \cdot k^{4} \log^2 N)$ time w.h.p.

\subsection{An Algorithm for the Fully Compressed Setting}

\newcommand{\Tr}{\mathsf{PT}}

Next, we focus on the fully compressed setting, where we want to solve approximate pattern
matching when both the text and the pattern are given as a straight-line programs.

We write \(N_{\G}\) for the set of~non-terminals of~a context-free grammar $\G$
and call the elements of~$\mathcal{S}_G \coloneqq N_\G \cup \Sigma$ \emph{symbols}.
Then, a \emph{straight line program} (\emph{SLP}) $\G$ is a context-free grammar
that consists of~a set $N_\G=\{ A_1 , \ldots , A_n \}$ of~non-terminals, such that each $A_i
\in N_\G$ is associated with a unique production rule $A_i \to f_\G(A_i) \in
(\Sigma \cup \{A_j : j < i \})^*$.
For SLPs given as input, we can assume without loss of~generality that each production
rule is of~the form $A \to BC$ for some symbols $B$ and $C$ (that is, the given SLP is in
Chomsky normal form).

Every symbol $A \in \mathcal{S}_\G$ generates a unique string, which we denote by $\gen(A) \in
\Sigma^*$. The string $\gen(A)$ can be obtained from $A$ by repeatedly replacing each
non-terminal by its production. In addition, $A$ is associated
with its \emph{parse tree} $\Tr(A)$ consisting of~a root labeled with $A$ to which zero or
more subtrees are attached:
\begin{itemize}
    \item If $A$ is a terminal, there are no subtrees.
    \item If $A$ is a non-terminal $A\to BC$, then $\Tr(B)$ and $\Tr(C)$ are attached (in this order).
\end{itemize}
\noindent Observe that if we traverse the leaves of~$\Tr(A)$ from left to right, spelling out the
corresponding non-terminals, then we obtain $\gen(A)$.
We say that $\G$ generates $\gen(\G) \coloneqq \gen(A_n)$.
The parse tree $\Tr_\G$ of~$\G$ is then the parse tree of~the starting symbol $A_n \in N_\G$.

First, as also observed in~\cite{unified}, given an SLP $\G$ of size $n$, generating a string $S$ of size $N$, we can efficiently implement the \modelname operations through
dynamic strings. Let us start with an empty collection $\X$ of dynamic strings.
Using $\cO(n)$ $\makestring(a)$ operations, for $a\in \Sigma$, and $\cO(n)$ $\concat$ operations (one for each non-terminal
of $\G$), we can insert $S$ to $\X$ in $\cO(n \log N)$ time w.h.p.
Then, we can perform each \modelname operation in $\cO(\log^2 N)$ time w.h.p., due to~\cref{thm:pildyn},
thus outperforming~\cref{thm:pilgc}.
Next, we outline a deterministic implementation of \modelname operations in the fully compressed setting.

Following~\cite{unified}, the handle of a fragment $S=X\fragmentco{\ell}{r}$
consists of a pointer to the SLP $\G\in \X$ generating $X$ along with the positions $\ell$ and $r$.
This makes operation $\extractOpName$ trivial.
As argued in~\cite{unified}, all remaining \modelname operations admit efficient implementations in the considered setting.
\begin{itemize}
	\item For operation $\lenOpName$, we precompute $|\gen(A)|$ for each non-terminal $A$ using dynamic programming.
	\item For operation $\accOpName$, we use the data structure of Bille et al.~\cite{BilleLRSSW15}.
	\item For operations $\lceOpName$ and \lcbOpName, we use the data structure of I~\cite{I17} that is based on the recompression technique,
		which is due to \cite{talg/Jez15,jacm/Jez16}.
	\item For operation $\ipmOpName$, we use a data structure presented in~\cite{unified,KK20} that is also based on the recompression technique.
\end{itemize}
The above discussion is summarized in the following statement.

\begin{theorem}[see \cite{BilleLRSSW15,I17,unified,KK20}]\label{thm:pilgc}
    Given a collection of~SLPs of~total size $n$, generating strings of~total length $N$,
    each \modelname operation can be performed in $\cO(\log^2 N \log\log N)$ time
    after an $\cO(n \log N)$-time preprocessing.\lipicsEnd
\end{theorem}

We are now ready to present an efficient algorithm for approximate pattern matching under edit distance
in the fully compressed setting.
We choose to state our results using the deterministic implementation of the \modelname model
in this setting, that is, \cref{thm:pilgc}.

We are given an SLP $\G_T$ of~size $n$ with $T \coloneqq \gen(\G_T)$, an SLP $\G_P$ of~size $m$ with
$P \coloneqq \gen(\G_T)$, and a threshold $k$ and are required to compute
the $k$-error occurrences of~$P$ in $T$.

Set $N\coloneqq|T|$, $M \coloneqq |P|$, and $\X \coloneqq \{\G_T, \G_P\}$.
The overall structure of~our algorithm is as follows:
We first preprocess the collection $\X$ in $\cO((n+m) \log N)$ time according to~\cref{thm:pilgc}.
Next, we traverse~$\G_T$ and compute, for every non-terminal~$A$ of~$\G_T$, the approximate
occurrences of~$P$ in $T$ that ``cross'' $A$. We combine
\cref{thm:pilgc} with \cref{thm:edalgII} to compute such occurrences.
Finally, we combine the computed occurrences using dynamic programming.

\newcommand{\cross}{\textsf{cross}}
Formally, for each non-terminal $A \in N_{\G_T}$, with production rule $A \to BC$,
let
\begin{itemize}
	\item $A_{\ell} \coloneqq \gen(B)\fragmentco{\max\{0, |\gen(B)|-M-k+1\}}{|\gen(B)|}$,
	\item $A_{r} \coloneqq \gen(C)\fragmentco{0}{\min\{M+k , |\gen(C)|\}}$, and
	\item $\cross(A) \coloneqq (\OccE_k(P, A_{\ell}A_r) \cap \fragmentco{0}{|A_{\ell}|}) \setminus \OccE_k(P, A_{\ell})$,
\end{itemize}
observing that $\OccE_k(P,\gen(A))$ can then be partitioned to
\begin{itemize}
	\item $\OccE_k(P,\gen(B))$,
	\item $|\gen(B)| - |\gen(A_\ell)| + \cross(A)$, and
	\item $|\gen(B)| + \OccE_k(P,\gen(C))$.
\end{itemize}
Next, observe that, by combining~\cref{thm:pilgc,thm:edalgII},
$\cross(A)$ 
can be computed in time $\cO(k^{3.5} \sqrt{\log M \log k} \log^2 N \log\log N)$
since $A_{\ell}A_r$ and $A_\ell$
are fragments of~$\gen(\G_T)$ of~length at most $2(M+k-1)$.
Now, the size of $|\OccE_k(P, T)|$ can be computed by a straightforward dynamic programming
approach: for a non-terminal $A \in N_{\G_T}$, with production rule $A \to BC$,
the number of $k$-errors occurrences of $P$ in $\gen(A)$ equals
$|\OccE_k(P,\gen(B))| + |\cross(A)| + |\OccE_k(P,\gen(C))|$.
Further, all approximate occurrences can be reported in time proportional to their number by
performing a traversal of~$\Tr_\G$, avoiding to explore subtrees that correspond to
fragments of~$T$ that do not contain $k$-error occurrences.\footnote{Compare
\cite[Main Theorem 2]{unified} for a similar algorithm.}

We obtain the following algorithm for pattern
matching with edits in the fully compressed setting.
\gcedalgmain*

The result encapsulated in \cref{gc_ed_alg_intro}
should be compared to~\cite[Main Theorem 2]{unified},
which summarizes an algorithm that computes $|\OccE_k(P, T)|$
in time $\cO(m\log N + n\, k^4 \log^2 N \log\log N)$ and can then report all elements of
$\OccE_k(P, T)$ within $\cO(|\OccE_k(P, T)|)$ extra time.

\clearpage
\partn{Seaweeds}

\newcommand{\AG}{\mathsf{AG}}
\newcommand{\dist}{\mathsf{dist}}
\newcommand{\R}{\mathbb{R}}
\newcommand{\LIS}{\mathsf{LIS}}
\newcommand{\lt}{\mathsf{lt}}
\newcommand{\br}{\mathsf{br}}
\newcommand{\wid}{\mathsf{width}}
\newcommand{\spn}{\mathsf{band}}
\newcommand{\dom}{\mathsf{dom}}
\renewcommand{\spn}{\mathsf{span}}
\newcommand{\ppr}[1]{{#1}^+}
\newcommand{\ppl}[1]{{}^+{#1}}
\newcommand{\mmr}[1]{{#1}^-}
\newcommand{\mml}[1]{{}^-{#1}}

\section{The Seaweed Monoid of Permutation Matrices}\label{sec:seaweeds}

We start by introducing the terminology behind the definition of permutation matrices
and their seaweed products.
We mostly follow~\cite[Section 3.2]{LIS} and~\cite[Chapters 2 and 3]{abs-0707-3619}, except that we use matrices whose rows and columns are indexed with non-empty integer intervals (finite or infinite).
Consistently with previous work, we think of the plane reflected along the horizontal axis; that is, $(1,1)$ is below and to the right of $(0,0)$.
For an interval $I\sub \Z$, denote
\begin{align*}\ppr{I} &= \{i\in \Z : i\in I \text{ or }i-1\in I\}, & \ppl{I} &= \{i\in \Z : i\in I \text{ or }i+1\in I\},\\
\mmr{I} &= \{i\in \Z : i\in I \text{ and }i+1 \in I\}, & \mml{I} &= \{i\in \Z : i\in I \text{ and }i-1 \in I\}.\end{align*}
Moreover, for $S\sub \Z$, we define $\spn(S)=\{i\in \Z :\exists_{s,s'\in S} s \le i \le s'\}$
as the smallest interval containing $S$.

The $(\min,+)$ product of matrices $A\in \Zz^{I\times K}$ and $B\in \Zz^{K\times J}$
is a matrix $A\odot B\in \Zz^{I\times J}$ with entries defined as follows for $i\in I$ and $j\in J$:
\[(A\odot B)\position{i,j} = \min_{k\in K}\{A\position{i,k}+B\position{k,j}\}.\]
For $A\in \Z^{I\times J}$, the \emph{density matrix} $A^\square \in \Z^{\mmr{I}\times \mmr{J}}$ has entries defined as follows for $i\in \mmr{I}$ and $j\in \mmr{J}$:
\[A^\square \position{i,j} = A\position{i+1,j}+ A\position{i,j+1} - A\position{i,j}-
A\position{i+1,j+1}.\]
A matrix $A\in \Z^{I\times J}$ is a \emph{Monge matrix} if all entries of the density matrix $A^\square$ are non-negative.

A matrix $A\in \{0,1\}^{I\times I}$ is a \emph{permutation matrix} if each row and each column contains exactly one entry equal to $1$.
We define $\spn(A)=\mml{\spn(\{i\in I : A\position{i,i}=0\})}$, and we say that the matrix is \emph{bounded} if $\spn(A)$ is finite.
Note that a permutation matrix $A$ can be represented with a permutation $\sigma:\ppl{\spn(A)}\to \ppl{\spn(A)}$
such that $A\position{i,j}=1$ if and only if $j = \sigma(i)$ (when $i\in \ppl{\spn(A)}$) or $j=i$ (otherwise).

For a bounded permutation matrix $A\in \{0,1\}^{I\times I}$,
the \emph{distribution matrix} $A^\Sigma \in \Zz^{\ppr{I}\times \ppr{I}}$,
has its entries defined as follows for $i,j\in \ppr{I}$ (the entries are finite because $A$ is bounded):
\[A^\Sigma\position{i,j} = \sum_{i'\ge i} \sum_{j'<j} A\position{i',j'}.\]
By constructing the data structure of Chan and P\v{a}tra\c{s}cu \cite{ChanP10} for answering
two-sided orthogonal range counting queries in 2D over the non-zero entries of $A$, one obtains the following.

\begin{fact}[\cite{abs-0707-3619,CKM20,LIS}]\label{fct:ra}
After $\Oh(w\sqrt{\log w})$-time preprocessing of a permutation representing a permutation matrix $A$ with $w=|\spn(A)|$,
any entry of $A^\Sigma$ can be computed in $\Oh({\log w}/{\log \log w})$ time.
\lipicsEnd
\end{fact}

The \emph{seaweed product} of bounded permutation matrices
$A,B\in \{0,1\}^{I\times I}$ is defined as $A \boxdot B \coloneqq  (A^\Sigma \odot B^\Sigma)^\square$.

\begin{theorem}[{Tiskin~\cite{abs-0707-3619,Tis15}}]\label{thm:prod}
	For any two bounded permutation matrices $A,B$, the seaweed product $C \coloneqq  A \boxdot B$
	is a permutation matrix with $\spn(C)\sub \spn(\spn(A)\cup \spn(B))$.
	Moreover, given the permutations representing $A$ and $B$, the permutation representing $C$ can be constructed
	in $\Oh(w \log w)$ time, where $w=|\spn(\spn(A)\cup \spn(B))|$.
    \lipicsEnd
\end{theorem}

\newcommand{\RD}[2]{\ensuremath{D_{#1|_{#2}}}}
\newcommand{\RP}[2]{\ensuremath{P_{#1|_{#2}}}}

For a matrix $A\in \Z^{I\times J}$ and an integer $s \in \Z$,
we define the (diagonal) \emph{shift} of $A$ by $s$ units as a matrix $A\downshift s\in \Z^{(I+s)\times (J+s)}$
such that $(A\downshift s)\position{i+s,j+s}=A\position{i,j}$ for $(i,j)\in I\times J$.

\begin{fact}
The $\downshift$ operation distributes over the seaweed product,
that is, $(A\downshift s)\boxdot (B\downshift s)=(A\boxdot B)\downshift s$
holds for all bounded permutation matrices $A,B$ and integers $s\in \Z$.
\lipicsEnd
\end{fact}

\subsection{Alignment Graphs and Distance Matrices}
\begin{definition}
Given a set $M\sub \Z^2$, we define the \emph{alignment graph} $\AG(M)$ with vertices $\Z^2$ and weighted edges:
\begin{itemize}
    \item $(x,y)\stackrel{1}{\longleftrightarrow} (x+1,y)$ for every $(x,y)\in \Z^2$ (horizontal edges),
    \item $(x,y)\stackrel{1}{\longleftrightarrow} (x,y+1)$ for every $(x,y)\in \Z^2$ (vertical edges),
    \item $(x,y)\stackrel{0}{\longleftrightarrow} (x+1,y+1)$ for every $(x,y)\in \Z^2 \sm M$ (diagonal edges).
\end{itemize}
We denote the underlying distance function on $\Z^2$ by $\dist_M$.
\lipicsEnd
\end{definition}

We introduce an order $\prec$ on $\Z^2$ so that $(x,y)\prec (x',y')$ if and only if $x < x'$ and $y < y'$.
For $S\sub \Z^2$, we denote by $\LIS(S)$ the maximum length of an increasing sequence (with respect to $\prec$) of points in $S$.

\begin{lemma}\label{lem:dist}
Consider a set $M\sub \Z^2$ and two points $p=(x,y),p'=(x',y')$ in $\Z^2$.
\begin{itemize}
    \item If $x\le x'$ and $y\le y'$, then $\dist_M(p,p') = |x'-x|+|y'-y|-2\LIS((\fragmentco{x}{x'}\times \fragmentco{y}{y'})\sm M)$.
    \item If $x \le x'$ and $y\ge y'$, then  $\dist_M(p,p') = |x'-x|+|y'-y|$.
    \item If $x \ge x'$ and $y\le y'$, then  $\dist_M(p,p') = |x'-x|+|y'-y|$.
    \item If $x \ge x'$ and $y\ge y'$, then  $\dist_M(p,p') = |x'-x|+|y'-y|-2\LIS((\fragmentco{x'}{x}\times \fragmentco{y'}{y})\sm M)$.
\end{itemize}
\end{lemma}
\begin{proof}
Let us fix a shortest path between $p$ and $p'$ in $\AG(M)$ and a box $B\coloneqq \fragment{i}{i'}\times \fragment{j}{j'}$
containing all vertices of this path. Observe that $\dist_M(p,p')=\dist_{M,B}(p,p')$,
where $\dist_{M,B}$ denotes the distance function of the subgraph of $\AG(M)$ induced by $B$.
By~\cite[Lemma 12]{LIS}, $\dist_{M,B}(p,p')$ satisfies the claimed formula in all four cases.
\end{proof}

\newcommand{\up}{\mathsf{up}}
\newcommand{\down}{\mathsf{down}}

\begin{definition}
For a \emph{box} $B=\fragment{i}{i'}\times \fragment{j}{j'}$, we define the \emph{left-top} and \emph{bottom-right boundaries} $(\lt^B_d)_{d\in \fragment{i-j'}{i'-j}},(\br^B_d)_{d\in \fragment{i-j'}{i'-j}}$
so that
\[\lt^B_d = \begin{cases}
    (i,i-d) & \text{for }d\in \fragment{i-j'}{i-j},\\
    (d+j, j) & \text{for }d\in \fragment{i-j}{i'-j},\\
\end{cases}\qquad
\br^B_d = \begin{cases}
    (d+j',j') & \text{for }d\in\fragment{i-j'}{i'-j'},\\
    (i',i'-d) & \text{for }d\in \fragment{i'-j'}{i'-j}.\\
\end{cases}
    \tag*{\lipicsEnd}
\]
\end{definition}

\begin{definition}\label{def:dmb}
    For a finite set $M\sub \Z^2$, we say that $B=\fragment{i}{i'}\times \fragment{j}{j'}$ is a \emph{bounding box}
    of $M$ if $M\sub \fragmentco{i}{i'}\times \fragmentco{j}{j'}$.
    We then define an infinite matrix $D_{M,B}\in \Z^{\Z\times \Z}$ as follows:
    \[D_{M,B}\position{a,b} \coloneqq  \begin{cases}
        \dist_M(\lt^B_a,\br^B_b) & \text{if }a,b\in \fragment{i-j'}{i'-j},\\
        |a-b| & \text{otherwise}.
    \end{cases}
    \tag*{\lipicsEnd}
\]
\end{definition}

\newcommand{\hB}{\hat{B}}
\newcommand{\hi}{\hat{\imath}}
\newcommand{\hj}{\hat{\jmath}}
\newcommand{\hM}{\hat{M}}
\begin{lemma}\label{lem:independent}
For every finite set $M\sub \Z^2$, the distance matrix $D_{M,B}$ does not depend on the choice of the bounding box $B$.
\end{lemma}
\begin{proof}
Consider an arbitrary bounding box $B=\fragment{i}{i'}\times \fragment{j}{j'}$
and the minimum bounding box $\hB=\fragment{\hi}{\hi'}\times \fragment{\hj}{\hj'}$;
if $M=\emptyset$, set $\hi=\hi'=i$ and $\hj=\hj'=j$.
We prove $D_{M,B}\position{a,b}=D_{M,\hB}\position{a,b}$ by analyzing several cases:
\begin{itemize}
    \item {\boldmath\bf$a,b\in \fragment{\hi-\hj'}{\hi'-\hj}$.} \cref{lem:dist} implies  $\dist_M(\lt_a^B,\lt_a^{\hB})=0=\dist_M(\br_b^{\hB},\br_b^{B})$.
        Thus, \[
            D_{M,B}\position{a,b}=\dist_M(\lt_a^B,\br_b^B)=\dist_M(\lt_a^{\hB},\br_b^{\hB})=D_{M,\hB}\position{a,b}.
        \]
    \item {\boldmath\bf$a,b\in \fragment{i-j'}{i'-j}$ and $a\notin \fragment{\hi-\hj'}{\hi'-\hj}$.}
    \cref{lem:dist} implies $\dist_M(\lt_a^B,\br_a^B)=0$ and $\dist_M(\br_a^B,\br_b^B)=|a-b|$.
    Thus,  \[
        D_{M,B}\position{a,b}=\dist_M(\lt_a^B,\br_b^B)=\dist_M(\br_a^B,\br_b^B)=|a-b|=D_{M,\hB}\position{a,b}.
    \]
    \item {\boldmath\bf$a,b\in \fragment{i-j'}{i'-j}$ and $b\notin \fragment{\hi-\hj'}{\hi'-\hj}$.}
    \cref{lem:dist} implies $\dist_M(\lt_b^B,\br_b^B)=0$ and $\dist_M(\lt_a^B,\lt_b^B)=|a-b|$.
    Thus, \[
        D_{M,B}\position{a,b}=\dist_M(\lt_a^B,\br_b^B)=\dist_M(\lt_a^B,\lt_b^B)=|a-b|=D_{M,\hB}\position{a,b}.
    \]
    \item {\bf Otherwise,} $D_{M,B}\position{a,b}=|a-b|=D_{M,\hB}\position{a,b}$.\qedhere
\end{itemize}
\end{proof}

\begin{definition}\label{def:dmpm}
For a finite set $M\sub \Z^2$, we define the \emph{distance matrix} $D_M\coloneqq D_{M,B}$
(for an arbitrary bounding box $B$ of $M$) and the \emph{seaweed matrix} $P_M \coloneqq
\onehalf D_M^\square$.
\lipicsEnd
\end{definition}

\begin{fact}\label{obs:shift}
For every finite set $M\sub \Z^2$ and vector $(u,v)\in \Z^2$,
we have $P_{M+(u,v)} = P_M \downshift (u-v)$, where $M+(u,v) = \{(x+u,y+v) : (x,y)\in M\}$. 
\lipicsEnd
\end{fact}

\begin{lemma}\label{lem:perm}
For every finite set $M\sub \Z^2$, the seaweed matrix $P_M$ is a bounded permutation matrix.
Moreover, $D_M\position{a,b}=2P^{\Sigma}_M\position{a,b}+a-b$ holds for all $a,b\in \Z^2$.
\end{lemma}
\begin{proof}
By~\cite[Lemma 19]{LIS}, for every bounding box $B=\fragment{i}{i'}\times \fragment{j}{j'}$,
the matrix $P_M$ restricted to entries $P_M\position{a,b}$ with  $a,b\in \fragmentco{i-j'}{i'-j}$
is a permutation matrix.
Since the bounding box can be chosen arbitrarily large,
the entire matrix $P_M$ is a bounded permutation matrix.
The second claim also follows from~\cite[Definition~18 and Lemma~19]{LIS}.
\end{proof}

\begin{definition}
    We say that two boxes $B=\fragment{i}{i'}\times \fragment{j}{j'}$ and $\hB=\fragment{\hi}{\hi'}\times \fragment{\hj}{\hj'}$  are
    \begin{itemize}
        \item \emph{vertically adjacent} if $i=\hi$, $i'=\hi'$, and $j'=\hj$,
        \item \emph{horizontally adjacent} if $i'=\hi$, $j=\hj$, and $j'=\hj'$.
            \lipicsEnd
    \end{itemize}
\end{definition}

\begin{lemma}\label{lem:adjacent}
    Consider finite sets $M,\hM\sub \Z^2$ with (horizontally or vertically) adjacent
    bounding boxes $B$ and $\hB$, respectively.
    Then, $P_{M\cup \hM} = P_{M}\boxdot P_{\hM}$.
\end{lemma}
\begin{proof}
    It suffices to prove that $D_{M\cup \hM}=D_M\odot D_{\hM}$.
    By symmetry, we assume without loss of generality that the boxes $B,\hB$ are vertically adjacent,
    that is, $B=\fragment{i}{i'}\times \fragment{j}{j'}$ and $\hB = \fragment{i}{i'}\times \fragment{j'}{j''}$
    for some integers $i\le i'$ and $j\le j'\le j''$.
    By \cref{lem:independent}, for every $a,c\in \Z$, we have
    $D_{M\cup \hM}\position{a,c} =\dist_{M\cup \hM}((a+j,j),(c+j'',j''))$,
    since \[
        \fragment{\min\{i,a+j,c+j''\}}{\max\{i',a+j,c+j''\}} \times \fragment{j}{j''}
    \]is a bounding box for $M\cup \hM$.
    In addition, since each path from $(a+j,j)$ to $(c+j'',j'')$ passes through a vertex of the form $(b+j',j')$ for some $b\in \Z$,
    we also have \[
        \dist_{M\cup \hM}((a+j,j),(c+j'',j'')) = \min_{b\in \Z} \{\dist_{M\cup
            \hM}((a+j,j),
        (b+j',j'))\;+\;\dist_{M\cup \hM}((b+j',j'), (c+j'',j''))\}.
    \]
    Now, for every $a,b,c\in \Z$, \cref{lem:independent,lem:dist} yield
    \[D_M\position{a,b}=\dist_M((a+j,j),(b+j',j'))=\dist_{M\cup \hM}((a+j,j),(b+j',j')),\]
    since 
    $\fragmentco{a+j}{b+j'} \times \fragmentco{j}{j'}$ is disjoint from $\hM$, and
    \[D_{\hM}\position{b,c}=\dist_{\hM}((b+j',j'),(c+j'',j''))=\dist_{M\cup \hM}((b+j',j'),(c+j'',j'')),\]
    since 
    $\fragmentco{b+j'}{c+j''} \times \fragmentco{j'}{j''}$ is disjoint from $M$.

    Thus, $D_{M\cup \hM}\position{a,c}=\min_{b\in
    \Z}\{D_M\position{a,b}+D_{\hM}\position{b,c}\}$ holds as claimed.
\end{proof}

\begin{lemma}\label{fct:representable}
Every bounded permutation matrix can be represented as $P_M$ for some finite set $M\sub \Z^2$.
\end{lemma}
\begin{proof}
Tiskin~\cite{abs-0707-3619} showed that every bounded permutation matrix can be represented as the seaweed product of transposition matrices, which can be defined as the unique permutation matrices $G_d:\{0,1\}^{\Z\times \Z}$
satisfying $\spn(G_d)=\{d\}$ for $d\in \Z$.
By \cref{def:dmb,def:dmpm}, for any $y,d\in \Zz$, we have $G_d = P_{\{(y+d,y)\}}$.
By \cref{lem:adjacent}, any product $G_{d_1}\boxdot \cdots \boxdot G_{d_t}$
can be represented as $P_{M}$ for $M=\{(y+d_y,y) : y\in \fragment{1}{t}\}$.
\end{proof}

For a set $M\sub \Z^2$, we define $\spn(M)=\spn(\{x-y : (x,y)\in M\})$.

\begin{lemma}\label{fct:out}
For every finite set $M\sub \Z^2$, we have $\spn(P_M)\sub \spn(M)$.
\end{lemma}
\begin{proof}
Let $a\in \Z \sm \ppl{\spn(M)}$.
This means that $a,a+1 \notin \spn(M)$
and thus $a,a+1\notin \{x-y : (x,y)\in M\}$.
By \cref{lem:dist}, this implies $D_M\position{a,a}=D_M\position{a+1,a+1}=0$ and
$D_M\position{a,a+1}=D_M\position{a+1,a}=1$.
Consequently, \[
    P_M\position{a,a}=\onehalf(D_M\position{a,a+1}+D_M\position{a+1,a}-D_M\position{a,a}-D_M\position{a+1,a+1})=1.
\]
Thus, $\{a\in \Z : P_M\position{a,a}=0\}\sub \ppl{\spn(M)}$.
Given that $\spn(M)$ is an interval, this yields $\spn(P_M)=\mml{\spn(\{a\in \Z :
P_M\position{a,a}=0\})}\sub \spn(M)$.
\end{proof}

\subsection{Restriction of Permutation Matrices}

For a finite set $M\sub \Z^2$ and an integer interval $I$, we denote $M|_I \coloneqq  \{(x,y)\in M : x-y\in I\}$;
observe that $\spn(M|_I) = \spn(\{x-y:  (x,y)\in M|_I\})\sub \spn(I) = I$.

\begin{lemma}\label{lem:rpm}
Consider a finite set $M\sub \Z^2$ and an integer interval $I=\fragmentoo{\ell}{r}$.
The \emph{restricted seaweed matrix} $\RP{M}{I}$ can be characterized as follows
based on the sets \begin{align*}
    L&=\{(a,b)\in \fragment{\ell}{r}^2 : a+b\le\ell+r\text{ and } D_M\position{a,b}\ge a+b-2\ell\},\\
R&=\{(a,b)\in \fragment{\ell}{r}^2 : a+b\ge \ell+r \text{ and }D_M\position{a,b}\ge 2r-a-b\}.\end{align*}
    For every $a,b\in \ppl{I}$, we have $\RP{M}{I}\position{a,b}=1$ if and only if at least one of the following cases holds:
 \begin{enumerate}[(1)]
     \item\label{it:L} $(a+1,b),(a,b+1)\in L$ and $(a+1,b+1)\notin L$, or
     \item\label{it:R} $(a+1,b),(a,b+1)\in R$ and $(a,b)\notin R$, or
     \item\label{it:M} $P_M\position{a,b}=1$, $(a,b)\notin L$, and $(a+1,b+1)\notin R$.
 \end{enumerate}

 Moreover, $(a+1,b),(a,b+1)\notin L$ holds for every $(a,b)\in \fragment{\ell}{r}^2\sm L$
 and $(a-1,b),(a,b-1)\notin R$ holds for every $(a,b)\in \fragment{\ell}{r}^2\sm R$.
\end{lemma}
\begin{proof}
    Let us first prove an auxiliary claim characterizing the \emph{restricted distance matrix} $\RD{M}{I}$.
    \begin{claim}\label{clm:rdm}
    For every $a,b\in \Z$, the restricted distance matrix satisfies
    \[\RD{M}{I}\position{a,b} =\begin{cases}
            \min\{ D_M\position{a,b}, a+b-2\ell,2r-a-b\} & \text{if }a,b\in \fragmentoo{\ell}{r},\\
        |a-b| & \text{otherwise}.\end{cases}\]
    \end{claim}
    \begin{claimproof}
            Consider the smallest bounding box $B$ of $M|_I$.
            The second case follows from the fact that $\spn(\RP{M}{I})\sub \spn(M|_I)\sub
            \fragmentoo{\ell}{r}$.
            As for the first case, note that $\AG(M)$ is a subgraph of $\AG(M|_I)$, so
            $\RD{M}{I}\position{a,b}\le D_M\position{a,b}$.
            Moreover, the triangle inequality yields
            \begin{align*}
                \RD{M}{I}\position{a,b} &\le
                \RD{M}{I}\position{a,\ell}+\RD{M}{I}\position{\ell,b}
                                        =|a-\ell|+|\ell-b|=a+b-2\ell\quad\text{and}\\
                \RD{M}{I}\position{a,b} &\le
                \RD{M}{I}\position{a,r}+\RD{M}{I}\position{r,b} =|a-r|+|r-b|=2r-a-b.
            \end{align*}
            Finally, observe that if a path from $\lt^B_a$ to $\br^B_b$ in $\AG(M|_I)$ does not reach
            any vertex $(x,y)$ with $x-y\notin I$, then path is also present in $\AG(M)$.
            Otherwise, the path must reach a vertex $(x,y)$ with $x-y\in \{\ell,r\}$, so
            its length is at least $\min\{|a-\ell|+|\ell-b|,
            |a-r|+|r-b|\} = \min\{a+b-2\ell,2r-a-b\}$.
    \end{claimproof}
    \Cref{clm:rdm} implies
    \begin{align*}
        L&=\{(a,b)\in \fragment{\ell}{r}^2 : \RD{M}{I}\position{a,b}=a+b-2\ell\}\quad    \text{and}\\
        R&=\{(a,b)\in \fragment{\ell}{r}^2 : \RD{M}{I}\position{a,b}=2r-a-b\}.
    \end{align*}
    By \cref{def:dmb,def:dmpm}, the neighboring entries of $\RD{M}{I}$ differ by at most one.
    Consequently, if $(a,b)\in \fragment{\ell}{r}^2\sm L$,
    then $\RD{M}{I}\position{a,b} < a+b-2\ell$, so
    $\RD{M}{I}\position{a+1,b} < a+b-2\ell+1=(a+1)+b-2\ell$
    and $\RD{M}{I}\position{a,b+1} < a+b-2\ell+1=a+(b+1)-2\ell$.
    Similarly, if $(a,b)\in \fragment{\ell}{r}^2\sm R$,
    then $\RD{M}{I}\position{a,b} < 2r-a-b$, so
    $\RD{M}{I}\position{a-1,b} < 2r-a-b+1 = 2r - (a-1)-b$
    and $\RD{M}{I}\position{a,b-1} < 2r-a-b+1 = 2r-a-(b-1)$.

    Thus, it remains to prove the statement characterizing $\RP{M}{I}$.
    In Case~\eqref{it:L}, we have
    \begin{align*}
        \RP{M}{I}\position{a,b} &= \RD{M}{I}\position{a+1,b}
        +\RD{M}{I}\position{a,b+1}-\RD{M}{I}\position{a,b}
        -\RD{M}{I}\position{a+1,b+1}\\
        &> (a+1+b-2\ell)+(a+b+1-2\ell)-(a+b-2\ell)-(a+1+b+1-2\ell)=0.
    \end{align*}
    In Case~\eqref{it:R}, we have
    \begin{align*}
        \RP{M}{I}\position{a,b} &=
        \RD{M}{I}\position{a+1,b}+\RD{M}{I}\position{a,b+1}
        -\RD{M}{I}\position{a,b}-\RD{M}{I}\position{a+1,b+1}\\
        &> (2r-a-1-b)+(2r-a-b-1)-(2r-a-b)-(2r-a-b-1)=0.
    \end{align*}
    In Case~\eqref{it:M}, we have
    \begin{align*}
        \RP{M}{I}\position{a,b} &= \RD{M}{I}\position{a+1,b}+\RD{M}{I}\position{a,b+1}
        -\RD{M}{I}\position{a,b}-\RD{M}{I}\position{a+1,b+1}\\
                                &= D_M\position{a+1,b}+D_M\position{a,b+1}
                                -D_M\position{a,b}-D_M\position{a+1,b+1}\\
                                &=P_M\position{a,b}=1.
    \end{align*}
    Thus, if at least of the three cases holds, then indeed $\RP{M}{I}\position{a,b}=1$.

    For a proof of the converse implication, suppose that $\RP{M}{I}\position{a,b}=1$ for some $a,b\in \fragmentco{\ell}{r}$.
    Since the neighboring entries of $\RD{M}{I}$ differ by at most one,
    the equality
    \[1=\RD{M}{I}\position{a+1,b}+\RD{M}{I}\position{a,b+1}-\RD{M}{I}\position{a,b}-\RD{M}{I}\position{a+1,b+1}\]
    implies $\RD{M}{I}\position{a+1,b+1}=\RD{M}{I}\position{a,b}$  and
    $\RD{M}{I}\position{a+1,b}=\RD{M}{I}\position{a,b+1} = \RD{M}{I}\position{a,b}+1$.
    If $(a,b)\in L$, then we have $(a+1,b),(a,b+1)\in L$ and $(a+1,b+1)\notin L$, that is, case~\eqref{it:L} holds.
    If $(a+1,b+1)\in R$, then we have $(a+1,b),(a,b+1)\in R$ and $(a,b)\notin R$, that is, case~\eqref{it:R} holds
    Finally, if $(a,b)\notin L$ and $(a+1,b+1)\notin R$, then $(a,b),(a+1,b),(a,b+1),(a+1,b+1)\notin L\cup R$, so
    \begin{align*}
        P_M\position{a,b} &=
        D_M\position{a+1,b}+D_M\position{a,b+1}-D_M\position{a,b}-D_M\position{a+1,b+1}\\
                          &=\RD{M}{I}\position{a+1,b}+\RD{M}{I}\position{a,b+1}-\RD{M}{I}\position{a,b}-\RD{M}{I}\position{a+1,b+1}\\
                          &=\RP{M}{I}\position{a,b}=1,
    \end{align*}
    that is, Case~\eqref{it:M} holds.
    \end{proof}

\begin{corollary}\label{cor:restrict}
Consider two finite sets $M,\hat{M}\sub \Z^2$ with $P_M = P_{\hat{M}}$.
For every integer interval $I\sub \Z$, we have $\RP{M}{I}=\RP{\hat{M}}{I}$.
\end{corollary}
\begin{proof}
If $I$ is a finite interval, then \cref{lem:rpm} provides a complete characterization of $\RP{M}{I}$ in terms of $D_M$ and $P_M$, whereas \cref{lem:perm} provides a complete characterization of $D_M$ in terms of $P_M$.
Consequently, $\RP{M}{I}=\RP{\hat{M}}{I}$ holds in this case.
If $I$ is an infinite interval, then we define $J=\spn(M\cup \hat{M})\cap I$
and observe that $M|_I=M|_J$ as well as $\hat{M}|_I = \hat{M}|_J$, so $\RP{M}{I}=\RP{M}{J}=\RP{\hat{M}}{J}=\RP{\hat{M}}{I}$.
\end{proof}

\cref{cor:restrict} combined with \cref{fct:representable} let us define the restriction operation
on bounded permutation matrices.
\begin{definition}\label{def:restr}
For a bounded permutation matrix $A$ and an interval $I$, we define the \emph{restriction of $A$ to $I$}, denoted $A|_I$, as $\RP{M}{I}$,
where $M\sub \Z^2$ is an arbitrary finite set such that $A=P_M$.
\lipicsEnd
\end{definition}

\begin{lemma}\label{fct:comm}
The following equalities hold for all bounded permutation matrices $A,B$, intervals $I,J\sub \Z$, and shifts $s\in \Z$:
\begin{enumerate}[(a)]
    \item $A|_I\boxdot B|_I = (A\boxdot B)|_I$;\label{it:distr}
    \item $(A|_I)|_J = A|_{I\cap J}$;\label{it:prod}
    \item $(A|_{I}) \downshift s = (A\downshift s)|_{I+s}$.\label{it:sh}
\end{enumerate}
\end{lemma}
\begin{proof}
\eqref{it:distr}
Let $M,\hat{M}\sub \Z^2$ be finite sets such that $A=P_M$ and $B=P_{\hat{M}}$ (these sets exist by \cref{fct:representable}). By \cref{def:dmb,def:dmpm}, we can shift $M$ and $\hat{M}$ along the diagonals
without influencing $P_M$ and $P_{\hat{M}}$. In particular, this lets us assume that $M$ and $\hat{M}$ admit vertically adjacent bounding boxes so that $A\boxdot B = P_{M\cup \hat{M}}$ holds by \cref{lem:adjacent}.
Since $M|_I$ and $\hat{M}_I$ also admit the same vertically adjacent bounding boxes,
\cref{lem:adjacent} also yields $A|_I\boxdot B|_I = P_{M|_I \cup \hat{M}|_I} = P_{(M\cup \hat{M})|_I}
= (A\boxdot B)|_I$.

\eqref{it:prod} Let $M\sub \Z^2$ be a finite set such that $A=P_M$.
Observe that $(M|_I)|_J=
\{(x,y)\in M|_I : x-y\in J\} = \{(x,y)\in M : x-y\in I\cap J\}= M|_{I\cap J}\}$.
By \cref{def:restr}, this implies $(A|_I)_J = A_{I\cap J}$.

\eqref{it:sh} Let $M\sub \Z^2$ be a finite set such that $A=P_M$.
Observe that $(M+(s,0))|_{I+s} = \{(x,y)\in M+(s,0) : x-y\in I+s\} = \{(x+s,y)\in M+(s,0) : x+s-y\in I+s\}
=\{(x+s,y) : (x,y)\in M \text{ and }x-y\in I\} = (M|_I) + (s,0)$.
Consequently, \cref{obs:shift} implies
$(A\downshift s)|_{I+s} = P_{(M+(s,0)) |_{I+s}} = P_{(M|_{I}) + (s,0)} = (A|_{I})\downshift s$.
\end{proof}

Next, we turn the combinatorial characterization of \cref{lem:rpm} into an efficient algorithm for restricting permutation matrices.
\begin{lemma}[{\tt Restrict}($\sigma$, $I$)]\label{lem:restrict}
Given a permutation matrix $A$ (represented by a permutation $\sigma$) and an integer interval $I$, the permutation matrix $A|_I$ (represented by a permutation $\sigma|_I$) can be constructed in $\Oh(|\spn(A)|)$ time.
\end{lemma}
\begin{algorithm}[h]
    \SetKwBlock{Begin}{}{end}
    \SetKwFunction{Restrict}{Restrict}
    \Restrict{$\sigma$, $I$}\Begin{
        $\fragmentoo{\ell}{r} \gets I\cap \spn(A)$\;
        \lForEach{$d\in \fragmentco{\ell}{r}$}{$\sigma|_I(d) \coloneqq  \sigma(d)$}
        $(a,b)\coloneqq (\ell,r)$\;\label{ln:Lstart}
        $v\coloneqq  2|\{d\in \dom(\sigma): d\ge \ell \text{ and } \sigma(d)<r\}|+\ell-r$\;
        \While{$(a,b)\ne (r,\ell)$}{
            \If{$a+1+b \le \ell+r$ \KwSty{and} ($\sigma(a)\ge b$ \KwSty{or} $v-1 \ge a+1+b-2\ell$)}{\label{ln:Ltest}
                \lIf{$\sigma(a)\ge b$}{$v\coloneqq v+1$}
                \lElse{$v\coloneqq v-1$}
                $a\coloneqq a+1$\;
            }\Else{
                $\sigma|_I(a) \coloneqq  b-1$\;\label{ln:Lset}
                \lIf{$\sigma^{-1}_M(b-1)< a$}{$v\coloneqq v+1$}
                \lElse{$v\coloneqq v-1$}
                $b\coloneqq b-1$\;
            }
        }\label{ln:Lend}
        \While{$(a,b)\ne (\ell,r)$}{\label{ln:Rstart}
            \If{$a-1+b \ge \ell+r$ \KwSty{and} ($\sigma(a-1)<b$ \KwSty{or} $v-1 \ge 2r-a+1-b$)}{\label{ln:Rtest}
                \lIf{$\sigma(a-1)<b$}{$v\coloneqq v+1$}
                \lElse{$v\coloneqq v-1$}
                $a\coloneqq a-1$\;
            }\Else{
                $\sigma|_I(a-1) \coloneqq  b$\;\label{ln:Rset}
                \lIf{$\sigma^{-1}_M(b)\ge a$}{$v\coloneqq v+1$}
                \lElse{$v\coloneqq v-1$}
                $b\coloneqq b+1$\;
            }
        }\label{ln:Rend}

        \Return{$\sigma|_I$}\;
    }
    \caption{Restricting the seaweed matrix.}\label{alg:restrict}
\end{algorithm}
\begin{proof}
Write $\fragmentoo{\ell}{r}\coloneqq  I\cap \spn(A)$;
observe that $A|_I = A|_{\fragmentoo{\ell}{r}}$, so we construct $\sigma|_I$ using the characterization of \cref{lem:rpm}.
For each $a\in \fragmentco{\ell}{r}$, we initialize $\sigma|_I(a)\coloneqq \sigma(a)$; this covers case~\eqref{it:M},
so it suffices to assign the values $\sigma|_I(a)$ corresponding to cases~\eqref{it:L} and~\eqref{it:R}.
As for case~\eqref{it:L}, we traverse the \emph{boundary} of the set $L$ (that is, pairs $(a,b)\in L$ with $(a+1,b+1)\notin L$).
A symmetric traversal of the boundary of the set $R$ (that is, pairs $(a,b)\in L$ with $(a-1,b-1)\notin R$) covers case~\eqref{it:R}.
The algorithm implementing this strategy is presented as~\cref{alg:restrict}.

Let us focus on the traversal of the boundary of $L$. Observe that $(r,\ell),(\ell,r)\in L$
due to $D_M\position{r,\ell}\ge |r-\ell| = r+\ell-2\ell$ and $D_M\position{\ell,r}\ge |\ell-r| = \ell+r-2\ell$.
Consequently, by the last part of \cref{lem:rpm}, the boundary of $L$ forms a path from $(r,\ell)$
to $(\ell,r)$ with individual steps going either down or to the left.
In Lines~\ref{ln:Lstart}--\ref{ln:Lend}, we traverse this path while maintaining $v =
D_M\position{a,b}$.
This value is initialized to $D_M\position{\ell,r} = 2P_M^\Sigma\position{\ell,r}+\ell-r$ according to \cref{lem:perm}
and the definition of the distribution matrix.

After visiting $(a,b)\in L$, we proceed to $(a+1,b)$ or $(a,b-1)$
depending on whether $(a+1,b)\in L$ or not.
According to the definition of $L$, this condition is equivalent to $a+1+b\le \ell+r$
and $D_M\position{a+1,b}\ge a+1+b-2\ell$. The first part is verified explicitly,
whereas for the second one, we observe that $D_M\position{a+1,b}-D_M\position{a,b} =
1-2\sum_{b'\in \fragmentco{\ell}{b}}P_M\position{a,b'}$,
that is, $D_M\position{a+1,b}=v+1$ if $\sigma(a)\ge b$ and $D_M\position{a+1,b}=v-1$ otherwise.
In the former case, $v\ge a+b-2\ell$ yields $v+1\ge a+1+b-2\ell$,
whereas in the latter one, we explicitly check whether $v-1 \ge a+1+b-2\ell$.
If this test, implemented in \cref{ln:Ltest}, reveals $(a+1,b)\in L$,
we move to $(a+1,b)$ and update $v$ as described above.
Otherwise, we proceed to $(a,b-1)$.
In order to update $v$, we observe that $D_M\position{a,b-1}- D_M\position{a,b}= 1-2\sum_{a'\in
\fragmentco{a}{r}}P_M\position{a',b-1}$,
that is, $D_M\position{a,b-1}=v+1$ if $\sigma(b-1)< a$ and $D_M\position{a,b-1}=v-1$ otherwise.
In this case, we also know that $\sigma|_I(a)$ is covered by case~\eqref{it:L} of \cref{lem:rpm}.
Thus, we set $\sigma|_I(a)=b-1$ in \cref{ln:Lset}; this value is either assigned correctly (if $(a+1,b-1)\in L$)
or overwritten in the next step (otherwise).

Lines~\ref{ln:Rstart}--\ref{ln:Rend} implement traversal of the boundary of $R$ and cover case~\eqref{it:R} of \cref{lem:rpm}. Observe that $(\ell,r),(r,\ell)\in R$
due to $D_M\position{\ell,r}\ge |\ell-r| = 2r-\ell-r$ and $D_M\position{r,\ell}\ge |r-\ell| = 2r-\ell-r$.
Consequently, by the last part of \cref{lem:rpm},
the boundary of $R$ forms a path from $(\ell,r)$ to $(r,\ell)$ with individual steps going either up or to the right.
We traverse this path maintaining $v=D_M\position{a,b}$.
Given that the traversal of the boundary of $L$ terminated at $(a,b)=(\ell,r)$, we do not need to initialize $v$.

After visiting $(a,b)\in R$, we proceed to $(a-1,b)$ or $(a,b+1)$
depending on whether $(a-1,b)\in R$ or not.
According to the definition of $R$, this condition is equivalent to $a-1+b\ge \ell+r$
and $D_M\position{a-1,b}\ge 2r-a+1-b$. The first part is verified explicitly,
whereas for the second one, we observe that \[
    D_M\position{a-1,b}-D_M\position{a,b} = 2\sum_{b'\in
    \fragmentco{\ell}{b}}P_M\position{a-1,b'}-1,
\]
that is, $D_M\position{a-1,b}=v+1$ if $\sigma(a-1)<b$ and $D_M\position{a-1,b}=v-1$ otherwise.
In the former case, $v\ge 2r-a-b$ yields $v+1\ge 2r-a+1-b$,
whereas in the latter one, we explicitly check whether $v-1 \ge 2r-a+1-b$.
If this test, implemented in \cref{ln:Rtest}, reveals $(a-1,b)\in R$,
we move to $(a-1,b)$ and update $v$ as described above.
Otherwise, we proceed to $(a,b+1)$.
In order to update $v$, we observe that $D_M\position{a,b+1}-D_M\position{a,b} =
2\sum_{a'\in \fragmentco{a}{r}}P_M\position{a',b}-1$,
that is, $D_M\position{a,b+1}=v+1$ if $\sigma(b)\ge a$ and $D_M\position{a,b+1}=v-1$ otherwise.
In this case, we also know that $\sigma|_I(a-1)$ is covered by case~\eqref{it:R} of \cref{lem:rpm}.
Thus, we set $\sigma|_I(a-1)=b$ in \cref{ln:Rset}; this value is either assigned correctly (if $(a-1,b+1)\in R$)
or overwritten in the next step (otherwise).

The running time of the algorithm is $\Oh(|\spn(A)|)$; this includes building the permutation $\sigma^{-1}_M$.
\end{proof}

\subsection{Computing the Seaweed Matrix}

\begin{lemma}\label{lem:psiphi}
    Consider a finite set $M\sub \Z^2$. For every $(x,y)\in \Z$, there are unique thresholds $\phi(x,y),\psi(x,y)\in \Z$ such that the following holds for any bounding box $B\in \fragment{i}{i'}\times \fragment{j}{j'}$ of $M$:
    \begin{align*}
        \dist_M(\lt^B_d,(x+1,y))-\dist_M(\lt^B_d,(x,y))&=\begin{cases} 1 &\text{if } d \le \psi(x,y), \\
        -1 & \text{otherwise.}\end{cases}\\
        \dist_M(\lt^B_d,(x,y+1))-\dist_M(\lt^B_d,(x,y))&=\begin{cases} -1 &\text{if } d \le \phi(x,y), \\
        1 & \text{otherwise.}\end{cases}
    \end{align*}
    Moreover,
    \begin{itemize}
        \item If $(x,y)\in M$, then $\phi(x+1,y)=\min\{\psi(x,y),\phi(x,y)\}$ and $\psi(x,y+1)=\max(\psi(x,y),\phi(x,y))$.
    \item Otherwise, $\phi(x+1,y)=\psi(x,y)$ and $\psi(x,y+1)=\phi(x,y)$.
    \end{itemize}
\end{lemma}
\begin{proof}
Let us first consider points $(x,y)\in \Z^2$ with $x<i$ or $y<j$.
In this case, we have $\dist_M(\lt^B_d,(x,y))=\dist_M(\lt^B_d,(x+1,y+1))=|x-y-d|$, $\dist_M(\lt^B_d,(x+1,y))=|x+1-y-d|$, $\dist_M(\lt^B_d,(x,y+1))=|x-y-1-d|$. Consequently, $\psi(x,y)=\phi(x+1,y)=x-y$ and $\phi(x,y)=\psi(x,y+1)=x-y-1$.

For the remaining points, we proceed by induction on $x+y$.
Suppose that the existence of unique values $\psi(x,y)$ and $\phi(x,y)$ has already been established.
If $(x,y)\notin M$, then $\dist_M((x,y),(x+1,y+1))=0$. We consider four cases depending on whether $d\le \psi(x,y)$ and $d\le \phi(x,y)$.
\begin{itemize}
\item If $d\le \psi(x,y)$ and $d\le \phi(x,y)$, then
    \begin{align*}
        \dist_M(\lt^B_d,(x,y))&=\dist_M(\lt^B_d,(x+1,y+1))=\dist_M(\lt^B_d,(x+1,y))-1\\
                              &=\dist_M(\lt^B_d,(x,y+1))+1.
    \end{align*}
\item If $d\le \psi(x,y)$ and $d> \phi(x,y)$, then
    \begin{align*}
        \dist_M(\lt^B_d,(x,y))&=\dist_M(\lt^B_d,(x+1,y+1))=\dist_M(\lt^B_d,(x+1,y))-1\\
    &=\dist_M(\lt^B_d,(x,y+1))-1.
    \end{align*}
\item If $d> \psi(x,y)$ and $d\le \phi(x,y)$, then
    \begin{align*}
        \dist_M(\lt^B_d,(x,y)) &= \dist_M(\lt^B_d,(x+1,y+1)) =
        \dist_M(\lt^B_d,(x+1,y))+1\\
                               &=\dist_M(\lt^B_d,(x,y+1))+1.
    \end{align*}
\item If $d> \psi(x,y)$ and $d> \phi(x,y)$, then
    \begin{align*}
        \dist_M(\lt^B_d,(x,y)) &= \dist_M(\lt^B_d,(x+1,y+1))=\dist_M(\lt^B_d,(x+1,y))+1\\
                               &=\dist_M(\lt^B_d,(x,y+1))-1.
    \end{align*}
\end{itemize}
Based on this case analysis, we see that $\phi(x+1,y)=\psi(x,y)$
and $\psi(x,y+1)=\phi(x,y)$ are unique and well-defined.

If $(x,y)\in M$, then \cref{lem:dist} implies \[
    \dist_M(\lt^B_d,(x+1,y+1))=1+\min\{\dist_M(\lt^B_d,(x+1,y)),\allowbreak
    \dist_M(\lt^B_d,(x,y+1))\}.
\] We consider the same four cases as previously:
\begin{itemize}
    \item If $d\le \psi(x,y)$ and $d\le \phi(x,y)$, then \[
        \dist_M(\lt^B_d,(x,y))=\dist_M(\lt^B_d,(x+1,y))-1=\dist_M(\lt^B_d,(x,y+1))+1,\]
        so \[\dist_M(\lt^B_d,(x+1,y+1))=\dist_M(\lt^B_d,(x,y)).\]
    \item If $d\le \psi(x,y)$ and $d> \phi(x,y)$, then
        \[\dist_M(\lt^B_d,(x,y))=\dist_M(\lt^B_d,(x+1,y))-1=\dist_M(\lt^B_d,(x,y+1))-1,\] so
        \[\dist_M(\lt^B_d,(x+1,y+1))=\dist_M(\lt^B_d,(x,y))+2.\]
    \item If $d> \psi(x,y)$ and $d\le \phi(x,y)$, then
        \[\dist_M(\lt^B_d,(x,y))=\dist_M(\lt^B_d,(x+1,y))+1=\dist_M(\lt^B_d,(x,y+1))+1,\]
        so \[\dist_M(\lt^B_d,(x+1,y+1))=\dist_M(\lt^B_d,(x,y)).\]
    \item If $d> \psi(x,y)$ and $d> \phi(x,y)$, then
        \[\dist_M(\lt^B_d,(x,y))=\dist_M(\lt^B_d,(x+1,y))+1=\dist_M(\lt^B_d,(x,y+1))-1,\] so
        \[\dist_M(\lt^B_d,(x+1,y+1))=\dist_M(\lt^B_d,(x,y)).\]
\end{itemize}
Based on this case analysis, we conclude that $\phi(x+1,y)=\min\{\psi(x,y),\phi(x,y)\}$ and $\psi(x,y+1)=\max(\psi(x,y),\phi(x,y))$ are unique and well-defined.
\end{proof}

\SetKwFunction{Seaweed}{Seaweed}
\SetKwFunction{suc}{successor}
\SetKwFunction{push}{insert}
\SetKwFunction{swap}{swap}
\SetKwFunction{extractMin}{extractMin}
\newcommand{\Que}{\mathcal{Q}}

\begin{lemma}[{\tt Seaweed}($I$,{\tt successor}($M,\star,\star$))]\label{prop:seaweed}
Consider a finite set $M\sub \Z^2$.
Suppose that we are given an interval $I\supseteq \spn(M)$ and function
$\suc(M,\star,\star)$ that,
given $d\in I$ and $y\in \{-\infty\}\cup \Z$, in $\Oh(1)$ time returns $\min\{y'\ge y : (y'+d,y')\in M\}$,
where $\min\emptyset=\infty$.
Then, a permutation $\sigma_M : \ppl{I} \to \ppl{I}$ representing $P_M$ can be constructed in $\Oh(|I|^2\log\log |I|)$ time.
\end{lemma}
\begin{algorithm}[t]
    \SetKwBlock{Begin}{}{end}
    $\Seaweed(I,\suc(M,\star,\star))$\Begin{
        $\fragmentoo{\ell}{r} \coloneqq  I$\;
        \lForEach{$d\in \fragmentco{\ell}{r}$}{$\sigma(d)\coloneqq d$}\label{ln:inisigma}
        \lForEach{$d\in \fragmentoo{\ell}{r}$}{$\Que.\push(\suc(M,d, -\infty),d)$}\label{ln:inique}
        \While{$\Que$ not empty}{
            $(y,d)\coloneqq  \Que.\extractMin()$\;
            \If{$y\ne \infty$ \KwSty{and} $\sigma(d-1) < \sigma(d)$}{
                $\swap(\sigma(d-1),\sigma(d))$\;
                \lIf{$d-1\in \fragmentoo{\ell}{r}$}{$\Que.\push(\suc(M,d-1, y+1),d-1)$}
                \lIf{$d+1\in \fragmentoo{\ell}{r}$}{$\Que.\push(\suc(M,d+1, y),d+1)$}
            }
        }
        \Return{$\sigma^{-1}$}\;
    }
    \caption{Constructing the seaweed matrix.}\label{alg:seaweed}
\end{algorithm}
\begin{proof}
Write $I\coloneqq\fragmentoo{\ell}{r}$.
Our solution, presented as \cref{alg:seaweed} maintains a permutation $\sigma : \ppl{I} \to \ppl{I}$
and implicitly iterates over all pairs $(y,d)\in \Z\times \fragment{\ell-1}{r}$ in the lexicographic order.
The main invariant is that, after processing $(y,d)$,
the permutation $\sigma$ is as follows:
\[\sigma(d') = \begin{cases}
\psi(y+1+d',y+1)& \text{for }d' \in \fragmentco{\ell}{d},\\
\phi(y+1+d,y) & \text{for }d' = d,\\
\psi(y+d', y) & \text{for }d' \in \fragmentoo{d}{r}.
\end{cases}
\]
Additionally, the algorithm maintains a priority queue $\Que$ that stores events
$(y',d')\in \Z\times \fragmentoo{\ell}{r}$ satisfying the following three invariants
for every $d'\in \fragmentoo{\ell}{r}$:
\begin{enumerate}[(1)]
    \item If $(y',d')\in \Que$ with $y'\ne \infty$, then $(y',d')\succ_{lex} (y,d)$ and $(y'+d',y')\in M$.
    \item If $\sigma(d'-1)<\sigma(d')$ and $d'\le d$, then $(\suc(M,d', y+1),d')\in \Que$.
    \item If $\sigma(d'-1)<\sigma(d')$ and $d'>d$, then $(\suc(M,d', y),d')\in \Que$.
\end{enumerate}

In the initialization phase (Lines~\ref{ln:inisigma}--\ref{ln:inique}), we set
$\sigma(d)=d$ for $d\in \ppl{I}$ and $\Que = \{(\suc(M,d,-\infty),d) : d\in
\fragmentoo{\ell}{r}\}$.
The processing of $(y,d)$ handled implicitly if $(y,d)\notin \Que$.
Otherwise, we extract $(y,d)$ from $\Que$ and check whether $\sigma(d-1)<\sigma(d)$.
If so, we swap $\sigma(d-1)$ with $\sigma(d)$ and insert to $\Que$ pairs $(\suc(M,d-1,
y+1),d-1)$ (provided that $d-1\in \fragmentoo{\ell}{r}$) and $(\suc(M,d-1, y),d+1)$
(provided that $d+1\in \fragmentoo{\ell}{r}$).
After processing all $(y,d)\in \Z\times \fragment{\ell-1}{r}$, the queue $\Que$ may still
contain entries of the form $\{\infty\}\times \fragmentoo{\ell}{r}$, which are extracted one by one until $\Que$ is empty.

Let us first prove that the invariants are satisfied after the initialization.
For this, we select $d=\ell-1$ and an arbitrary value $y$ with $M\sub \Z\times \fragmentco{y}{\infty}$.
Since $\psi(y+d',d')=d'$ holds for all $d'\in \ppl{I}$, the permutation $\sigma$ satisfies the invariant.
As for the invariants regarding $\Que$, note that $\Que = \{(\suc(M,d',y),d') : d'\in
\fragmentoo{\ell}{r}\}$.
In particular, $\Que$ contains all the required pairs. Moreover, if $(y',d')\in \Que$ with $y'\ne \infty$,
then $y'\ge y$ (and thus $(y',d')\succ_{lex} (y,d)$) and $(y'+d',y')\in M$ (by definition
of $\suc(M,\star,\star)$).

Next, consider processing $(y,d)\in \fragment{\ell-1}{r}$.
If $d=\ell-1$, the invariants for $(y,\ell-1)$ are exactly the same as the invariants for $(y-1,r)$.
If $d=\ell$, then the only difference between the invariants for $(y,\ell-1)$ and $(y,\ell)$
is that $\sigma(\ell)$ needs to be updated from $\psi(y+\ell,y)$ to $\phi(y+1+\ell,y)$.
However, due to $(y+\ell,y)\notin M$, \cref{lem:psiphi} guarantees $\phi(y+1+\ell,y)=\psi(y+\ell,y)$.
If $d=r$, then the only difference between the invariants for $(y,r-1)$ and $(y,r)$
is that $\sigma(r-1)$ needs to be updated from $\phi(y+r,y)$ to $\psi(y+r,y+1)$.
However, due to $(y+r,y)\notin M$, \cref{lem:psiphi} guarantees $\psi(y+r,y+1)=\phi(y+r,y)$.
Hence, in each of the above three cases, the void implementation (guaranteed by $\Que \sub
\Z\times \fragmentoo{\ell}{r}$) is correct.

In the main case of $d\in \fragmentoo{\ell}{r}$, the values $\sigma(d-1),\sigma(d)$
should be updated from $\phi(y+d,y),\psi(y+d,y)$
to $\psi(y+d,y+1),\phi(y+d+1,y)$.
If $(y+d,y)\notin M$, then $(y+d,y)\notin \Que$ and our implementation is void.
This is correct because \cref{lem:psiphi} yields $\psi(y+d,y+1)=\phi(y+d,y)$ and
$\phi(y+d+1,y)=\psi(y+d,y)$.
If $(y+d,y)\in M$ and $\sigma(d-1)<\sigma(d)$,
then the invariant on $\Que$ guarantees $(y+d,y)\in \Que$.
Hence, our algorithm swaps $\sigma(d-1)$ with $\sigma(d)$.
This is correct because \cref{lem:psiphi} yields $\psi(y+d,y+1)=\max(\psi(y+d,y),\phi(y+d,y))=\psi(y+d,y)$ and
$\phi(y+d+1,y)=\min\{\psi(y+d,y),\phi(y+d,y)\}=\phi(y+d,y)$.
If  $(y+d,y)\in M$ and $\sigma(d-1)\ge \sigma(d)$,
the queue $\Que$ may contain $(y+d,y)$ or not.
In both scenarios, our algorithm keeps $\sigma(d-1)$ and $\sigma(d)$ intact
This is correct because \cref{lem:psiphi} yields $\psi(y+d,y+1)=\max(\psi(y+d,y),\phi(y+d,y))=\phi(y+d,y)$ and
$\phi(y+d+1,y)=\min\{\psi(y+d,y),\phi(y+d,y)\}=\psi(y+d,y)$.

Next, we shall prove that the invariants regarding $\Que$ remain satisfied.
Since $(y,d)$ is removed from $\Que$, all the remaining elements $(y',d')$ satisfy $(y',d')\succ_{lex} (y,d)$
and $(y',d')\in M$.
By definition of $\suc(M,\star,\star)$, this is also true for the newly inserted elements, if any.
Next, consider $d'\in \fragmentoo{\ell}{r}$ with $\sigma(d'-1)<\sigma(d)$.
If $d'\in \fragmentoo{\ell}{d-1}$, the entries $\sigma(d'-1)$ and $\sigma(d')$
were kept intact, so $(\suc(M,d',y+1),d')$ is still guaranteed to be contained in $\Que$.
The same is true for $d'=d-1$ if we did not swap $\sigma(d-1)$ with $\sigma(d)$.
If we did, however, then $(\suc(M,d-1,y+1),d-1)$ is contained in $\Que$ because it was inserted explicitly.
Symmetrically, for $d'\in \fragmentoo{d+1}{r}$, the entries $\sigma(d'-1)$ and $\sigma(d')$
were kept intact, so $(\suc(M,d',y),d')$ is still guaranteed to be contained in $\Que$.
The same is true for $d'=d+1$ if we did not swap $\sigma(d-1)$ with $\sigma(d)$.
If we did, however, then $(\suc(M,d+1,y),d)$ is contained in $\Que$ because it was inserted explicitly.
It remains to consider $d'=d$.
Due to $\psi(y+d,y+1)<\phi(y+d,y)$, \cref{lem:psiphi} guarantees $(y,y+d)\notin M$.
Consequently, $(\suc(M,d,y),d)=(\suc(M,d,y+1),d)$ is still guaranteed to be contained in $\Que$.

We conclude the correctness analysis by considering $d=r$ and an arbitrary value $y$ with $M\cap \Z\times \fragmentoc{-\infty}{y}$
The invariants show that $\sigma(d')=\psi(y+d',d')$ for $d'\in \ppl{I}$
and that $\Que$ does not contain any entry $(y',d')$ with $y'\ne \infty$ (so no further iterations alter $\sigma$).
Recall that  $P_M\position{a,b}=1$ implies $D_M\position{a,b}=D_M\position{a+1,b+1}$ and
$D_M\position{a,b+1}=D_M\position{a+1,b}=D_M\position{a,b}+1$,
that is, $D_M\position{a,b+1}-D_M\position{a,b}=1$ and
$D_M\position{a+1,b+1}-D_M\position{a+1,b}=-1$.
Using \cref{def:dmb,def:dmpm,lem:psiphi}, we conclude that $a \le \sigma(b)$ and $a+1>\sigma(b)$,
that is, that $a=\sigma(b)$.
Since $P_M$ is a permutation matrix, this guarantees that $\sigma^{-1}$ is the permutation representing $P_M$.

As for the running time,
we observe that each swap of subsequent entries of $\sigma$ increases the number of inversions in $\sigma$,
so the total number of swaps is $\Oh(|I|^2)$.
Consequently, the total number of operations on $\Que$ is $\Oh(|I|+|I|^2)=\Oh(|I|^2)$.
Using the state-of-the-art priority queries with integer keys~\cite{Han2004,Thorup2007}, we can achieve $\Oh(\log \log |I|)$ time per operation,
for a total running time of $\Oh(|I|^2 \log \log |I|)$.
\end{proof}

\newcommand{\Xr}{\hat{X}}
\newcommand{\Yr}{\hat{Y}}

\section{Applications of Seaweeds}\label{sec:DPM_sol}

For two strings $X,Y\in \Sigma^*$, we define $M(X,Y)=\{(x,y)\in \fragmentco{0}{|X|}\times
\fragmentco{0}{|Y|} : X\position{x}\ne Y\position{y}\}$.
Moreover, we denote $D_{X,Y}=D_{M(X,Y)}$ and $P_{X,Y}=P_{M(X,Y)}$.
\cref{lem:dist,def:dmb} yields the following characterization:
\begin{fact}\label{obs:dd}
For $X,Y\in \Sigma^*$ and $0\le \ell \le r\le |X|$, we have $\DD(X\fragmentco{\ell}{r},Y)
= D_{X,Y}\position{\ell,r-|Y|}$.
\lipicsEnd
\end{fact}

\begin{lemma}\label{fct:concat}
All strings $X,Y,Z\in \Sigma^*$
satisfy $P_{X,YZ}=P_{X,Y}\boxdot (P_{X,Z}\downshift (-|Y|))$ and $P_{XY,Z}=P_{X,Z} \boxdot (P_{Y,Z}\downshift |X|)$.
\end{lemma}
\begin{proof}
    Observe that $M(X,YZ)=M(X,Y)\cup (M(X,Z)+(0,|Y|))$
    and the two sets of the right-hand side have vertically adjacent bounding boxes.
    Hence, \cref{obs:shift,lem:adjacent} imply $P_{X,YZ}=P_{X,Y}\boxdot P_{M(X,Z)+(0,|Y|)}=P_{X,Y}\boxdot (P_{X,Z}\downshift (-|Y|))$.
    Similarly, $M(XY,Z)=M(X,Z)\cup (M(Y,Z)+(|X|,0))$,
    and the two sets of the right-hand side have horizontally adjacent bounding boxes.
    Hence, \cref{obs:shift,lem:adjacent} imply $P_{XY,Z}=P_{X,Z}\boxdot P_{M(Y,Z)+(|X|,0|)}=P_{X,Z} \boxdot (P_{Y,Z}\downshift |X|)$.
\end{proof}

The principle behind \cref{fct:concat} was used in the algorithm of~\cite{CKM20}
to maintain $\DD(X,Y)$ subject to edit operations in $X,Y$. The algorithm of~\cite{CKM20}
actually maintains the seaweed matrix $P_{X,Y}$, so we get the following result:
\begin{fact}[\cite{CKM20}]\label{prp:dynamic}
    There exists a dynamic algorithm that maintains $P_{X,Y}$ subject to character insertions and deletions in $X,Y\in \Sigma^*$. The initialization costs $\Oh(|X|\cdot |Y|)$ time
    and the updates cost $\Oh((|X|+|Y|)\log^2(|X|+|Y|))$ time.
\lipicsEnd
\end{fact}

\begin{lemma}\label{lem:whole}
Consider families $\X,\Y\sub \Sigma^{\le n}$.
The seaweed matrices $P_{X,Y}$ for $X\in \X$ and $Y\in \Y$
can all be constructed in $\Oh(n^2 + (1+d_X)(1+d_Y)n\log^2 n)$ time,
where $d_X = \ed(\X)$, and $d_Y = \ed(\Y)$.
\end{lemma}
\begin{proof}
We iterate over $(X,Y)\in \X\times \Y$ while maintaining $P_{X,Y}$ via \cref{prp:dynamic}.
The initialization costs $\Oh(n^2)$ time whereas
updating $(X,Y)$ to $(X',Y')$ costs $\Oh((\ed(X,X')+\ed(Y,Y'))n\log^2 n)$ time.
To bound the total contribution of the update costs, let $\Xr$ and $\Yr$ be such that $\ed(\X)=\sum_{X\in \X}\ed(X,\Xr)$
and $\ed(\Y)=\sum_{Y\in \Y}\ed(Y,\Yr)$.
Then, $\ed(X,X')+\ed(Y,Y')\le \ed(X,\Xr)+\ed(X',\Xr) + \ed(Y,\Yr)+\ed(Y',\Yr)$.
Consequently, the total update time is $\Oh((|\Y|\cdot \ed(\X) + |\X|\cdot \ed(\Y))n\log^2 n)
= \Oh((1+d_X)(1+d_Y)n\log^2 n)$ because $|\Y|\le 1+\ed(\Y)$ and $|\X| \le 1+\ed(\X)$.
\end{proof}

\begin{lemma}\label{lem:prefsum}
There exists a dynamic algorithm that maintains a matrix $\boxdot_{i=1}^t (A_i \downshift \sum_{j=1}^i \delta_i)$
subject to insertions and deletions of pairs $(\delta_i,A_i)$
consisting of an integer $\delta_i$ and a bounded permutation matrix $A_i$.
The update cost is $\Oh(w \log w \log t)$ time and the initialization
costs $\Oh(wt\log w)$ time, where $w=|\spn(\bigcup_{i=1}^t\spn(A_i\downshift (\sum_{j=1}^{i} \delta_j)))|$.
\end{lemma}
\begin{proof}
We maintain a balanced binary tree with $t$ leaves representing $\fragment{1}{t}$.
A node $\nu$ whose subtree contains leaves representing $\fragment{\ell}{r}$
maintains $A_{\nu} \coloneqq  \boxdot_{i=\ell}^r (A_i \downshift \sum_{j=\ell}^i \delta_i)$
as well as $s_\nu = \sum_{i=\ell}^r \delta_i$.
If $\nu$ is the $i$th leaf, then $A_{\nu}=A_i\downshift \delta_i$ and $s_\nu = \delta_i$.
If $\nu$ is an internal node with children $\nu_L,\nu_R$,
then $A_{\nu} = A_{\nu_L}\boxdot (A_{\nu_R}\downshift s_{\nu_L})$ and $s_{\nu}=s_{\nu_L}+s_{\nu_R}$;
this matrix can be constructed in $\Oh(w\log w)$ time using \cref{thm:prod}.
Any insertion or deletion requires updating $\Oh(\log t)$ nodes, whereas the initialization requires computing $A_\nu$
for all $\Oh(t)$ nodes. Thus, the running times are $\Oh(w \log w \log t)$ and $\Oh(wt\log w)$, respectively.
\end{proof}

\begin{lemma}\label{lem:prod}
Consider strings $X,Y$, an integer interval $I=\fragment{\ell}{r}$, and
a decomposition $Y=Y_1\cdots Y_t$ such that $Y_i=Y\fragmentco{y_i}{y'_i}$ for $i\in \fragment{1}{t}$.
Moreover, for $i\in \fragment{1}{t}$, let
$X_i = X\fragmentco{x_i}{x'_i}$, where $x_i \le \max(0, y_i+\ell)$ and $x'_i \ge \min\{|X|,
y'_i+r\}$,
and $I_i \supseteq  I - x_i + y_i$ is an integer interval.
Then, $P_{X,Y}|_I = \left(\boxdot_{i=1}^t \left(\left(P_{X_i,Y_i} |_{I_i}\right)\downshift (x_i-y_i)\right)\right)|_I$.
\end{lemma}
\begin{proof} We start with an auxiliary claim.
    \begin{claim}\label{clm:prod}
        For every $i\in \fragment{1}{t}$, we have $P_{X,Y_i}|_{I+y_i}=\left((P_{X_i,Y_i} |_{I_i})|_{I-x_i+y_i}\right)\downshift x_i$.
        \end{claim}
        \begin{claimproof}
            Let $M_i = M(X,Y_i)$ so that $P_{M_i}=P_{X,Y_i}$.
            Observe that $M_i |_{I+y_i} = \{(x,y)\in M_i : x-y\in I+y_i\}
            =\{(x,y)\in M_i : x-y-y_i\in I\}
            = \{(x,y)\in M_i: x-y-y_i\in \fragment{\ell}{r}\}=\{(x,y)\in M_i: x\in \fragment{y+y_i+\ell}{y+y_i+r}\}
            \sub \fragmentco{0}{|X_i|}\times \fragmentco{y_i+\ell}{y'_i+r}$.
            At the same time, $M_i |_{I+y_i}\sub \fragmentco{0}{|X|}\times \fragmentco{0}{|Y_i|}$,
            so $M_i |_{I+y_i}\sub \fragmentco{x_i}{x'_i}\times \fragmentco{0}{|Y_i|}$.
            In particular, $M_i|_{I+y_i} = (M_i\cap (\fragmentco{x_i}{x'_i}\times \fragmentco{0}{|Y_i|}))|_{I+y_i}
            = (M(X_i,Y_i)-(0,y_i))|_{I+y_i}$,
            so $P_{X,Y_i}|_{I+y_i} = (P_{X_i,Y_i}\downshift x_i)|_{I+y_i}
            = (P_{X_i,Y_i}|_{I-x_i+y_i})\downshift x_i$ holds by \cref{fct:comm}\eqref{it:sh}.
            Moreover, due to  \cref{fct:comm}\eqref{it:prod} and $I_i \supseteq I-x_i+y_i$,
            the equality $P_{X,Y_i}|_{I+y_i}=(P_{X_i,Y_i}|_{I-x_i+y_i})\downshift x_i=((P_{X_i,Y_i}|_{I_i})|_{I-x_i+y_i})\downshift x_i$ holds as claimed.
        \end{claimproof}
    \begin{align*}
        P_{X,Y} |_I
        &= \left(\boxdot_{i=1}^t (P_{X,Y_i}\downshift (-y_i))\right)|_I &\text{\cref{fct:concat}}\\
        &= \left(\left(\boxdot_{i=1}^t (P_{X,Y_i}\downshift (-y_i))\right)|_I\right)|_I  &\text{\cref{fct:comm}\eqref{it:prod}}\\
        &=\left(\boxdot_{i=1}^t \left((P_{X,Y_i}\downshift (-y_i))|_I\right)\right)|_I &\text{\cref{fct:comm}\eqref{it:distr}}\\
        &=\left(\boxdot_{i=1}^t ((P_{X,Y_i}|_{I+y_i})\downshift (-y_i))\right)|_I &\text{\cref{fct:comm}\eqref{it:sh}}\\
        &=\left(\boxdot_{i=1}^t \left(\left(\left(P_{X_i,Y_i} |_{I_i}\right)|_{I-x_i+y_i}\right)\downshift (x_i-y_i)\right)\right)|_I &\text{\cref{clm:prod}}\\
        &= \left(\boxdot_{i=1}^t \left(\left(\left(P_{X_i,Y_i} |_{I_i}\right)\downshift (x_i-y_i)\right)|_I\right)\right)|_I &\text{\cref{fct:comm}\eqref{it:sh}}\\
        &= \left(\left(\boxdot_{i=1}^t \left(\left(P_{X_i,Y_i} |_{I_i}\right)\downshift (x_i-y_i)\right)\right)|_I\right)|_I &\text{\cref{fct:comm}\eqref{it:distr}}\\
        &=  \left(\boxdot_{i=1}^t \left(\left(P_{X_i,Y_i} |_{I_i}\right)\downshift (x_i-y_i)\right)\right)|_I. &\text{\cref{fct:comm}\eqref{it:prod}}
    \end{align*}
    This completes the proof.
\end{proof}

\newcommand{\xr}{\hat{x}}
\newcommand{\yr}{\hat{y}}

\begin{corollary}\label{cor:prod}
Consider strings $X,\Xr,Y,\Yr$, alignments $\A_X: \Xr \onto X$ and $\A_Y : \Yr \onto Y$,
 an integer $k\ge \ed^{\A_X}(X,\Xr)+\ed^{\A_Y}(Y,\Yr)$, an interval $\fragment{\ell}{r}$,
and a decomposition $\Yr=\Yr_1\cdots \Yr_{t}$ such that $\Yr_i = \Yr\fragmentco{\yr_i}{\yr'_i}$ for $i\in \fragment{1}{t}$.
For $i\in \fragment{1}{t}$, let \begin{itemize}
\item $\Xr_i = \Xr\fragmentco{\xr_i}{\xr'_i}$, where $\xr_i = \max(0,\yr_i+\ell -k)$ and
$\xr'_i = \min\{|\Xr|,\yr'_i+r+k\}$.
\item $Y_i = Y\fragmentco{y_i}{y'_i}=\A_Y(\Yr_i)$.
\item $X_i = X\fragmentco{x_i}{x'_i}=\A_X(\Xr_i)$.
\item $I_i = \fragment{\ell-k-\xr_i+\yr_i}{r+k-\xr_i+\yr_i}$
\end{itemize}
Then, $P_{X,Y}|_I = \left(\boxdot_{i=1}^t \left(\left(P_{X_i,Y_i} |_{I_i}\right)\downshift (x_i-y_i)\right)\right)|_I$.
\end{corollary}
\begin{proof}
By \cref{fct:ali}, $Y=Y_1\cdots Y_t$ is a decomposition of $Y$.
Due to \cref{lem:prod}, it suffices to prove that $x_i \le \max(0, y_i+\ell)$,
$x'_i \ge \min\{|X|,y'_i+r\}$, and $I_i \supseteq I-x_i+y_i$.
If $\xr_i = 0$, then $x_i=0$.
Otherwise, $x_i \le \xr_i + \ed^{\A_X}(X,\Xr) = \yr_i+\ell-k+ \ed^{\A_X}(X,\Xr) \le y_i + \ed^{\A_Y}(Y,\Yr)+\ell-k+\ed^{\A_X}(X,\Xr)\le y_i+\ell$ holds as claimed.
If $\xr'_i = |\Xr|$, then $x'_i = |X|$.
Otherwise, $x'_i \ge \xr'_i - \ed^{\A_X}(X,\Xr) = \yr'_i+r+k-\ed^{\A_X}(X,\Xr) \le y'_i - \ed^{\A_Y}(Y,\Yr)+r+k-\ed^{\A_X}(X,\Xr)\ge y'_i +r$ holds as claimed.
Finally, note that $\ell-k-\xr_i + \yr_i \le \ell-k-x_i+\ed^{\A_X}(X,\Xr)+y_i+\ed^{\A_Y}(Y,\Yr)\le \ell-x_i+y_i$
and $r+k-\xr_i+\yr_i \ge r+k-x_i - \ed^{\A_X}(X,\Xr)+y_i - \ed^{\A_Y}(Y,\Yr) \ge r-x_i+y_i$,
so $I_i \supseteq I -x_i + y_i$ holds as claimed.
\end{proof}

\newcommand{\bS}{\bar{S}}
\begin{lemma}\label{fct:center}
Given a string family $\S$, a string $\bS\in \S$ such that $\sum_{S\in S}\ed(S,\bS)\le 2\ed(\S)$
can be constructed in $\Oh(1+\ed(\S)^3)$ time in the \modelname model.
\end{lemma}
\begin{proof}
The algorithm computes $\ed(S,S')$ for all pairs of distinct strings $S,S'\in \S$
and returns the string $\bS\in \S$ minimizing $\sum_{S\in \S} \ed(S,\bS)$.

As for correctness, consider $\Sr\in \Sigma^*$ such that $\ed(\S)=\sum_{S\in \S}\ed(S,\Sr)$.
Then, there exists $\bS'\in \S$ such that $|\S|\ed(\bS',\Sr)\le \sum_{S\in \S}\ed(S,\Sr) = \ed(\S)$.
Consequently, $\sum_{S\in \S} \ed(S,\bS') \le \sum_{S\in \S} (\ed(S,\Sr)+\ed(\bS',\Sr))
\le \ed(\S) + |\S|\ed(\bS',\Sr) \le 2\ed(\S)$.
Now, by definition of $\bS$,
we have
$\sum_{S\in \S} \ed(S,\bS) \le \sum_{S\in \S} \ed(S,\bS') \le 2\ed(\S)$.

As for the running time note that computing $\ed(S,S')$ costs $\Oh(\ed(S,S')^2)
= \Oh((\ed(S,\Sr)+\ed(S',\Sr))^2)$ time.
 Across all pairs $(S,S')$ this is bounded by
 $\Oh(|\S|\sum_{S\in \S} \ed(S,\Sr)^2 +\sum_{S,S'\in \S}\ed(S,\Sr)\ed(S',\Sr))
 =\Oh(|\S|\ed(\S)^2)=\Oh(1+\ed(\S)^3)$.
\end{proof}

\begin{lemma}\label{lem:prepr}
    Consider families $\X,\Y\sub \Sigma^{\le n}$,
    a positive integer $d\ge \ed(\X)+ \ed(\Y)$,
    as well as an integer interval $I$.
    The matrices $P_{X,Y}|_I$ for $X\in \X$ and $Y\in \Y$
    can all be constructed in $\Oh((d^3 + d|I|^2)\log^2 (d+|I|))$ time in the \modelname model.
\end{lemma}
\begin{proof}
First, we compute $\Xr\in \X$ and $\Yr\in \Y$ such that $\ed(\X,\Xr)\le 2\ed(\X)$
and $\ed(\Y,\Yr)\le 2\ed(\Y)$ (\cref{fct:center})
as well as optimal alignments $\A_X : \Xr\onto X$
and $\A_Y: \Yr \onto Y$ for all $X\in \X$ and $Y\in \Y$ (represented by the underlying breakpoints).

Next, we iteratively construct of a partition $\Yr=\Yr_1\cdots \Yr_t$ of $\Yr$,
where $\Yr_i = \Yr\fragmentco{\yr_i}{\yr'_i}$ for $i\in \fragment{1}{t}$.
In the $i$-th iteration, we set $\yr_i = 0$ (if $i=1$) or $\yr_i=\yr'_{i-1}$ (otherwise).
Next, we define $\yr'_i$. If $\yr_i > |\Yr|-(|I|+4d)$, we set $\yr'_i = |\Yr|$.
Otherwise, we define $\Yr_i = \Yr\fragmentco{\yr_i}{\yr'_i}$ to be the longest possible fragment of $\Yr_i$
starting position $\yr_i$ such that either (a) $|\Yr_i|=|I|+4d$ or (b) $Y_i \coloneqq  \A_Y(\Yr_i)$ matches $\Yr_i$
for all $Y\in \Y$ and $X_i \coloneqq  \A_X(\Xr_i)$ matches $\Xr_i$ for all $X\in \X$, where $\Xr_i$ is defined as in \cref{cor:prod} for $k=2d$.
This construction partitions $\Yr$ into \emph{perfect} fragments satisfying condition (b)
and the remaining \emph{imperfect} fragments.

Let us denote $\mathcal{U}_i = \{\A_X(\Xr_i) : X\in \X\}$ and $\Y_i = \{A_Y(\Yr_i) : Y\in \Y\}$.
Our goal is to compute $P_{X_i,Y_i}|_{I_i}$ (with $I_i$ defined in \cref{cor:prod} for $k=2d$)
for all $(X_i,Y_i)\in \mathcal{U}_i \times \Y_i$.
If $\Yr_i$ is a perfect fragment, then $|\mathcal{U}_i|=|\Y_i|=1$, and we apply \cref{prop:seaweed}.
Otherwise, we use \cref{lem:whole} to first construct $P_{X_i,Y_i}$, and then we derive $P_{X_i,Y_i}|_{I_i}$ using \cref{lem:restrict}.

Finally, we initialize the data structure of \cref{lem:prefsum}
with $A_i = P_{\Xr_i,\Yr_i}|_{I_i}$ and $\delta_1=0$ and $\delta_i= \xr_i-\yr_i-\xr_{i-1}+\yr_{i-1}$ for $i\in \fragment{2}{t}$.
For each $(X,Y)\in \X\times \Y$, we substitute $A_i\coloneqq P_{X_i,Y_i}|_{I_i}$ whenever $(X_i,Y_i)\ne (\Xr_i,\Yr_i)$
and $\delta_i\coloneqq x_i-y_i-x_{i-1}+y_{i-1}$ whenever $x_i-y_i-x_{i-1}+y_{i-1}\ne \xr_i-\yr_i-\xr_{i-1}+\yr_{i-1}$.
Then, we retrieve the matrix $A\coloneqq \boxdot_{i=1}^t (A_i \downshift \sum_{j=1}^i \delta_i)$ from the dynamic algorithm of \cref{cor:prod} and return $A|_I$, computed using \cref{lem:restrict}.
Finally, we undo all the substitutions applied for $(X,Y)$.

Correctness of the algorithm follows directly from \cref{cor:prod}.
It remains to analyze the running time.
First, we note that the applications of \cref{fct:center} take $\Oh(d^3)$ time.
Constructing the alignment $\A_X:\Xr \onto X$ costs $\Oh(1+\ed(X,\Xr)^2)$ time,
which sums up to $\Oh(d^2)$ across all $X\in \X$.
Symmetrically, all the alignments $\A_Y:\Yr \onto Y$
are built in $\Oh(d^2)$.

In order to efficiently implement the partitioning of $\Xr$,
we construct the sets
\begin{align*}
    \brkp_{\X} &= \{\xr\in \fragment{0}{|\Xr|} : (\xr,x)\in \brkp_{\A_X}\text{ for some
    }X\in \X\text{ and }x\in \fragment{0}{|X|}\}\quad\text{and}\\
    \brkp_{\Y} &= \{\yr\in \fragment{0}{|\Yr|} : (\yr,y)\in \brkp_{\A_Y}\text{ for some
    }Y\in \Y\text{ and }y\in \fragment{0}{|Y|}\}.
\end{align*}
For each $X\in \X$, the contribution of $X$ to $\brkp_{\X}$ can be constructed in $\Oh(1+\ed(X,\Xr))$ time by
scanning the (breakpoints behind) $\A_X$.
Consequently, the set  $\brkp_{\X}$ is of size $\Oh(d)$ and can be constructed in $\Oh(d\log d)$ time (this includes sorting and removing duplicates).
A symmetric argument shows that $\brkp_{\Y}$ is of size $\Oh(d)$ and can be constructed in $\Oh(d \log d)$ time.

Recall that for $\Yr_i = \Yr\fragmentco{\yr_i}{\yr'_i}$, we have $\Xr_i = \Xr\fragmentco{\xr_i}{\xr'_i}$,
where $\xr_i = \max(0, \yr_i+\ell-k)$ and $\xr'_i = \min\{|\Xr|,\yr'_i+r+k\}$.
Thus, $\Yr\fragmentco{\yr_i}{\yr'_i}$ is a perfect fragment as long as $\fragmentco{\yr_i+\ell-k}{\yr'_i+r+k}\cap \brkp_{\X}=\emptyset$
and $\fragmentco{\yr_i}{\yr'_i}\cap \brkp_{\Y}=\emptyset$.
In $\Oh(\log d)$ time (by binary search over $\brkp_{\X}$ and $\brkp_{\Y}$), we can verify this condition for $\yr'_i = \yr_i+|I|+4d$ and, if satisfied, determine the maximum possible $\yr'_i$.
Thus, the partition is constructed in $\Oh((d+t)\log d)$ time.

Our next goal is to prove that $t=\Oh(d)$.
For this, we note that we cannot simultaneously have $\fragmentco{\xr_i}{\xr'_i}\cap \brkp_{\X}=\emptyset=\fragmentco{\xr_{i+1}}{\xr'_{i+1}}$ and $\fragmentco{\yr_i}{\yr'_i}\cap \brkp_{\Y}=\emptyset=\fragmentco{\yr_{i+1}}{\yr'_{i+1}}\cap \brkp_{\Y}$ for any $i\in \fragmentco{1}{t}$.
At the same time, $\fragmentco{\xr_i}{\xr'_i}\cap \fragmentco{\xr_{i+2}}{\xr'_{i+2}}=\emptyset = \fragmentco{\yr_i}{\yr'_i}\cap \fragmentco{\yr_{i+2}}{\yr'_{i+2}}$ holds for $i\in \fragment{1}{t-2}$ (because $\xr_{i+2}\ge  \yr'_{i+1} +\ell -k \ge \yr_{i+1}+|I|+4d +\ell - k= \xr_{i+1}+(r-\ell)+2k +r - k = \xr'_{i}+r+k \le \yr'_i$).
Consequently, the number $t'$ of imperfect fragments is at most $2|\brkp_{\Y}|+2|\brkp_{\X}|=\Oh(d)$
and the partition size $t$ satisfies $t\le 2t'+1=\Oh(d)$.
Moreover, each $X\in \X$ satisfies  $\sum_{i=1}^t \ed(\Xr_i,X_i) \le 2\ed(\Xr,X)$
and each $Y\in \Y$ satisfies $\sum_{i=1}^t \ed(\Yr_i,Y_i) \le 2\ed(\Yr,Y)$.
In particular, we conclude that $\sum_{i=1}^t \ed(\mathcal{U}_i) = \Oh(d)$
and $\sum_{i=1}^t \ed(\Y_i) = \Oh(d)$.

Each application of \cref{lem:whole} costs $\Oh((|I|+d)^2 + (1+\ed(\mathcal{U}_i))(1+\ed(\Y_i))(|I|+d)\log^2 (|I|+d))$
time, which is $\Oh(d(|I|+d)^2 + d^2(|I|+d)\log^2 (|I|+d))=\Oh((d^3 + d|I|^2)\log^2 (|I|+d))$ in total.
On the other hand, each application of \cref{prop:seaweed} costs $\Oh((d+|I|)^2 \log \log (|I|+d))$,
which is $\Oh((d^3 + d|I|^2)\log \log (|I|+d))$ in total;
the subsequent calls to \cref{lem:restrict} are dominated by this running time.

Observe that $|\spn(\bigcup_{i=1}^t\spn(A_i\downshift (\sum_{j=1}^{i} \delta_j)))|=\Oh(|I|+d)$
holds at all times, so the initialization of \cref{cor:prod} costs $\Oh((|I|+d)t\log (|I|+d))=
\Oh((d^2 + d|I|)\log (|I|+d))$ time,
whereas each update costs $\Oh((|I|+d)\log t\log (|I|+d))=\Oh((|I|+d)\log^2(|I|+d))$ time.
The total number of updates is $\Oh(|\X|\cdot \ed(\Y)+\ed(X)\cdot |\Y|)=\Oh(d^2)$,
so this sums up to $\Oh((d^3+d^2|I|)\log^2(|I|+d))$.
Finally, the total time needed to restrict the returned matrices is $\Oh(|\X|\times |\Y|\times (|I|+d))
= \Oh(d^3+d^2|I|)$.
\end{proof}

\begin{lemma}\label{lem:puzzle}
    Let $U_1,\ldots,U_z$ and $V_1,\ldots,V_z$ denote $\Delta$-puzzles with
    values $U$ and $V$, respectively.
    Moreover, consider $k,w\in \Zz$ such that \[
        k+\sum_{i=1}^z \big||U_i|-|V_i|\big|\le w \le \Delta/2.
    \]
    Further, define strings $U'_i$ and intervals $I_i$ so that
    \begin{itemize}
        \item $U'_1 = U_1\fragmentco{0}{|U_1|+w-\Delta}$ and $I_1 = \fragment{-k}{w}$;
        \item $U'_i = U_i\fragmentco{w}{|U_i|+w-\Delta}$ and $I_i = \fragment{0}{\Delta}$ for $i\in \fragment{2}{z-1}$;
        \item $U'_z = U_z\fragmentco{w}{|U_z|}$ and $I_z = \fragment{0}{\Delta}$.
    \end{itemize}
    Then, we have \[
    P_{V,U}|_{\fragment{-k}{|V|-|U|+k}} =
        \Big(
            \boxdot_{i=1}^z
            \Big(
                (P_{V_i,U'_i}|_{I_i})\downshift
                \Big(
                    \sum_{j=1}^{i-1}(|V_j|-\Delta-|U'_j|)
                \Big)
            \Big)
        \Big)|_{\fragment{-k}{|V|-|U|+k}}.
    \]
\end{lemma}
\begin{proof}
    Observe that $U'_i = U\fragmentco{u_i}{u'_i}$, where $u_i = \sum_{j=1}^{i-1} |U'_i|$
    and $u'_i = \sum_{j=1}^{i} |U'_i|$.
    Moreover, $V_i = V\fragmentco{v_i}{v'_i}$, where $v_i = \sum_{j=1}^{i-1} (|V_i|-\Delta)$
    and $v'_i = \Delta+\sum_{j=1}^{i} (|V_i|-\Delta)$.
    By \cref{lem:prod}, it suffices to prove that $v_i \le \max(0, u_i-k)$,
    $v'_i \ge \min\{|V|,u'_i+|V|-|U|+k\}$, and $I_i \supseteq \fragmentco{-k}{|V|-|U|+k+1}-v_i+u_i$
    hold for $i\in \fragment{1}{z}$.

    As for the first inequality, we have $v_i = 0$ if $i=1$.
    Otherwise, $u_i = w + \sum_{j=1}^{i-1} (|U_i|-\Delta)$
    and $v_i = \sum_{j=1}^{i-1} (|V_i|-\Delta)$,
    so $v_i-u_i = \sum_{j=1}^{i-1} (|V_i|-|U_i|)-w \le w-k-w =-k$ holds as claimed.
    As for the second inequality, we have $v'_i = |V|$ if $i=z$.
    Otherwise, $u'_i = w+\sum_{j=1}^{i} (|U_i|-\Delta)$
    and $v'_i = \Delta + \sum_{j=1}^{i} (|V_i|-\Delta)$,
    so $v'_i-u'_i = \sum_{j=1}^{i} (|V_i|-|U_i|)+\Delta-w
    = |V|-|U| - (\sum_{j={i+1}}^z|V_i|-|U_i|)+\Delta - w \ge |V|-|U|-w+k+\Delta-w
    \ge |V|-|U|+k$ holds as claimed.
    Next, we need to prove that $I+u_i-v_i\sub I_i$.
    This is true for $i=1$, when $u_i = v_i = 0$ and $I = \fragment{-k}{|V|-|U|+k+1}\sub \fragment{-k}{w}=I$.
    Otherwise, $u_i - v_i = \sum_{j=1}^{i-1} (|U_i|-|V_i|)+w
    \ge k-w+w \ge k$
    and  \[
        u_i - v_i = \sum_{j=1}^{i-1} (|U_i|-|V_i|)+w
            = |U|-|V|+\sum_{j=i}^{z} (|V_i|-|U_i|)+w
            \le |U|-|V|+w-k+w = |U|-|V|+2w-k,
        \]
    that is, $u_i-v_i\sub \fragment{k}{|U|-|V|+2w-k}$.
    Hence,
    so $\fragment{-k}{|V|-|U|+k}+v_i - u_i \sub \fragment{0}{2w} \sub \fragment{0}{\Delta}$
    holds as claimed.
    \end{proof}

\begin{lemma}\label{lem:dpm}
    Consider an instance of the \DPM problem.
    We can maintain a permutation matrix $A$
    such that $A=P_{V,U}|_{\fragment{-k}{|V|-|U|+k+1}}$
    holds
    whenever $U_1,\ldots,U_z$ and $V_1,\ldots,V_z$ are $\Delta$-puzzles with values $U$
    and $V$, respectively.
    \begin{itemize}
        \item The preprocessing of each family $\S_i$ costs
            $\Oh((d_i^3+d_i\Delta^2)\log^2(d_i+\Delta))$ time,
            where $d_i = \min\{1,\ed(\S_i)\}$,
        \item the initialization of $\I$ costs $\Oh(z\Delta \log \Delta)$ time, and
        \item the updates of $\I$ cost $\Oh(\Delta \log z \log \Delta)$ time.
    \end{itemize}
\end{lemma}
\begin{proof}
    Set $w \coloneqq \floor{{\Delta}/{2}}$.
    The preprocessing consists of the following steps: \begin{enumerate}
        \item We build $\mathcal{U}_\beta = \{S\fragmentco{0}{|S|+w-\Delta}: S\in \Sb\}$
            and process $(\mathcal{U}_\beta,\Sb)$ using \cref{lem:prepr} for $I=\fragment{-k}{w}$.
        \item We build $\mathcal{U}_\mu = \{S\fragmentco{w}{|S|+w-\Delta} : S\in \Sm\}$
            and process $(\mathcal{U}_\mu,\Sm)$ using \cref{lem:prepr} for $I=\fragment{0}{\Delta}$.
        \item We build $\mathcal{U}_{\varphi} = \{S\fragmentco{w}{|S|} : S\in \Sf\}$
            and process $(\mathcal{U}_{\varphi},\Sf)$ using \cref{lem:prepr} for $I=\fragment{0}{\Delta}$.
    \end{enumerate}
    At initialization time, we initialize the data structure of \cref{lem:prefsum}
    with $A_i = P_{V_i,U'_i}|_{I_i}$ for $i\in \fragment{1}{z}$,  $\delta_1=0$, and $\delta_i=|U'_{i-1}|-|V_{i-1}|-\Delta$ for $i\in \fragment{2}{z}$, where $U'_i$ and $I_i$ are defined in \cref{lem:puzzle}.
    At update time, we update the sequence $A_i$ and $\delta_i$ accordingly.
    At query time, we retrieve the matrix $A\coloneqq \boxdot_{i=1}^t (A_i \downshift \sum_{j=1}^i \delta_i)$
    from the dynamic algorithm of \cref{lem:prefsum}, and return $A|_{\fragment{-k}{|V|-|U|+k}}$,
    retrieved using \cref{lem:restrict}.

    As for correctness, we note that all the possible matrices $A_i = P_{V_i,U'_i}|_{I_i}$ have been constructed at  during the preprocessing phase. Moreover, \cref{lem:puzzle} guarantees that $A|_{\fragment{-k}{|V|-|U|+k}}=P_{V,U}|_{\fragment{-k}{|V|-|U|+k}}$ provided that $U_1,\ldots,U_z$ and $V_1,\ldots,V_z$ are $\Delta$-puzzles with values $U$ and $V$, respectively.

    It remains to analyze the running time.
    Due to $\ed(\mathcal{U}_i)\le \ed(\S_i)$, preprocessing each family $\S_i$ using \cref{lem:prepr}
    costs $\Oh((d_i^3+d_i\Delta^2)\log^2(d_i+\Delta))$ time.

Observe that $|\spn(\bigcup_{i=1}^t\spn(A_i\downshift (\sum_{j=1}^{i} \delta_j)))|=\Oh(\Delta)$
holds at all times, so the initialization of \cref{cor:prod} costs $\Oh(\Delta z\log \Delta)$ time,
whereas each update costs $\Oh(\Delta\log z \log \Delta)$ time.
Each query costs $\Oh(\Delta)$ time.
\end{proof}

\begin{lemma}\label{fct:dd}
Consider strings $U,V$ and an integer $k\in \Zz$ such that $I\coloneqq \fragment{-k}{|V|-|U|+k}$ is non-empty.
Given $P_{V,U}|_I$, the set $\OccD_k(U,V)$ can be constructed in $\Oh(|I|{\log |I|}/{\log \log |I|})$ time.
\end{lemma}
\begin{proof}
    By \cref{obs:dd}, we have $\DD(V\fragmentco{i}{j},U)=D_{V,U}\position{i,j-|U|}$.
    Moreover, \cref{clm:rdm} yields
    \begin{align*}
        &D_{M(V,U)|_I}\position{i,j-|U|} \\
        &\quad= \min\{D_{V,U}\position{i,j-|U|},\quad
        i+j-|U|+2(k+1),\quad 2(|V|-|U|+k+1)-i-j+|U|\}\\
        &\quad=\min\{D_{V,U}\position{i,j-|U|},\quad 2k+2+i+j-|U|,\quad
                                        2k+2 +2|V|-|U|-i-j\}.
    \end{align*}
    If $j \ge |U|-k$, then \[
        2k+2+i+j-|U| \ge k+2+i > k.
    \]
    If $i\le |V|-|U|+k$, then \[
        2(|V|-|U|+k+1)-i-j+|U| > |V|-j+k+2 > k.
    \]
    Hence, for $(i,j)\in \fragment{0}{|V|-|U|+k}\times \fragment{|U|-k}{|V|}$,
    we have  $\DD(U,V\fragmentco{i}{j})\le k$ if and only if
    $D_{M(V,U)|_I}\position{i,j-|U|}\le k$.
    On the other hand  $\DD(U,V\fragmentco{i}{j})\ge |j-i-|U\|>k$ holds
    whenever $j < |U|-k$ or $i>|V|-|U|+k$.
    Hence, our task reduces to checking, for each $i\in \fragment{0}{|V|-|U|+k}$,
    whether  $D_{M(V,U)|_I}\position{i,j-|U|}\le k$ holds for some $j\in \fragment{|U|-k}{|V|}$.
    For this, we recall that $D_{M(V,U)|_I}$ is a Monge matrix and note that it remains
    a Monge matrix when restricted to $\fragment{0}{|V|-|U|+k}\times \fragment{-k}{|V|-|U|}$.
    The SMAWK algorithm~\cite{SMAWK} finds row-minima in an $m\times m$ Monge matrix using $\Oh(m)$
    queries asking for values of the matrix entries.
    By \cref{fct:ra,lem:perm}, after $\Oh(|I|\sqrt{\log |I|})$-time preprocessing,
    we have $\Oh({\log |I|}/{\log \log |I|})$-time access to entries of $D_{M(V,U)|_I}$.
    Since $m=|V|-|U|+k+1 \le |I|$, the final running time is $\Oh(|I|{\log |I|}/{\log \log |I|})$.
\end{proof}

\begin{proposition}\label{prp:dpm}
    There is a data structure for a $\DPM(k, \Delta, \Sb, \Sm, \Sf)$ problem variant,
    reporting $\OccD_k(U,V)$ instead of $\OccE_k(U,V)$,
    with $\Oh(\Delta \log z \log \Delta)$-time updates and queries,
    $\Oh(\Delta z \log \Delta)$-time initialization, and $\Oh((d^3+\Delta^2 d)\log^2 (d+\Delta))$-time preprocessing, where
    $d=\ed(\Sb)+\ed(\Sm)+\ed(\Sf)$.
\end{proposition}
\begin{proof}
    Let $w = \floor{{\Delta}/{2}}$.   We use \cref{lem:dpm} to maintain
    a permutation matrix $A$ such that $A=P_{V,U}|_{\fragment{-k}{|V|-|U|+k}}$ holds
    whenever $U_1,\ldots,U_z$ and $V_1,\ldots,V_z$ are $\Delta$-puzzles with values $U,V$, respectively.

    It remains to implement queries. If $|V| < |U|-k$, then we report that $\OccD(U,V)=\emptyset$.
    Otherwise, we derive $A=P_{V,U}|_{\fragment{-k}{|V|-|U|+k}}$.
    Next, we compute $\OccD_k(U,V)$ using \cref{fct:dd}.
    These two steps add $\Oh(w\log w / \log \log w)$ to the query time,
    which is dominated by the update time $\Oh(w\log z \log w)$.
\end{proof}

\dpm*
\begin{proof}
We apply \cref{prp:dpm} using the distortion-free embedding from $\ed$ to $\DD$~\cite{abs-0707-3619}.
This embedding is defined by mapping every string $S\in \Sigma^*$
to a string $f(S)\coloneqq \bigodot_{i=0}^{|S|-1} S\position{i}\$$, where $\$\notin \Sigma$.
Observe that $\ed(S,T)=\onehalf\DD(f(S),f(T))$
and $\OccE_k(P,T) = \{i\in \fragment{0}{|T|-|P|+k} : 2i\in \OccD_{2k}(f(P),f(T))\}$.
Thus, we use \cref{prp:dpm} we the strings mapped through $f$ and the integer parameters
$k$ and $\Delta$ doubled. Note that \modelname operations on the family $f(\X)$ can be easily implemented
using the \modelname operations on~$\X$.
\end{proof}

\clearpage
\phantomsection
\addcontentsline{toc}{section}{\sffamily \bfseries  References}
\bibliographystyle{alphaurl}
\bibliography{ms}
\clearpage
\appendix
\section{Notation Overview for Part I}

    \begin{tabularx}{1.05\textwidth}{c X}
        {\bf Notation} & {\bf Explanation}\\\toprule
        \(i, j\) & Integer indices.\\
        \(s, \posa, \posb, x, y\) & Positions in strings, typically occurring in pairs as
        boundaries of fragments.\\
        \(F, G, S, U, V, X, Y\) & Strings or fragments of strings; possibly
        with additional meaning in their respective contexts.\\
        \((a_{X}, a_{Y})\) & Elements of an alignment of (fragments of) \(X \onto Y\). We
        may deviate from this notation in the special case of alignments of (fragments of)
        a string to (fragments of) itself.\\
        \midrule
        \(T\), \(|T| = n\) & Text \(T\) of length \(n\).\\
        \(P\), \(|P| = m\) & Pattern \(P\) of length \(m\).\\
        \(k\) & A non-negative threshold \(k\) used to denote the allowed number of edits
        when searching for occurrences of  \(P\) in \(T\).\\
        \(\OccE_k(P, T)\) & The set of all starting positions of \(k\)-edit occurrences of
        \(P\) in \(T\).\\[2ex]
        \textit{The Standard Trick}& We may assume \(n < \threehalfs m + k\) with a
        running-time overhead of \(\Oh(n/m)\) (see \cref{sec:oldalgo}).\\
        \midrule
        \(Q\), \(|Q| = q\) & Primitive string \(Q\) of length \(q\).\\
        {\small \(\A_T : T \onto \Q\fragmentco{x_T}{y_T}\), \(d_T\)} & An alignment of \(T\) onto a
        fragment of \(\Q\) with cost \(d_T\). We have \(x_T \in \fragmentco{0}{q}\).\\
            {\small \(\A_P : P \onto \Q\fragmentco{0}{y_P}\), \(d_P\)} & An alignment of \(P\) onto a
        fragment of \(\Q\) with cost \(d_P\).\\[2ex]
        \(\kappa \coloneqq d_P + k + d_T\) & For an alignment \(\A: P \onto
        T\fragmentco{\posa}{\posb}\) of cost at most \(k\), the bound \(\kappa\) is an
        upper bound on the number of edits of the induced alignment \(\A_T\circ
        \A\circ\A_P^{-1} : \Q\fragmentco{0}{y_P} \onto \Q\fragmentco{\posa'}{\posb'}\).\\
        \(\tau \coloneqq q\ceil{\kappa / 2q}\) & Auxiliary parameter.
        \\[2ex]
        \(d\) & Positive integer that satisfies \(d_T/3  \le d\) and \(d_P \le d\) and \(2k
        \le d\) and \(8q/m \le d\).\\
        \midrule
        \(J\) & See \cref{def:jdef}; we have
        \(J =
        \fragment{\ceil{{(x_T-\kappa)}/{\tau}}}{\floor{{(y_T-y_P+\kappa)}/{\tau}}}\).\\
        \(R_j \coloneqq T\fragmentco{r_j}{r'_j}\)& See \cref{def:jdef}; we have
        \(r_j \coloneqq \min\{a_T \mid (a_T,a_Q)\in \A_T
        \text{ for }  j\tau-\kappa \le a_Q\}\)
        and
        \(r'_j \coloneqq \max\{a_T \mid (a_T,a_Q)\in \A_T\text{ for } a_Q \le
        j\tau+y_P + \kappa\}.\) Furthermore, \(r_{\max J + 1} \coloneqq n\).\\
        \midrule
        \(S=\bigodot_{i=1}^{\beta_S} S\fragmentco{s_{i-1}}{s_i}\) & Tile partition of a
        string \(S\)
        with respect to some alignment \(\A_S: S \onto \Q\fragmentco{x_S}{y_S}\) with
        \(S\fragmentco{s_{i-1}}{s_i} = \A_S^{-1}(\Q\fragmentco{\max\{x_S,
        (j-1)\tau\}}{\min\{y_S,j\tau\}})\).\\
        \(\beta_S \coloneqq \ceil{y_S/\tau}\)& The number of the last non-zero tile of \(S\).\\
        \(P=\bigodot_{i=1}^{\beta_P} P\fragmentco{p_{i-1}}{p_i}\) & \(\tau\)-tile partition of
        \(P\).
        We assume that \(\beta_P \ge 20\) (by \cref{lem:smallbeta}).\\
        \(T=\bigodot_{i=1}^{\beta_T} T\fragmentco{t_{i-1}}{t_i}\) & \(\tau\)-tile partition of
        \(T\).\\[2ex]
        $\Delta \coloneqq  6\kappa$& Overlap of neighboring puzzle pieces.\\
        $z \coloneqq \beta_P - 17$& Number of puzzle pieces.\\
        \(P_1, \dots, P_z\) & A \(\Delta\)-puzzle with value \(P\), where
        $P_1 \coloneqq P\fragmentco{p_0}{p_2+\Delta}$ and\\
        \(\val_{\Delta}(P_1, \dots, P_z) = P\)
                            &$P_{i} \coloneqq
        P\fragmentco{p_{i}}{p_{i+1}+\Delta}$ for $i\in \fragmentoo{1}{z}$ and $P_z
        \coloneqq
        P\fragmentco{p_{z}}{|P|}$.\\
        \(T_{j,1},\ldots,T_{j,z}\) & A \(\Delta\)-puzzle with value \(R_j\), where
        \(T_{j,1} \coloneqq T\fragmentco{r_j}{t_{j+2}+\Delta}\) and\\
        \(\val_\Delta(T_{j,1},\ldots,T_{j,z}) = R_j\)   &
        \(T_{j,i} \coloneqq T\fragmentco{t_{j+i}}{t_{j+i+1}+\Delta}\) for \(i \in
        \fragmentoo{1}{z}\) and \(T_{j,z} \coloneqq
        T\fragmentco{t_{j+z}}{r'_j}\).\\
        \(T_{\min J + 2}, \dots, T_i, \dots, T_{\max J + z -1}\)
        & Shorthand for the pieces \(T_{j,i'}\) with \(j + i' = i\) and \(i
        \in\fragmentoo{1}{z}\); that is \(T_i \coloneqq T\fragmentco{t_i}{t_{i+1} +
        \Delta}\).
        \\\midrule\pagebreak\\
        \midrule
        \(\Sb \coloneqq \{P_1\} \cup \{T_{j,1} \mid j\in J\}\)
        & The set of {leading} puzzle pieces.\\
        \(\Sm \coloneqq \{P_{i} : i\in \fragmentoo{1}{z}\}\;\cup\) &
        The set of  {internal} puzzle pieces.\\
        \(\{T_{i} : i \in \fragmentoo{\min J +1}{\max J +z} \}\) \\
        \(\Sf \coloneqq \{P_z\} \cup \{T_{j,z} : j\in J\}\)&
        The set of {trailing} puzzle pieces.\\[2ex]
        \(\ed(\S)\coloneqq \min_{\Sr\in \Sigma^*} \sum_{S\in \S}\ed(S,\Sr)\)
        & The median edit distance of the family \(\S\).\\[2ex]
        \(\sp{T}\), \(\sp{P}\) & The special internal pieces in \(T\) and \(P\), that is,
        the pieces in \(\Sm\) that differ from \(\Q\fragmentco{0}{\tau +
        \Delta}\).\\
        \midrule
        \(\I = (U_1, V_1)(U_2,V_2)\dots(U_z,V_z)\)
        & Sequence of ordered pairs of strings, a
        \emph{DPM-sequence}. Occasionally, we call the elements \((U_i, V_i)\) a
        \emph{DPM-pair}.\\
        \(\OccE_k(\I) \coloneqq \OccE_k( U, V ) \)
        & Set of \(k\)-error occurrences of the string
        \(U \coloneqq \val_{\Delta}( U_1,\dots,U_z ) \)
        in the string \( V \coloneqq \val_{\Delta}( V_1,\dots,V_z ) \).\\
        \(\tor(\I) \coloneqq \sum_{i=1}^z \big||U_i|-|V_i|\big|\) & The \torn of the
        DPM-sequence \(\I\).\\
        \(\I_j \coloneqq (P_1, T_{j,1}) (P_2, T_{j,2}) \cdots (P_z, T_{j,z})\)
        & DPM-sequence
        representing the strings \(P = \val_{\Delta}( P_1, \dots, P_z )\) and
        \(R_j = \val_{\Delta}( T_{j,1}, \dots, T_{j,z})\);
        we have \(\tor(\I_j) \le 3\kappa - k = \Delta/2 - k\)
         by \cref{fact:lengths_diff}.\\
        \(\mathcal{Q} \coloneqq
        (\Q\fragmentco0{\tau+\Delta},\)
        & The plain (internal) DPM-pair.\\
        \(\Q\fragmentco0{\tau+\Delta})\) & \\
        \midrule
        \(\mathcal{L}^P, \lpref, L^P_1\)
        & The set of locked fragments of $P$ that contains a $\lpref$-locked prefix $L^P_1$, computed using $\locked(P,Q,d_P,\lpref)$.\\
        \(\mathcal{L}^T\)
        & The set of locked fragments of $T$, computed using $\locked(T,Q,d_T,0)$.\\
        \(\ktotm\)
        & A slack allowance for computing overlaps of fragments for the purpose of marking.\\
        \(\markf(v,\ktotm, L^P, L^T )\)
        & The number of marks given to a position $v$ due to the (possible) overlap of $L^P \in \mathcal{L}^P$ with $L^T \in \mathcal{L}^T$.\\
        \(\markf(v) = \markf(v,\ktotm, \mathcal{L}^P, \mathcal{L}^T )\)
        & The total number of marks given to a position $v$ of $T$.\\
        \(\markf(v,L)\)
        & The number of marks given to a position $v$ of $T$ due to overlaps of pairs of locked fragments containing locked fragment $L$.\\
        \(\tradeoff , \Hv, \light\)
        & A positive integer threshold used in the partition of the positions $\fragment{0}{n-m+k}$
        of $T$ to a set $\Hv$ of heavy positions and a set $\light$ of light positions; see \cref{def:light}.\\
        \midrule
        \(\rho(\star)\)
        & A mapping from positions of $T$ to integers in $\fragmentco{x_T}{y_T}$ such that $\A_T(T\fragmentco{v}{n})=\Q\fragmentco{\rho(v)}{y_T}$.\\
        \(\mathcal{D}(v)\)
        & For a light position $v$ of $T$, the set\\
        &$\{L^P_1\} \cup \{L \in \mathcal{L}^P\cup \mathcal{L}^T_{\fragmentco{v}{v+m}} : \markf(v,L) < \edl{L}{Q}\}$.\\
        \bottomrule
    \end{tabularx}

\clearpage
\section{Notation Overview for Part II}

    \begin{tabularx}{1.05\textwidth}{c X}
        {\bf Notation} & {\bf Explanation}\\\toprule
            $\ppr{I},\ppl{I},\mmr{I},\mml{I}$ & An interval $I$ extended or shrunk by one element to the left or right.\\
            $\spn(S)$ & The smallest integer interval containing a set $S\sub \Z$.\\\midrule
            $A\odot B$ & The min-plus product of matrices $A,B$.\\
            $A^\square$ & The density matrix of a matrix $A$.\\
            $\spn(A)$ & The span of a permutation a matrix $A$, defined as $\mml{\spn(\{i\in I : A[i,i]=0\})}$.\\
            $A^\Sigma$ & The distribution matrix of a matrix $A$.\\
            $A \boxdot B$ & The seaweed product of permutation matrices $A,B$; defined as $(A^\Sigma \odot B^\Sigma)^\square$.\\
            $A\downshift s$ & The diagonal shift of a matrix $A$ by $s$ units; $(A\downshift s)[i+s,j+s]=A[i,j]$.\\\midrule
            $M$ & A finite subset of $\Z^2$.\\
            $\AG(M)$ & The alignment graph of $M$.\\
            $\dist_M$ & The distance function on $\Z^2$ induced by distances in $\AG(M)$.\\
            $\prec$ & A partial order on $\Z^2$ defined so that $(x,y) \prec (x',y')$ if and only if $x < x'$ and $y<y'$.\\
            $\LIS(S)$ & The maximum length of a $\prec$-chain within $S\sub \Z^2$.\\\midrule
            $\lt^B$, $\br^B$ & The left-top and the bottom-right boundary of a bounding box $B$.\\
            $D_{M,B}$, $D_M$ & The distance matrix of $M$ with respect to a bounding box $B$ (shown to be independent of $B$).\\
            $P_M$ & The seaweed matrix of $M$, defined as $\frac12D_M^\square$.\\\midrule
            $\spn(M)$ & The span of $M$, defined as $\spn(\{x-y : (x,y)\in M\})$.\\
            $M_I$ & The restriction of $M$ to interval $I$, defined as $\{(x,y)\in M : x-y\in I\}$.\\
            $A_I$ & The restriction of a permutation matrix $A$ to interval $I$, defined as $\RP{M}{I}$ if $A=P_M$ (shown to be independent of $M$).\\\midrule
            $M(X,Y)$ & The set of mismatches between two strings $X,Y$, defined as $\{(x,y) : X[x]\ne Y[y]\}$.\\
            $D_{X,Y}$, $P_{X,Y}$ & The distance matrix and the permutation matrix of $M(X,Y)$.\\
            $\DD(X,Y)$ & The deletion distance between strings $X,Y$.\\
            $\DD(\X)$ & The median deletion distance of a finite string family $\X$.\\
        \bottomrule
    \end{tabularx}

\end{document}